\documentclass[twoside, 11pt, openright]{book}

\usepackage[Sonny]{fncychap}

\usepackage{adjustbox}
\usepackage{geometry}
\usepackage{xcolor}
\usepackage{graphicx}
\usepackage{tikz}
\usepackage{cancel}
\usepackage{appendix}

\usepackage{dsfont}
\usepackage{amsmath}
\usepackage{mathtools}
\usepackage{amsfonts}
\usepackage{amssymb}
\usepackage{amsthm}
\usepackage{bbold}
\usepackage{ifthen}
\usepackage{float}
\usepackage{multirow}
\usepackage{mathrsfs}
\usepackage{empheq}
\usepackage{booktabs}

\usepackage{accents}
\usepackage{tensor}
\usepackage[vcentermath,enableskew]{youngtab}
\usepackage{ytableau}
\usepackage{cancel}

\usepackage{textcase}
\let\MakeUppercase\MakeTextUppercase

\newcommand{\notprop}{\propto\kern-1\@ptsize pt \diagup}
\usepackage[square,numbers,sort&compress]{natbib}
\usepackage[backref=page]{hyperref}
\hypersetup{colorlinks=true,urlcolor=[rgb]{0,0.2,0.5},
	linkcolor=[rgb]{0.2,0.2,0.8},citecolor=[rgb]{0.8,0.1,0.1}}
\usepackage{cleveref}

\numberwithin{equation}{chapter}

\counterwithin*{equation}{section}

\newcounter{mysubequations}

\newtheorem{proposition}{Proposition}[section]
\newtheorem{lemma}[proposition]{Lemma}
\newtheorem{corollary}[proposition]{Corollary}
\newtheorem{theorem}[proposition]{Theorem}

\newtheorem{definition}[proposition]{Definition}

\theoremstyle{definition}

\usepackage[framemethod=tikz]{mdframed}
\usepackage{lipsum}

\usepackage{tablefootnote} 
\makeatletter 
\AfterEndEnvironment{mdframed}{%
 \tfn@tablefootnoteprintout%
 \gdef\tfn@fnt{0}%
}

\newenvironment{example}[1][]
{\refstepcounter{proposition}\par\medskip\noindent%
	\textbf{Example~\theproposition. #1} \rmfamily\par\nopagebreak%
	\begin{mdframed}[
		linewidth=1pt,
		linecolor=black,
		bottomline=false,topline=false,rightline=false,
		innerrightmargin=0pt,innertopmargin=0pt,innerbottommargin=0pt,
		innerleftmargin=1em,
		skipabove=.5\baselineskip
		]\itshape}
	{\end{mdframed}}

\newenvironment{remark}[1][]
  {\refstepcounter{proposition}\par\smallskip\noindent%
   \textbf{\underline{Remark}~\theproposition. #1} \rmfamily\par\nopagebreak%
  \begin{mdframed}[
     linewidth=1pt,
     linecolor=black,
     bottomline=false,topline=false,rightline=false,
     innerrightmargin=0pt,innertopmargin=0pt,innerbottommargin=0pt,
     innerleftmargin=1em,
     skipabove=.5\baselineskip
   ]}
  {\end{mdframed}}
  
  	\definecolor{darkcyan}{HTML}{0091A4}  
	\definecolor{brightcyan}{HTML}{dcf0f2}
	
	\definecolor{darkgray}{rgb}{0.3,0.3,0.3}
	\definecolor{brightgray}{rgb}{0.95,0.95,0.95}
	
	\newlength{\defparindent}
	\newmdenv[
	topline=false,  
	rightline=false,  
	bottomline=false,  
	leftline=true,  
	linecolor=darkgray,  
	linewidth=3pt,  
	backgroundcolor=brightgray,
	settings={\setlength{\parindent}{0cm}}
	]{mathematicasindent}
	
	\newmdenv[  
	topline=false,
	rightline=false,  
	bottomline=false,  
	leftline=true,  
	linecolor=darkgray,  
	linewidth=3pt,  
	backgroundcolor=brightgray,  
	]{mathematica0}   
	
	\newmdenv[  
	topline=false,
	rightline=false,  
	bottomline=false,  
	leftline=false,  
	linecolor=darkgray,  
	linewidth=3pt,  
	backgroundcolor=brightgray 
	]{cbox} 
	
	\newenvironment{mathematicas}
	[1][]{\begin{mathematicasindent}[frametitle={Mathematica commands: ~#1}]\setlength{\parindent}{\defparindent}\ignorespaces
		}
	{\end{mathematicasindent}}
	\newenvironment{mathematica}[1][]
	{\begin{mathematica0}[frametitle={Mathematica command: ~#1}]
		}
	{\end{mathematica0}}

\definecolor{orange}{rgb}{0.99,0.34,0.07}
\definecolor{jaune}{rgb}{0.85,0.5,0.07}
\definecolor{orangebracelet}{rgb}{0.996,0.678,0.255}

\usepackage{esvect} 

\newcommand{\be}{\begin{eqnarray}}
	\newcommand{\ee}{\end{eqnarray}}

\newcommand{\diff}{\mathrm{d}}

\newcommand{\bigzero}{\mbox{\normalfont\Large\bfseries 0}}

\usepackage{bm}
\newcommand{\sn}{\mathfrak{S}_{n}}
\newcommand{\Sn}[1]{\mathfrak{S}_{#1}}
\newcommand{\cn}{\mathcal{C}_n}
\newcommand{\Cn}[1]{\mathcal{C}_{#1}}
\newcommand{\bn}{B_n}
\newcommand{\Bn}[1]{B_{#1}}
\newcommand{\dbn}{\mathcal{B}_n}

\newcommand{\Stab}{\textup{Stab}}
\newcommand{\End}{\textup{End}}
\newcommand{\Hom}{\textup{Hom}}
\newcommand{\GL}{\textup{GL}}
\newcommand{\Or}{\textup{O}}
\newcommand{\C}{\mathbb{C}}
\newcommand{\R}{\mathbb{R}}
\newcommand{\Q}{\mathbb{Q}}
\newcommand{\stab}{\textup{st}}
\newcommand{\spec}{\textup{spec}}
\newcommand{\st}{\,|\,}
\newcommand{\lp}{\textup{(}}
\newcommand{\rp}{\textup{)}}

\newcommand{\Par}{\mathcal{P}}

\newcommand{\Irr}{\textup{Irr}}
\newcommand{\tab}{\mathrm{t}}

\newcommand{\Tab}{\textup{Tab}}

\newcommand{\blist}[1]{\{#1\}}
\newcommand{\down}[1]{\underline{#1}}

\newcommand{\GCT}{\textup{GCT}}
\newcommand{\CT}{\textup{CT}}

\DeclareFontFamily{U} {MnSymbolC}{}

\DeclareFontShape{U}{MnSymbolC}{m}{n}{
	<-6> MnSymbolC5
	<6-7> MnSymbolC6
	<7-8> MnSymbolC7
	<8-9> MnSymbolC8
	<9-10> MnSymbolC9
	<10-12> MnSymbolC10
	<12-> MnSymbolC12}{}

\DeclareSymbolFont{MnSyC} {U} {MnSymbolC}{m}{n}
\DeclareMathSymbol{\slashdiv}{\mathbin}{MnSyC}{29}

\newcommand*{\Scale}[2][4]{\scalebox{#1}{\ensuremath{#2}}}
\newcommand\scalemath[2]{\scalebox{#1}{\mbox{\ensuremath{\displaystyle #2}}}}

\newcommand{\bb}[1]{\mathbf{#1}}
\newcommand{\db}[1]{\dot{\mathbf{#1}}}
\newcommand{\dd}[1]{\ddot{\mathbf{#1}}}

\newcommand{\mione}{\scalebox{0.6}{-1}}
\newcommand{\mitwo}{\scalebox{0.6}{-2}}
\newcommand{\mithree}{\scalebox{0.6}{-3}}
\newcommand{\one}{\scalebox{0.6}{1}}
\newcommand{\two}{\scalebox{0.6}{2}}
\newcommand{\three}{\scalebox{0.6}{3}}
\newcommand{\zero}{\scalebox{0.6}{0}}

\newcommand{\LRp}{\, {\scriptstyle \otimes}\,}
\newcommand{\svdots}{%
	\vbox{
		\scriptsize \baselineskip 2.5pt \lineskiplimit 0pt
		\hbox {.}\hbox {.}\hbox {.}\kern-0.75pt
	}%
}

\makeatletter
\MHInternalSyntaxOn
\def\MT_leftarrow_fill:{%
	\arrowfill@\leftarrow\relbar\relbar}
\def\MT_rightarrow_fill:{%
	\arrowfill@\relbar\relbar\rightarrow}
\newcommand{\xrightleftarrows}[2][]{\mathrel{%
		\raise.55ex\hbox{%
			$\ext@arrow 0359\MT_rightarrow_fill:{\phantom{#1}}{#2}$}%
		\setbox0=\hbox{%
			$\ext@arrow 3095\MT_leftarrow_fill:{#1}{\phantom{#2}}$}%
		\kern-\wd0 \lower.55ex\box0}}
\MHInternalSyntaxOff
\makeatother

\newcommand{\stimes}{\,{\scriptstyle \times}\,}

\newcommand{\id}{\mathds{1}}

\newcommand{\forceindent}{\leavevmode{\parindent=1em\indent}} 
\DeclareMathSymbol{\shortminus}{\mathbin}{AMSa}{"39}
\newcommand{\overbar}[1]{\mkern 1.5mu\overline{\mkern-1.5mu#1\mkern-1.5mu}\mkern 1.5mu}
\newcommand{\Dim}{\textup{d}}

\include{pdg}

\author{Thomas Helpin}
\date{19 Décembre 2023}
\title{Décompositions irréductibles des tenseurs via l'algèbre de Brauer et applications à la gravitation métrique-affine.}
\specialite{Physique}
\directeur{VOLKOV Mikhail S.}

\rapporteura{\textbf{BOULANGER Nicolas}}{Professeur, Université de Mons}
\rapporteurb{\textbf{GARC\'IA-PARRADO Alfonso}}{Professeur, Université de Cordoue}

\jurya{\textbf{BEKAERT Xavier}}{Maitre de conférences HDR, Université de Tours}
\juryb{\textbf{BOULANGER Nicolas}}{Professeur, Université de Mons}
\juryc{\textbf{GARC\'IA-PARRADO Alfonso}}{Professeur, Université de Cordoue}
\juryd{\textbf{GICQUAUD Romain}}{Maitre de conférences HDR, Université de Tours}
\jurye{\textbf{KOIVISTO Tomi S.}}{Professeur,  Université de Tartu}
\juryf{\textbf{NOUI Karim (Président du jury)}}{Professeur,  Université Paris-Saclay}
\juryg{\textbf{VOLKOV Mikhail S.}}{Professeur, Université de Tours}

\begin{document}
	\let\MakeUppercase\relax	
	
	\pagedegarde

	\frontmatter
	\newgeometry{hmargin=2.5cm,vmargin=3.5cm}
	\chapter*{Remerciements}

	
	Je tiens tout d'abord à exprimer ma sincère gratitude envers les rapporteurs de ma thèse, Nicolas Boulanger et Alfonso García-Parrado. Leur lecture détaillée et leurs commentaires ont grandement contribué à l'amélioration de ce manuscrit. Je remercie également les membres du jury, Xavier Bekaert, Romain Gicquaud, Karim Noui et Tomi S. Koivisto, pour avoir accepté de participer à l'examen de cette thèse ainsi qu'à ma soutenance. Je tiens à exprimer ma reconnaissance envers mon directeur de thèse, Mikhail S. Volkov, qui a rendu cette thèse possible et qui m'a fait confiance pour suivre mes propres projets de recherche. Aussi, je souhaite exprimer ma sincère gratitude envers mon collaborateur, Yegor Goncharov, pour son aide durant la rédaction de cette thèse et pour m'avoir introduit à l'algèbre de Brauer. Nos recherches mutuelles et nos nombreuses discussions ont grandement contribué à l'élaboration de ce manuscrit.\medskip
	
	Mes sincères remerciements vont également à tous mes collègues et amis. Merci aux membres du bureau 1360, Romain, Guillaume, Yegor et Andrii, pour l'ambiance incroyable qu'ils y ont apportée durant ces dernières années. Une mention spéciale pour tous les membres du \textit{club amande}, fournisseur quotidien de viennoiseries, organisé par notre cher président Guillaume. J'en profite pour remercier Floriane qui, à maintes reprises, nous a fait le plaisir de sa visite, gagnant ainsi le statut de stagiaire du bureau, et pour son amitié en général.  Merci aussi à l'ensemble des doctorants du laboratoire, Yohan, Léa, Antonin, Théo, Jad, Igor, Maxime, Romane, Marion, Julien et Vagif, pour tous les bons moments passés ensemble, autant au sein du laboratoire qu'en dehors. Je remercie aussi les anciens doctorants et post-doctorants que j'ai eu la chance de cotoyer : Mahmut, Rima, Hakim, Abraham, Nathan, Andreas, Aymane, Eda, Florestan, Jean-David, Charles, Clément et Julien.\medskip 
	
	Aussi, je tiens à remercier l'ensemble des membres de l'équipe pédagogique de l'institut Denis Poisson et du Greman, ainsi que Laetitia pour sa bonne humeur au quotidien.\medskip  
 	
	Un grand merci à tous mes amis de Tours et d'ailleurs. Pour ceux qui n'ont pas été mentionnés précédemment, je suis certain qu'ils se reconnaîtront. Enfin, je suis profondément reconnaissant envers ma famille pour son soutien indéfectible tout au long de cette période.

	

	\newgeometry{hmargin=2.5cm,vmargin=0.5cm}
	\chapter*{Résumé}
	
	\noindent Dans la première partie de cette thèse, on utilise la théorie des représentations des groupes et des algèbres afin d'obtenir une décomposition irréductible des tenseurs dans le contexte de la gravitation métrique-affine. En particulier, on considère l'action du groupe orthogonal $\Or(1, \Dim-1)$ sur le tenseur de Riemann associé à une connexion affine, avec torsion et non-métricité, définit sur une variété pseudo-Riemannienne. Cette connexion est l'ingrédient caractéristique de la gravitation métrique-affine.
	
	La décomposition irréductible du tenseur de Riemann effectuée dans cette thèse est conçue pour l'étude des théories invariantes projectives de la gravitation métrique-affine. Les propriétés suivantes impliquent l'unicité de la décomposition. Premièrement, elle s'effectue en deux étapes: on considère l'action de $\GL(\Dim,\mathbb{R})$ puis l'action de $\Or(1,\Dim-1)$. Deuxièmement, le nombre de tenseurs invariants projectifs dans la décomposition est maximal. Troisièmement, la décomposition est orthogonale par rapport au produit scalaire canonique induit par la métrique. La même procédure est appliquée à la distorsion de la connexion affine. Enfin, à partir de ces décompositions, nous obtenons les Lagrangiens quadratiques généraux en la distortion et en la courbure de Riemann.
	
	\medskip
	\noindent Dans la deuxième partie, nous construisons les opérateurs de projection utilisés pour obtenir les décompositions mentionnées précédemment. Ces opérateurs sont réalisés en termes de l'algèbre du groupe symétrique $\C\sn$ et de l'algèbre de Brauer $\bn(\Dim)$, qui sont respectivement liées à l'action de $\GL(\Dim,\C)$ (et sa forme réelle $\GL(\Dim,\mathbb{R}$)) et à l'action de $\Or(\Dim,\C)$ (et sa forme réelle $\Or(1,\Dim-1)$) sur les tenseurs via la dualité de Schur-Weyl.
	
	Tout d'abord, nous proposons une approche alternative aux formules connues pour les idempotents centraux de $\C\sn$. Ces éléments réalisent une décomposition réductible unique, connue sous le nom de décomposition isotypique. Cette décomposition s'avère remarquablement pratique pour aboutir à la décomposition irréductible par rapport à $\GL(\Dim,\mathbb{R})$.
	
	Ensuite, nous construisons les éléments de $\bn(\Dim)$ qui réalisent la décomposition isotypique d'un tenseur par rapport à l'action de $\Or(\Dim,\C)$. Cette décomposition est irréductible sous $\Or(\Dim,\C)$ lorsqu'elle est appliquée à un tenseur $\GL(\Dim,\C)$ irréductible d'ordre 5 ou moins. En conséquence directe de la construction, nous proposons une solution au problème de décomposition d'un tenseur arbitraire en sa partie sans trace, doublement sans trace, et ainsi de suite.
	
	Enfin, dans le dernier chapitre, nous présentons une technique qui optimise l'utilisation des opérateurs de projection par des systèmes de calcul formel.  Ces résultats ont conduit au développement de plusieurs paquets Mathematica liés au bundle \textit{xAct} pour le calcul tensoriel en théorie des champs. Des fonctions particulières sont présentées tout au long du manuscrit.\medskip
	
	\textbf{Mots-clés} : Gravitation métrique-affine, tenseur de Riemann, invariance projective, décomposition irreductible, dualités de Schur-Weyl, algèbre de Brauer.

	\chapter*{Abstract}
	
	\noindent In the first part of this thesis, we make use of representation theory of groups and algebras to perform an irreducible decomposition of tensors in the context of metric-affine gravity. In particular, we consider the action of the orthogonal group $\Or(1,\Dim-1)$ on the Riemann tensor associated with an affine connection defined on a $\Dim$-dimensional pseudo-Riemannian manifold. This connection, with torsion and non-metricity, is the characteristic ingredient of metric-affine theories of gravity.  \smallskip

	The irreducible decomposition of the Riemann tensor carried out in this thesis is devised for the study of projective invariant theories of metric-affine gravity. The following properties imply the uniqueness of the decomposition. Firstly, it is performed in two steps: we consider the action of $\GL(\Dim,\R)$ and then the action of $\Or(1,\Dim-1)$. Secondly, the number of projective invariant irreducible tensors in the decomposition is maximal. Thirdly, the decomposition is orthogonal with respect to the canonical scalar product induced by the metric. The same procedure is applied to the distortion of the affine connection. Finally, from these decompositions, we derive the general quadratic Lagrangians in the distortion and in the Riemann tensor.

	\medskip
	\noindent In the second part, we construct the projection operators used for the aforementioned decomposition. They are realized in terms of the symmetric group algebra $\C\sn$ and of the Brauer algebra $\bn(\Dim)$ which are related respectively to the action of $\GL(\Dim,\C)$ (and its real form $\GL(\Dim,\mathbb{R})$) and to the action of $\Or(\Dim,\C)$ (and its real form $\Or(1,\Dim-1)$) on tensors via the Schur-Weyl duality.

	\smallskip
	First of all, we give an alternative approach to the known formulas for the central idempotents of $\C\sn$. These elements provide a unique reducible decomposition, known as the isotypic decomposition. For our purposes, this decomposition is remarkably handy to arrive at the sought after irreducible decomposition with respect to $\GL(\Dim,\mathbb{R})$.

	\smallskip		
	Then, we construct the elements in $\bn(\Dim)$ which realize the isotypic decomposition of a tensor under the action of $\Or(\Dim,\C)$. This decomposition is irreducible under $\Or(\Dim,\C)$ when applied to an irreducible $\GL(\Dim,\C)$ tensor of order $5$ or less. As a by product of the construction, we give a solution to the problem of decomposing an arbitrary tensor into its traceless part, doubly traceless part and so on.

	\smallskip
	Finally, in the last chapter  we present a technique which optimizes the use of the projection operators by computer algebra systems. These results led to the development of several Mathematica packages linked to the \textit{xAct} bundle for tensor calculus in field theory. Particular functions are presented along the manuscript. \medskip  
	
	\textbf{Keywords:} Metric-affine gravity, Riemann tensor, projective invariance, irreducible decomposition, Schur-Weyl dualities, Brauer algebra.
	
	\newgeometry{hmargin=2.2cm,vmargin=2.8cm}
	\section*{Notations and conventions: metric-affine gravity}
	
	\begin{center}
	{\renewcommand{\arraystretch}{1.3}
		\begin{tabular}{ll}
			Notations & \\
			\hline
			$\mathcal{M}$ &  smooth manifold of dimension $\Dim$\\
			$g_{\alpha\beta}$ & components of the symmetric non degenerate Lorentzian metric tensor $g$\\
			$\Bigl \{ \tensor{}{^{\sigma}_{\alpha\beta}}\Bigr \}$ & components of the Levi-Civita connection $\accentset{g}{\nabla}$\\
			$\tensor{R}{_{\alpha \beta\sigma}^{\gamma}}$ & components of the Riemann tensor associated with the connection $\accentset{g}{\nabla}$\\
			$\tensor{R}{_{\alpha \beta}}$ & components of the Ricci tensor associated with the connection $\accentset{g}{\nabla}$\\
			$R$ & Ricci scalar associated with the connection $\accentset{g}{\nabla}$\\
			$\tensor{\Gamma}{^{\sigma}_{\alpha\beta}}$ & components of a connection with torsion and non-metricity $\nabla$\\
			$\tensor{T}{^{\,\sigma}_{\alpha\beta}}$ & components of the torsion tensor $T$ \\
			$\tensor{Q}{_{\sigma}_{\alpha\beta}}$& components of the non-metricity tensor $Q$\\
			$\tensor{C}{^{\sigma}_{\alpha\beta}}$ & components of the distortion tensor $C$\\
			$\tensor{\mathcal{R}}{_{\alpha \beta\sigma}^{\gamma}}$ & components of the Riemann tensor associated with the connection $\nabla$\\
			$\tensor{\accentset{(1)}{\mathcal{R}}}{_{\alpha \beta}}$ & components of the Ricci tensor associated with the connection $\nabla$\\
			$\tensor{\accentset{(2)}{\mathcal{R}}}{_{\alpha \beta}}$ & components of the co-Ricci tensor associated with the connection $\nabla$ and metric $g$\\
			$\tensor{\accentset{(3)}{\mathcal{R}}}{_{\alpha \beta}}$ & components of the homothetic tensor associated with the connection $\nabla$\\
			$\mathcal{R}$ & Ricci scalar associated with the connection $\nabla$ and metric $g$\\[6pt]
			\hline
		\end{tabular}
	}
	\end{center}
\vspace{0.1cm}
	
	\begin{center}
	{\renewcommand{\arraystretch}{1.3}
		\begin{tabular}{ll}
			Conventions - terminology  & \\
			\hline
			indices $\alpha\,,\,\beta\,,\,\gamma\,,\,\delta\,,\,\sigma$ & coordinate frame (holonomic) indices \\
			indices $a \,,\, b\,,\, c \,,\, d$  & moving frame indices  \\
			$\tensor{T}{_{(\alpha \beta)}}$ & totally symmetric part of $T_{\alpha\beta}$: $\dfrac{1}{2}\left(\tensor{T}{_{\alpha \beta}}+\tensor{T}{_{\beta \alpha}}\right)$ \\
			$\tensor{T}{_{[\alpha \beta]}}$ & totally antisymmetric part of  $T_{\alpha\beta}$: $\dfrac{1}{2}\left(\tensor{T}{_{\alpha \beta}}-\tensor{T}{_{\beta \alpha}}\right)$ \\
			$\tensor{T}{_{(\alpha| \beta |\gamma)}}$ & symmetric part of $T_{\alpha\beta\gamma}$ with respect to $\alpha\,,\gamma$: $\dfrac{1}{2}\left(\tensor{T}{_{\alpha \beta \gamma}}+\tensor{T}{_{\gamma \beta \alpha}}\right)$ \\
			$\tensor{T}{_{[\alpha| \beta |\gamma]}}$ & antisymmetric part  of $T_{\alpha\beta\gamma}$ with respect to $\alpha\,,\gamma$: $\dfrac{1}{2}\left(\tensor{T}{_{\alpha \beta \gamma}}-\tensor{T}{_{\gamma \beta \alpha}}\right)$ \\
			Einstein summation  & $\tensor{T}{_\alpha^\alpha_\beta}=\displaystyle{\sum_{k=1}^{\Dim}}\,\tensor{T}{_k^k_\beta}$ \\
			metric contraction  & $g^{\beta\gamma}\,\tensor{T}{_{\,\alpha_{1}\ldots\alpha_{i-1}\,\beta\,\alpha_{i+1}\ldots\alpha_{j-1}\,\gamma\,\alpha_{j+1}\ldots \alpha_{q}}}=\tensor{T}{_{\,\alpha_{1}\ldots\alpha_{i-1}}^{\,\beta\,}_{\alpha_{i+1}\ldots\alpha_{j-1}\,\beta\,\alpha_{j+1}\ldots \alpha_{q}}}$\\
			& $g_{\beta\gamma}\,\tensor{T}{^{\,\alpha_{1}\ldots\alpha_{i-1}\,\beta\,\alpha_{i+1}\ldots\alpha_{j-1}\,\gamma\,\alpha_{j+1}\ldots \alpha_{p}}}=\tensor{T}{^{\,\alpha_{1}\ldots\alpha_{i-1}}_{\,\beta\,}^{\alpha_{i+1}\ldots\alpha_{j-1}\,\beta\,\alpha_{j+1}\ldots \alpha_{p}}}$\\
			$(f+1)$-traceless tensor & tensor for which any $f+1$ metric contractions yield zero\\
			1-traceless tensor $\tensor{\underline{T}}{_{\,\alpha_{1}\ldots \alpha_{q}}}$ & totally traceless part of $\tensor{T}{_{\,\alpha_{1}\ldots \alpha_{q}}}$: for any $1\leqslant i <j\leqslant q$\\
			& $g^{\alpha_i\alpha_j}\,\tensor{\underline{T}}{_{\,\alpha_{1}\ldots\alpha_{i-1}\,\alpha_{i}\,\alpha_{i+1}\ldots\alpha_{j-1}\,\alpha_{j}\,\alpha_{j+1}\ldots \alpha_{q}}}=0$ \\[6pt]
			\hline
		\end{tabular}
	}
\end{center}

\newpage
	\section*{Notations: representation theory}

	\begin{center}
	{\renewcommand{\arraystretch}{1.3}
	\begin{tabular}{ll}
		\hline
		$V$, $V^{\otimes n}$ & complex vector space of dimension $\Dim$, $n$-fold tensor product of $V$\\
		$\GL(\Dim,\C)$ (resp. $\GL(\Dim,\mathbb{R})$) & complex (resp. real) Lie group of invertible $\Dim\times\Dim$ matrices\\
		$\Or(\Dim,\C)$ (resp. $\Or(1\,,\,\Dim-1)$) & complex (resp. real) Lie group of orthogonal (resp. pseudo-orthogonal)\\
		& $\Dim\times\Dim$ matrices\\
		$\sn$ & symmetric group which acts on $V^{\otimes n}$ (permutation group on $n$ symbols)  \\
		$\C\sn$ & symmetric group algebra over the complex field  \\
		$\bn(\Dim)$ (or simply $\bn$) & Brauer algebra which acts on $V^{\otimes n}$  \\
		$\mu$ & integer partitions of $n$ or Young diagrams with $n$ boxes associated\\
		&  with the irreducible representations of $\GL(\Dim,\C)$ or $\C\sn$ in $V^{\otimes n}$\\
		$\lambda$ & integer partitions or Young diagrams associated with\\
		& the irreducible representations of $\Or(\Dim,\C)$  or $\bn(\Dim)$ in $V^{\otimes n}$\\
		$V^\mu$ (resp. $D^\lambda$) & irreducible representations of $\GL(\Dim,\C)$ (resp. $\Or(\Dim,\C)$) in $V^{\otimes n}$\\
		$m_\mu$ (resp. $m_\lambda$) & multiplicity of $V^{\mu}$ (resp. $D^{\lambda}$) in $V^{\otimes n}$ \\
		$L^\mu$ (resp. $M^\lambda_n$) & irreducible representations of $\C\sn$ (resp. $\bn(\Dim)$) in $V^{\otimes n}$ \\
		$f$ (or $f_\lambda$) & non-negative integers implicitly parameterizing the irreducible \\
		&representations $M^\lambda_n$ of $\bn(\Dim)$ in $V^{\otimes n}$: $\lambda\vdash n-2\, f$; $f_\lambda=\dfrac{n-|\lambda|}{2}$\\
		$g_\mu$ (resp. $g_\lambda$) & multiplicity of $L^{\mu}$ (resp. $M^{\lambda}_n$) in $V^{\otimes n}$ \\
		$\tensor{C}{^\mu_\lambda_\nu}$ & Littlewood-Richardson coefficients\\
		$\mathcal{Z}_n$ (resp. $\mathcal{Z}_{\bn}$) & center of $\C\sn$ (resp. center of $\bn(\Dim)$) \\
		$Z^\mu$ & central Young idempotents (belong to $\mathcal{Z}_n$);\\
		&  projection operators onto the isotypic components $\left(V^\mu\right)^{\oplus m_{\mu}}$\\
		$P^\lambda_n$ & elements of $\bn(\Dim)$ which belong to $\mathcal{Z}_{\bn}$ when $\Dim \geqslant n-1$;\\
		&projection operators onto the isotypic components $\left(D^\lambda\right)^{\oplus m_{\lambda}}$\\
		$\tab$ & either paths in the Bratteli diagram associated with $\C\sn$\\
		&(standard Young tableaux)\\
		&or paths in the Bratteli diagram associated with $\bn(\Dim)$\\ 
		&(oscillating tableaux)\\
		$\mathcal{Y}^{\,\tab}$ & Young symmetrizers ($t$ is a standard tableau) \\
		$Y^{\,\tab}$ & Young seminormal idempotents ($t$ is a standard tableau): \\
		& primitive pairwise orthogonal idempotents in $\C\sn$\\
		$P_n^{\,\tab}$ & when $\Dim\geqslant n-1$, primitive pairwise orthogonal idempotents in $\bn(\Dim)$ \\
		& ($\tab$ is an oscillating tableau) \\
		$\cn$ & centralizer of $\sn$ in $\bn(\Dim)$  \\[6pt]
		\hline
	\end{tabular}
	}
	\end{center}

	\newgeometry{hmargin=2.5cm,vmargin=3.5cm}

	\tableofcontents
	\mainmatter	
	\clearpage
	\chapter*{Introduction}
\addcontentsline{toc}{chapter}{\protect\numberline{}Introduction}
\pagestyle{headings}
\markboth{Introduction}{}

While a large part of this thesis is more mathematical and dedicated to some algebraic aspects related to the irreducible decomposition of tensors, it is originally motivated by a class of modified theories of gravity called metric-affine gravity.  

\subsection*{General Relativity}
\addcontentsline{toc}{subsection}{\protect\numberline{} General Relativity}


During the $19$th century, the astronomer Urbain Le Verrier studied the orbit of Mercury around the Sun and observed an anomalous drift of its perihelion\footnote{The point of its orbit closer to the sun.} compared to the prediction of Newton's theory of gravity. 
In order to accommodate his observations within the theory of Newton, he postulated the existence of Vulcain: an hypothetical planet orbiting between the Sun and Mercury. Despite many efforts by the astronomers at that time, this planet was never observed. The other approach considered was to modify Newton's law of universal gravitation in the regime of a strong gravitational field, and the solution was provided in $1915$ by Albert Einstein with the theory of General Relativity.\medskip 

Within the Newtonian theory of gravity, when an object loses mass, the gravitational potential $\Phi$ and gravitational acceleration field $\vv{G}$ it generates are instantaneously modified: the interaction  propagates at infinite speed. This fact is in contradiction with the principles of Special Relativity \cite{gourgoulhon2016special}, a theory introduced by Einstein in $1905$, which postulates that no field can propagate faster than the constant speed of light $c$. To incorporate this causal constraint of physical interactions, the concept of absolute time of Newtonian physics had to be abandoned. In Special Relativity, space and time are joined together within a whole continuum, the flat Minkowski spacetime. Inertial observers, defined as having zero acceleration\footnote{More precisely, an inertial observer is modeled by an orthonormal frame transported along a timelike curve with zero four-acceleration and zero four-rotation.}, follow the straight lines in the Minkowski space. Each point on the observer's trajectory is equipped with a vector space (its tangent space) with its constant spacetime metric (the Minkowski metric $\eta_{ab}$) allowing for the measurement of spacetime distances and the construction of their causal structure (\textit{the light cone}).\medskip
\newpage 

In the theory of General Relativity, Einstein incorporated the gravitational interactions into the fabric of a curved spacetime, thereby generalizing the flat Minkowski space of Special Relativity. The spacetime of General Relativity is described by a differentiable manifold\footnote{A differentiable manifold is roughly speaking a space which can be ‘patched' by domains where one can define a coordinate system that maps the points of each domains to points of $\mathbb{R}^4$. Additionally one can define partial derivatives of functions on the manifold through the usual partial derivatives of functions on $\mathbb{R}^4$.} whose tangent spaces are endowed with a metric $g_{\alpha\beta}$ which may vary from point to point. This metric field, which replaces at the same time the Minkowski metric of Special Relativity and the gravitational potential $\Phi$ of Newton's theory of gravitation, is the core ingredient of General Relativity. It governs the curvature of spacetime: the Riemann curvature tensor $\tensor{R}{_{\alpha\beta\rho}^\sigma}$.\medskip

The latter is constructed from an affine connection which extends the concept of parallel transport of vectors along curves on an affine spaces\footnote{In the flat Minkowski space of Special Relativity vectors are transported between two infinitesimally close points $A$ and $B$ by a translation $\vv{AB}$.} to curves on a manifold, and hence the notion of flat straight lines to curved straight lines: the \textit{geodesics}. The affine connection of General Relativity is the Levi-Civita connection\footnote{The Levi-Civita connection coefficients will be denoted $\Bigl \{ \tensor{}{^\alpha_{\beta\rho}}\Bigr \}$ in the rest of the manuscript.} $\tensor{\accentset{g}{\Gamma}}{^\alpha_{\beta\rho}}$, which \textit{preserves the metric} under parallel transport and has no \textit{torsion}. As a result of these properties, $\tensor{\accentset{g}{\Gamma}}{^\alpha_{\beta\rho}}$ is constructed from the metric and its partial derivatives only, and replaces the gravitational acceleration field $\vv{G}$ of Newton which is given by the derivatives of the potential $\Phi$. The geodesics of $\tensor{\accentset{g}{\Gamma}}{^\alpha_{\beta\rho}}$ are the paths followed by small masses (observers) in a strong gravitational field, much like the trajectory of Mercury around the Sun.\medskip

%
%
%
%
%
%
%

The following question arises: how does the Sun influences the curvature of spacetime~?  The answer is given by Einstein field's equations for the metric $g_{\alpha\beta}$:
\begin{equation}\label{eq:EH_FH}
	\begin{array}{ll}
		R_{\alpha\beta}-\dfrac{1}{2} g_{\alpha\beta} R=T_{\alpha\beta}\,,\\
	\end{array}
\end{equation}
where  $R_{\alpha\beta}=\tensor{R}{_{\alpha\sigma\beta}^\sigma}$ (\textit{Ricci tensor}) and $R=g^{\alpha\beta}\tensor{R}{_{\alpha\sigma\beta}^\sigma}$ (\textit{Ricci scalar}) are constituents of the Riemann tensor which depends only  on the metric, its first and second derivatives, while $T_{\alpha\beta}$ is called the stress-energy tensor\footnote{This tensor generalizes to relativistic physics the stress tensor used in the mechanics of continuous media.} (also known as energy-momentum tensor). This latter tensor characterizes the flow of energy and momentum of the matter content of spacetime.\medskip 

The theory of General Relativity provides a coherent description of space, time, gravity and its interaction with matter, covering a vast range of scales with unprecedented agreement with observations. Among the successes of General Relativity, let us mention the prediction of the emission of gravitational waves propagating at the speed of light $c$. Originating from extremely violent astronomical processes like black hole and neutron star mergers, they were finally detected in the last
years by LIGO and Virgo Collaborations \cite{PhysRevLett.116.241103,PhysRevLett.119.161101}. 

\medskip

\newpage

\subsection*{Modified gravity}
\addcontentsline{toc}{subsection}{\protect\numberline{} Modified gravity}
The dynamics of a classical field theory is most often determined by its Lagrangian, whose construction is based on symmetry principles. For example, the Lagrangian of Yang-Mills\footnote{The Yang-Mills theory is at the core of our understanding of the fundamental interactions of particle physics.} theory, which is constructed from quadratic terms in the field strength (curvature) of the gauge field\footnote{After quantization, the gauge fields (for example the electromagnetic field) of the theory are called gauge bosons (for example photons).} (connection), is invariant under particular spacetime-dependent transformations associated with Lie groups. For General Relativity, the Lagrangian responsible for the left hand side of Einstein's field equations \eqref{eq:EH_FH} is the Ricci scalar $R$ which is invariant under arbitrary change of coordinates.\medskip

One of the motivations for modifying General Relativity is related to the quest for a quantum description of gravity. A first attempt to overcome the obstacles \cite{AIHPA_1974__20_1_69_0} for a quantum theory of gravity based on General Relativity involved introducing quadratic terms in curvature into the Lagrangian. Unfortunately, this approach did not prove entirely satisfactory \cite{PhysRevD.16.953}, and the problem remains open.\medskip

%

Nowadays, there are strong enough motivations to consider modified theories of gravity already at the classical level. These motivations are related to astrophysical and cosmological observations, like the gravitational lensing effect\footnote{Gravitational lensing corresponds to the deflection of light caused by the curvature of spacetime induced by a distribution of mass.} \cite{Massey_2010}, the rotational curve of galaxies \cite{vera_review,Jiao} or the present cosmic acceleration of the universe \cite{Riess_1998,Perlmutter_1999}. In order to accommodate these observations with the theory of General Relativity, the standard model of cosmology $\Lambda$CDM  introduces special ad-hoc contents to our universe called \textit{dark matter}, and \textit{dark energy}. Within this model the dark matter is viewed as a particular type of exotic matter which interacts very weakly with the rest of the ordinary matter. Its distribution is determined to account for the observed velocity of stars in galaxies and bending of light by the gravitational field. Dark energy on its part, is encoded in the theory by the introduction of the \textit{cosmological constant $\Lambda$} in the Lagrangian of General Relativity, and its value is determined by the observations of the accelerated expansion of the universe. As a consequence, the $\Lambda$CDM model does not account for the exact nature of these two constituents. The main goal of modified theories of gravity is to provide alternatives to the standard model of cosmology.\medskip 

A first approach relies on extending the Lagrangian of General Relativity by adding higher order curvature terms or more general tensor invariants constructed from $\tensor{R}{_{\alpha\beta\rho}^\sigma}$, $\tensor{R}{_{\alpha\beta}}$ and $R$. For example, there is the Starobinsky model of inflation \cite{Starobinsky:1980te}, or the $f(R)$ theories \cite{sotiriou1,DeFelice2010}. The Lagrangians of these theories can be rewritten as modifications of General Relativity, including a specific coupling of curvature to an additional dynamical scalar field, thus making them equivalent to certain tensor-scalar theories \cite{Dicke_1962,PhysRevLett.29.137,David_Wands_1994}. The latter fall within the framework of Horndeski theories and their extensions \cite{Horndeski1974,Deffayet:2009wt,Deffayet:2011gz,PhysRevD.89.064046,Gleyzes:2014dya,Langlois_2016,BenAchour:2016fzp,Kobayashi_2019}, aiming to construct and explore, in all generality, cosmological models of dark energy and inflation, including a dynamical scalar field.\medskip

A second approach relies on the extension of the geometry modeling the gravitational interactions. In its broadest form known as metric-affine gravity, this approach provides a framework offering numerous alternatives to the standard model of cosmology \citep{Iosifidis_2023, Gialamas_2023, Aoki:2023sum, Aoki_2020, Koivisto_2006, Koivisto_2010, Jimenez_2016, Shimada:2018lnm, Aoki:2018lwx, BELTRANJIMENEZ2016400, Li_2012, PhysRevD.100.064018, GALTSOV2019453, Helpin:2019kcq, PhysRevD.96.084023, PhysRevD.99.104020, PhysRevD.100.044037, Enqvist_2012}. Metric-affine gravity originates both in the works of Weyl and his attempt to unify gravity and electromagnetism \cite{WeylMAG, RevModPhys.72.1, Eddington}, and in the works of Cartan, who was the first to propose a generalization of General Relativity including torsion \cite{cartan1922generalisation} (for additional historical information on metric-affine gravity, see \cite{Goenner2014, blagojevic2012gauge}). Note that metric-affine gravity encompasses various classes of theories, such as Poincaré gauge theories \cite{hehl2023lectures, Obukhov_2018}, theories on Weyl-Cartan spacetime  \cite{Tresguerres1995, DirkPuetzfeld_2001, Puetzfeld2002, Jorge, Babourova_2006}, teleparallel gravity \cite{jimenez2018teleparallel, BELTRANJIMENEZ2020135422, Pereira_2019, Conroy2018, beltran2019geometrical, Koivisto_coincident, Cai_2016}, and Born-Infeld-inspired theories \cite{BELTRANJIMENEZ20181}.

\subsection*{Metric-affine gravity and projective invariance} 
\addcontentsline{toc}{subsection}{\protect\numberline{} Metric-affine gravity and projectve invariance}
In metric-affine gravity, just like in General Relativity, spacetime is modeled by a pseudo-Riemannian manifold, that is a differentiable manifold endowed with its  Lorentzian metric $g_{\alpha\beta}$. The particularity of metric-affine gravity relies in the introduction of another affine connection $\tensor{\Gamma}{^\alpha_{\beta\rho}}$ on top of the pseudo-Riemannian manifold. Contrary to the Levi-Civita connection $\tensor{\accentset{g}{\Gamma}}{^\alpha_{\beta\rho}}$, the affine connection $\tensor{\Gamma}{^\alpha_{\beta\rho}}$ has non-vanishing \textit{torsion}:
\begin{equation}
\tensor{T}{^\alpha_{\beta\rho}}=2\,\tensor{\Gamma}{^\alpha_{[\beta\rho]}}\neq 0\,,
\end{equation}
and does not preserve the metric under parallel transport, the \textit{non-metricity} tensor does not vanish: 
\begin{equation}
\tensor{Q}{_{\alpha\beta\rho}}=\partial_{\alpha} g_{\beta\rho}-2\,\tensor{\Gamma}{_{(\beta|\alpha|\rho)}}\neq 0 \,.
\end{equation}

One of the motivations for such generalization of the geometrical framework of gravity comes from the continuum limit of crystalline materials with defects \cite{hehl2007elie,Hehl_mom1}. In the continuum approach, a real crystal is modeled by a smooth manifold endowed with geometric fields which characterize the density of defects in the crystal. In this respect, distributed point defects and dislocations in crystals are encoded respectively via the non-metricity tensor $\tensor{Q}{_{\alpha\beta\rho}}$ \cite{Lazar_2007,Mindlin1964,Falk1981,kroner1986continuized,Kupferman2018,yavari2012weyl} and the torsion tensor $\tensor{T}{^\alpha_{\beta\rho}}$ \cite{RazKupferman2015,Kleinert_2000,Yavari2012},  while the application of a stress to the crystal leads to the deformation and propagation of these defects. In General Relativity the stress-energy tensor  $T_{\alpha\beta}$ is defined as the variational derivative of the matter action\footnote{The matter action is the coordinate invariant integral of the matter Lagrangian over the domain of study of space-time.} with respect to the metric. Similarly, the hypermomentum tensor of metric-affine gravity, denoted $\tensor{\Delta}{^\alpha_\beta_\rho}$, is defined as the variational derivative of the matter action with respect to the affine connection. Just as $T_{\alpha\beta}$ is responsible for the ‘elastic' deformations of spacetime, $\tensor{\Delta}{^\alpha_\beta_\rho}$ sources the deformations of ‘spacetime defects' \cite{Hehl_mom1,HehlKerlickHeyde_mom2,OBUKHOV_mom}. The hypermomentum tensor encompasses hypothetical properties of the micro-structure of matter\footnote{This tensor is often decomposed into separated significant physical quantities. They are the spin density current, as well as the dilation and shear currents \cite{Hehl_mom1}.}, which are expected to emerge at very high energy \cite{PhysRevD.10.1066,Yuval_Neeman_1997,puetzfeld2007propagation,Yasskin_1980} like for example during the early stages of the universe.  \medskip

%
To the connection $\tensor{\Gamma}{^\alpha_{\beta\rho}}$ one associates a Riemann curvature tensor $\tensor{\mathcal{R}}{_{\alpha\beta\sigma}^\rho}$, so that one may consider the two Riemann tensors:
\begin{equation}
\tensor{\mathcal{R}}{_{\alpha\beta\sigma}^\rho} \hspace{0.5cm} \text{\,(\textit{affine} Riemann tensor)\,}\,,\hspace{ 2cm}	\tensor{R}{_{\alpha\beta\sigma}^\rho} \hspace{0.5cm} \text{\,(\textit{metric} Riemann tensor)\,}.
\end{equation}
The metric Riemann tensor $\tensor{R}{_{\alpha\beta\sigma}^\rho}$ admits a unique decomposition into elementary pieces: 
\begin{equation}\label{eq:Ricci_decomposition}
	\tensor{R}{_{\alpha\beta\rho}^\sigma}=\tensor{W}{_{\alpha\beta\rho}^\sigma}+\tensor{E}{_{\alpha\beta\rho}^\sigma}+\tensor{S}{_{\alpha\beta\rho}^\sigma}\,,
\end{equation}
where $\tensor{W}{_{\alpha\beta\rho}^\sigma}$ is the Weyl tensor which is the totally traceless part of the Riemann tensor, while $\tensor{E}{_{\alpha\beta\rho}^\sigma}$ and $\tensor{S}{_{\alpha\beta\rho}^\sigma}$, are respectively constructed from the Ricci tensor and the Ricci scalar. This decomposition, known as the Ricci decomposition \cite{hawking_ellis_1973,lee2018introduction}, can be seen as the unique irreducible decomposition of the Riemann tensor with respect to the action of the group preserving the metric: the Lorentz group $\Or(1, \Dim-1)$\footnote{We consider a pseudo-Riemannian manifold of arbitrary dimension $\Dim$. Specific cases will be examined in detail.}. More precisely, in an orthonormal moving frame defined on an open domain $U$ in spacetime\footnote{Any open domain $U$ can be equipped with an orthonormal moving frame \cite[Prop. 2.8]{lee2018introduction}.}, each part of \eqref{eq:Ricci_decomposition} at a point in $U$ is an irreducible tensor\footnote{The vector space to which an irreducible tensor belongs is preserved with respect to the group action. Moreover, the dimension of this space is the smallest possible with this property.}. This decomposition plays  a structural role in the analyses of General Relativity, and more generally in our understanding of gravity. For example, the Weyl tensor is a measure of the deviation of a pseudo-Riemannian manifold from conformal flatness. It is the part of curvature at a point which is determined by the matter distribution at other points and governs the propagation of gravitational waves in vacuum \cite{hawking_ellis_1973}. The Ricci tensor and Ricci scalar constitute the components of curvature at a point that are determined by the matter distribution at that same point \eqref{eq:EH_FH}.\medskip

In metric-affine gravity, the metric Riemann tensor $\tensor{R}{_{\alpha\beta\sigma}^\rho}$ is contained in the affine Riemann tensor $\tensor{\mathcal{R}}{_{\alpha\beta\sigma}^\rho}$ and the latter reduces to the former in the limit of zero non-metricity and zero torsion. In this respect, it is reasonable to consider $\tensor{\mathcal{R}}{_{\alpha\beta\sigma}^\rho}$ as more fundamental and construct the  Lagrangians from its scalar invariants. While the only scalar entering $\tensor{\mathcal{R}}{_{\alpha\beta\sigma}^\rho}$ is the Ricci scalar $\mathcal{R}= g^{\alpha\beta}\tensor{\mathcal{R}}{_{\alpha\sigma\beta}^\sigma}$, which at linear level support no dynamics for $\tensor{\Gamma}{^\alpha_{\beta\rho}}$, kinetic terms $‘\partial \Gamma\partial\Gamma\,$', and hence dynamics is naturally provided by the introduction of curvature squared terms in the Lagrangian.\smallskip 

The main goal of the first chapter is to extend the Ricci decomposition of $\tensor{R}{_{\alpha\beta\sigma}^\rho}$ to the affine Riemann tensor $\tensor{\mathcal{R}}{_{\alpha\beta\sigma}^\rho}$ (see Theorem \ref{theo:irreducible_Riemann_O} and Proposition \ref{prop:small_d_Riemann}) and to construct the most general quadratic curvature Lagrangian from this decomposition (see Proposition \ref{prop:Lagrangian_Riemann}). As a matter of fact, due to greater complexity of the Riemann tensor associated with a connection with non-metricity and torsion, its irreducible decomposition is not unique. As such we give an alternative to the decomposition and quadratic action obtained in $1992$ by McCrea \cite{McCrea_1992} in the context of metric-affine gauge theory of gravity. While in \cite{McCrea_1992,HEHL_1999} it is claimed that the decomposition performed therein is unique and canonical, our approach shows that an alternative decomposition with minimal natural requirements is available. We will also be interested in the irreducible decomposition of the distortion tensor $\tensor{C}{^\alpha_{\beta\rho}}=\tensor{\Gamma}{^\alpha_{\beta\rho}}-\tensor{\accentset{g}{\Gamma}}{^\alpha_{\beta\rho}}$, which characterizes the deviation between the post-Riemannian geometry and pseudo-Riemannian geometry (see Theorem \ref{theo:irreducible_distortion_O} and Proposition \ref{prop:small_d_distortion}).\medskip


To achieve our goal, we are guided by the notion of projective transformation. The latter is a change of connection which preserves the direction of vectors under parallel transport
along any curve \cite{thomas1926asymmetric,eisenhart1929non}, and hence is related to the notion of scale invariance. In terms of the components of the connection, a projective transformation is given by 
\begin{equation}\label{eq:projective_transformation_intro}
	\tensor{\Gamma}{^\sigma_\alpha_\beta}\to\tensor{\overline{\Gamma}}{^\sigma_\alpha_\beta}= \tensor{\Gamma}{^\sigma_\alpha_\beta} + \tensor{\delta}{^\sigma_{\beta}} \tensor{\xi}{_{\alpha}}\,,
\end{equation}
where $\xi_\alpha$ is an arbitrary vector field. As a result of the construction, the obtained quadratic action is adapted for the study of projective invariant theories of metric-affine gravity.\smallskip 

It has been known for some time that the Einstein-Hilbert action of metric-affine gravity, constructed from the metric-affine Ricci scalar $\mathcal{R}$, is invariant under the projective transformation \eqref{eq:projective_transformation_intro} \cite{schrodinger_1985,Hehl1978}. However, the identification of this symmetry as a gauge symmetry was only made more recently \cite{BJulia_1998, Dadhich2012}. The second Noether's theorem \cite{haro_2022} indicates an algebraic \textit{off-shell} identity\footnote{This means that the identity is always valid independently of the equations of motion} associated with the Einstein-Hilbert action of metric-affine gravity. Indeed, the Palatini tensor, which arises from the variation of the Einstein-Hilbert action with respect to the connection:
\begin{equation}
	\dfrac{\delta S_{\textup{EH}}}{\delta \tensor{\Gamma}{^\alpha_\beta_\rho}}=\tensor{P}{_\alpha^\beta^\rho}\,,
\end{equation}
has the following property
\begin{equation}
	\tensor{P}{_\alpha^\beta^\alpha}=0.
\end{equation}
As a consequence of this identity, the hypermomentum tensor associated with the matter Lagrangian is subjected to the constraints  \cite{Sandberg_1975,Hehl1981,iosifidis2019scale,Iosifidis_2020,Garcia-Parrado_2021} :
\begin{equation}\label{eq:contrainte}
	\tensor{\Delta}{_\alpha^\beta^\alpha}=0,
\end{equation}
which are verified identically for a projective invariant matter Lagrangian.
It is argued in \cite{Hehl1981} that projective invariant theories are physically inconsistent and that the projective symmetry has to be broken by appropriate use of Lagrange multipliers. The authors' point of view can be grasped from the following two citations:
\begin{quote}
	‘‘\textit{We believe that the sources (matter fields) for a theory are fundamental and should determine the appropriate geometrical framework, not vice versa.}''
\end{quote}
Referring to the equation of motion for the connection $\tensor{P}{_\alpha^\beta^\rho}=\tensor{\Delta}{_\alpha^\beta^\rho}$ of the metric-affine Einstein-Hilbert action with matter, the authors continue as follows:
\begin{quote}
	‘‘\textit{The inconsistency is at once apparent. Since the Palatini tensor is traceless in its first and last indices, such a theory would be possible only for matter with a vanishing dilation current $\tensor{\Delta}{_\alpha^\beta^\alpha}$. Because the geometry is projectively invariant, it cannot respond to the degrees of freedom of the matter associated with projective transformations.}''
\end{quote}
However, it is possible to adopt an alternative perspective by imposing projective gauge symmetry on the matter Lagrangian. This viewpoint is discussed in the following articles \cite{jimenez2018teleparallel,Bernal:2016lhq,Sandberg_1975,Bejarano_2020,Garcia-Parrado_2021}. In this thesis, the projective symmetry is mainly used from a mathematical standpoint, serving as a tool to structure the irreducible decompositions of the Riemann tensor and distortion tensor, as well as for the construction of the associated quadratic Lagrangians.

\newpage

\subsection*{Irreducible decompositions of tensors via the Brauer algebra} 
\addcontentsline{toc}{subsection}{\protect\numberline{} Irreducible decompositions of tensors via Brauer algebra}
Since the seminal work of E. Wigner \cite{Wigner_1939} the elementary particles in field theories are understood as irreducible components of the underlying symmetry group/algebra of the theory. As a consequence, decomposing tensors into irreducible constituents with respect to the action of groups and algebras is of cornerstone importance in modern physics. 
In case when a theory includes a non-degenerate metric, elementary constituents are irreducible with respect to action of the metric-preserving group: the (pseudo-)orthogonal group.\medskip

In the second part of this thesis we construct the projection operators which were used to perform an irreducible decomposition of the Riemann tensor and of the distortion tensor with respect to the Lorentz group $\Or(1,\Dim-1)$. The construction is explicit and systematic as it applies to tensors with arbitrary symmetries and of arbitrary fixed ranks. As a preliminary step toward the irreducible decomposition with respect to $\Or(1,\Dim-1)$ we first carry out the irreducible decomposition with respect to $\GL(\Dim,\mathbb{R})$. This allows one to have more control over the properties of the decomposition, in particular with respect to the projective transformation of the connection.\medskip 

The construction is general enough to find applications in several areas of physics and applied mathematics: for example in higher-spin theory \cite{bekaert2021unitary,Bekaert_2007,Boul-Bek_mixed_GL,Nicolas_Boulanger_2009,Fronsdal_1978,Bonelli_2003,LABASTIDA1986101,SIEGEL1987125,Brink_2000,Isaev_2020}, in elasticity and generalized elasticity theory \cite{Toupin1962,Itin_2015,auffray_2015,Lazar_mech}, and in random tensor models and networks \cite{keppler2023duality,Collins2006,collins2009some}.\medskip

For the sake of generality and simplicity, in the second part of this thesis we consider a complex vector space $V$ endowed with a non-degenerate symmetric metric $g$. The irreducible decompositions of tensors in $V^{\otimes n}$ with respect to the action of $\GL(\Dim,\C)$ can be done using certain projection operators in the symmetric group algebras $\C\sn$, while when considering the action of orthogonal groups, irreducible tensors can be obtained using projection operators in the Brauer algebras $\bn(\Dim)$: introduced in 1937 by  R. Brauer \cite{Brauer_1937}. This fact is a consequence of Schur-Weyl dualities \cite{stevens2016schur,doty2007new}: on the one hand, for $\GL(\Dim,\C)$ and $\mathbb{C}\Sn{n}$, on the other hand, for $\Or(\Dim,\C)$ and $B_{n}(\Dim)$. This interplay can be summarised via the following diagram\footnote{Here we reproduce the diagram similar to the one presented in \cite{CDVM_Br_blocks}.}:
\begin{figure}[H]
	\centering
	\includegraphics[scale=1.1]{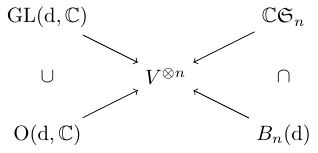}
	\caption{Schur-Weyl dualities for the classical Lie groups $\GL(\Dim,\C)$ and $\Or(\Dim,\C)$. The image in $\End(V^{\otimes n})$ of the horizontally aligned groups/algebras are mutual centralizer of each other.}
	\label{fig_proj:SW_dualities}
\end{figure}
\vskip 4pt

Any permutation $s\in \mathfrak{S}_n$ can be conveniently represented by a diagram with $2$ horizontal rows of $n$ vertices, with each vertex in the top row being connected to a vextex in the bottom row by a line.
For example, the permutations of $\mathfrak{S}_2$ are given by: 
\begin{equation*}
	\id=\raisebox{-.4\height}{\includegraphics[scale=0.32]{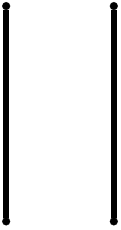}}\,, \hspace{0.3cm} (12)=\raisebox{-.4\height}{\includegraphics[scale=0.32]{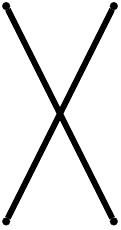}}\,, \hspace{0.3cm}
\end{equation*}
while the permutations of $\mathfrak{S}_3$ are given by: 
\begin{equation*}\label{eq:perms3intro}
	\id=\raisebox{-.4\height}{\includegraphics[scale=0.51]{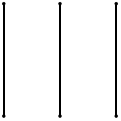}}\,, \hspace{0.3cm} (12)=\raisebox{-.4\height}{\includegraphics[scale=0.51]{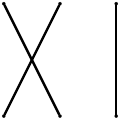}}\,, \hspace{0.3cm} (23)=\raisebox{-.4\height}{\includegraphics[scale=0.55]{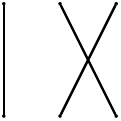}}\,, \hspace{0.3cm} (13)=\raisebox{-.4\height}{\includegraphics[scale=0.51]{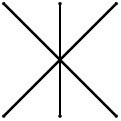}}\,, \hspace{0.3cm} (123)=\raisebox{-.4\height}{\includegraphics[scale=0.51]{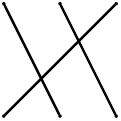}}\,, \hspace{0.3cm} (132)=\hspace{0.3cm} \raisebox{-.4\height}{\includegraphics[scale=0.51]{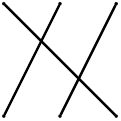}}\,.
\end{equation*}
Placing labels $i=1\,,\, \ldots  \,,\, n$, on the top row and $j=1\,,\,\ldots \,,\, n$, on the bottom row of a permutation diagram $s$, one can construct the associated tensor operator $\mathfrak{s}\in \End(V^{\otimes n})$ whose components are given by\footnote{This presentation is taken from \cite{keppler2023duality}.}: 
\begin{equation}\label{eq:operator_sn_intro}
	\tensor{\mathfrak{s}}{_{\,\, b_1\ldots b_n}^{a_1\ldots a_n}}=  \prod_{\Scale[0.7]{\begin{array}{c}
				{\scriptstyle i \text{ in the top row}}\\[-4pt]
				{\scriptstyle \text{connected to } j \text{ in to bottom row} }
	\end{array}}} \tensor{\delta}{^{a_i}_{b_j}}\,.
\end{equation}

It is well known that the projection to a particular irreducible representation of $V^{\otimes n}$ under $\GL(\Dim,\C)$ can be done using the so-called Young symmetrizers \cite{fulton2013representation} which are particular linear combinations of permutations in $\sn$. In this relation, tensors with Young symmetries are irreducible tensors under $\GL(\Dim,\C)$. The Young symmetrizers do not form a partition of unity of $\C\sn$ and thus are not adapted to the full irreducible decomposition of an arbitrary tensor\footnote{Already for $n=3$ the Young symmetrizers do not yield an orthogonal irreducible decomposition of $V^{\otimes n}$.}. In fact, there exist projection operators in $\C\sn$ which are well suited for this task: the \textit{Young seminormal idempotents} of $\C\sn$ \cite{jucys1966young,Murphy,Thrall_seminormalYoung_1941,molev2006fusion,doty2019canonical,Keppeler:2013yla}.   \medskip

The Brauer algebra $B_n(\Dim)$ plays a similar role when one considers the case of the orthogonal group $\Or(\Dim,¯\C)$. The diagrams of Brauer algebras extend the set of permutation diagrams: any pair of vertices is allowed to be connected. In particular, there can be paths connecting two vertices in the same row. 
For example, in addition to the permutations $\Sn{2}$ on has only one more diagram in  $B_2(\Dim)$:
\vskip 4pt
\begin{equation*}
	\raisebox{-.4\height}{\includegraphics[scale=0.29]{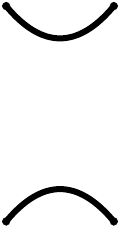}}\,
\end{equation*}
\vskip 4pt
\noindent while the additional diagrams of $B_3(\Dim)$ are: 
\vskip 4pt
\begin{equation*}\label{eq:brauer3intro}
	\raisebox{-.4\height}{\includegraphics[scale=0.5]{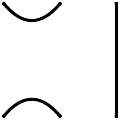}}\,, \hspace{0.5cm} \raisebox{-.4\height}{\includegraphics[scale=0.5]{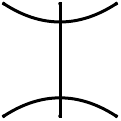}}\,, \hspace{0.5cm} \raisebox{-.4\height}{\includegraphics[scale=0.5]{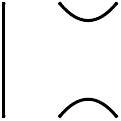}}\,, \hspace{0.5cm} \raisebox{-.4\height}{\includegraphics[scale=0.5]{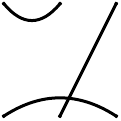}}\,, \hspace{0.5cm} \raisebox{-.4\height}{\includegraphics[scale=0.5]{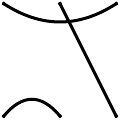}}\,, \hspace{0.5cm} \raisebox{-.4\height}{\includegraphics[scale=0.5]{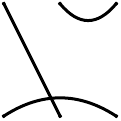}}\,,\hspace{0.5cm}
	\raisebox{-.4\height}{\includegraphics[scale=0.5]{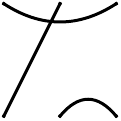}}\,,\hspace{0.5cm}  \raisebox{-.4\height}{\includegraphics[scale=0.5]{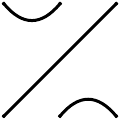}}\,, \hspace{0.5cm} \raisebox{-.4\height}{\includegraphics[scale=0.5]{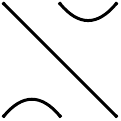}}\,.		
\end{equation*}
\vskip 4pt
\noindent From the point of view of tensor calculus, the arcs in the diagrams of $B_n(\Dim)$ mark the presence of the metric tensor $g$. Similarly to \eqref{eq:operator_bn_intro}, from a diagram in $\bn(\Dim)$ one can construct the associated tensor operator $\mathfrak{b}\in\End(V^{\otimes n})$ whose components are given by: 
\begin{equation}\label{eq:operator_bn_intro}
	\tensor{\mathfrak{b}}{_{\, b_1\ldots b_n}^{a_1\ldots a_n}}= \hspace{-0.5cm}\prod_{\Scale[0.7]{\begin{array}{c}
				{\scriptstyle i \text{ in the top row}}\\[-4pt]
				{\scriptstyle \text{connected to } j \text{ in to bottom row} }
	\end{array}}} \tensor{\delta}{^{a_i}_{b_j}}\,  \prod_{\Scale[0.7]{\begin{array}{c}
				{\scriptstyle i \text{ in the top row}}\\[-4pt]
				{\scriptstyle \text{connected to } j \text{ in to top row} }
	\end{array}}} \tensor{g}{^{a_ia_j}}\, \prod_{\Scale[0.7]{\begin{array}{c}
				{\scriptstyle i \text{ in the bottom row}}\\[-4pt]
				{\scriptstyle \text{connected to } j \text{ in to bottom row} }
	\end{array}}} \tensor{g}{_{b_i}_{b_j}}\,.
\end{equation}
\vskip 4 pt

To illustrate our point, let us write the projections operators in $\bn(\Dim)$ which realize the decomposition of an arbritary tensor of order $2$ under the action of $\Or(\Dim,\C)$: 
\begin{equation}\label{eq:projection_operator_2}
	P^{\Yboxdim{3pt}\yng(1,1)}_{2}=\dfrac{1}{2}\left(\,\,\raisebox{-.4\height}{\includegraphics[scale=0.25]{fig/id2.pdf}}\,-\,\raisebox{-.4\height}{\includegraphics[scale=0.25]{fig/s12.pdf}}\,\right)\,, \hspace{1cm}
	P^{\Yboxdim{3pt}\yng(2)}_{2}=\dfrac{1}{2}\left(\,\,\raisebox{-.4\height}{\includegraphics[scale=0.25]{fig/id2.pdf}}\,+\raisebox{-.4\height}{\includegraphics[scale=0.25]{fig/s12.pdf}}\,\,\right)-\,\,\dfrac{1}{\Dim}\,\,\raisebox{-.4\height}{\includegraphics[scale=0.25]{fig/b2.pdf}}\,\,, \hspace{1cm}
	P^{\emptyset}_{2}=\,\dfrac{1}{\Dim}\,\,\raisebox{-.4\height}{\includegraphics[scale=0.25]{fig/b2.pdf}}\,\,.
\end{equation}
Here $P^{\Yboxdim{3pt}\yng(1,1)}_{2}$ projects a tensor on the space of antisymmetric tensor, while $P^{\Yboxdim{3pt}\yng(2)}_{2}$ and $P^{\emptyset}_{2}$ projects a tensor on the space of symmetric traceless tensors and tracefull tensors respectively. As operators on $V^{\otimes n}$ they can be written as 
\begin{equation}
	\begin{aligned}
		&\tensor{{(\mathfrak{P}^{\Yboxdim{3pt}\yng(1,1)}_{2})}}{_{\,\, b_1b_2}^{a_1a_2}}=\dfrac{1}{2}\left(\tensor{\delta}{^{a_1}_{b_1}}\tensor{\delta}{^{a_1}_{b_1}}-\tensor{\delta}{^{a_1}_{b_2}}\tensor{\delta}{^{a_2}_{b_1}}\right)\,,\\
		&\tensor{{(\mathfrak{P}^{\Yboxdim{3pt}\yng(2)}_{2})}}{_{\,\, b_1b_2}^{a_1a_2}}=\dfrac{1}{2}\left(\tensor{\delta}{^{a_1}_{b_1}}\tensor{\delta}{^{a_1}_{b_1}}+\tensor{\delta}{^{a_1}_{b_2}}\tensor{\delta}{^{a_2}_{b_1}}\right) \,-\, \dfrac{1}{\Dim}\,\,\tensor{g}{^{a_1a_2}}\tensor{g}{_{b_1b_2}},\\
		&\tensor{{(\mathfrak{P}^{\emptyset}_{2})}}{_{\,\, b_1b_2}^{a_1a_2}}=\dfrac{1}{\Dim}\,\,\tensor{g}{^{a_1a_2}}\tensor{g}{_{b_1b_2}}\,.
	\end{aligned}
\end{equation}
The Young diagrams $\lbrace\Yboxdim{5.5pt}\yng(2)\,,\,\Yboxdim{5.5pt}\yng(1,1)\,,\,\emptyset\rbrace$ which appear in the above example parametrize the irreducible tensor representations of $\Or(\Dim,\C)$ \cite{Hamermesh:100343,Ful_Har}.\medskip


\paragraph{Isotypic decomposition.} The goal of the second part of this thesis is to provide new formulas for the construction of the projection operators which generalizes \eqref{eq:projection_operator_2} to tensors of arbitrary rank $n$ and arbitrary symmetries. In particular we describe an explicit and systematic way to construct the projection operators for the \textit{isotypic decomposition} of $V^{\otimes n}$ with respect to $\GL(\Dim,\C)$ and with respect to the orthogonal group $\Or(\Dim,\C)$.\medskip

The unique isotypic decomposition of a completely reducible representation of a group is a ‘coarse-grain' decomposition which consists in decomposing a representation into a direct sum of representations, each being a direct sum of all pairwise-equivalent irreducible representations. It may be seen as a structural intermediary step toward the desired irreducible decomposition. For example, an isotypic component of $V^{\otimes n}$ with respect to the action of $\GL(\Dim,\C)$ is parametrized by a Young diagram $\mu$ with $n$ boxes and is the direct sum of all pairwise equivalent irreducible representation of $V^{\otimes n}$ which are parametrized by $\mu$. More concretely, an isotypic tensor with respect to $\GL(\Dim,\C)$ is the sum of the irreducible tensors which are parametrized by standard tableaux of shape $\mu$ and obtained via the Young seminormal projection operators\footnote{If one replaces the Young seminormal idempotents by Young symmetrizers the assertion remains true only for tensors of order $n<5$ \cite{stembridge2011orthogonal}.}. The projection operators in $\C\sn$ realizing the isotypic decomposition of $V^{\otimes n}$ under $\GL(\Dim,\C)$ are the central idempotent of $\C\sn$, denoted $Z^{\mu}$ in this manuscript \cite{Janusz_1966,ceccherini2010representation,doty2019canonical}. As demonstrated in \cite{doty2019canonical}, there is a standard inductive procedure to arrive at an irreducible decomposition of $V^{\otimes n}$ with respect to both $\GL(\Dim,\C)$ and $\Or(\Dim,\C)$, using only isotypic projection operators. As a result, by solving the problem of the isotypic decomposition one actually solves the problem of the irreducible decomposition.\medskip 

In chapter \ref{chap:projectors_GL}, we will review some important aspects of the representation theory of the $\C\sn$ and present a new formula for the central Young idempotents (see Theorem \ref{Theorem:centralYoung}). The formula is constructed from a particular element $T_n\in \C\sn$:\medskip
\begin{equation}
	T_n=\sum_{1\leqslant i \, < \, j\leqslant n}\raisebox{-.3\height}{\includegraphics[scale=0.65]{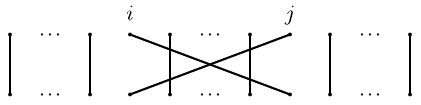}}\,,
\end{equation}
and an \textit{averaged} line induction map $\mathcal{L}\,:\, \C\Sn{n-1}\to \C\Sn{n}$ which is described in section \ref{subsec:inductiveformula}.\medskip

In chapter \ref{chap:projectors_O}, we give a construction for the operators in the Brauer algebra which projects $V^{\otimes n}$ on its various isotypic components with respect to $\Or(1,\Dim-1)$ (see Theorem \ref{theo:non_inductive_traceless_projectors} and Theorem \ref{theo:isotypic_projectors_2}). The obtained formulas are constructed from a particular element $A_n\in \bn(\Dim)$:
\begin{equation}
	A_n=\sum_{1\leqslant i \, < \, j\leqslant n}\raisebox{-.3\height}{\includegraphics[scale=0.65]{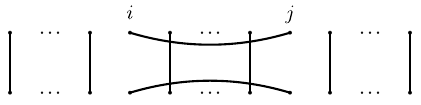}}\,,
\end{equation}
and an \textit{averaged} arc induction map $\mathcal{A}\,:\, \Bn{n-2}\to \bn$ which is described in section \ref{subsection:average_arc}.\medskip

\paragraph{Optimization for tensor calculus with computer algebra system.} In the last chapter, we introduce a method aimed at optimizing the results for application to tensor calculus. In particular, we formulate the known bijection between the  centralizer algebra of $\mathbb{C}\Sn{n}$, $\cn\subset B_{n}(\Dim)$, in terms of ternary bracelets, and propose a technique which allows one to expand the obtained formulas \eqref{eq:ctYoung_traceless5}-\eqref{eq:branching_f_traceless_central_idempotent_class_A25} successively by expressing the multiplication by $T_n$ and $A_n$ as differential operators on the space of ternary bracelets.\smallskip

The computations required for the explicit irreducible decomposition of tensors with respect to $\Or(1,\Dim-1)$ are too heavy to be performed manually, even for tensors of order $4$. This statement also applies to tensor computations in the framework of metric-affine gravity with quadratic curvature Lagrangians. Most of the time, one has to resort to a computer algebra system. The results obtained during this thesis are implemented in the Mathematica language within several packages designed for the irreducible decomposition of tensors with the Brauer algebra \cite{BrauerAlgebraPackage,xBrauerPackage}, and for tensor calculus in metric-affine gravity \cite{xMAGPackage}. A package named \textit{SymmetricFunctions} \cite{SymmetricFunctionsPackage}, allowing to work with symmetric functions and polynomials and containing numerous combinatorial tools used in the irreducible decomposition of tensors, has also been implemented. The packages developed during this thesis have the following inclusion structure:
\begin{equation}
\textit{SymmetricFunctions}\,\subset\,\textit{BrauerAlgebra}\,\subset\,\textit{xBrauer}\subset\, \textit{xMAG},
\end{equation}
and are linked to the \textit{xAct} bundle \cite{xActPackage} for tensor calculus, a widely utilized tool in gravity research. The \textit{SymmetricFunctions} and \textit{BrauerAlgebra} packages can be used independently of \textit{xAct}.\smallskip

	\chapter{Irreducible decompositions of tensors in metric-affine gravity}\label{chap:Irreducible_MAG}

\section{Introduction}\label{sec:introduction}

An irreducible decomposition with respect to $\Or(1,\Dim-1)$ of the Riemann tensor of metric-affine gravity was given by McCrea  \cite{McCrea_1992,Hehl:1994ue} (see also \cite{Nicolas_2006} and the appendix of \cite{jimenez2022metric}) in the language of exterior calculus. This decomposition was subsequently used for the formulation of a general quadratic curvature Lagrangian and has become a fundamental element in the study of quadratic curvature theories in the context of metric-affine gauge gravity \cite{obukhov1997irreducible,heinicke2005einstein,Nicolas_2006,PhysRevD.57.3457,PhysRevD.56.7769,jimenez2022metric,Tresguerres1995,HEHL_1999,Garc_a_2000,TRESGUERRES1995405}.\medskip

The main goal of this chapter is to provide an alternative to the aforementioned decomposition and quadratic Lagrangian. In the first section, we introduce the geometric framework of metric-affine gravity and outline the conventions that will be followed throughout the chapter. We will also recall the geometrical meaning of the notion of projective equivalent connections in the cases where torsion is non zero (see definition \ref{def:proj_tor}).\smallskip 

In section \ref{sec:Schur_Weyl_dualities}, we review the main algebraic tools necessary for the irreducible decompositions of the distortion tensor and of the Riemann tensor of metric-affine gravity. In particular, we recall the important consequences of the Schur-Weyl duality for $\GL(\Dim,\C)$ and the symmetric group $\C\sn$, and for the orthogonal group $\Or(\Dim,\C)$ and the Brauer algebra $\bn(\Dim)$. For our purposes, the main consequence of these dualities  being that the irreducible decomposition of a tensor with respect to $\GL(\Dim,\R)$ and with respect to $\Or(1,\Dim-1)$ can be done systematically using some projection operators in $\C\sn$ and in $\bn(\Dim)$ respectively. We also highlight the fact that for any tensor of order less or equal to five, each of its irreducible components with respect to $\GL(\Dim,\R)$ decomposes uniquely as an irreducible tensor with respect to $\Or(1,\Dim-1)$ (see Lemma \ref{lem:multiplicity_free_GL_O}). 
As a consequence, if the irreducible decomposition of a tensor, of order less or equal to $5$, with respect to $\GL(\Dim,\R)$ is ‘\textit{canonical}' in some sense, then its decomposition with respect to $\Or(1,\Dim-1)$  also exhibits this property. This fact leads to a natural two-step procedure for the $\Or(1,\Dim-1)$-irreducible decomposition of tensors in metric-affine gravity: first perform the decomposition with respect to $\GL(\Dim,\R)$ and then ‘\textit{subtract traces}', that is decompose each $\GL(\Dim,\R)$-irreducible tensors into $\Or(1,\Dim-1)$-irreducible tensors.\medskip

In sections \ref{sec:Distortion_Decomposition} and \ref{sec:Riemann_Decomposition}, we give the detailed procedure leading to a unique $\Or(1,\Dim-1)$-irreducible decomposition of the distortion tensor and of the Riemann tensor. The results are given for arbitrary high dimension $\Dim$ (see Theorem \ref{theo:irreducible_distortion_O} and Theorem \ref{theo:irreducible_Riemann_O}) and for small $\Dim$ (see Proposition \ref{prop:small_d_distortion} and Proposition \ref{prop:small_d_Riemann}).\medskip 

Finally, in the last section we obtain the general Lagrangians at most quadratic in the Riemann tensor and in the distortion tensor, while also establishing the conditions for their projective invariance (see Propositions \ref{prop:Lagrangian_Riemann} and \ref{prop:Lagrangian_distortion}). We will also present the general Lagrangian which reduces to the Einstein-Hilbert Lagrangian in the limit of zero distortion \eqref{eq:lagrangienR_EH}. In  \cite{percacci2020new} they were able to give the conditions for physically viable theories (no ghosts and no tachyons) around Minkwoski background, in the cases of $T\neq 0$, $Q=0$ (non zero torsion, zero non-metricity) and $T=0$, $Q\neq 0$ (zero torsion, non zero non-metricity). We hope that the proposed Lagrangian will help in the generalization of the above result to theories with both torsion and non-metricity. 


\paragraph{Conventions.} Greek letters are used for spacetime indices and Latin letters for moving frame indices. The Greek letters $\mu$, $\rho$, $\nu$, and $\lambda$ are reserved for labeling irreducible representations of $\GL(\Dim,\mathbb{C})$ and of $\Or(\Dim,\C)$, as well as the irreducible representations of their associated discrete algebras $\C\sn$ and $\bn(\Dim)$ within the context of Schur-Weyl duality.

%



\section{The geometry of metric-affine gravity}
We consider a $\Dim$-dimensional smooth manifold $\mathcal{M}$ equipped with a symmetric non-degenerate metric tensor field $g=g_{\alpha\beta} \mathrm{d}x^\alpha\otimes\mathrm{d}x^\beta$. The fundamental Theorem of pseudo-Riemannian geometry states that every pseudo-Riemannian manifold has a metric-preserving torsion-free affine connection. This connection is called the Levi-Civita connection $\accentset{g}\nabla$ and for any vector fields $X$ and $Y$ it satisfies:
\begin{equation}
	\begin{array}{ll}
\accentset{g}{\nabla} g\,=0 \,, \hspace{1cm} &\textit{metric preserving}\,,\\
\accentset{g}{\nabla}_X Y-\accentset{g}{\nabla}_Y X\,= [\, X \,,\, Y\,]\,, \hspace{1cm} &\textit{zero torsion}\,,
	\end{array}
\end{equation}
where $[X,Y]$ is the Lie Bracket of the vector fields $X$ and $Y$.
\begin{remark}[Notations]
The definition of an affine connection (also called linear connection in our context), is given in the appendix \ref{subsec:definitions_geo_diff}.
\end{remark}

\medskip 
\subsection{Connection and distortion tensor in metric-affine gravity}\label{subsec:connection}

In metric-affine gravity the spacetime manifold $\mathcal{M}$ is equipped with a supplementary affine connection $\nabla$. The Levi-Civita connection coefficients and the connection coefficients of $\nabla$ are defined by:

\begin{equation}\label{connection_coef_coord}
\accentset{g}{\nabla}_\alpha (\partial_\beta):=\Bigl \{ \tensor{}{^{\sigma}_{\alpha\beta}}\Bigr \}\,\partial_\sigma\,,\hspace{1cm} \nabla_\alpha (\partial_\beta):=\tensor{\Gamma}{^{\sigma}_{\alpha\beta}}\, \partial_\sigma\,.
\end{equation}

\paragraph{Torsion and non-metricity.} The connection $\nabla$ has both \textit{non-metricity} $Q$ and \textit{torsion} $T$; for any vector fields $X,Y,Z$ one has 
\begin{equation}
Q(Z,X,Y):=\left( \nabla_Z g \right)(X,Y)\,,\hspace{1cm} T(X,Y):=\nabla_{X} Y -\nabla_Y X -[X,Y]\,.
\end{equation}
In a coordinate system the components of non-metricity and torsion are given respectively by 
\begin{equation}\label{torsion}
\tensor{Q}{_{\sigma\alpha\beta}}=\nabla_{\sigma} g_{\alpha\beta}\,,\hspace{0.5cm}	\tensor{T}{^\sigma_{\alpha\beta}}=2\tensor{\Gamma}{^{\sigma}_{[\alpha\beta]}}\,.
\end{equation} 

Here and in what follows we adopt the usual notation, where symmetrized and antisymmetrized indices are grouped within round and square brackets, respectively. In addition, we use a vertical line to denote a break in a group
of (anti-)symmetrized indices, for example $T_{(\mu|\alpha|\nu)} =\frac{1}{2}\left(T_{\mu\alpha\nu} + T_{\nu\alpha\mu}\right)$. The non-metricity tensor has two independent traces, while torsion has only one trace: 
\begin{equation}\label{eq:vectors_nm_torsion}
\begin{aligned}
&Q_\sigma=g^{\alpha\beta}\tensor{Q}{_{\sigma\alpha\beta}}\, \hspace{1cm} \textit{Weyl co-vector}\\
&\tilde{Q}_\alpha=g^{\sigma\beta}\tensor{Q}{_{\sigma\alpha\beta}}\, \\
&\tensor{T}{_\alpha}=\tensor{T}{^\beta_\alpha_\beta}\hspace{1.8cm} \textit{torsion vector}
\end{aligned}
\end{equation}

\paragraph{Distortion tensor.} The distortion tensor characterizes the deviation of the geometry of metric-affine gravity from the purely (pseudo)-Riemaniann one. Its components  $\tensor{C}{^{\sigma}_{\alpha\beta}}$ are defined by the relation
\begin{equation}\label{eq:def_distortion}
	\tensor{C}{^{\sigma}_{\alpha\beta}}:=\tensor{\Gamma}{^{\sigma}_{\alpha\beta}}-\Bigl \{ \tensor{}{^{\sigma}_{\alpha\beta}}\Bigr \}.
\end{equation}
In all generality the distortion tensor has no symmetries with respect to permutations of indices. As a consequence one can define three independent traces: 
\begin{equation}\label{eq:def_traces_distortion}
	\tensor{\accentset{(1)}{C}}{_\alpha}:=\tensor{C}{_\alpha^\beta_\beta}\, , \hspace{0.8cm} \tensor{\accentset{(2)}{C}}{_\alpha}:=\tensor{C}{^\beta_\alpha_\beta}\, ,  \hspace{0.8cm}
	\tensor{\accentset{(3)}{C}}{_\alpha}:=\tensor{C}{^\beta_\beta_\alpha}\, .
\end{equation}
In terms the traces of the non-metricity and torsion tensors one has: 
\begin{equation}
\tensor{\accentset{(1)}{C}}{_\alpha}=\dfrac{1}{2}\, Q_\alpha-\tilde{Q}_\alpha+\tensor{T}{_\alpha}\,,\hspace{0.5cm}\tensor{\accentset{(2)}{C}}{_\alpha}=-\dfrac{1}{2}Q_\alpha \,,\hspace{0.5cm}\tensor{\accentset{(3)}{C}}{_\alpha}=-\dfrac{1}{2}\tilde{Q}_\alpha-\tensor{T}{_\alpha}\,.
\end{equation} 
With these definitions and conventions at hand, the distortion tensor can be expressed in terms of the torsion and non-metricity: 
\begin{equation}
	\tensor{C}{^\sigma_{\alpha\beta}}=\frac{1}{2}\tensor{T}{^\sigma_{\alpha\beta}}+\tensor{T}{_{(\alpha|}^\sigma_{\,|\beta)}}+\frac{1}{2}\tensor{Q}{^\sigma_{\alpha\beta}}-\tensor{Q}{_{(\alpha|}^\sigma_{\,|\beta)}}.
\end{equation}
The above expression makes it clear that the geodesics of a torsion-full connection explicitly depend on torsion. For completeness let us mention that the right-hand side of the previous equation is often rearranged through the definition of the contorsion tensor 
\begin{equation}
\tensor{K }{^\sigma_{\alpha\beta}}:=\frac{1}{2}\tensor{T}{^\sigma_{\alpha\beta}}+\tensor{T}{_{(\alpha|}^\sigma_{\,|\beta)}}\,,
\end{equation}
and of the disformation tensor 
\begin{equation}
\tensor{D}{^\sigma_{\alpha\beta}}:=\frac{1}{2}\tensor{Q}{^\sigma_{\alpha\beta}}-\tensor{Q}{_{(\alpha|}^\sigma_{\,|\beta)}}\,.
\end{equation}
In the restriction of the geometry of metric-affine gravity to the cases of zero non-metricity or zero torsion the distortion tensor reduces to the contorsion tensor or to the disformation tensor respectively.
\subsection{Projective transformations}\label{subsec:Projective}
In this section we recall the geometrical meaning of projective equivalence of connections, in particular in the presence of torsion. The content of the discussion was already familiar to the geometers of the beginning of the $20$th century \cite{thomas1926asymmetric,eisenhart1929non,schouten2013ricci}. Our understanding of this topic was significantly enriched during the summer of 2021 when, with my collaborator Yegor Goncharov, we attended a series of enlightening lectures on Cartan geometry and tractor geometry, delivered by Jordan François and Yannick Herfray at the University of Mons. I would like to express my sincere gratitude to Nicolas Boulanger for making this possible and for his warm hospitality during our visit.

\paragraph{Projective transformations (symmetric case).} Projective transformations of a torsion-free connection is related to the geometry of geodesic paths whose construction was initiated by H. Weyl, L.P. Eisenhart, O.Veblen, J.M. Thomas and T.Y. Thomas \cite{Weyl1921,Eisenhart_1922,Veblen_1922,Veblen_1923,Veblen_1926,yerkes1926projective,whitehead1931representation} (see also \cite{kobayashi1964projective,crampin2007projective,gover2017projectively,Lazzarini:2021kwi} for more modern approaches). A geodesic path (or just geodesic) is a curve $\gamma : \, t \mapsto \gamma(t)$ on $\mathcal{M}$ which is a solution of the differential equations: 
\begin{equation}\label{eq:geodesics}
	\frac{d^2 \gamma^\sigma}{dt^2}+ \tensor{\Gamma}{^\sigma_\alpha_\beta} \frac{d \gamma^\alpha}{dt}\dfrac{d \gamma^\beta}{dt}=0\,.
\end{equation}
Under reparametrization $t\to s(t)$ along the path $\gamma$ the previous equation becomes
\begin{equation}\label{eq:geodesic_param}
	\frac{d^2 \gamma^\sigma}{ds^2}+ \tensor{\Gamma}{^\sigma_\alpha_\beta} \frac{d \gamma^\alpha}{ds}\dfrac{d \gamma^\beta}{ds}=f(s) \dfrac{d \gamma^\sigma}{ds}\,,\quad \text{with}\quad f(s)=-\left(\frac{ds}{dt}\right)^{\shortminus 2}\,\frac{d^2 s}{dt^2}\,.
\end{equation}
Multiplying the previous equation by $\dfrac{d\gamma^\delta}{ds}$ on both sides and eliminating $f(s)$ one obtains the following equation:
\begin{equation}\label{eq:geodesic_unparam}
	\frac{d\gamma^\delta}{ds}\left(\frac{d^2 \gamma^\sigma}{ds^2}+ \tensor{\Gamma}{^\sigma_\alpha_\beta} \frac{d \gamma^\alpha}{ds}\dfrac{d \gamma^\beta}{ds}\right)-\frac{d\gamma^\sigma}{ds}\left(\frac{d^2 \gamma^\delta}{ds^2}+ \tensor{\Gamma}{^\delta_\alpha_\beta} \frac{d \gamma^\alpha}{ds}\dfrac{d \gamma^\beta}{ds}\right)=0\,.
\end{equation}
This equation is the equivalent parameter independent form of \eqref{eq:geodesics} and \eqref{eq:geodesic_param} (see for example \cite{Veblen_1923,berwald2013projective}). 

\begin{definition}\label{def:proj_sym}
Two connections $\nabla$ and $\overline\nabla$ are said to be \textit{symmetric}-projectively\footnote{In the mathematical literature one simply says that the two connections are projectively equivalent. For the Cartan geometry associated with the projective structure (a class of projectively equivalent connection) on the manifold see \cite{kobayashi1964projective}.} equivalent if each $\nabla$-geodesic is, after reparametrization, a $\overline\nabla$-geodesic.
\end{definition}

\begin{proposition} 
Two connections $\nabla$ and $\overline\nabla$ are \textit{symmetric}-projectively equivalent, if and only if their connection coefficients are related by a projective transformation \smallskip
\begin{equation}\label{eq:projective_transformation_sym}
	\tensor{\Gamma}{^\sigma_\alpha_\beta}\to\tensor{\overline{\Gamma}}{^\sigma_\alpha_\beta}= \tensor{\Gamma}{^\sigma_\alpha_\beta} + \tensor{\delta}{^\sigma_{(\alpha}} \tensor{\xi}{_{\beta)}}\,,
\end{equation}
for some vector field $\xi$.
\end{proposition}
\begin{proof}
Let us inquire under which circumstances two connections $\nabla$ and $\overline\nabla$ with connection coefficients $\tensor{\Gamma}{^\sigma_\alpha_\beta}$ and $\tensor{\overline{\Gamma}}{^{\,\sigma}_\alpha_\beta}$ describe the same unparametrized geodesics \eqref{eq:geodesic_unparam}. One has directly
\begin{equation}\label{eq:condition_pro}
	\frac{d\gamma^\rho}{ds}\frac{d \gamma^\alpha}{ds}\dfrac{d \gamma^\beta}{ds}\left(\tensor{\delta}{^{\delta}_\rho}\tensor{\Lambda}{^\sigma_{\alpha\beta}}-\tensor{\delta}{^{\sigma}_\rho}\tensor{\Lambda}{^\delta_{\alpha\beta}}\right)=0\,,
\end{equation}
where
\begin{equation}
	\tensor{\Lambda}{^\delta_{\alpha\beta}}=\tensor{\Gamma}{^\delta_{(\alpha\beta)}}-\tensor{\overline{\Gamma}}{^{\,\delta}_{(\alpha\beta)}}\,,
\end{equation}
are the components of a tensor. Since \eqref{eq:condition_pro} must hold for any unparametrized geodesics, the derivatives $\frac{d \gamma^\sigma}{ds}$ may be chosen arbitrarily, this gives 
\begin{equation}
	\tensor{\delta}{^{\delta}_{(\rho\,|}}\tensor{\Lambda}{^\sigma_{|\,\alpha\beta)}}=\tensor{\delta}{^{\sigma}_{(\rho\,|}}\tensor{\Lambda}{^\delta_{|\,\alpha\beta)}}\,.
\end{equation}
Contracting the indices $\delta$ and $\rho$ yields
\begin{equation}
	\tensor{\Lambda}{^\sigma_{\alpha\beta}}=\dfrac{2}{\Dim+1}\tensor{\delta}{^\sigma_{(\alpha}} \tensor{\Lambda}{^\rho_{\beta)}_\rho}\,.
\end{equation}
\end{proof}

%

%


\paragraph{Projective transformations (asymmetric case).} There is another understanding of projective transformation which is more common in the literature of metric-affine gravity \cite{Garcia-Parrado_2021,JANSSEN2018462,Dadhich2012,olmo_Ricci_based,iosifidis2019scale,Iosifidis_2020,Alfonso_2017,Iosifidis_2019}. This notion does not assume a torsion-free connection and is related to the symmetry of the equations of parallel transport of a vector $X\big{|_{\gamma(0)}}$ along a given smooth curve $\gamma : \, t \mapsto \gamma(t)$ on $\mathcal{M}$:
\begin{equation}
\nabla_{\dot\gamma} X=0\,,
\end{equation}
which in components reads
\begin{equation}\label{eq:parallel_transport_0}
	\frac{d X^{\sigma}}{dt} + \tensor{\Gamma}{^{\sigma}_{\alpha\beta}}\,\frac{d\gamma^{\alpha}}{dt} X^{\beta} = 0\,.
\end{equation}
For any solution $X$ of \eqref{eq:parallel_transport_0} we will say that the vectors $X$ are $\nabla$-parallel along $\gamma$. The vectors $\overline{X}$ along $\gamma$ defined by their components as
\begin{equation}
\overline{X}^\alpha= \lambda(t)X^\alpha\,,
\end{equation}
have the same direction as $X$ along $\gamma$. For any rescaling $\lambda(t)$ there is a function $g(t) = \frac{d \ln \lambda(t)}{dt}$ such that the components of the rescaled vector satisfy 
\begin{equation}\label{eq:parallel_transport_param}
	\frac{d \overline{X}^{\sigma}}{dt} + \tensor{\Gamma}{^{\sigma}_{\alpha\beta}}\,\frac{d\gamma^{\alpha}}{dt} \overline{X}^{\beta} = g(t)\tensor{\overline{X}}{^{\sigma}}\,.
\end{equation}
Dropping the bars and multiplying the previous equation by $X^\delta$ on both sides one obtains the following equation:
\begin{equation}\label{eq:parallel_transport_1}
	X^\delta\left(\frac{d X^{\sigma}}{dt} + \tensor{\Gamma}{^{\sigma}_{\alpha\beta}}\,\frac{d\gamma^\alpha}{dt}X^\beta\right) -X^\sigma\left(\frac{d X^{\delta}}{dt} + \tensor{\Gamma}{^{\delta}_{\alpha\beta}}\frac{d\gamma^\alpha}{dt}X^\beta\right)=0\,.
\end{equation}
This equation is the rescaling independent form of the equation of parallelism \eqref{eq:parallel_transport_0}. That is for any solution $X$ of \eqref{eq:parallel_transport_1} there exist $\overline{X}^\alpha= \lambda(t)X^\alpha$ for some function $\lambda(t)$ such that $\overline{X}$ is a solution of \eqref{eq:parallel_transport_0} (see for example \cite[p. 13]{eisenhart1929non}).\footnote{In the sense of Eisenhart, $\overline{X}$ is also parallel along $\gamma$ \cite[p. 13]{eisenhart1929non}, but we do not use this terminology here.}

\begin{definition}\label{def:proj_tor}
Two connections $\nabla$ and $\overline\nabla$ are said to be \textit{asymmetric}-projectively equivalent\footnote{In the literature of metric-affine gravity one would say that the two connections are projectively equivalent, and at the end we will follow this convention.} if for any curve $\gamma$, each $\nabla$-parallel vectors along $\gamma$ are, after rescaling, $\overline\nabla$-parallel vectors along $\gamma$.
\end{definition}

In other word, a change of \textit{asymmetric}-projectively equivalent connections $\nabla\to\overline\nabla$ preserves the directions of vectors under parallel transport.
\begin{proposition} 
Two connections $\nabla$ and $\overline\nabla$ are \textit{asymmetric}-projectively equivalent, if and only if their connection coefficients are related by an \textit{asymmetric}-projective transformation \smallskip
\begin{equation}\label{eq:projective_transformation}
	\tensor{\Gamma}{^\sigma_\alpha_\beta}\to\tensor{\overline{\Gamma}}{^\sigma_\alpha_\beta}= \tensor{\Gamma}{^\sigma_\alpha_\beta} + \tensor{\delta}{^\sigma_{\beta}} \tensor{\xi}{_{\alpha}}\,,
\end{equation}
for some arbitrary vector field $\xi$.
\end{proposition}
\begin{proof} 
Let us inquire under which circumstances two connections $\nabla$ and $\overline\nabla$ with connection coefficients $\tensor{\Gamma}{^\sigma_\alpha_\beta}$ and $\tensor{\overline{\Gamma}}{^{\,\sigma}_\alpha_\beta}$ have, for any curve, the same equations of parallelism up to rescaling of vectors \eqref{eq:parallel_transport_1}.  One has directly
\begin{equation}\label{eq:condition_pro0}
	\frac{d \gamma^\alpha}{ds}X^\rho X^\beta\left(\tensor{\delta}{^{\delta}_\rho}\tensor{\Lambda}{^\sigma_{\alpha\beta}}-\tensor{\delta}{^{\sigma}_\rho}\tensor{\Lambda}{^\delta_{\alpha\beta}}\right)=0\,,
\end{equation}
where
\begin{equation}
	\tensor{\Lambda}{^\delta_{\alpha\beta}}=\tensor{\Gamma}{^\delta_{\alpha\beta}}-\tensor{\overline{\Gamma}}{^{\,\delta}_{\alpha\beta}}\,,
\end{equation}
are the components of a tensor. Since the equation must hold for any curve and for any $X^\beta$
\begin{equation}
\tensor{\delta}{^{\delta}_{(\rho|}}\tensor{\Lambda}{^\sigma_{|\alpha|\beta)}}-\tensor{\delta}{^{\sigma}_{(\rho|}}\tensor{\Lambda}{^\delta_{|\alpha|\beta)}}=0\,.
\end{equation}
Contracting the indices $\delta$ and $\rho$ yields
\begin{equation}
\tensor{\Lambda}{^\sigma_{\alpha\beta}}=\dfrac{1}{\Dim}\tensor{\delta}{^\sigma_{\alpha}} \,\, \tensor{\Lambda}{^\rho_{\beta}_\rho}\,.
\end{equation}
\end{proof}
%

\vskip 6pt

Note that the \textit{asymmetric}-projective transformations do not preserve the class of torsion-free connections. From \eqref{eq:projective_transformation} it is clear that two symmetric connections which may be \textit{symmetric}-projectively equivalent can not be \textit{asymmetric}-projectively equivalent, while two \textit{asymmetric}-projectively equivalent connections are necessarily \textit{symmetric}-projectively equivalent. In fact, performing the transformation $\eqref{eq:projective_transformation}$ within the unparametrized geodesic equation \eqref{eq:geodesic_unparam} boils down to a \textit{symmetric}-projective transformation. In the sequel we focus exclusively on the bigger \textit{asymmetric}-projective transformations and we call them projective transformations. 

\newpage
\begin{remark}
An interesting problem would be to investigate the Cartan geometry associated with the asymmetric projective structure.
\end{remark}

%
\subsection{Conventions and properties of the Riemann tensor}\label{subsec:Riemann}
The Riemann curvature tensor $\mathcal{R}$ associated with the connection $\nabla$ is defined for any vector field $X,Y,Z$ and any one form field $\omega$ as follows:
\begin{equation}\label{eq:Geo_defRiemann}
	\mathcal{R}(X,Y,Z,\omega):=\omega(\nabla_X\nabla_Y Z -\nabla_Y\nabla_X Z - \nabla_{[X,Y]} Z).
\end{equation}
In a coordinate system, the above definition for the Riemann tensor induces the following expression for its components:
\begin{equation}\label{eq:defRiemann}
	\tensor{\mathcal{R}}{_{\alpha \beta\sigma}^{\gamma}}=\, 2 \, \tensor{\partial}{_{[ \alpha | \,}}\tensor{\Gamma}{^{\gamma}_{|\beta]\,}_{\sigma}}+\, 2\, \tensor{\Gamma}{^{\delta}_{[\alpha|\, \sigma \,}}\tensor{\Gamma}{^{\gamma}_{|\,\beta ]\, \delta \,}}.
\end{equation}
In the general context of metric affine theories of gravity the components of the Riemann tensor have no \textit{à priori} symmetries with respect to permutations of indices except for the obvious antisymmetry on the first two indices. Therefore, one can define three independent tensors of order 2 by contracting a pair of indices. They are the Ricci tensor $\tensor{\accentset{(1)}{\mathcal{R}}}{_{\alpha\beta}}$, the co-Ricci tensor $\tensor{\accentset{(2)}{\mathcal{R}}}{_{\alpha\beta}}$ and the homothetic tensor $\tensor{\accentset{(3)}{\mathcal{R}}}{_{[\alpha\beta]}}$:  
\begin{equation}\label{eq:defTracesRiemann}
	\tensor{\accentset{(1)}{\mathcal{R}}}{_{\alpha\beta}}=\tensor{\mathcal{R}}{_{\alpha\sigma\beta}^\sigma},    \qquad     \tensor{\accentset{(2)}{\mathcal{R}}}{_{\alpha\beta}}=\tensor{\mathcal{R}}{_{\sigma\alpha}^\sigma_\beta},   \qquad   \tensor{\accentset{(3)}{\mathcal{R}}}{_{\alpha\beta}}=\tensor{\mathcal{R}}{_{\alpha\beta\sigma}^\sigma}\, .
\end{equation}
Expanding the homothetic tensor in terms of the distortion tensor one find that it has a very simple expression
\begin{equation}\label{eq:homothetic_tensor}
	\tensor{\accentset{(3)}{\mathcal{R}}}{_{\alpha\beta}}=\, 2 \tensor{\accentset{g}{\nabla}}{_{[\alpha}}\tensor{\accentset{(2)}{C}}{_{\beta]}}\,.
\end{equation} 
The homothetic tensor quantifies how the length of a vector changes when transported along a closed loop. For a volume preserving connection, $\overline{\nabla}_\alpha(\sqrt{|g|})=0$, this tensor is identically zero \cite{eisenhart1929non}. \smallskip

The unique double trace of the Riemann tensor is the Ricci scalar: $\mathcal{R}=\tensor{\accentset{(1)}{\mathcal{R}}}{^{\ a}_{a}}=\tensor{\accentset{(2)}{\mathcal{R}}}{^{\ a}_{a}}$. Finally, we define $\tensor{\accentset{(1)}{\underline{\mathcal{R}}}}{_{\alpha\beta}}$ and $\tensor{\accentset{(2)}{\underline{\mathcal{R}}}}{_{\alpha\beta}}$ as the traceless part of the Ricci and co-Ricci tensor respectively: 
\begin{equation}
	\tensor{\accentset{(1)}{\underline{\mathcal{R}}}}{_{\,\alpha\beta}}\equiv\tensor{\accentset{(1)}{\mathcal{R}}}{_{\alpha\beta}}-\dfrac{1}{\Dim}g_{\alpha\beta}\mathcal{R},\qquad \tensor{\accentset{(2)}{\underline{\mathcal{R}}}}{_{\,\alpha\beta}}\equiv \tensor{\accentset{(2)}{\mathcal{R}}}{_{\alpha\beta}}-\dfrac{1}{\Dim}g_{\alpha\beta}\mathcal{R}\,.
\end{equation}
Here and in what follows underlined tensors are to be understood as totally traceless tensors.\\

Under a projective transformation \eqref{eq:projective_transformation} the Riemann tensor transforms as
\begin{equation}\label{eq:projective_trans_Riemann}
	\tensor{\mathcal{R}}{_{\alpha \beta\sigma}^{\gamma}}\to \tensor{\mathcal{R}}{_{\alpha \beta\sigma}^{\gamma}} + 2\, \tensor{\delta}{^\gamma_\sigma} \tensor{{\accentset{g}{\nabla}}}{_{[\alpha}} \tensor{\xi}{_{\beta]}}.
\end{equation}
The only traces components of the Riemann tensor which transform non-trivially  are the antisymmetric part of the Ricci and co-Ricci tensor, and the homothetic tensor: 
\begin{equation}\label{eq:proj_trans_Riccis}
	\begin{aligned}
		&\tensor{\accentset{(1)}{\mathcal{R}}}{_{[\alpha\beta]}}\to \tensor{\accentset{(1)}{\mathcal{R}}}{_{[\alpha\beta]}} + 2 \,\tensor{{\accentset{g}{\nabla}}}{_{[\alpha}} \tensor{\xi}{_{\beta]}}\, ,\\
		&\tensor{\accentset{(2)}{\mathcal{R}}}{_{[\alpha\beta]}}\to \tensor{\accentset{(2)}{\mathcal{R}}}{_{[\alpha\beta]}} - 2 \,\tensor{{\accentset{g}{\nabla}}}{_{[\alpha}} \tensor{\xi}{_{\beta]}}\, ,\\
		&\tensor{\accentset{(3)}{\mathcal{R}}}{_{\alpha\beta}}\to \tensor{\accentset{(3)}{\mathcal{R}}}{_{\alpha\beta}} + 2 \,\Dim \,\tensor{{\accentset{g}{\nabla}}}{_{[\alpha}} \tensor{\xi}{_{\beta]}} \,. 
	\end{aligned}
\end{equation}

\subsection{Linear frames and coframes}\label{subsec:Frames}

To define the point-wise action of the general linear group $\GL(\Dim,\mathbb{R})$ and of the (pseudo)-orthogonal group $\Or(1,\Dim-1)$ on tensors we use the notion of moving frame $\lbrace \bm{e}_a \st a=1\,,\ldots \Dim \rbrace$ (resp. coframe $\lbrace \bm{\theta}^a\st a=1\,,\ldots \Dim \rbrace$) on an open subset $U\subset\mathcal{M}$. They are vector fields (resp. covector fields) that are linearly independent in the tangent space $T_p U$ (resp. cotangent space of $T^*_p U$) at each point $p\in U$.\smallskip

With respect to a coordinate frame and co-frame on $U$ one as
\begin{equation}\label{eq:def_frames}
	\bm{e}_a:= \tensor{e}{_a^\alpha} \partial_\alpha \hspace{1.4cm} \bm{\theta}^a:=\tensor{e}{^a_\alpha} dx^\alpha, \hspace{1.4cm}\text{(for } a=0,1,\ldots,\Dim-1)
\end{equation}
where the components of the frame in a coordinate system $\tensor{e}{_a^\mu}$ are smooth functions on $U$ and the matrix $\left[\tensor{e}{_a^\mu}(p)\right]$ at a point $p\in U$ is an element of $\GL(\Dim,\mathbb{R})$. Note that $\bm{e}_a$ is also called $\textit{tetrad}$ or $\textit{vielbein}$ and the coefficients $\tensor{e}{_a^\alpha}$ are the tetrad coefficients. 
A moving frame $\bm{e}_a$ is a coordinate frame if and only if it satisfies the integrability conditions 
\begin{equation}\label{eq:holonomic_frame}
	[\bm{e}_a,\bm{e}_b]=0,
\end{equation}
where $[\cdot\,,\cdot\,]$ is the Lie bracket of two vector fields. In this case the moving frame ${\bm{e}_a}$ is said to be holonomic. For nonholonomic frame \eqref{eq:holonomic_frame} is not satisfied and in particular one has 
\begin{equation}\label{eq:nonholonomic_frame}
	[\bm{e}_a,\bm{e}_b]=\tensor{\Omega}{^c_{ab}} \bm{e}_c,
\end{equation} 
where $\tensor{\Omega}{^c_{ab}}=2\tensor{e}{_{[a|}^\alpha}\,\tensor{e}{_{|b]}^\beta}\left(\tensor{\partial}{_{\beta}}\tensor{e}{^c_{\alpha}}\right)$ are the coefficients of anholonomy.\medskip


The connection coefficients in the nonholonomic frame are define similarly to \eqref{connection_coef_coord}: 
	\begin{equation}\label{eq:def_connection_coefficients_spin}
		\nabla_b (\bm{e}_c):=\tensor{\omega}{^{a}_{bc}}\,\bm{e}_a,
	\end{equation} 
From the coordinate expression of the moving frame \eqref{eq:def_frames} and \eqref{eq:def_connection_coefficients_coordinate} we get the following relations between the connection coefficients in a moving frame and in a coordinate frame. 
\begin{equation*}
	\tensor{\omega}{^a_{bc}}=\tensor{\Gamma}{^{\alpha}_{\beta\sigma}}\tensor{e}{^a_\alpha}\tensor{e}{_b^\beta}\tensor{e}{_c^\sigma}-\tensor{e}{_b^\gamma}\tensor{e}{_c^\sigma}\partial_{\sigma}\tensor{e}{^a_\gamma}.
\end{equation*}
Note that this expression can be obtained directly from $\nabla(\bm{e})=0$ where $\bm{e}=\tensor{e}{^a_\mu}\, \bm{e}_a\,\otimes \,\diff x^\mu=\tensor{\delta}{^a_b}\bm{e}_a\,\otimes \bm{\theta}^b$ which is sometimes called the "tetrad postulate". 
The decomposition \eqref{eq:def_distortion} of the spacetime connection $\Gamma$ translates to
\begin{equation}
	\tensor{\omega}{^a_{bc}}=\tensor{\accentset{\circ}{\omega}}{^a_{bc}}+\tensor{C}{^a_{bc}},
\end{equation}
where the connection coefficients of the Levi-Civita connection in a moving frame are defined by the relation  $\tensor{\accentset{g}{\nabla}}{_b}(\tensor{\bm{e}}{_c}):=\tensor{\accentset{\circ}{\omega}}{^a_{bc}}\,\bm{e}_a$. The relation between the connection coefficients of the Levi-Civita connection in a moving frame and in a coordinate frame are given by :
\begin{equation}
\tensor{\accentset{\circ}{\omega}}{^a_{bc}}=\Bigl \{ \tensor{}{^{\alpha}_{\beta\sigma}}\Bigr \}_g\tensor{e}{^a_\alpha}\tensor{e}{_b^\beta}\tensor{e}{_c^\sigma}-\tensor{e}{_b^\gamma}\tensor{e}{_c^\sigma}\partial_{\sigma}\tensor{e}{^a_\gamma}\,.
\end{equation}

For completeness, we introduce the $\nabla$-connection one-form, the $\accentset{g}{\nabla}$-connection one-form and the distortion tensor one form
\begin{equation}
\tensor{\accentset{\circ}{\omega}}{^a_{b}}:=\tensor{\accentset{\circ}{\omega}}{^a_{\sigma b}}\diff x^\sigma\,,\hspace{1cm}	\tensor{\omega}{^a_{b}}:=\tensor{\omega}{^a_{\sigma b}}\diff x^\sigma\,, \hspace{1cm} 
\tensor{C}{^a_b}:=\tensor{C}{^a_{\sigma b}} \diff x^\sigma\,.
\end{equation}
For the symmetric and antisymmetric part of the distortion one form we have:
\begin{equation}
\tensor{C}{^{(ab)}}=-\dfrac{1}{2}\tensor{Q}{_\alpha^{ab}}\, \diff x^{\alpha}\,,\hspace{1cm}\tensor{C}{^{[ab]}}=\left(\tensor{Q}{^{[ab]}_\alpha}+\tensor{K}{^{[a|}_\alpha^{|b]}}\right)\diff x^{\alpha}\,.
\end{equation}
In the context of Cartan geometry and metric-affine gauge theories of gravity one also introduces the non-metricity one form and the torsion $2$-form:
\begin{equation}
\tensor{Q}{^{ab}}=\tensor{Q}{_\alpha^{ab}}\diff x^{\alpha}\,, \hspace{1cm} \tensor{\mathcal{T}}{^a}=\dfrac{1}{2}\tensor{T}{^a_\alpha_\beta}\diff x^{\alpha}\wedge\diff x^{\beta}\,.\footnote{We use a specific notation for the torsion $2$-form in order to avoid ambiguities with the torsion vector defined in \eqref{eq:vectors_nm_torsion}.}
\end{equation}
Note that in this context $Q=\tensor{Q}{^{a}_a}$ is called the Weyl one-form.\medskip

Any tensor field expressed in a local coordinate system can be translated to a $\GL(\Dim,\mathbb{R})$-tensor field via the isomorphisms $\tensor{e}{_a^\alpha}$. In particular the spacetime metric $g_{\alpha\beta}$ can be identified as a $\GL(\Dim,\mathbb{R})$ covariant metric $g_{ab}$ 
\begin{equation}
	g_{ab}=g_{\alpha\beta}\,\tensor{e}{_a^\alpha}\tensor{e}{_b^\beta}.
\end{equation}

\begin{remark}[Connection coefficients in an orthonormal frame]
Let $\lbrace \bm{e}_a, \st a=1,\ldots, \Dim \rbrace $ be a linear frame on $U\subset M$. The metricity condition $\accentset{g}{\nabla}_\mu g_{ab}=0$ leads to
\begin{equation}\label{eq:metricity_local}
\partial_{\sigma} g_{ab}= 2 \,\tensor{\accentset{\circ}{\omega}}{_{(a|\sigma|b)}}\,.
\end{equation}
Similarly, the non-metricity tensor $\nabla_\mu g_{ab}=\tensor{Q}{_\mu_{ab}}$ translates to 
\begin{equation}\label{eq:nonmetricity_local}
	\begin{aligned}
	\tensor{Q}{_{\sigma a b}}&=\partial_{\sigma} g_{ab}-2\,\tensor{\omega}{_{(a|\sigma|b)}}\, ,\\
	\end{aligned}
\end{equation}
Now choose an orthonormal frame on the local subset $U\subset\mathcal{M}$ which we still denote by $\bm{e}_a$. In this case, in term of this orthonormal frame one has: 
\begin{equation}
g(\bm{e}_a(x),\bm{e}_b(x))=\eta_{ab}, \hspace{1cm} \text{for all $x\in U$} 
\end{equation}
where $\eta_{ab}$ is the Minkowski metric. Then the equations \eqref{eq:metricity_local} and \eqref{eq:nonmetricity_local} reduce respectively to 
\begin{equation}
\tensor{\accentset{\circ}{\omega}}{_{(a|\sigma|b)}}=0\,, \hspace{0.5cm} \text{and} \hspace{0.5cm} \,\tensor{\omega}{_{(a|\sigma|b)}}=-\frac{1}{2}\, \tensor{Q}{_\sigma_{ab}}\,.
\end{equation}
\end{remark}
\medskip

\begin{remark} 
The analysis of metric-affine gravity as a gauge theory of gravity is beyond the scope of the present thesis. 
For special account on metric-affine gauge theory of gravity see for example \cite{jimenez2022metric,Hehl:1994ue} and references therein.\smallskip 

As a side remark related to gauge theories let us mention the introduction of the \textit{dressing field method} which clarifies the status of the notion of symmetry breaking in field theory \cite{berghofer2021gauge2,Attard2018}.
\end{remark}

Similarly, the expression of the distortion and Riemann tensor components in the (co)-frame basis are related to their coordinate basis components \eqref{eq:def_distortion}-\eqref{eq:defRiemann} by 
\begin{equation}
	\tensor{C}{^a_{bc}}=\tensor{e}{^a _\alpha}\tensor{e}{_b^\beta}\tensor{e}{_c^\sigma}\tensor{C}{^\alpha_{\beta\sigma}}\,, \hspace{1cm}\text{and}\hspace{1cm} \tensor{\mathcal{R}}{_{abc}^d}=\tensor{e}{_a^\alpha}\,\tensor{e}{_b^\beta}\,\tensor{e}{_c^\sigma}\,\tensor{e}{^d_\gamma}\,\tensor{\mathcal{R}}{_{\alpha\beta\sigma}^\gamma}\,.
\end{equation}
Choosing an orthonormal frame on a open subset $U$, the $\GL(\Dim,\mathbb{R})$-tensor fields translate to $\Or(1,\Dim-1)$-tensor fields. We recall that the main goal of this chapter is to obtain a uniquely defined irreducible decomposition of the tensor fields $\tensor{C}{^a_{bc}}$ and $\tensor{\mathcal{R}}{_{abc}^d}$ with respect to the action of $\Or(1,\Dim-1)$. \medskip

In the next section we present some of the mathematical tools required for a systematic approach to the irreducible decomposition of the tensor product of a complex vector space endowed with a symmetric metric, with respect to the action of both $\GL(\Dim,\C)$ and $\Or(\Dim,\C)$\footnote{We use this approach because the representation theory of complex Lie groups is more systematic and easier.}. In particular, we recall that the irreducible decomposition of tensors with respect to $\GL(\Dim,\C)$ (respectively $\Or(\Dim,\C)$) can be obtained systematically using certain projectors in $\C\sn$ (respectively in the Brauer algebra $B_n(\Dim)$).  This result is applicable to our framework \cite[Theo 5.7. F]{Weyl}, where we consider instead the tensor product of a real vector space endowed with a symmetric metric, acted upon by $\GL(\Dim,\R)$ and $\Or(1,\Dim-1)$\footnote{The main reason for this is because the projectors used for the decomposition have coefficients in $\Q$. 
}.\medskip

Recalling that the pairwise inequivalent irreducible tensor representations of $\GL(\Dim,\R)$ and $\Or(1,\Dim-1)$ are indexed by integer partitions/Young diagrams, the reader interested mainly in applications to the irreducible decomposition of tensors in metric-affine gravity may skip the following section. Reference to the important points of section \ref{sec:Schur_Weyl_dualities} will be made along the discussions in sections \ref{sec:Distortion_Decomposition} and \ref{sec:Riemann_Decomposition}.

\section{Irreducible decompositions of tensors and Schur-Weyl dualities}\label{sec:Schur_Weyl_dualities}
We consider a complex vector space $V$ with basis $\{\bm{e}_1,\ldots,\bm{e}_\Dim\}$ endowed with a non-degenerate symmetric metric $g$. We write $V^{\otimes n}$ for the $n$-fold tensor product of $V$. 

\paragraph{Young diagrams.} The set of pairwise inequivalent irreducible tensor representations of $\GL(\Dim,\C)$ and $\Or(\Dim,\C)$ are indexed by particular integer partitions. An integer partition $\mu=\left(\mu_1,\mu_2,\ldots,\mu_l\right)$ of $n$, denoted $\mu\vdash n$, is a sequence of non increasing positive integers such that $|\mu|=\sum_{i=1}^{l}\mu_i=n$. To each such integer partition we identify a \textit{Young diagram}: an array of $n$ boxes arranged in $l$ left justified rows. For example, the partitions of $n=4$ are: 

\begin{equation*}
	\left(4\right)=\Yboxdim{8pt}\yng(4)\hspace{0.8cm} \left(3,1\right)=\Yboxdim{8pt}\yng(3,1)\hspace{0.8cm} \left(2,2\right)=\Yboxdim{8pt}\yng(2,2)\hspace{0.8cm} \left(2,1,1\right)=\Yboxdim{8pt}\yng(2,1,1)\hspace{0.8cm}\left(1,1,1,1\right)=\Yboxdim{8pt}\yng(1,1,1,1)
\end{equation*}
To any partition $\mu= (\mu_1,\dots,\mu_r)$ of $n$ one constructs the {\it dual partition} $\mu^{\prime} = (\mu_1^{\prime},\dots,\mu^{\prime}_{\mu_1})$ by transposing the corresponding Young diagram. 


\begin{mathematicas}[\textit{Built-in functions}]
	IntegerPartitions[n]\,\\
	\textit{Returns a list of all integer partitions of $n$.}\\
	
	Head[expr]\,\\
	\textit{Returns the head of $expr$. For example, if expr is a list \textup{Head[$expr$]} returns \textup{List} while if expr is an integer \textup{Head[$expr$]} returns \textup{Integer}}.
\end{mathematicas}

\begin{mathematica}[\textit{SymmetricFunctions}]
	TableauForm[$\mu$]\,\\
	\textit{Returns a Young diagram with head \textup{Graphics} associated with the integer partition $\mu$.}
\end{mathematica}

\subsection{Irreducible decompositions of tensors with respect to $\GL(\Dim,\mathbb{C})$.}

\paragraph{Action of $\GL(\Dim,\C)$ on tensors.}The tensor product space $V^{\otimes n}$ is a representation of $\GL(\Dim,\mathbb{C})$ defined by the following canonical left action of any element $\Lambda\in \GL(\Dim,\mathbb{C})$ on tensors $T\in V^{\otimes n}$: 
\begin{equation}\label{eq:GLaction}
	\begin{aligned}
		\Lambda\cdot(T^{a_1\ldots a_n}\bm{e}_{a_1}\otimes\ldots\bm{e}_{a_n}):&=T^{a_1\ldots a_n}\left(\tensor{\Lambda}{^{b_1}_{a_1}}\bm{e}_{b_1}\otimes\ldots\otimes\tensor{\Lambda}{^{b_n}_{a_n}}\bm{e}_{b_n}\right)\,,\\
		&=\tensor{\Lambda}{^{a_1}_{b_1}}\ldots\tensor{\Lambda}{^{a_n}_{b_n}}T^{b_1\ldots b_n}\bm{e}_{a_1}\otimes\ldots\otimes\bm{e}_{a_n}\, ,\\
		&=(\Lambda\cdot T)^{a_1\ldots a_n}\bm{e}_{a_1}\otimes\ldots\otimes\bm{e}_{a_n}.
	\end{aligned}
\end{equation}

\paragraph{Irreducible representations of $\GL(\Dim,\mathbb{C})$ in $V^{\otimes n}$.}
Let $\Par_n(\Dim)$ denote the following set of Young diagrams: 
\begin{equation}\label{eq:P_d}
\Par_n(\Dim)=\lbrace \mu\vdash n \st \mu_1^{\prime}\leqslant \Dim \rbrace\,,
\end{equation}
$\mu_1^{\prime}\leqslant \Dim$ is a constraint on the number of rows in the first column of the diagrams $\mu$.\medskip

Due to classical results of Weyl, the vector space $V^{\otimes n}$ is completely reducible under the action \eqref{eq:GLaction} of $\GL(\Dim,\mathbb{C})$, and its irreducible decomposition is given by
\begin{equation}\label{eq:GL_decomposition}
	V^{\otimes n}\cong\bigoplus_{\mu\in \Par_n(\Dim)}
	 \left(V^{\mu}\right)^{\oplus m_{\mu}}\,,
\end{equation}
where the $V^{\mu}$'s are pairwise inequivalent irreducible tensor representation of $\GL(\Dim,\mathbb{C})$ and $m_{\mu}$ is the \textit{multiplicity} of $V^{\mu}$ in $V^{\otimes n}$. In other words, $m_\mu$ is the number of equivalent\footnote{Equivalent irreducible representations are related to each other by an isomorphism called intertwining operator.} irreducible representations of $\GL(\Dim,\mathbb{C})$, parameterized by the integer partition $\mu$, which appear in $V^{\otimes n}$. For the definition of direct sum see appendix \ref{subsec:definitions}.\medskip 
\newpage
\begin{remark} The irreducible decomposition \eqref{eq:GL_decomposition} is sometimes written 
\begin{equation}\label{eq:GL_decomposition_alt}
	V^{\otimes n}\cong\bigoplus_{\mu\in \Par_n(\Dim)}
	m_{\mu}V^{\mu}\,,
\end{equation}
where $m_\mu V^{\mu}$ is a direct sum of $m_\mu$ copies of $V^\mu$ \textup{(}see for example \cite{ceccherini2010representation} for more details\textup{)}.
\end{remark}

\paragraph{Isotypic components.} The unique $\GL(\Dim,\C)$-isotypic decomposition of $V^{\otimes n}$ may be written as
\begin{equation}
	V^{\otimes n}=\bigoplus_{\mu\in \Par_n(\Dim)}\mathcal{V}^\mu 
\end{equation}
where $\mathcal{V}^\mu$ is the isotypic component of weight $\mu$ (also called the $\mu$-isotypic component)  of $V^{\otimes n}$, that is:
\begin{equation}
\mathcal{V}^\mu \cong\left(V^{\mu}\right)^{\oplus m_\mu}\,.
\end{equation}
The definition of isotypic components and isotypic decomposition of a representation of a finite dimensional algebra\footnote{Note that representations of a group of linear transformations (for example $\GL(\Dim,\C)$ or $\Or(\Dim,\C)$) can be viewed as representations (modules) of the corresponding finite dimensional group algebra (for example $\C \GL(\Dim,\C)$ or $\C\Or(\Dim,\C)$). See for example \cite[p. $178$]{goodman2009symmetry}} is given in appendix \ref{subsec:rep_theo_Algebra} (see also \cite[Def. $1.2.9$, Def. $7.2.7$]{ceccherini2010representation} and \cite[Sec. $4.1.6$]{goodman2009symmetry}.\medskip

For any tensor $T\in V^{\otimes n}$, we write
\begin{equation}\label{eq:isotypic_decomp_GL_Gen}
T=\sum_{\mu\in\Par_n(\Dim)} T^\mu,\hspace{1cm}\text{with}\hspace{1cm} T^\mu\in \mathcal{V}^\mu\,,
\end{equation}
for the unique \textit{isotypic decomposition} of $T$ with respect to $\GL(\Dim,\C)$. The tensor $T^\mu$ is the \textit{isotypic component} of weight $\mu$ of $T$.\smallskip 
\begin{remark}\label{rem:isotypic_component_irreducible}\hphantom{0}\medskip
	\begin{itemize}
		\item[\it i)] The isotypic decomposition of $V^{\otimes n}$ is an orthogonal decomposition with respect to the canonical scalar product induced by any non-degenerate symmetric metric on $V$.
		\item[\it{ii)}] Isotypic components $T^\mu$ with multiplicity $m_\mu=1$ are irreducible tensors. For example, the totally symmetric tensors (parametrized by one row Young diagrams) and totally antisymmetric tensors (parametrized by one column Young diagrams) are at the same time isotypic components and irreducible components.
		\item[\it{iii)}] Isotypic components $T^\mu$ with multiplicity $m_\mu>1$ are not irreducible tensors, and the irreducible decomposition of $T^\mu$ is in this case not unique.\medskip
		
		When all the multiplicities $m_\mu$ of the irreducible representation $V^\mu\subset V^{\otimes n}$ are equal to $1$ the irreducible decomposition is said to be \textit{multiplicity free}. In this case, the isotypic decomposition of a tensor is also the unique irreducible decomposition.
	\end{itemize}
\end{remark}
\begin{example}
	For $n=2$ and $\Dim\geq 2$ any tensor can be decomposed into two irreducible tensors, a symmetric one and an antisymmetric one:
	\begin{equation}
		T=T^{\Yboxdim{4pt}\yng(2)}+T^{\Yboxdim{4pt}\yng(1,1)}\,.
	\end{equation}
	Here the Young diagrams $\Yboxdim{5pt}\yng(2)$ (resp.$\Yboxdim{5pt}\yng(1,1)$) parameterizes the symmetric (resp. antisymmetric) representation, and the irreducible decomposition is unique.
\end{example}   

\paragraph{Isotypic components in matrix form and $\GL(\Dim,\C)$ action again.} Upon decomposing $V^{\otimes n}$ into itsotypic components $\left(V^{\mu}\right)^{\oplus m_\mu}$ one can represent any $T\in V^{\otimes n} $ as follows: 
\begin{equation}\label{eq:iso_column_vec}
\begin{pmatrix}
T^{\mu_1}\\
\vdots\\
T^{\mu_k}
\end{pmatrix}\,,\hspace{1cm}
\text{with}\hspace{1cm} T^{\mu_i}\in \left(V^{\mu_i}\right)^{\oplus m_{\mu_i}}
\end{equation}
where $k=|\mathcal{P}_n(\Dim)|$. Note that upon choosing a basis in each vector space $\left(V^{\mu_i}\right)^{\oplus m_{\mu_i}}$, the elements $T^{\mu_i}$ in \eqref{eq:iso_column_vec} may be represented as column vector of length $m_{\mu_i}\dim(V^{\mu_i})$.\smallskip

The isotypic components $\left(V^{\mu}\right)^{\oplus m_\mu}$ are invariant subspaces with respect to $\GL(\Dim,\C)$ action, hence any operator $\mathrm{O}\in \End(V^{\otimes n})$ generated by the action of $\GL(\Dim,\C)$ on $V^{\otimes n}$ is a linear combination of elements which assume the following diagonal block matrix structure:\medskip

\begin{equation}
\begin{pmatrix}
\,\rho^{\mu_1}(\mathrm{\Lambda}) & \dots & \bigzero \\
\vdots & \ddots & \vdots \\
 \bigzero & \dots & \rho^{\mu_k}(\mathrm{\Lambda})
\end{pmatrix}\,\hspace{0.5cm} \text{where} \hspace{0.5cm} \Lambda\in \GL(\Dim,\C)\,.
\end{equation}
With this notation, each block $\rho^{\mu_i}(\Lambda)$ is an operator in $\GL(\left(V^{\mu_i}\right)^{\oplus m_{\mu_i}})$\,. More precisely the linear map 
\begin{equation}
\rho^{\mu_i} : \GL(\Dim,\C) \to \GL(\left(V^{\mu_i}\right)^{\oplus m_{\mu_i}})\,
\end{equation}
is called a representation\footnote{It obeys $\rho^{\mu_i}(\lambda_1\lambda_2)=\rho^{\mu_i}(\lambda_1)\rho^{\mu_i}(\lambda_2)$ for all $\Lambda_1\,,\Lambda_2\in  \GL(\Dim,\C)$.} of $\GL(\Dim,\C)$, and $\left(V^{\mu_i}\right)^{\oplus m_{\mu_i}}$ is the representation space. Each operator $\rho^{\mu_i}(\mathrm{\Lambda})$ may be seen as a matrix\footnote{These matrices do not belong the full complex matrix algebra, but to a particular subspace which is discussed below.} of size  $\left(m_{\mu_i}\dim(V^{\mu_i})\right)\times \left(m_{\mu_i}\dim(V^{\mu_i})\right)$.

\begin{remark}
	Depending on the context, by representation we either mean the linear map, the vector space or both. In the present thesis, a representation will almost always refer to the representation space.
\end{remark}

\paragraph{Irreducible decomposition of tensors.} The classical Schur-Weyl duality for $\GL(\Dim,\mathbb{C})$ and $\mathfrak{S}_n$ (see appendix \ref{subsec:Double Centralizer Theorem}) implies that the decomposition \eqref{eq:GL_decomposition} can be performed explicitly using certain projectors in the symmetric group algebra $\C\sn$. These projectors are the primitive idempotents in $\C\sn$.\medskip 

In order to be as explicit as possible we use the pairwise orthogonal primitive idempotents in $\C\sn$ \cite{jucys1966young,Murphy,Thrall_seminormalYoung_1941,molev2006fusion} which have various names in the literature (Young's seminormal idempotents \cite{doty2019canonical}, Young's seminormal units \cite{Garsia2020}, Hermitian Young operators \phantomsection\cite{Keppeler_2014}). We will called them the Young seminormal idempotents and denote them by $Y^{\,\tab}$ with $\mathrm{t}$ a standard Young tableau\footnote{A standard tableau is a Young diagram with $n$ boxes filled with integers in $\lbrace 1\,,\ldots, n\rbrace $ in a precise manner which is related to the representation theory of $\sn$ (see section \ref{sec:seminormalYoung} for more details and also \cite{fulton2013representation}).}. A tensor $T\in V^{\otimes n}$ can be decomposed into irreducible $\GL(\Dim,\C)$ tensors as follows,
		\begin{equation}\label{eq:Irreducible_decomposition_GL}
			T=\sum_{\mu\in \Par_n(\Dim)} \,\, \sum_{\scriptstyle \mathrm{t}\in \text{Tab}(\mu)} T^{\,\tab}\,, \hspace{0.5cm}\text{where}\hspace{0.5cm} T^{\,\tab}=Y^{\,\tab}(T)\,,
		\end{equation} 
		and $\text{Tab}(\mu)$ is the set of standard Young tableaux of shape $\mu$. An important fact is that the number of standard tableaux of shape $\mu$ is equal to the multiplicity of the irreducible representation $V^\mu$:
		\begin{equation}
			|\text{Tab}(\mu)|=m_\mu\,.
		\end{equation}
\begin{remark}
If we replace the Young seminormal idempotents by the Young symmetrizers $\mathcal{Y}^{\,\tab}$ \cite[Chapter 4]{fulton2013representation}, the equality \eqref{eq:Irreducible_decomposition_GL} does not hold for $n>4$ . This can be seen from the fact the Young symmetrizers are not pairwise orthogonal (for $n>4$)~\cite{stembridge2011orthogonal}, and hence do not form a partition of unity in $\C\sn$.
\end{remark}
Let $\tab_j(\mu)\in \textup{Tab}(\mu)$ with $j=1,\ldots, m_\mu$. Upon decomposing $V^{\otimes n}$ into irreducible components (with Young seminormal idempotents) one can represent any $T\in V^{\otimes n} $ as follows: 

\begin{equation}
	\begin{pmatrix}
		T^{\tab_1(\mu_1)}\\
		\vdots\\
		T^{\tab_{m_{\mu_1}}(\mu_1)}\\
  		\cmidrule(lr){1-1}
  		\vdots\\
  		\vdots\\
  		\cmidrule(lr){1-1}
		T^{\tab_1(\mu_k)}\\
		\vdots\\
		T^{\tab_{m_{\mu_k}}(\mu_k)}
	\end{pmatrix}\,,\hspace{1cm}
	\text{with}\hspace{1cm} T^{\tab_{j}(\mu_i)}\in V^{\tab_{j}(\mu_i)}
\end{equation}
where $k=|\mathcal{P}_n(\Dim)|$ and in relation with the isotypic decomposition \eqref{eq:iso_column_vec} one has
\begin{equation}
T^{\mu_i}=\sum_{j=1}^{m_{\mu_i}} T^{\tab_{j}(\mu_i)}\,.
\end{equation}

The decomposition being irreducible, any operator $\mathrm{G}\in \End(V^{\otimes n})$ generated by the action of $\GL(\Dim,\C)$ on $V^{\otimes n}$ assumes the following diagonal block matrix structure: 
\begin{equation}
	\begin{pmatrix}
		\begin{bmatrix}
			\,\, G^{\tab_1(\mu_1)} &  & \bigzero \\
			& \ddots&   \\
			\bigzero & & G^{\tab_{m_{\mu_1}}(\mu_1)}
		\end{bmatrix}
		&  & & &\bigzero \\
		& &
		\begin{matrix}
			\ddots &  &  \\
			& \ddots&   \\
			& &\ddots
		\end{matrix}
		& &	\\
		\bigzero &  & &  &
		\begin{bmatrix}
			 G^{\tab_1(\mu_k)} &  & \bigzero \\
			& \ddots&   \\
			\bigzero & &  G^{\tab_{m_{\mu_k}}(\mu_k)}
		\end{bmatrix}
	\end{pmatrix}
\end{equation}
Upon choosing a basis in each vector space $V^{\tab_{j}(\mu_i)}$ one has that each block $G^{\tab_{j}(\mu_i)}$ belongs to the full matrix algebra $M_{n_i}(\C)$ ($n_i\times n_i$ matrices with entry in $\C$), where $n_i=\dim(V^{\mu_i})$. This is a particular illustration of the Artin-Wedderburn Theorem \ref{theo:Wedderburn} which states that any completely reducible $\C$-algebra decomposes as a direct sum of full matrix algebras over $\C$. 
 


\paragraph{Permutation diagrams.}Any permutation $s\in \mathfrak{S}_n$ can be represented by a diagram with $2$ horizontal rows of $n$ vertices. Each vertex in the top row is connected to a vextex in the bottom row by a line.
For example the permutations of $\mathfrak{S}_3$ are given by: 
\begin{equation*}\label{eq:perms3}
	\id=\raisebox{-.4\height}{\includegraphics[scale=0.51]{fig/id3.pdf}}\,, \hspace{0.3cm} (12)=\raisebox{-.4\height}{\includegraphics[scale=0.51]{fig/s3a.pdf}}\,, \hspace{0.3cm} (23)=\raisebox{-.4\height}{\includegraphics[scale=0.55]{fig/s3b.pdf}}\,, \hspace{0.3cm} (13)=\raisebox{-.4\height}{\includegraphics[scale=0.51]{fig/s3c.pdf}}\,, \hspace{0.3cm} (123)=\raisebox{-.4\height}{\includegraphics[scale=0.51]{fig/s3d.pdf}}\,, \hspace{0.3cm} (132)=\hspace{0.3cm} \raisebox{-.4\height}{\includegraphics[scale=0.51]{fig/s3e.pdf}}\,.
\end{equation*}
The product of two permutations $s_1$ and $s_2$ is obtained by placing $s_1$ below $s_2$ and unfolding the path of the resulting lines. For example for $s_1=\raisebox{-.35\height}{\includegraphics[scale=0.35]{fig/s3c.pdf}}$ and $s_2=\raisebox{-.35\height}{\includegraphics[scale=0.35]{fig/s3d.pdf}}$ :
\begin{equation}
	s_1s_2=\begin{aligned}
		&\raisebox{-.4\height}{\includegraphics[scale=0.5]{fig/s3d.pdf}}\\[-6pt]
		&{\vspace{0.4cm}\raisebox{-.4\height}{\includegraphics[scale=0.5]{fig/s3c.pdf}}}
	\end{aligned}=\raisebox{-.4\height}{\includegraphics[scale=0.5]{fig/s3a.pdf}}
\end{equation}
The elements of the group algebra $\mathbb{C}\mathfrak{S}_n$ are the vectors $v=\sum_{s\in \mathfrak{S}_n} v_s s$ with $v_s\in \mathbb{C}$, and the multiplication in the algebra is defined naturally from the multiplication in the group.\medskip 

for any $s\in \mathfrak{S}_n$, the operation $(\cdot)^*$ of flipping a permutation diagram with respect to the middle horizontal line defined in the first chapter correspond to the inverse operation in $\mathfrak{S}_n$: $s^*=s^{\shortminus 1}$. This flip operation is extended to the algebra by linearity, and any element $z\in \mathbb{C}\mathfrak{S}_n $ such that 
\begin{equation}\label{eq:flip_inv_Sn}
	z=z^{*}
\end{equation}
is said to be \textit{flip invariant}.
\paragraph{Action of $\mathbb{C}\mathfrak{S}_n$ on tensors.} The tensor product space $V^{\otimes n}$ is a representation of the symmetric group $\mathfrak{S}_n$ defined by the following canonical right action on tensors: 
\begin{equation}\label{eq:Snaction}
	\begin{aligned}
		(T^{a_1a_2\ldots a_n}\bm{e}_{a_1}\otimes\bm{e}_{a_2}\otimes\ldots\bm{e}_{a_n})\cdot s:&=T^{a_1a_2\ldots a_n}\bm{e}_{a_{s^{\shortminus1}(1)}}\otimes\bm{e}_{a_{s^{\shortminus1}(2)}}\otimes\ldots\otimes\bm{e}_{a_{s^{\shortminus1}(n)}}\,,\\
		&=T^{a_{s(1)}a_{s(2)}\ldots a_{s(n)}}\bm{e}_{a_1}\otimes\bm{e}_{a_2}\otimes\ldots\otimes\bm{e}_{a_n}\, ,\\
		&=(T\cdot s)^{a_1a_2\ldots a_n} \, \bm{e}_{a_1}\otimes\bm{e}_{a_2}\otimes\ldots\bm{e}_{a_n}.
	\end{aligned}
\end{equation}
In terms of permutation diagrams, this action is realized by the following set of instructions: 
\begin{enumerate}
	\item  Write $T$ on the left below the diagram and place the indices $a_1$, $a_2$, ..., $a_n$ on the top row of the diagram.
	\item  Permute the indices according to the lines that connect the top row to the bottom row.
\end{enumerate}
For example, with the permutation diagrams $s_1$ and $s_2$ introduced above one has:
\vspace{0.5cm}
\begin{equation}
	\left(T \cdot s_1\right)^{a_1a_2a_3}=T\hspace{-0.15cm}\begin{matrix}
		\vspace{-1.3cm}\\
		\Scale[0.8]{{\scriptstyle a_1\,a_2\,a_3}}\\
		\raisebox{0.3\height}{\includegraphics[scale=0.25]{fig/s3c.pdf}}
	\end{matrix}=\tensor{T}{^{a_3a_2a_1}}\,,\hspace{3cm}
	\left(T \cdot s_2\right)^{a_1a_2a_3}=T\hspace{-0.15cm}\begin{matrix}
		\vspace{-1.3cm}\\
		\Scale[0.8]{{\scriptstyle a_1\,a_2\,a_3}}\\
		\raisebox{0.3\height}{\includegraphics[scale=0.25]{fig/s3d.pdf}}
	\end{matrix}=\tensor{T}{^{a_3a_1a_2}}\,.
\end{equation}
\vskip 4 pt
\begin{remark}
The Lie group $\GL(\Dim, \C)$ acts identically on all basis factors of $V^{\otimes n}$ while $\sn$ permutes them. Hence, the actions of $\GL(\Dim, \C)$ and of $\sn$ on $V^{\otimes n}$ commutes with each other:
\begin{equation}
	(\Lambda\cdot T)\cdot s = \Lambda\cdot( T\cdot s)\,, \hspace{2cm} \text{for all $\Lambda\in \GL(\Dim,\mathbb{C})$, $s\in \sn$ and $T\in V^{\otimes n}$ }.
\end{equation}
\end{remark}
These two actions have stronger connection, namely the algebras of operators in $\End(V^{\otimes n})$ generated by them are the mutual centralizer of each other \cite[Lem. 3.3]{stevens2016schur}. This is the subject of the Double Centralizer Theorem \cite{stevens2016schur} (see also \cite[Chapter 4]{goodman2000representations}) which in turns implies the Schur-Weyl duality.

\vskip 4 pt
To any permutation $s$ we identify an operator $\mathfrak{r}(s)=\mathfrak{s}\in \text{End}(V^{\otimes n})$ such that $\mathfrak{s}(T)= T\cdot s$, whose components are given by: 
\begin{equation}\label{eq:operator_sn}
	\tensor{\mathfrak{s}}{_{b_1b_2\ldots b_n}^{a_1a_2\ldots a_n}}=  \prod_{\Scale[0.7]{\begin{array}{c}
				{\scriptstyle i \text{ in the top row}}\\[-4pt]
				{\scriptstyle \text{connected to } j \text{ in to bottom row} }
	\end{array}}} \tensor{\delta}{^{a_i}_{b_j}}\,,
\end{equation}
and
\begin{equation}
(T\cdot s)^{a_1\ldots a_n}=\tensor{\mathfrak{s}}{_{b_1\ldots b_n}^{a_1\ldots a_n}}\, T^{b_1\ldots b_n}\,.
\end{equation}
The action of $s\in \mathfrak{S}_n$ on $V^{\otimes n}$ is extended to the group algebra $\mathbb{C}\mathfrak{S}_n$ by linearity.
\begin{remark}\label{rem:injectivity_sn}
	The homomorphism $\mathfrak{r}$ is injective if and only if $\Dim\geqslant n $.
\end{remark}



\paragraph{Schur-Weyl duality.}  As a consequence of the Double Centralizer Theorem, the vector space $V^{\otimes n}$ is completely reducible under the action of $\mathbb{C}\mathfrak{S}_n$, and in particular 
\begin{equation}
	V^{\otimes n}\cong\bigoplus_{\mu\in \Par_n(\Dim)} \left(L^{\mu}\right)^{\oplus g_{\mu}}\,,
\end{equation}
where $L^{\mu}$ are pairwise inequivalent irreducible tensor representation of $\mathbb{C}\mathfrak{S}_n$ and $g_{\mu}$ is the multiplicity of $L^{\mu}$ in $V^{\otimes n}$. 
The Schur-Weyl duality between $\GL(\Dim,\mathbb{C})$ and $\mathfrak{S}_n$ leads to the following important facts (see appendix \ref{subsec:Double Centralizer Theorem} for more details):
\begin{itemize}
	\item The multiplicity of the irreducible representation $L^\mu$ of $\mathbb{C}\mathfrak{S}_n$ in $V^{\otimes n}$ is equal to the dimension of the irreducible representation $V^\mu$ of $\GL(\Dim,\C)$:
	\begin{equation}
		g_{\mu}=\text{dim}(V^{\mu})
	\end{equation}
	\item The multiplicity of the irreducible representation $V^\mu$ of $\GL(\Dim,\mathbb{C})$ in $V^{\otimes n}$ is equal to the dimension of the irreducible representation $L^\mu$ of $\mathbb{C}\mathfrak{S}_n$:
	\begin{equation}
		m_\mu=\dim(L^\mu)
	\end{equation}
	\item Their exist a unique isotypic decomposition of $V^{\otimes n}$ with respect to $\GL(\Dim)$:
	\begin{equation}\label{eq:isotypic_decomposition_Schur_Weyl}
		V^{\otimes n}=\bigoplus_{\mu\in \Par_n(\Dim)}(V^{\otimes n})\cdot Z^{\mu},\hspace{0.5cm} \text{such that}\hspace{0.5cm} (V^{\otimes n})\cdot Z^{\mu}\cong\left(V^{\mu}\right)^{\oplus m_\mu}\,,
	\end{equation}
	where $Z^\mu$ are the central idempotents (projectors) of $\mathbb{C}\mathfrak{S}_n$. Their construction is the subject of chapter \ref{chap:projectors_GL}. 
	\item An irreducible decomposition of $V^{\otimes n}$ with respect to $\GL(\Dim,\mathbb{C})$ is given by: 
	\begin{equation}\label{eq:GL_decomposition_proj}
		V^{\otimes n}=\bigoplus_{\begin{array}{c}
				{\scriptstyle \mu\in \Par_n(\Dim)}\\
				{\scriptstyle \tab\,\in \Tab(\mu)}
		\end{array}}(V^{\otimes n})\cdot  \varepsilon^{\,\tab}\,,\hspace{0.5cm} \text{such that} \hspace{0.5cm}	(V^{\otimes n})\cdot  \varepsilon^{\,\tab}\cong V^{\mu(\,\tab)}\,,
	\end{equation}
	 where $\mu(\,\tab)$ denote the Young diagram corresponding to the shape of the standard tableau $\,\tab$,\footnote{Equivalently $\mu(\,\tab)$ is the last Young diagram in the path $\,\tab$ of the Bratteli diagram associated with the group algebra $\C\sn$.} and $ \varepsilon^{\,\tab}\in \mathbb{C}\mathfrak{S}_n$ are \textit{primitive} idempotents in $\C\sn$. The primitive idempotents $\varepsilon^{\,\tab}$ are either the Young symmetrizer $\mathcal{Y}^{\,\tab}$ or the Young seminormal idempotents $Y^{\,\tab}$.

\end{itemize}

\begin{mathematicas}[\textit{SymmetricFunctions}]
	StandardTableaux[\,$\mu$\,]\,\\
	\textit{Returns the standard tableaux with shape $\mu$.}\bigskip
	
	StandardTableaux[\,$n$\,]\,\\
	\textit{Returns all the standard tableaux associated with the integer partitions of $n$.}\\
	
	TableauForm[\,$\tab$\,]\,\\
	\textit{Returns the standard tableau with head \textup{Graphics} associated with the input standard tableau $\tab$ with head \textup{List}. This function can also be applied to the more evolved combinatorial tableaux called the semistandard Young tableaux which are also implemented in the SymmetricFunctions package.}
\end{mathematicas}

\begin{mathematicas}[\textit{xBrauer}]
	GLIsotypicProject[$\textit{T}\,[inds]$, $\mu$]\footnote{This function is equivalent to the older function GLCentralIrreducibleProject[$\textit{T}\,[inds]$, $\mu$] which can still be used in the new version of the package.}\,\\
	\textit{Projects the tensor $T$ with indices $inds$ onto the isotypic component $(V^{\mu})^{\oplus m_{\mu}}$ of $V^{\otimes n}$with respect to $\GL$ action using the central Young idempotent $Z^\mu$ associated with the integer partition $\mu$.}\\
	
	GLIrreducibleProject[$\textit{T}\,[inds]$, Method $\rightarrow$ SemiNormalYoungUnit, $\tab$]\,\\
	\textit{Projects the tensor $T$ with indices $inds$ onto the irreducible representation $V^{\mu}$ of $V^{\otimes n}$ with respect to $\GL$ action using the Young seminormal idempotent $Y^\tab$ associated with the standard tableau $\tab$ whose shape correspond to the Young diagram $\mu$.}\\
	
	GLIrreducibleProject[$\textit{T}\,[inds]$, Method $\rightarrow$ YoungSymmetrizer, $\tab$]\,\\
	\textit{Projects the tensor $T$ with indices $inds$ onto the irreducible component $V^{\mu}$ of $V^{\otimes n}$ with respect to $\GL$ action using the Young symmetrizer $\mathcal{Y}^\tab$ associated with the standard tableau $\tab$ whose shape correspond to the Young diagram $\mu$.}
\end{mathematicas}


\paragraph{Orthogonality of the irreducible decomposition with respect to $\GL(\Dim,\C)$.} There is a canonical scalar product on $V^{\otimes n}$ which is induced by the metric: 
\begin{equation}\label{eq:scalar_product}
	\begin{aligned}
		\langle \, T \, ,\, S \,\rangle&=\tensor{{T}}{^{a_1\ldots a_n}}\tensor{{S}}{^{b_1\ldots b_n}} \langle \,\bm{e}_{a_1}\otimes\ldots\otimes\bm{e}_{a_n}\,,\,\bm{e}_{b_1}\otimes\ldots\otimes\bm{e}_{b_n}\,\rangle\,,\\
		&:=\tensor{{T}}{^{a_1\ldots a_n}}\tensor{{S}}{^{b_1\ldots b_n}}g(\bm{e}_{a_1},\bm{e}_{b_1})\ldots g(\bm{e}_{a_n},\bm{e}_{b_n})\,,\\
		&=\tensor{{T}}{^{a_1\ldots a_n}}\tensor{{S}}{_{a_1\ldots a_n}}\,.
	\end{aligned}
\hspace{0.5cm}\text{for all $T,S\in V^{\otimes n}$}
\end{equation}
Here we somehow precipitate the introduction of the  scalar product in the sense that we are dealing with the action of $\GL(\Dim,\C)$ on $V^{\otimes n}$. Nevertheless, we do so keeping in mind that the goal is the irreducible decomposition with respect to $\Or(\Dim,\C)$ which is the group preserving the metric and hence the scalar product \eqref{eq:scalar_product}.\medskip

For any $z\in \mathbb{C}\mathfrak{S}_n$, we define the adjoint of $\mathfrak{r}(z)$ by the relation  
\begin{equation}
\langle \, T \, ,\, (\mathfrak{r}(z))(S)\,\rangle=\langle \,(\mathfrak{r}(z))^{\dagger}(T)\, ,\, S\,\rangle\,, \hspace{1cm} \text{for all $T,S\in V^{\otimes n}$.}
\end{equation}
For any $s\in \mathfrak{S}_n$ one has $\tensor{{T}}{^{a_1\ldots a_n}}\tensor{{S}}{_{a_{s(1)}\ldots a_{s(n)}}}=\tensor{{T}}{^{a_{s^{\shortminus 1}(1)}\ldots a_{s^{\shortminus 1}(n)}}}\tensor{{S}}{_{a_1\ldots a_n}}$, hence
\begin{equation}
\langle \, T \, ,\, S\cdot s \,\rangle=\langle \, T \cdot s^{*}\, ,\, S \,\rangle\,.
\end{equation}
where we recall that $(\cdot)^{*}$ is the flip operation on permutation diagrams. Then one has $(\mathfrak{r}(s))^{\dagger}=\mathfrak{r}(s^{*})$ which by linearity yields
\begin{equation}
(\mathfrak{r}(z))^{\dagger}=\mathfrak{r}(z^{*}),\hspace{1cm} \text{for all $z\in \C\sn$,} 
\end{equation}

\begin{definition}\label{def:orthogonal_decomposition}
An irreducible decomposition of $V^{\otimes n}$ with respect to $\GL(\Dim,\C)$
\begin{equation}\label{eq:GL_decomposition_proj2}
	V^{\otimes n}=\bigoplus_{\begin{array}{c}
			{\scriptstyle \mu\in \Par_n(\Dim)}\\
			{\scriptstyle \tab\,\in \Tab(\mu)}
	\end{array}}(V^{\otimes n})\cdot  \varepsilon^{\,\tab}\,,\hspace{0.5cm} \text{such that} \hspace{0.5cm}	(V^{\otimes n})\cdot  \varepsilon^{\,\tab}\cong V^{\mu(\,\tab)}\,,
\end{equation}
is said to be \textit{orthogonal} if it is orthogonal with respect to the scalar product \eqref{eq:scalar_product}. 
\end{definition}
The following lemma is a direct consequence of elementary results of linear algebra which is described by Lemma \ref{lem:orthogonal_direct_sum_decomposition} whose proof is given in the appendix \ref{subsec:definitions}.\medskip
\begin{lemma}\hphantom{a}\smallskip
\begin{itemize}
\item[\textit{i})] If the primitive idempotents $\lbrace\varepsilon^{\,\tab} \st \tab\,\in \Tab(\mu) \hspace{0.3cm}\text{with}\hspace{0.3cm} \mu\in \Par_n(\Dim) \rbrace$ realize an orthogonal decomposition of $V^{\otimes n}$ then their images in $\End(V^{\otimes n})$ form a partition of unity of $\End(V^{\otimes n})$.
\item[\textit{ii})] The primitive idempotents $\lbrace\varepsilon^{\,\tab} \st \tab\,\in \Tab(\mu) \hspace{0.3cm}\text{with}\hspace{0.3cm} \mu\in \Par_n(\Dim) \rbrace$ realize an orthogonal decomposition of $V^{\otimes n}$ if and only if their image in $\End(V^{\otimes n})$ are self-adjoint.
\end{itemize}
\end{lemma}
\vskip 6pt 

For $\Dim\geqslant n$\footnote{In this case the homomorphism $\mathfrak{r}$ is injective.} an element of $\C\sn$ is flip invariant if and only if its image in $\End(V^{\otimes n})$ is a self-adjoint operator. 
\begin{remark}
For $\Dim < n $ if an element of $\C\sn$ is flip invariant then its image in $\End(V^{\otimes n})$ is a self-adjoint operator. One can always add non-flip invariant elements of the kernel of $\mathfrak{r}$ to a flip invariant element of $\sn$ such that the equivalence does not hold anymore.
\end{remark}


Unfortunately, the Young symmetrizers $\mathcal{Y}^{\,\tab}$ are neither pairwise orthogonal (for $n>4$)~\cite{stembridge2011orthogonal} nor flip invariant for ($n>2$). Hence, they do not realize on orthogonal irreducible decomposition of $V^{\otimes n}$. On the contrary, the Young seminormal idempotents $Y^{\,\tab}$ enjoy these two properties for any $n$ (see Lemma \eqref{lem:flip_inv_SN}), and as result they realize an orthogonal irreducible decomposition of $V^{\otimes n}$.
%
%
%

\paragraph{Non-uniqueness of an orthogonal irreducible decomposition.} The Young seminormal idempotents $Y^{\,\tab}$ as constructed via the relation \eqref{eq:Young_seminormal_idempotents} may be seen as a standard choice for a complete set of primitive flip invariant pairwise orthogonal idempotents in $\mathbb{C}\mathfrak{S}_n$. But, let us stress that this choice is not unique. To see that, let us recall that the central Young idempotent $Z^\mu$ is the sum of Young seminormal idempotents $Y^{\,\tab}$ where $\tab$ are a standard tableau of shape $\mu$\footnote{This relation does not hold if one replaces the Young seminormal idempotents $Y^{\,\tab}$ by Young symmetrizers $\mathcal{Y}^{\,\tab}$.}:
	\begin{equation}\label{eq:central_decomp}
		Z^{\mu}=\displaystyle{\sum_{\tab \,\in \,\textup{Tab}({\mu})}}\, Y^{\,\tab}\,.
	\end{equation}
In view of this relation, one has the following freedom for a complete set of primitive pairwise orthogonal idempotent in $\C\sn$ (and hence for the irreducible decomposition with respect to $\GL(\Dim,\C)$):
\begin{equation}\label{eq:freedom_seminormal}
	Z^{\mu}=z\,Z^{\mu}\,z^{\shortminus1}=\displaystyle{\sum_{\,\tab\in \text{Tab}({\mu})}}\, z\, Y^{\,\tab} \,z^{\shortminus1}\,,
\end{equation}
where $z$ is any invertible element of $\mathbb{C}\mathfrak{S}_n$. In order to preserve flip invariance one must have $z^{\shortminus 1}=z^{*}$ (which is true for any $z\in \mathfrak{S}_n$).\medskip


\subsection{Irreducible decompositions of tensors with respect to $\Or(\Dim,\mathbb{C})$.}\label{subsection:Schur_Weyl_O}

\paragraph{Action of $\Or(\Dim,\mathbb{C})$ on tensors.} We consider here the subgroup $\Or(\Dim,\mathbb{C})\subset \GL(\Dim,\mathbb{C})$ of transformations preserving the form of the metric $g_{ab}$. 
More precisely, an invertible transformation $R \in \Or(\Dim,\mathbb{C})$ is such that
\begin{equation}\label{eq:OrthogonalGroup}
	g(R(\bm{e}_a),R(\bm{e}_b))=g(\bm{e}_a,\bm{e}_b)=g_{ab}.
\end{equation} 
The action of $\Or(\Dim,\mathbb{C})$ on $V^{\otimes n}$ is as in \eqref{eq:GLaction}. The property \eqref{eq:OrthogonalGroup} implies that the vector space of tensors generated by the metric $g$ is a trivial representation of $\Or(\Dim,\mathbb{C})$, that is all elements of $\Or(\Dim,\mathbb{C})$ act as the identity on $g$. 
\paragraph{(f+1)-traceless tensors and the trace decomposition problem.} For any tensor $T\in V^{\otimes n}$ we introduce the trace operation $\text{tr}_{ij}\, : \,  V^{\otimes n} \to V^{\otimes n-2}$ define on the components of $T$ as: 
\begin{equation}
	\begin{aligned}
		\text{tr}_{ij}(T^{a_1\ldots a_n})&=g_{a_ia_j}T^{a_1\ldots a_{i}\ldots a_{j}\ldots a_n}\,,\\
		&=\tensor{T}{^{a_1\ldots \, c \,\ldots\,}_{c \,}^{\ldots a_n}}\,.
	\end{aligned}
\end{equation}

A $1$-traceless tensor $T$ (or simply traceless tensor) is such that  $\text{tr}_{ij}(T^{a_1\ldots a_n})=0$ for all pairs $i<j\leqslant n$.
In addition, we define the 2-traceless tensors as the elements of $V^{\otimes n}$  which are such that $\text{tr}_{kl}(\text{tr}_{ij}(T^{a_1\ldots a_n}))=0$  for all pairs $1\leqslant i<j\leqslant n$ and $1\leqslant k<l\leqslant n$. The $(f+1)$-traceless tensors are defined analogously.\medskip


The tensor power of metrics $g^{\otimes f}$ is a trivial representation of $\Or(\Dim,\mathbb{C})$, and as a consequence, the subspace $g^{\otimes f} \otimes V^{\otimes n-2f}\subset V^{\otimes n}$ is an invariant subspace under the action of $\Or(\Dim,\mathbb{C})$. Further invariant-subspace in $g^{\otimes f} \otimes V^{\otimes n-2f}$ are the irreducible representation of $V^{\otimes n-2f}$ which are either traceless subspace or proportional to $g$. This argument allows one to identify an important class of $(f+1)$-traceless tensors: the one which have at least $f$ factors of the metric. In particular, their exists a \textit{unique} reducible decomposition \cite[See the discussion below Theo. 5.6]{weyl1946classical} (see also \cite[Chapter $10$]{Hamermesh:100343}): 
\begin{equation}\label{eq:trace_decomposition}
	V^{\otimes n}=\bigoplus_{f=0,\ldots,\lfloor \frac{n}{2}\rfloor-1}\left( D^{(f)}\right) \oplus \overline{D} \,,
\end{equation}
where $\lfloor x \rfloor$ denotes the integer part of $x$ and the full-trace space $\overline{D}$ is either as the space of scalar invariant when $n$ is even or the space of vectors when $n$ is odd. In this unique decomposition the $(f+1)$-traceless subspace $D^{(f)}$ are constituted of tensors of order $n$ whose components are a linear combinations of products of $f$ metrics with traceless tensors of order $n-2f$. For example, any $2$-traceless tensor $T^{(1)}\in D^{(1)}\subset V^{\otimes n}$ has the form 
\begin{equation}
\begin{aligned}
\hspace{-0.2cm}\tensor{T}{^{(1)}^{\,a_1\ldots a_n}}=\tensor{g}{^{a_1a_2}}\tensor{T{^{(a_1,a_2)}}}{^{\cancel{a}_1\cancel{a}_2a_3\ldots a_n}}+\ldots&+\tensor{g}{^{a_ia_j}}\tensor{T{^{(a_i,a_j)}}}{^{a_1\dots \cancel{a}_i\ldots \cancel{a}_j \ldots a_n}}\\[6pt]
&+\ldots +\tensor{g}{^{a_{n-1}a_{n}}}\tensor{T{^{(a_{n-1},a_{n})}}}{^{a_1\ldots \cancel{a}_{n-1}\cancel{a}_n}},
\end{aligned}
\end{equation}
where each tensor $T{^{(a_i,a_j)}}$ are traceless tensors of order $n-2$ and $\cancel{a}_i$ means that the index $a_i$ is missing. Schematically, any tensor $T^{(f)}\in D^{(f)}$ as the following structure 
\begin{equation}
	g^{\otimes f} \otimes T_{n-2f}\,,
\end{equation}
where $T_{n-2f}$ is a traceless tensor.\medskip

The decomposition of a tensor according to \eqref{eq:trace_decomposition} is sometimes referred to as the trace decomposition problem \cite{Mikes2011}. Note that the projectors \eqref{eq:non_inductive_traceless_projectors_0}-\eqref{eq:main_res_f_traceless_central_idempotent} constructed within the Brauer algebra in chapter \ref{chap:projectors_O} yield a systematic solution \eqref{eq:proj_trace_decomposition} to this problem.\footnote{Similar construction within the Walled-Brauer algebra \cite{bulgakova:tel-02554375} should solve the analogous problem \cite{krupka_trace_decomposition2006,krupka1995trace} in the absence of a metric.}\medskip

\begin{mathematica}[\textit{xBrauer}]
	MetricOfTensor[$\textit{T}$ ]\,\\
	\textit{Returns the metric (say $g$) associated with the tensor $T$. When $T$ is assigned to a metric $g$ its indices are then raised and lowered with $g$.}\\
	\textit{One can assign a metric $g$ to a tensor $T$ at definition time of $T$ using the option}\\ 
	\textit{\textup{Master} $\rightarrow$ $g$ in the command \textup{DefTensor} of the $\textit{xTensor}$ package.}\\
	\textit{One can assign a metric $g$ to a tensor $T$ after definition with the command}\\
	 \textup{MetricOfTensor[$\textit{T}$ ]}$\,\,\widehat{}= g$.
\end{mathematica}

\begin{mathematicas}[\textit{xBrauer}]
	TracelessProject[$\textit{T}\,[inds]$, $f+1$, $g$]\,\\
	\textit{Projects the tensor $\textit{T}$ with indices $inds$ onto the space of  $(f+1)$-traceless tensors $D^{f}$ using metric $g$. If \textup{MetricOfTensor[$\textit{T}$ ]} is non-empty the input $g$ is optional.}\\
	
	TraceProject[$\textit{T}\,[inds]$, $g$]\,\\
	\textit{Projects the tensor $\textit{T}$ with indices $inds$ onto the space of  full trace tensors $D^{\emptyset}$ or $D^{(1)}$ depending on the order of $T$. If \textup{MetricOfTensor[$\textit{T}$ ]} is non-empty the input $g$ is optional.}
\end{mathematicas}

\paragraph{Irreducible representation of $\Or(\Dim,\mathbb{C})$.} Let $\Lambda_n(\Dim)$ denote the following set of Young diagrams: 
\begin{equation}\label{eq:Lambda_d}
	\Lambda_n(\Dim)=\lbrace \lambda\vdash n-2f \,,  \hspace{0.3cm} \textit{with} \hspace{0.3cm}  f=0,\ldots, \lfloor \tfrac{n}{2}\rfloor \hspace{0.2cm} \st \hspace{0.2cm} \lambda_1^{\prime}+\lambda_2^{\prime}\leqslant\Dim\, \rbrace\,,
\end{equation}
where $\lambda_1^{\prime}+\lambda_2^{\prime}\leqslant\Dim$ is a constraint on the number of rows in the first two columns of the diagrams $\lambda$.\medskip

Due to classical result of Weyl \cite[Theo 5.7. A, Theo 5.7. C]{Weyl}, the vector space $V^{\otimes n}$ is completely reducible under the action of $\Or(\Dim,\mathbb{C})$ and in particular:
\begin{equation}\label{eq:O_decomposition}
	V^{\otimes n}\cong\bigoplus_{\lambda \in \Lambda_n(\Dim)} \left(D^{\lambda}\right)^{\oplus m_\lambda}\,,
\end{equation}
where $D^{\lambda}$ are pairwise inequivalent irreducible tensor representation of $\Or(\Dim,\mathbb{C})$ and $m_{\lambda}$ is the multiplicity of $D^{\lambda}$ in $V^{\otimes n}$. The $(f+1)$-traceless representations $D^{(f)}$ and full-trace representation $\overline{D}$ introduced above decompose into irreducible representations as \cite[Theo. 5.7.F, Theo. 5.7.G]{Weyl}
\begin{equation}\label{eq:irreducible_trace_decomposition}
	D^{(f)}\cong\bigoplus_{\begin{array}{c}
			{\scriptstyle \lambda \vdash n-2f}\\
			{\scriptstyle\lambda_1^{\prime}+\lambda_2^{\prime}\leqslant\Dim}
	\end{array}} \left(D^{\lambda}\right)^{\oplus m_\lambda}\,,\hspace{2cm}
	\overline{D}\cong\left\{\begin{array}{ll}
	(D^\emptyset)^{\oplus m_\emptyset}\, \hspace{1cm} &\text{for $n$ even,}\\
	(D^{\Yboxdim{3pt}\yng(1)})^{\oplus m_{\Yboxdim{3pt}\yng(1)}}\, \hspace{1cm} &\text{otherwise,}
	\end{array}\right.
\end{equation}
and hence the irreducible tensors $T\in D^{\lambda}\subset V^{\otimes n}$ with $|\lambda|=n-2f$ are $(f+1)$-traceless tensors. Schematically, any tensor $T\in D^{\lambda}\subset V^{\otimes n}$ with $f_\lambda=\frac{n-|\lambda|}{2}$ has the following structure 
\begin{equation}
g^{\otimes f_\lambda} \otimes T^{(\lambda)}_{n-2f_\lambda}\,
\end{equation}
where $T^{(\lambda)}_{n-2f_\lambda}\in D^{\lambda}\subset V^{\lambda} $ is a traceless irreducible tensor of order $n-2f_\lambda$.
\medskip

To perform the decomposition \eqref{eq:O_decomposition} explicitly and in a systematic manner, one needs to identify the centralizer algebra  of the action of the orthogonal group on $V^{\otimes n}$. This was done in $1937$ by R. Brauer \cite{Brauer_1937} who introduced a diagrammatic algebra which is now known as the Brauer algebra $B_n(\Dim)$. 

\paragraph{Brauer diagrams.}The Brauer algebra $B_n(\Dim)$, is a natural extension of the symmetric group algebra. The vector basis of $B_n(\Dim)$, denoted $\dbn$, are diagrams with $2$ horizontal rows of $n$ vertices, where all vertices are pairwise connected. Recall, that permutation diagrams have each vertex in the top row connected to a vertex in the bottom row and as such are basis vectors of $B_n(\Dim)$ and $\mathbb{C}\mathfrak{S}_n\subset B_n(\Dim)$. The diagrams in $B_n(\Dim)$ are more general and allow for the connection of two vertices in same row. For example, in addition to the permutations \eqref{eq:perms3} the diagrams of $B_3(\Dim)$ are 
\begin{equation*}
	\raisebox{-.4\height}{\includegraphics[scale=0.5]{fig/d3a.pdf}}\,, \hspace{0.5cm} \raisebox{-.4\height}{\includegraphics[scale=0.5]{fig/d3b.pdf}}\,, \hspace{0.5cm} \raisebox{-.4\height}{\includegraphics[scale=0.5]{fig/d3c.pdf}}\,, \hspace{0.5cm} \raisebox{-.4\height}{\includegraphics[scale=0.5]{fig/d3d.pdf}}\,, \hspace{0.5cm} \raisebox{-.4\height}{\includegraphics[scale=0.5]{fig/d3e.pdf}}\,, \hspace{0.5cm} \raisebox{-.4\height}{\includegraphics[scale=0.5]{fig/d3f.pdf}}\,,\hspace{0.5cm}
	\raisebox{-.4\height}{\includegraphics[scale=0.5]{fig/d3g.pdf}}\,,\hspace{0.5cm}  \raisebox{-.4\height}{\includegraphics[scale=0.5]{fig/d3h.pdf}}\,, \hspace{0.5cm} \raisebox{-.4\height}{\includegraphics[scale=0.5]{fig/d3i.pdf}}\,.		
\end{equation*}

The multiplication between two diagrams $b_1$ and $b_2$ in $\bn(\Dim)$ is defined by placing $b_1$ below $b_2$ and unfolding the path of the resulting lines, with the additional rule stating that when $l$ loops appear in the superposition of $b_1$ below $b_2$ the resulting diagrams gets an additional factor $\Dim^l$. For example,
with $b_1=\raisebox{-.4\height}{\includegraphics[scale=0.5]{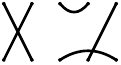}}$ and $b_2=\raisebox{-.4\height}{\includegraphics[scale=0.5]{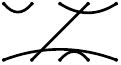}}$ one has: 
\begin{equation}\label{eq:multiplication_Bn}
	b_1\, b_2=\begin{aligned}
		&\raisebox{-.4\height}{\includegraphics[scale=0.9]{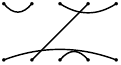}}\\[-8pt]
		&{\vspace{0.4cm}\raisebox{-.4\height}{\includegraphics[scale=0.9]{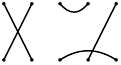}}}
	\end{aligned}=\Dim \,\, \raisebox{-.4\height}{\includegraphics[scale=0.9]{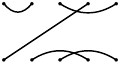}}\,.
\end{equation}
\newpage
\begin{remark}
	The flip operation $(\cdot)^{*}$ defined on $\C\sn$ in the previous section also applies to any element of $\bn(\Dim)$.
\end{remark}

\begin{mathematicas}[\textit{BrauerAlgebra}]\label{math:brauerdiagram}
	BrauerList\,\\
	\textit{\textup{BrauerList} is the head for the basis element of $\bn$. Usage: \textup{BrauerList[$b$]} represents a basis element of $\bn$ where $b$ is a list. The formatting of the input list $b$ is presented in section \ref{sec:parametrization_cn}}.\\
	
	BrauerDiagram[$b$]\,\\
	\textit{Returns the diagram associated with the \textup{BrauerList} $b$. The input $b$ can also be a linear combination of objects with heads \textup{BrauerList}.}\\
	
	BrauerElements[n]\,\\
	\textit{Returns a list of all basis element of $\bn$ with head \textup{BrauerList}.}\\
	
	BrauerProduct[$b_1$, $b_2$]\,\\
	\textit{Returns the product of the diagram $b_1$ with $b_2$, where $b_1$ and $b_2$ have head \textup{BrauerList}. The inputs $b_1$ and $b_2$ can also be linear combinations of objects with heads \textup{BrauerList}.}\\
	
	Flip[$b$]\,\\
	\textit{Flip the Brauer diagram $b$ with head \textup{BrauerList} with respect to the middle horizontal line. The input $b$ can also be a linear combination of objects with heads \textup{BrauerList}.}
\end{mathematicas}

\paragraph{Action of $B_n(\Dim)$ on tensors.}The \textit{right} action $T\cdot b$ of any diagram $b \in B_n(\Dim)$ on a tensor $T=T^{a_1a_2\ldots a_n}\bm{e}_{a_1}\otimes\bm{e}_{a_2}\otimes\ldots\otimes\bm{e}_{a_n}\in V^{\otimes n}$ can be described by the following set of instructions:
\begin{enumerate}
	\item  Write $T$ on the left below the diagram and place the indices $a_1$, $a_2$, ..., $a_n$ on the top row of the diagram.  
	\item  Permute the indices according to the lines that connect the top row to the bottom row.
	\item  Contract pairs of indices joined by arcs in the bottom row.
	\item  Insert $g^{a_ia_j}$ for each arc connecting two vertices ($i,j$) in the top row.
\end{enumerate}
For example, for the diagrams $b_1, b_2$ introduced above  
\vspace{0.5cm}
\begin{equation}
	\begin{aligned}
	&\left(T \cdot b_1\right)^{a_1a_2a_3a_4a_5}=T\hspace{-0.15cm}\begin{matrix}
	\vspace{-1.3cm}\\
	\Scale[0.8]{{\scriptstyle a_1a_2a_3a_4a_5}}\\
	\raisebox{0.3\height}{\includegraphics[scale=0.4]{fig/d1.pdf}}
\end{matrix}=\tensor{g}{^{a_3a_4}}\tensor{T}{^{a_2a_1ba_5}_b}\,,\\[15pt]
&\left(T \cdot b_2\right)^{a_1a_2a_3a_4a_5}=T\hspace{-0.15cm}\begin{matrix}
	\vspace{-1.3cm}\\
	\Scale[0.8]{{\scriptstyle a_1a_2a_3a_4a_5}}\\
	\raisebox{0.3\height}{\includegraphics[scale=0.4]{fig/d2.pdf}}
\end{matrix}=\tensor{g}{^{a_1a_2}}\tensor{g}{^{a_3a_5}}\tensor{T}{^{ba_4c}_{\,cb}}\,.
	\end{aligned}
\end{equation}
To any diagram $b\in B_n(\Dim) $ we identify an operator $\mathfrak{r}(b)=\mathfrak{b}\in \text{End}(V^{\otimes n})$ such that $\mathfrak{b}(T)= T\cdot b$, whose components are given by: 
\begin{equation}\label{eq:operator_bn}
	\tensor{\mathfrak{b}}{_{b_1\ldots b_n}^{a_1\ldots a_n}}= \hspace{-0.5cm}\prod_{\Scale[0.7]{\begin{array}{c}
				{\scriptstyle i \text{ in the top row}}\\[-4pt]
				{\scriptstyle \text{connected to } j \text{ in to bottom row} }
	\end{array}}} \tensor{\delta}{^{a_i}_{b_j}}\,  \prod_{\Scale[0.7]{\begin{array}{c}
				{\scriptstyle i \text{ in the top row}}\\[-4pt]
				{\scriptstyle \text{connected to } j \text{ in to top row} }
	\end{array}}} \tensor{g}{^{a_ia_j}}\, \prod_{\Scale[0.7]{\begin{array}{c}
				{\scriptstyle i \text{ in the bottom row}}\\[-4pt]
				{\scriptstyle \text{connected to } j \text{ in to bottom row} }
	\end{array}}} \tensor{g}{_{b_i}_{b_j}}\,,
\end{equation}
and 
\begin{equation}
(T\cdot b)^{a_1\ldots a_n}=\tensor{\mathfrak{b}}{_{b_1\ldots b_n}^{a_1\ldots a_n}}\, T^{b_1\ldots b_n}.
\end{equation}
The action of a diagram $b$ on $V^{\otimes n}$ is extended to the whole algebra $\bn$ by linearity.
\begin{mathematicas}[\textit{xBrauer}]
	TracePermuteIndices[$T[inds]$, $b$]\,\\
	\textit{Returns the image of the tensor $T$ with indices $inds$ under the action of the Brauer diagram $b$ with head \textup{BrauerList}. The input $b$ can also be a linear combination of elements with head \textup{BrauerList}.}\\
	
	BrauerToTensor[$b$, $inds$]\,\\
	\textit{Returns the image in $\End(V^{\otimes n})$ of the Brauer diagram $b\in \bn$ with head \textup{BrauerList}. The input $inds$ is an optional list of $2n$ indices.}
\end{mathematicas}

\begin{remark}\label{rem:injectivity_bn}
The homomorphism $\mathfrak{r}$ is injective if and only if $\Dim\geqslant n $ \cite{Brauer_1937}.
\end{remark}
\paragraph{Schur-Weyl duality.} As a consequence of the Double Centralizer Theorem, $V^{\otimes n}$ is completely reducible with respect to the action of $B_n$, and in particular 
\begin{equation}\label{eq:decomposition_V_bn}
	V^{\otimes n}\cong \bigoplus_{\lambda\vdash \Lambda_n(\Dim)} \left(M^{\lambda}_n\right)^{\oplus g_{\lambda}}\,,
\end{equation}
where $M^{\lambda}_n$ are pairwise inequivalent irreducible tensor representation of $B_n(\Dim)$ and $g_{\lambda}$ is the multiplicity of $M^{\lambda}_n$ in $V^{\otimes n}$. The Schur-Weyl duality for $\Or(\Dim,\mathbb{C})$ and $B_n(\Dim)$ leads to the following important facts:\\
\begin{itemize}
	\item The multiplicity of the irreducible representation $M^\lambda_n$ of $\bn(\Dim)$ in $V^{\otimes n}$ is equal to the dimension of the irreducible representation $D^\lambda$ of $\Or(\Dim,\mathbb{C})$:
	\begin{equation}
		g_{\lambda}=\text{dim}(D^{\lambda})
	\end{equation}
	\item The multiplicity of the irreducible representation $D^\lambda$ of $\Or(\Dim,\mathbb{C})$ in $V^{\otimes n}$ is equal to the dimension of the irreducible representation $M^\lambda_n$ of $B_n(\Dim)$: 
	\begin{equation}
		m_\lambda=\dim(M^\lambda_n)
	\end{equation}
	\item The \textit{unique} isotypic decomposition of $V^{\otimes n}$ with respect to $\Or(\Dim,\mathbb{C})$ is given by
	\begin{equation}\label{eq:central_decomposition_O}
		V^{\otimes n}=\bigoplus_{\lambda\in \Lambda_n(\Dim)}(V^{\otimes n})\cdot P^{\lambda}_n
	\end{equation}
	where each  $P^{\lambda}_n\in B_n(\Dim)$ is such that 
	\begin{equation}
	(V^{\otimes n})\cdot P^{\lambda}_n\cong\left(D^{\lambda}\right)^{\oplus m_\lambda}\,, \hspace{0.5cm} \text{and} \hspace{0.5cm} (V^{\otimes n})\cdot (P^{\lambda}_n P^{\lambda}_n)=(V^{\otimes n})\cdot P^{\lambda}_n\,.
	\end{equation}
	The construction of these elements is the subject of chapter \ref{chap:projectors_O} (see \eqref{eq:non_inductive_traceless_projectors_0} and \eqref{eq:main_res_f_traceless_central_idempotent} for the formulas). 
	\item An irreducible $\Or(\Dim,\mathbb{C})$ decomposition of $V^{\otimes n}$ can be obtained by the action of certain elements $P^{\tab}_n\in \bn(\Dim)$:\footnote{See the discussion around equation \eqref{eq:irreducible_projector_O} at the end of chapter \ref{chap:projectors_O}.}
	\begin{equation}\label{eq:irreducible_decomposition_O}
		V^{\otimes n}=\bigoplus_{\begin{array}{c}
				{\scriptstyle \tab\in \text{Path}(n,\Dim)}
		\end{array}}(V^{\otimes n})\cdot  P^{\,\tab}_n\,, \hspace{0.5cm} \text{such that}
	\end{equation}
	\begin{equation}
		(V^{\otimes n})\cdot P^{\,\tab}_n\cong D^{\lambda(\,\tab)}\,, \hspace{0.5cm} \text{and} \hspace{0.5cm} (V^{\otimes n})\cdot ( P^{\,\tab}_n P^{\,\tab}_n)=(V^{\otimes n})\cdot P^{\,\tab}_n\,,
	\end{equation}
	 where $\text{Path}(n,\Dim)$ is the set of \textit{permissible}\footnote{A path is permissible if all its Young diagrams are in $\Lambda_n(\Dim)$ \cite{Nazarov}.} paths (also referred to as \textit{up-down tableaux} in \cite{doty2019canonical} or \textit{oscilLating tableaux} in \cite{Isaev_2020}) in the Bratteli diagram associated with $\bn(\Dim)$, and $\lambda(\,\tab)$  is the last Young diagram in the path $\,\tab$ of the Bratteli diagram associated with $\bn(\Dim)$. For the definition of Bratteli diagram see Definition \ref{def:bratteli_diagram}.

\end{itemize}

\begin{mathematicas}[\textit{xBrauer}]
	OIsotypicProject[$\textit{T}\,[inds]$, $\lambda$]\footnote{This function is equivalent to the older function OCentralIrreducibleProject[$\textit{T}\,[inds]$, $\lambda$] which can still be used in the new version of the package.}\,\\
	\textit{Projects the tensor $T$ with indices $inds$ onto the isotypic component $(D^{\lambda})^{\oplus m_{\lambda}}$ of $V^{\otimes n}$ with respect to the (pseudo)-orthogonal group action using elements $P^\lambda_n\in \bn$ where $\lambda$ is an integer partition of $n-2f$.}\\
	
	OIrreducibleProject[$\textit{T}\,[inds]$, $\tab$]\,\\
	\textit{Projects the tensor $T$ with indices $inds$ onto the irreducible representation $D^{\lambda}$ of $V^{\otimes n}$ with respect to the (pseudo)-orthogonal group action using the element $P^{\,\tab}_n\in \bn$  where $\tab$ is a path ending by the Young diagram $\lambda$ at level $n$ in the Bratteli diagram associated with the algebra $\bn$.}
\end{mathematicas}



\paragraph{Orthogonality of an irreducible decomposition and quadratic tensor invariants.} We recall that the canonical scalar product on $V^{\otimes n}$ is defined in \eqref{eq:scalar_product}.\medskip

Similarly to the case of $\C\sn$, for any $z\in \bn$ we define the adjoint operator of $\mathfrak{r}(z)$ by the relation 
\begin{equation}
	\langle \, T \, ,\, (\mathfrak{r}(z))(S)\,\rangle=\langle \,(\mathfrak{r}(z))^{\dagger}(T)\, ,\, S\,\rangle\,, \hspace{1cm} \text{for all $T,S\in V^{\otimes n}$.}
\end{equation}
\vskip 4pt 

Let $b\in \dbn$ be in a single diagram in $\bn$, and let $\mathfrak{b}=\mathfrak{r}(b)$ be its image in $\End(V^{\otimes n})$. Then one has
\begin{equation}
\tensor{{\mathfrak{b}^{\dagger}}}{^{\,a_1\ldots a_n}_{b_1\ldots b_n}}=\tensor{\mathfrak{b}}{_{\,b_1\ldots b_n}^{a_1\ldots a_n}}\,.
\end{equation}
From the definition \eqref{eq:operator_bn} of the homomorphism $\mathfrak{r}$ one concludes that 
\begin{equation}
\left(\mathfrak{r}(b)\right)^{\dagger}=\mathfrak{r}(b^{*})\,,
\end{equation}
and by linearity, for any $z\in \bn$ one has that $\left(\mathfrak{r}(z)\right)^{\dagger}=\mathfrak{r}(z^{*})$. As a consequence when $\Dim\geqslant n$ any $\mathfrak{r}(z)\in \bn$ is self-adjoint if and only if $z$ is flip invariant.\medskip


Consider an orthogonal irreducible decomposition of a tensor $T\in V^{\otimes n}$ with respect to $\Or(\Dim,\C)$:
\begin{equation}\label{eq:irreducible_decomposition_T}
	T=\displaystyle{\sum_{i\in I}T\cdot \accentset{(i)}{P}}=\displaystyle{\sum_{i\in I}\accentset{(i)}{T}},
\end{equation}
where $I$ is the index set parametrizing all the irreducible tensor representations $D^i\subset V^{\otimes n}$ of $\Or(\Dim,\C)$ and $\accentset{(i)}{P}\in \bn(\Dim)$ are the projectors in $\bn(\Dim)$ realizing the decomposition. For any pair of equivalent representations $D^{i}\cong D^{j}$ we will write $i\sim j$. Then, one has the following proposition.

\begin{proposition}\label{prop:lagrangians}
	Let $T\in V^{\otimes n}$ with irreducible decomposition given by \eqref{eq:irreducible_decomposition_T}. For any irreducible representation $i\sim j$, let $b\in \bn$ be such that $b_{ij}=\accentset{(i)}{P} \, b \, \accentset{(j)}{P}\neq 0$. Then the quadratic tensor invariants of $T$ are generated by the following sets
	\begin{itemize}
\item[i)]	$\lbrace \langle \, \accentset{(i)}{T}\,,\, \accentset{(i)}{T}\, \rangle \, \st \, i\in I \, \rbrace\,$\,,
\item[ii)]  $\lbrace\langle \, \accentset{(i)}{T}\cdot b_{ij}\,,\, \accentset{(j)}{T}\, \rangle \st i,j\in I \hspace{0.3cm} \text{with} \hspace{0.3cm} i\sim j \,\rbrace $\,.
	\end{itemize} 
\end{proposition}

 This result follows from from the non-degeneracy of the scalar product \eqref{eq:scalar_product} and from Schur's lemma. The proof is given in appendix \ref{subsec:proof_of_prop_lagrangians}. This proposition will be used in section \ref{sec:Lagrangians} for the construction of quadratic Lagrangians in the distortion tensor and in the Riemann curvature tensor of metric-affine gravity.\medskip


\subsection{Branching rules $\GL(\Dim,\mathbb{C})\downarrow \Or(\Dim,\mathbb{C})$ and the two-step decomposition.}\label{subsection:branching_GL_O}
Irreducible tensor representations of $\GL(\Dim,\mathbb{C})$ are representations of its subgroup $\Or(\Dim,\mathbb{C})$.  Upon restriction to $\Or(\Dim,\mathbb{C})\subset \GL(\Dim,\mathbb{C})$, the irreducible $\GL(\Dim,\C)$ representations decompose into a direct sum of irreducible $\Or(\Dim,\mathbb{C})$ representations
\begin{equation}\label{eq:branching_GL_O}
V^{\mu} \cong \bigoplus_{\begin{array}{c}
		{\scriptstyle \text{even }\nu \, \subset \mu}\\
\end{array}}\big(D^{\lambda}\big)^{\oplus \tensor{\overline{C}}{^{\,\mu}_\lambda_\nu}(\Dim)}\quad\text{(upon $\GL(\Dim,\mathbb{C})\downarrow \Or(\Dim,\mathbb{C})$)}\,,
\end{equation}
where by even Young diagram we mean that each row contains an even number of boxes. The branching rules \eqref{eq:branching_GL_O} are extensively presented in the literature \cite{Koike_Terada_banching,Enright_Willenbring_branching,Kwon_branching,Jang_Kwon_branching}. 
\vskip 6pt
The integer $\tensor{\overline{C}}{^{\,\mu}_\lambda}(\Dim)=\sum_{ \text{even }\nu\,\subset \, \mu} \tensor{\overline{C}}{^{\,\mu}_\lambda_\nu}(\Dim)$ gives the multiplicity of the irreducible representation $D^\lambda$ in $V^\mu$, the branching decomposition \eqref{eq:branching_GL_O} is said to be \textit{multiplicity free} if $\tensor{\overline{C}}{^{\,\mu}_\lambda}(\Dim)\in \{0,1\}$ for all $\lambda\in \Lambda_n(\Dim)$.

\begin{remark}\label{rem:littlewood_restriction_rule}\hphantom{0}\medskip 
\begin{itemize}
\item[\it i)] When $\mu\in\Lambda_{n}(\Dim)$ that is $\mu\vdash n$ and $\mu_1^{\prime}+\mu_2^{\prime}\leqslant\Dim$ one has 
\begin{equation}\label{eq:Littlewood_restriction_rule_1}
	\tensor{\overline{C}}{^{\,\mu}_\lambda_\nu}(\Dim)=  \tensor{C}{^{\,\mu}_\lambda_\nu}\,, \hspace{1cm} \text{(the Littlewood-Richardson coefficients)}
\end{equation}
and \eqref{eq:branching_GL_O} is referred to as the \textit{Littlewood's restriction rules}. Two combinatorial methods for the computation of the Littlewood-Richardson coefficients are presented in the appendix \ref{subsec:Littlewood_Richardson_rules}.
\item[\it ii)] 	The regime $\Dim\geqslant n$ is called the \textit{stable} case. There, all integer partitions $\mu\vdash n$ are both in $\Par_n(\Dim)$ and in $\Lambda_n(\Dim)$ so that the Littlewood's restriction rule applies. Hence $\tensor{\overline{C}}{^{\,\mu}_\lambda_\nu}(\Dim)=\tensor{C}{^{\,\mu}_\lambda_\nu}$ and the multiplicities $\tensor{C}{^\mu_\lambda}$ for fixed $\mu$ can be obtained efficiently using the \textit{jeu de taquin} formulation of the Littlewood-Richardson rule (see appendix \ref{subsec:Littlewood_Richardson_rules}).\medskip
\end{itemize}
\end{remark}

\begin{lemma}\label{lem:multiplicity_free_GL_O}
For all $\Dim\in \mathbb{N}$ and $\mu\vdash n$ with $n\leqslant 5$ the branching decomposition \eqref{eq:branching_GL_O} is multiplicity free.
\end{lemma}
\begin{proof}
The result follows by direct computation of $\tensor{C}{^\mu_\lambda}=\sum_{ \text{even }\nu\,\subset \, \mu} \tensor{C}{^{\,\mu}_\lambda_\nu}$ for all $\mu\vdash n$ with $n\leqslant 5$.
\end{proof}
\begin{mathematica}[\textit{SymmetricFunctions}]
BranchingRule[$\mu$, GeneralLinearGroup[$\Dim$], OrthogonalGroup[$\Dim$]]\,\\
\textit{Returns the integer partitions $\lambda$ with multiplicities corresponding to the irreducible representation of $\Or(\Dim)$ present in the restriction of an irreducible representation of $\GL(\Dim)$ to $\Or(\Dim)$.}
\end{mathematica}

As a consequence of the multiplicity free property of \eqref{eq:branching_GL_O} for $n\leqslant 5$, the restriction of an irreducible representation under $\GL(\Dim,\C)$ to an irreducible representation under $\Or(\Dim,\C)$ is defined uniquely. In this respect, we will perform the explicit decomposition of the distortion and of the Riemann tensor following a two-step approach like it is often done in the higher-spin community \cite{bekaert2021unitary,Fronsdal_1978,Bonelli_2003,LABASTIDA1986101,SIEGEL1987125,Brink_2000}. The transition from an irreducible decomposition under $\GL(\Dim,\mathbb{R})$ to an irreducible decomposition under $\Or(1,\Dim-1)$, will be obtained from the action of the uniquely defined isotypic projection operators $P^\lambda_n\in B_n$ on each $\GL(\Dim,\mathbb{R})$-irreducible tensor. For the remaining of this chapter this procedure will be referred to as the \textit{two-step decomposition}.\medskip 

We will use the following notations in sections \ref{sec:Distortion_Decomposition} and \ref{sec:Riemann_Decomposition}. For $T^{(\overline{\mu})}\in V^{\mu}\subset V^{\otimes n}$ an irreducible tensor with respect to $\GL(\Dim,\mathbb{C})$ of order less than or equal to $5$, we denote by $T^{(\overline{\mu}, \lambda)}\in D^{\lambda}\subset V^\mu$ the irreducible tensor with respect to $\Or(\Dim,\mathbb{C})$ obtained by the action of $P^{\lambda}_n$ on $T^{(\overline{\mu})}$; that is $T^{(\overline{\mu},\lambda)}:=T^{(\overline{\mu})}\cdot P^{\lambda}_n\,$.

\begin{remark}
The last point \eqref{eq:irreducible_decomposition_O} indicates that the irreducible decomposition of a tensor with respect to the action of $\Or(\Dim,\mathbb{C})$ can be systematically achieved using certain projection operators $P^{\,\tab}_n$ which belong to the Brauer algebra. Consequently it is possible to totally bypass the decomposition with respect to $\GL(\Dim,\mathbb{C})$ to arrive at an irreducible decomposition with respect to $\Or(\Dim,\mathbb{C})$. This fact is an illustration of the statement given by Weyl~\cite[p.136]{weyl1946classical}:

\begin{quotation}
	\textit{‘‘ The importance of the full linear group $\GL(\Dim,\mathbb{C})$ lies in the fact that any group $G$ of linear transformations is a subgroup of $\GL(\Dim,\mathbb{C})$ and hence decomposition of tensor space with respect to $\GL(\Dim,\mathbb{C})$ must precede decomposition relative to $G$. One should, however, not overemphasize this relationship; for after all each group stands in its own right and does not deserve to be looked upon merely as a subgroup of something else, be it even Her All-embracing Majesty $\GL(\Dim,\mathbb{C})$."}
	\flushright Hermann Weyl
\end{quotation}
\end{remark}

\section{Irreducible decompositions of the distortion tensor}\label{sec:Distortion_Decomposition}
In this section, unless otherwise specified we assume $\Dim\geqslant 3$. We write $V$ for the tangent space at point $p$ of $U\subset \mathcal{M}$. The explicit tensor computations presented in this section and in section \ref{sec:Riemann_Decomposition} can be performed using the Mathematica package \textit{xMAG} \cite{xMAGPackage} which we recall contains all the other packages (\textit{SymmetricFunctions}, \textit{BrauerAlgebra} and \textit{xBrauer})  developed during this thesis. A Mathematica notebook going through all the steps presented here is available on gitHub \cite{xMAGdistortion}. Therein one can also find the irreducible decompositions of the non-metricity and torsion tensors. 
\subsection{$\GL(\Dim,\mathbb{R})$ isotypic decomposition}\label{subsec:CentralGLDistortion}
The distortion tensor of metric-affine gravity \eqref{eq:def_distortion} has no symmetry of indices and as such is an element of the tensor product $V^{\otimes 3}\,=\, V\otimes V \otimes V=V^{(\Yboxdim{3pt}\yng(1))}\otimes V^{(\Yboxdim{3pt}\yng(1))} \otimes V^{(\Yboxdim{3pt}\yng(1))}$. The Littlewood-Richardson rule (see \ref{subsec:Littlewood_Richardson_rules} for more details) gives the decomposition of the tensor product of two irreducible representations of $\GL(\Dim,\mathbb{R})$: 
\begin{equation}
	V^{\mu}\otimes V^{\rho}= \bigoplus_{\nu\,\in \Par_n(\Dim)} (V^{\nu})^{\oplus \tensor{C}{^\nu_{\mu\rho}}}\,, \hspace{0.5cm} \text{with } \text{$n=|\mu|+|\rho|$}\,,
\end{equation}
where $\tensor{C}{^\nu_{\mu\rho}}$ are again the Littlewood-Richardson coefficients. Iterative application of the Littlewood-Richardson rule yields the following direct sum decomposition of $V^{\otimes 3}$ with respect to $\GL(\Dim,\mathbb{R})$:
\begin{equation}\label{eq:centralGL_decompositionV3}
	V^{\otimes 3}\,=V^{(\Yboxdim{3pt}\yng(3))}\oplus 2 \, V^{\,(\Yboxdim{3pt}\yng(2,1))\,} \oplus V^{(\Yboxdim{3pt}\yng(1,1,1))}.
\end{equation}
\begin{mathematica}[\textit{SymmetricFunctions}]
	LittlewoodRichardsonRule[$\mu_1$, $\ldots$, $\mu_k$]\,\\
	\textit{returns a list of the partitions (with multiplicities) parameterizing to irreducible representations of $\GL(\Dim)$ appearing in the \
		tensor product of the irreducible representations $V^{\mu_1}\otimes\ldots \otimes V^{\mu_k}$.}
\end{mathematica}

One can project an arbitrary tensor onto the isotypic component $m_\mu V^{\mu}$ of $V^{\otimes n}$ under $\GL(\Dim,\mathbb{R})$ using the \textit{uniquely} defined central idempotent $Z^{\mu}\in \mathbb{C}\mathfrak{S}_n$ (see \eqref{eq:isotypic_decomposition_Schur_Weyl}). The isotypic decomposition of the distortion tensor is 
\begin{equation}\label{eq:GLcentral_distortion}
	\tensor{C}{_{abc}}=\sum_{\mu\vdash 3} \tensor{(C\cdot Z^{\mu})}{_{abc}},
\end{equation}
where $Z^{(\Yboxdim{3pt}\yng(3))}$ projects a tensor onto the space of totally symmetric tensors $V^{(\Yboxdim{3pt}\yng(3))}$, $Z^{(\Yboxdim{3pt}\yng(1,1,1))}$ projects onto space of totally antisymmetric tensors $V^{(\Yboxdim{3pt}\yng(1,1,1))}$ and finally $Z^{(\Yboxdim{3pt}\yng(2,1))}$ projects onto the direct sum $V^{\,(\Yboxdim{3pt}\yng(2,1))} \, \oplus \, V^{\,(\Yboxdim{3pt}\yng(2,1))}$. The above central Young idempotents are expressed in the conjugacy class sum basis of the center $\mathcal{Z}_3$ of $\C\Sn{3}$ in equations \eqref{eq:centralYoung_3}.\smallskip

We define the isotypic tensors (see \eqref{eq:isotypic_decomp_GL_Gen}) $\tensor{{\accentset{\left(\Yboxdim{2.5pt}\yng(3)\right)}{C}}}{_{abc}}$, $\tensor{{\accentset{\left(\Yboxdim{2.5pt}\yng(1,1,1)\right)}{C}}}{_{abc}}$ and $\tensor{{\accentset{\left(\Yboxdim{2.5pt}\yng(2,1)\right)}{C}}}{_{abc}}$ as the image of the action of the central Young idempotents on the distortion tensor, namely: 
\begin{equation}
	\begin{aligned}
		&\tensor{{\accentset{\left(\Yboxdim{2.5pt}\yng(3)\right)}{C}}}{_{abc}}:=\tensor{(C\cdot Z^{\Yboxdim{2.5pt}\yng(3)})}{_{abc}}=\tensor{C}{_{(abc)}}\,,\\
		&\tensor{{\accentset{\left(\Yboxdim{2.5pt}\yng(1,1,1)\right)}{C}}}{_{abc}}:=\tensor{(C\cdot Z^{\Yboxdim{2.5pt}\yng(1,1,1)})}{_{abc}}=\tensor{C}{_{[abc]}}\,, \\
		&\tensor{{\accentset{\left(\Yboxdim{2.5pt}\yng(2,1)\right)}{C}}}{_{abc}}:=\tensor{(C\cdot Z^{\Yboxdim{2.5pt}\yng(2,1)})}{_{abc}}=\tensor{C}{_{abc}}-\tensor{{\accentset{\left(\Yboxdim{2.5pt}\yng(3)\right)}{C}}}{_{abc}}-\tensor{{\accentset{\left(\Yboxdim{2.5pt}\yng(1,1,1)\right)}{C}}}{_{abc}}\,. \\
	\end{aligned}
\end{equation}

The component $\tensor{{\accentset{\left(\Yboxdim{2.5pt}\yng(2,1)\right)}{C}}}{_{abc}}$ of the isotypic decomposition of the distortion has no symmetry of indices. Instead it enjoys the following algebraic identities $\tensor{{\accentset{\left(\Yboxdim{2.5pt}\yng(2,1)\right)}{C}}}{_{(abc)}}=0$, and $\tensor{{\accentset{\left(\Yboxdim{2.5pt}\yng(2,1)\right)}{C}}}{_{[abc]}}=0$, which follows directly from the orthogonality property of the direct sum decomposition \eqref{eq:centralGL_decompositionV3}. 

\subsection{$\GL(\Dim,\mathbb{R})$ irreducible decomposition and projective invariance}\label{subsec:GLdistortion}
Totally symmetric, and totally antisymmetric tensors are irreducible. The only part of the decomposition \eqref{eq:GLcentral_distortion} which is not irreducible is $\tensor{{\accentset{\left(\Yboxdim{2.5pt}\yng(2,1)\right)}{C}}}{_{abc}}$. In order to obtain an irreducible decomposition of the distortion tensor one has to decompose $\tensor{{\accentset{\left(\Yboxdim{2.5pt}\yng(2,1)\right)}{C}}}{_{abc}}$ into two irreducible pieces.\medskip


Because the irreducible decomposition of $\tensor{{\accentset{\left(\Yboxdim{2.5pt}\yng(2,1)\right)}{C}}}{_{abc}}$ is not unique we require a reasonable guideline to fix a preferred irreducible decomposition of the distortion tensor. Such a guideline can be  provided by the geometrical notion of projective invariance (see section \ref{subsec:Projective}).\medskip 

Note that the isotypic tensor $\tensor{{\accentset{\left(\Yboxdim{2.5pt}\yng(2,1)\right)}{C}}}{_{abc}}$ transforms non trivially under projective change of connection\footnote{Note also that $\tensor{{\accentset{\left(\Yboxdim{2.5pt}\yng(1,1,1)\right)}{C}}}{_{abc}}$ is projective invariant, while $\tensor{{\accentset{\left(\Yboxdim{2.5pt}\yng(3)\right)}{C}}}{_{abc}}$ is not.}:
\begin{equation*}
	\tensor{{\accentset{\left(\Yboxdim{2.5pt}\yng(2,1)\right)}{C}}}{_{abc}}\to \tensor{{\accentset{\left(\Yboxdim{2.5pt}\yng(2,1)\right)}{C}}}{_{abc}} + \tensor{{\delta_\xi \accentset{\left(\Yboxdim{2.5pt}\yng(2,1)\right)}{C}}}{_{abc}}, \hspace{1cm} \text{with} \hspace{1cm} \tensor{{\delta_\xi \accentset{\left(\Yboxdim{2.5pt}\yng(2,1)\right)}{C}}}{_{abc}}=\frac{2}{3}\left(g_{ac}\xi_{b}-g_{b(c}\xi_{a)}\right)\,.
\end{equation*}
In order to settle the decomposition of $\tensor{{\accentset{\left(\Yboxdim{2.5pt}\yng(2,1)\right)}{C}}}{_{abc}}$ into two irreducible tensors we require one of them to be projective invariant. We are looking for elements $Y_{S},Y_{A}\in \mathbb{C}\mathfrak{S}_3$ such that $\accentset{\left(\Yboxdim{3pt}\yng(2,1)\right)}{C}= \accentset{\left(\Yboxdim{3pt}\yng(2,1)\right)}{C}\cdot Y_{S}+\accentset{\left(\Yboxdim{3pt}\yng(2,1)\right)}{C}\cdot Y_{A}$, with $\delta_\xi (\accentset{\left(\Yboxdim{3pt}\yng(2,1)\right)}{C}\cdot Y_{S})\neq 0$ and  $\delta_\xi (\accentset{\left(\Yboxdim{3pt}\yng(2,1)\right)}{C}\cdot Y_{A})= 0$. That is, $Y_{S}$ and $Y_{A}$ decompose $\accentset{\left(\Yboxdim{3pt}\yng(2,1)\right)}{C}$ into a projective variant and projective invariant part.  On the other hand, we want $Y_{S}$ and $Y_{A}$ to realize an orthogonal irreducible decomposition of $\accentset{\left(\Yboxdim{3pt}\yng(2,1)\right)}{C}$.\medskip

From the symmetry of $\tensor{{\delta_\xi\accentset{\left(\Yboxdim{2.5pt}\yng(2,1)\right)}{C}}}{_{abc}}$
\begin{equation*}
	\tensor{{\delta_\xi\accentset{\left(\Yboxdim{2.5pt}\yng(2,1)\right)}{C}}}{_{(a|b|c)}}=\tensor{{\delta_\xi\accentset{\left(\Yboxdim{2.5pt}\yng(2,1)\right)}{C}}}{_{abc}}\,,
\end{equation*}
we have a natural candidate for the sought-after decomposition
\begin{equation}\label{eq:irre_decomp_C21}
	\tensor{{\accentset{\left(\Yboxdim{3pt}\yng(2,1)\right)}{C}}}{_{abc}}=\tensor{{\accentset{\left(\Yboxdim{3pt}\yng(2,1)S\right)}{C}}}{_{abc}}+\tensor{{\accentset{\left(\Yboxdim{3pt}\yng(2,1)A\right)}{C}}}{_{abc}}\,, 	\hspace{1cm} \text{with} \hspace{1cm} \tensor{{\accentset{\left(\Yboxdim{3pt}\yng(2,1)S\right)}{C}}}{_{abc}}:=\tensor{{\accentset{\left(\Yboxdim{3pt}\yng(2,1)\right)}{C}}}{_{(a|b|c)}} \hspace{0.5cm} \text{and} \hspace{0.5cm} \,\tensor{{\accentset{\left(\Yboxdim{3pt}\yng(2,1)A\right)}{C}}}{_{abc}}:=\tensor{{\accentset{\left(\Yboxdim{3pt}\yng(2,1)\right)}{C}}}{_{[a|b|c]}}\,.
\end{equation}
The tensors $\tensor{{\accentset{\left(\Yboxdim{3pt}\yng(2,1)S\right)}{C}}}{_{abc}}$ and $\tensor{{\accentset{\left(\Yboxdim{3pt}\yng(2,1)A\right)}{C}}}{_{abc}}$ are clearly non-zero and orthogonal with respect to the scalar product \eqref{eq:scalar_product}. They are both irreducible tensors as they belong respectively to two complementary $\GL(\Dim,\mathbb{R})$-invariant subspaces in $V^{\Yboxdim{3pt}\yng(2,1)}\oplus V^{\Yboxdim{3pt}\yng(2,1)}$.\medskip 

This decomposition is realized by the action on $\accentset{\left(\Yboxdim{3pt}\yng(2,1)\right)}{C}$ of the following operators
\begin{equation}
	Y_{S}=s\, Z^{\Yboxdim{3.5pt}\yng(2)}\, s^{\shortminus 1}\,, \hspace{0.5cm} Y_{A}=s\, Z^{\Yboxdim{3.5pt}\yng(1,1)}\, s^{\shortminus 1}\hspace{1cm}\text{with}\hspace{1cm} s=\raisebox{-.4\height}{\includegraphics[scale=0.4]{fig/s3b.pdf}}\,.
\end{equation}
The irreducible tensors $\accentset{\left(\Yboxdim{3pt}\yng(2,1)S\right)}{C}$ (resp. $\accentset{\left(\Yboxdim{3pt}\yng(2,1)A\right)}{C}$\,) can also be obtained by the action on the distortion tensor of the Young seminormal idempotent $Y^{\,\Yboxdim{3.5pt}\yng(2,1)S}$ (resp. $Y^{\,\Yboxdim{3.5pt}\yng(2,1)A}$\,) defined by 
\begin{equation}
	Y^{\,\Yboxdim{3.5pt}\yng(2,1)S}:=s \, Y^{\,\Scale[0.5]{\young(12,3)}} \, s^{\shortminus 1}\, \hspace{0.5cm}\text{and}\hspace{0.5cm}  Y^{\,\Yboxdim{3.5pt}\yng(2,1)A}:=s \, Y^{\,\Scale[0.5]{\young(13,2)}} \, s^{\shortminus 1}\,, \hspace{1cm}\text{with}\hspace{1cm} s=\raisebox{-.4\height}{\includegraphics[scale=0.4]{fig/s3b.pdf}}\,.
\end{equation}
The \textit{standard} Young seminormal idempotents $Y^{\,\Scale[0.5]{\young(12,3)}}$ and $Y^{\,\Scale[0.5]{\young(13,2)}}$ are expressed in term of the central Young idempotents as (see formula \eqref{eq:Young_seminormal_idempotents_intro} and section \ref{sec:seminormalYoung} for more details):
\begin{equation}
Y^{\,\Scale[0.5]{\young(12,3)}}=Z^{\Yboxdim{4pt}\yng(2)}Z^{\Yboxdim{4pt}\yng(2,1)}\,\hspace{1cm} Y^{\,\Scale[0.5]{\young(13,2)}}=Z^{\Yboxdim{4pt}\yng(1,1)}Z^{\Yboxdim{4pt}\yng(2,1)}\,.
\end{equation}
Note that because $Z^{\Yboxdim{4pt}\yng(2,1)}$ commutes with $\sn$, we have
\begin{equation}
Y^{\,\Yboxdim{3.5pt}\yng(2,1)S}=Y_{S}\, Z^{\Yboxdim{3pt}\yng(2,1)}\,, \hspace{0.5cm} Y^{\,\Yboxdim{3.5pt}\yng(2,1)A}=Y_{A}\, Z^{\Yboxdim{3pt}\yng(2,1)}\,.
\end{equation}

\begin{proposition}\label{prop:irreducible_distortion_GL}
	The decomposition \eqref{eq:irre_decomp_C21} is the unique orthogonal irreducible decomposition of $\tensor{{\accentset{\left(\Yboxdim{2.5pt}\yng(2,1)\right)}{C}}}{_{abc}}$ with respect to $\GL(\Dim,\mathbb{R})$ which maximizes the number of projective invariant tensors within its parts. 
\end{proposition}
See appendix \ref{subsec:proof_of_prop_irreducible_distortion_GL} for the proof, while we recall that the definition of an orthogonal irreducible decomposition is given in \eqref{def:orthogonal_decomposition}.\medskip

Summarizing, we have the following corollary for the irreducible decomposition of the distortion with respect to $\GL(\Dim,\mathbb{R})$.
\begin{corollary}\label{cor:GLdecomposition_Distortion}
	The unique orthogonal irreducible decomposition of the distortion tensor with respect to $\GL(\Dim,\mathbb{R})$ maximizing the number of projective invariant tensors within its parts is:
	\begin{equation}
		\tensor{C}{_{abc}}=\tensor{{\accentset{\left(\Yboxdim{2.5pt}\yng(3)\right)}{C}}}{_{abc}}+\tensor{{\accentset{\left(\Yboxdim{2.5pt}\yng(2,1)S\right)}{C}}}{_{abc}}+\tensor{{\accentset{\left(\Yboxdim{2.5pt}\yng(2,1)A\right)}{C}}}{_{abc}}+\tensor{{\accentset{\left(\Yboxdim{2.5pt}\yng(1,1,1)\right)}{C}}}{_{abc}}\,,
	\end{equation}
	with
	\begin{equation}
		\tensor{{\accentset{\left(\Yboxdim{2.5pt}\yng(3)\right)}{C}}}{_{abc}}=\tensor{C}{_{(abc)}}\,,\hspace{0.5cm}
		\tensor{{\accentset{\left(\Yboxdim{2.5pt}\yng(1,1,1)\right)}{C}}}{_{abc}}=\tensor{C}{_{[abc]}}\,,\hspace{0.5cm}
		\tensor{{\accentset{\left(\Yboxdim{2.5pt}\yng(2,1)S\right)}{C}}}{_{abc}}=\tensor{C}{_{(a|b|c)}}-\tensor{{\accentset{\left(\Yboxdim{2.5pt}\yng(3)\right)}{C}}}{_{abc}}\,,\hspace{0.5cm}
		\tensor{{\accentset{\left(\Yboxdim{2.5pt}\yng(2,1)A\right)}{C}}}{_{abc}}=\tensor{C}{_{[a|b|c]}}-\tensor{{\accentset{\left(\Yboxdim{2.5pt}\yng(1,1,1)\right)}{C}}}{_{abc}}\,.\hspace{0.5cm}
	\end{equation}
	The projective invariant tensors of the decomposition are $\tensor{{\accentset{\left(\Yboxdim{2.5pt}\yng(2,1)A\right)}{C}}}{_{abc}}$ and $\tensor{{\accentset{\left(\Yboxdim{2.5pt}\yng(1,1,1)\right)}{C}}}{_{abc}}\,$.
\end{corollary}
\subsection{A unique $\Or(1,\Dim-1)$ irreducible decomposition}\label{subsec:Odistortion}
Having performed an orthogonal $\GL(\Dim,\mathbb{R})$ irreducible decomposition of the distortion tensor as a first step we are know in position to precise this decomposition to the subgroup $\Or(1,\Dim-1)$ as a second step. The branching rules \eqref{eq:branching_GL_O} give us the following decomposition the irreducible representations \eqref{eq:centralGL_decompositionV3} of $\GL(\Dim,\mathbb{R})$ into irreducible representation of $\Or(1,\Dim-1)$:
\begin{equation}\label{eq:branching_distortion}
	V^{\Yboxdim{3pt}\yng(3)}= D^{\Yboxdim{3pt}\yng(3)}\oplus D^{\Yboxdim{3pt}\yng(1)}\,, \hspace{0.5cm} V^{\Yboxdim{3pt}\yng(2,1)}= D^{\Yboxdim{3pt}\yng(2,1)}\oplus D^{\Yboxdim{3pt}\yng(1)}\,, \hspace{0.5cm} V^{\Yboxdim{3pt}\yng(1,1,1)}= D^{\Yboxdim{3pt}\yng(1,1,1)}\,,
\end{equation}
and one has $V^{\otimes 3}=D^{\Yboxdim{3.5pt}\yng(3)}\oplus 2\, D^{\Yboxdim{3.5pt}\yng(2,1)}\oplus D^{\Yboxdim{3.5pt}\yng(1,1,1)}\oplus 3 \, D^{\Yboxdim{3.5pt}\yng(1)}$. Because some irreducible representations have non trivial multiplicities: $m(\Yboxdim{3.5pt}\yng(2,1))=2$ and $m(\Yboxdim{3.5pt}\yng(1))=3$, the irreducible decomposition of tensors with respect to $\Or(1,\Dim-1)$ is not defined uniquely. Nevertheless, the branching decompositions \eqref{eq:branching_distortion} are multiplicity free, hence projection of a $\GL(\Dim,\mathbb{R})$ irreducible tensor onto an irreducible representation of $\Or(1,\Dim-1)$ is defined uniquely.
\begin{mathematica}[\textit{SymmetricFunctions}]
	NewellLittlewoodRule[$\lambda_1$, $\ldots$, $\lambda_k$]\,\\
	\textit{Returns a list of the partitions (with multiplicities) parameterizing the irreducible representations of $\Or(\Dim,\C)$ appearing in the \
	tensor product of the irreducible representations $D^{\lambda_1}\otimes\ldots \otimes D^{\lambda_k}$.}
\end{mathematica}

\begin{theorem}\label{theo:irreducible_distortion_O}
	Within the two-step decomposition: 
	\begin{itemize}
		\item[i)] The unique orthogonal irreducible decomposition of the distortion tensor with respect to\\ $\Or(1,\Dim-1)$  which maximizes the number of projective invariant tensors within its part is 
		\begin{equation}\label{eq:Irreducible_distortion}
				\hspace{-1.3cm}\tensor{C}{_{\, abc}}=\underbrace{\tensor{{\accentset{\left(\Yboxdim{2.5pt}\yng(3)\right)}{\underline{C}}}}{_{\,abc}}+\tensor{{\accentset{\left(\Yboxdim{2.5pt}\yng(2,1)S\right)}{\underline{C}}}}{_{\,abc}}+\tensor{{\accentset{\left(\Yboxdim{2.5pt}\yng(2,1)A\right)}{\underline{C}}}}{_{\,abc}}+\tensor{{\accentset{\left(\Yboxdim{2.5pt}\yng(1,1,1)\right)}{\underline{C}}}}{_{\,abc}}}_{\text{traceless}}+\underbrace{\tensor{{\accentset{\left(\Yboxdim{2.5pt}\yng(3),\Yboxdim{2.5pt}\yng(1)\right)}{C}}}{_{\, abc}}+\tensor{{\accentset{\left(\Yboxdim{2.5pt}\yng(2,1)S,\Yboxdim{2.5pt}\yng(1)\right)}{C}}}{_{\, abc}}+\tensor{{\accentset{\left(\Yboxdim{2.5pt}\yng(2,1)A,\Yboxdim{2.5pt}\yng(1)\right)}{C}}}{_{\, abc}}}_{\text{full trace}}\,,
		\end{equation}
		where each irreducible components are given by $\,\tensor{\accentset{\left(\Yboxdim{3pt}\mu\,, \lambda \right)}{C}}{_{abc}}=\tensor{(\,\accentset{(\mu)}{C}\,\cdot P^{\lambda}_3)}{_{abc}}$ and $\tensor{{\accentset{\left(\Yboxdim{2.5pt}\yng(1,1,1)\right)}{\underline{C}}}}{_{\,abc}}=\tensor{{\accentset{\left(\Yboxdim{2.5pt}\yng(1,1,1)\right)}{C}}}{_{abc}}$.\smallskip 
		\item[ii)] Introducing the following vector building blocks
			\begin{equation}\label{eq:building_blocks_distortion}
				\tensor{{\accentset{(1)}{B}}}{_a}=\tensor{\accentset{(1)}{C}}{_a}+\tensor{\accentset{(2)}{C}}{_a}+\tensor{\accentset{(3)}{C}}{_a}\,,\hspace{0.5cm}
				\tensor{{\accentset{(2)}{B}}}{_a}=\tensor{\accentset{(1)}{C}}{_a}-\,2\, \tensor{\accentset{(2)}{C}}{_a}+\tensor{\accentset{(3)}{C}}{_a}\,,\hspace{0.5cm}
				\tensor{{\accentset{(3)}{B}}}{_a}=\tensor{\accentset{(1)}{C}}{_a}-\tensor{\accentset{(3)}{C}}{_a}\,,
			\end{equation}
	the explicit expressions of the irreducible tensor components are,
	\begin{equation}\label{eq:O_decomposition_distortion1}
	\hspace{-0.5cm}\begin{array}{lll}
			\tensor{{\accentset{\left(\Yboxdim{2.5pt}\yng(3),\Yboxdim{2.5pt}\yng(1)\right)}{C}}}{_{\, abc}}=\dfrac{g_{(ab}\tensor{{\accentset{(1)}{B}}}{_{c)}}}{\Dim+2}\, \,,\hspace{0.4cm}
			&\tensor{{\accentset{\left(\Yboxdim{2.5pt}\yng(2,1)S,\Yboxdim{2.5pt}\yng(1)\right)}{C}}}{_{\, abc}}=\dfrac{ \, \tensor{{\accentset{(2)}{B}}}{_{(a|}}\tensor{g}{_{b|c)}} -g_{ac} \tensor{{\accentset{(2)}{B}}}{_b} \,}{3\left(\Dim-1\right)}\,,\hspace{0.4cm}
			&\tensor{{\accentset{\left(\Yboxdim{2.5pt}\yng(2,1)A,\Yboxdim{2.5pt}\yng(1)\right)}{C}}}{_{\, abc}}=\dfrac{\tensor{{\accentset{(3)}{B}}}{_{[a|}}\tensor{g}{_{b|c]}}}{\Dim-1}\,,
		\end{array}
		\end{equation}
	and for the totally traceless parts :
		\begin{equation}\label{eq:O_decomposition_distortion2}
		\hspace{-0.5cm}\begin{array}{lll}
			\tensor{{\accentset{\left(\Yboxdim{2.5pt}\yng(3)\right)}{\underline{C}}}}{_{\,abc}}=\tensor{{\accentset{\left(\Yboxdim{2.5pt}\yng(3)\right)}{C}}}{_{abc}}-\tensor{{\accentset{\left(\Yboxdim{2.5pt}\yng(3),\Yboxdim{2.5pt}\yng(1)\right)}{C}}}{_{abc}}\,,\hspace{0.4cm}
			&\tensor{{\accentset{\left(\Yboxdim{2.5pt}\yng(2,1)S\right)}{\underline{C}}}}{_{\,abc}}=\tensor{{\accentset{\left(\Yboxdim{2.5pt}\yng(2,1)S\right)}{C}}}{_{abc}}-\tensor{{\accentset{\left(\Yboxdim{2.5pt}\yng(2,1)S,\Yboxdim{2.5pt}\yng(1)\right)}{C}}}{_{\,abc}}\,,\hspace{0.4cm}
			&\tensor{{\accentset{\left(\Yboxdim{2.5pt}\yng(2,1)A\right)}{\underline{C}}}}{_{\,abc}}=\tensor{{\accentset{\left(\Yboxdim{2.5pt}\yng(2,1)A\right)}{C}}}{_{abc}}-\tensor{{\accentset{\left(\Yboxdim{2.5pt}\yng(2,1)A,\Yboxdim{2.5pt}\yng(1)\right)}{C}}}{_{\,abc}}\,.
		\end{array}
	\end{equation}

		\item[iii)]The traceless tensors $\,\tensor{{\accentset{\left(\Yboxdim{2.5pt}\yng(3)\right)}{\underline{C}}}}{_{\,abc}}$, $\tensor{{\accentset{\left(\Yboxdim{2.5pt}\yng(2,1)S\right)}{\underline{C}}}}{_{\,abc}}$, $\tensor{{\accentset{\left(\Yboxdim{2.5pt}\yng(2,1)A\right)}{\underline{C}}}}{_{\,abc}}$, $\tensor{{\accentset{\left(\Yboxdim{2.5pt}\yng(1,1,1)\right)}{\underline{C}}}}{_{\,abc}}$, and $\,\,\,\tensor{{\accentset{\left(\Yboxdim{2.5pt}\yng(2,1)A,\Yboxdim{2.5pt}\yng(1)\right)}{C}}}{_{\, abc}}$ are the projective invariant tensors of the decomposition \eqref{eq:Irreducible_distortion}.
	\end{itemize}
\end{theorem}
\begin{proof}
	The uniqueness of the decomposition is a result of Corollary \eqref{cor:GLdecomposition_Distortion} and of the multiplicity free property of the branching decompositions \eqref{eq:branching_distortion}. From the Schur-Weyl duality \eqref{eq:central_decomposition_O} one has indeed $\,\tensor{\accentset{\left(\Yboxdim{3pt}\mu\,, \lambda \right)}{C}}{_{abc}}=\tensor{(\,\accentset{(\mu)}{C}\,\cdot P^{\lambda}_3)}{_{abc}}$ and the expression of the projector $P^{\Yboxdim{3.pt}\yng(1)}_3$ used for the decomposition is given in equation \eqref{eq:proj_Distortion} in terms of the conjugacy class sum basis of the centralizer $\mathcal{C}_3$ of $\Sn{3}$ in $\Bn{3}(\Dim)$. The second point and the definition of the building blocks $\tensor{{\accentset{(1)}{B}}}{_a}$, $\tensor{{\accentset{(2)}{B}}}{_a}$, $\tensor{{\accentset{(3)}{B}}}{_a}$ are the results of direct calculations which are detailed in the Mathematica notebook \cite{xMAGdistortion}. Let $\nabla$ and $\hat\nabla$ be two projectively equivalent connection: $\tensor{\hat C}{_{abc}} = \tensor{C}{_{abc}}+\tensor{\delta_\xi C}{_{abc}}$. Because $\tensor{\delta_\xi C}{_{abc}}$ is proportional to the metric, its traceless part is zero. In other words $\tensor{\delta_\xi C}{_{abc}}$ belong to the full trace subspace of $V^{\otimes 3}$ which is orthogonal to the traceless subspace \eqref{eq:trace_decomposition}.
	As a direct consequence $\,\tensor{{\accentset{\left(\Yboxdim{2.5pt}\yng(3)\right)}{\underline{C}}}}{_{\,abc}}$, $\tensor{{\accentset{\left(\Yboxdim{2.5pt}\yng(2,1)S\right)}{\underline{C}}}}{_{\,abc}}$, $\tensor{{\accentset{\left(\Yboxdim{2.5pt}\yng(2,1)A\right)}{\underline{C}}}}{_{\,abc}}$, and  $\tensor{{\accentset{\left(\Yboxdim{2.5pt}\yng(1,1,1)\right)}{\underline{C}}}}{_{\,abc}}$ are projective invariant. Because $\tensor{{\accentset{\left(\Yboxdim{2.5pt}\yng(2,1)A\right)}{C}}}{_{abc}}$ and $\tensor{{\accentset{\left(\Yboxdim{2.5pt}\yng(2,1)A\right)}{\underline{C}}}}{_{\,abc}}$ are projective invariant so is $\,\,\,\tensor{{\accentset{\left(\Yboxdim{2.5pt}\yng(2,1)A,\Yboxdim{2.5pt}\yng(1)\right)}{C}}}{_{\, abc}}$.
\end{proof}
\begin{remark} In terms of the traces of non-metricity and torsion \eqref{eq:vectors_nm_torsion} the building blocks \eqref{eq:building_blocks_distortion} are given by
	\begin{equation*}
		\tensor{{\accentset{(1)}{B}}}{_a}=-\tensor{\tilde{Q}}{_a}-\frac{1}{2}\,\tensor{Q}{_a}\,, \hspace{0.5cm} \tensor{{\accentset{(2)}{B}}}{_a}=\tensor{Q}{_a}-\tensor{\tilde{Q}}{_a}\,,\hspace{0.5cm} \tensor{{\accentset{(3)}{B}}}{_a}=\tensor{Q}{_a}-\tensor{\tilde{Q}}{_a}+\,2\, \tensor{T}{_a}\,. 
	\end{equation*} 
\end{remark}
For the cases $\Dim< 3$ one has the following proposition.
\begin{proposition}\label{prop:small_d_distortion} For $\Dim < 3$ the irreducible decomposition of the distortion tensor of Theorem \ref{theo:irreducible_distortion_O} reduces to:
\begin{itemize}
\item[\it{i)}] For $\Dim=1$, 
\begin{equation}
	\tensor{C}{_a_b_c}=\tensor{{\accentset{\left(\Yboxdim{2.5pt}\yng(3),\Yboxdim{2.5pt}\yng(1)\right)}{C}}}{_{\, abc}}\,.
\end{equation}
\item[\it{ii)}] For $\Dim=2$, 
\begin{equation}
	\tensor{C}{_a_b_c}=\tensor{{\accentset{\left(\Yboxdim{2.5pt}\yng(3)\right)}{\underline{C}}}}{_{\,abc}}+\tensor{{\accentset{\left(\Yboxdim{2.5pt}\yng(3),\Yboxdim{2.5pt}\yng(1)\right)}{C}}}{_{\, abc}}+\tensor{{\accentset{\left(\Yboxdim{2.5pt}\yng(2,1)S,\Yboxdim{2.5pt}\yng(1)\right)}{C}}}{_{\, abc}}+\tensor{{\accentset{\left(\Yboxdim{2.5pt}\yng(2,1)A,\Yboxdim{2.5pt}\yng(1)\right)}{C}}}{_{\, abc}}\,.
\end{equation}
\end{itemize}
\end{proposition}
\begin{proof}\hphantom{c}\smallskip
\begin{itemize}
	\item[\it i)] This part is quite obvious but let us still detail the proof in the light of \eqref{eq:GL_decomposition} and \eqref{eq:O_decomposition}. All traceless tensors in \eqref{eq:O_decomposition_distortion2} are identically zero because the length of the first two columns of each diagram is greater than 1, that is they do not belong to $\Lambda_3(1)$ (recall the definition of $\Lambda_n(\Dim)$ in \eqref{eq:Lambda_d} and \eqref{eq:O_decomposition}).\smallskip
	
	 The expressions for $\,\,\tensor{{\accentset{\left(\Yboxdim{2.5pt}\yng(2,1)A,\Yboxdim{2.5pt}\yng(1)\right)}{C}}}{_{\, abc}}\,\,$ and $\,\,\tensor{{\accentset{\left(\Yboxdim{2.5pt}\yng(2,1)S,\Yboxdim{2.5pt}\yng(1)\right)}{C}}}{_{\, abc}}\,\,$ in \eqref{eq:O_decomposition_distortion1} are singular. In this case, because $\,\Yboxdim{5.5pt}\yng(2,1)\,\notin \mathcal{P}_3(1)$ (recall the definition of $\mathcal{P}_n(\Dim)$ \eqref{eq:P_d} and \eqref{eq:GL_decomposition}) the associated irreducible $\GL(\Dim,\R)$ representation is not present in the decomposition. Hence, the tensors $\,\,\tensor{{\accentset{\left(\Yboxdim{2.5pt}\yng(2,1)A,\Yboxdim{2.5pt}\yng(1)\right)}{C}}}{_{\, abc}}$ and $\,\,\tensor{{\accentset{\left(\Yboxdim{2.5pt}\yng(2,1)S,\Yboxdim{2.5pt}\yng(1)\right)}{C}}}{_{\, abc}}$ are identically zero. Finally, one has $\Yboxdim{5pt}\yng(3)\in \mathcal{P}_3(1)$ and $\Yboxdim{5pt}\yng(1)\in \Lambda_3(1)$ and the result follows.\smallskip
	 
\item[\it ii)] For the traceless irreducible tensors one has $\Yboxdim{5pt}\yng(3)\in \Lambda_3(2)$,  while $\Yboxdim{5pt}\yng(2,1)\,,\,\Yboxdim{5pt}\yng(1,1,1)\notin\Lambda_3(2)$. Hence the only non-zero traceless component is $\tensor{{\accentset{\left(\Yboxdim{2.5pt}\yng(3)\right)}{\underline{C}}}}{_{\,abc}}$.\medskip 

For the vector component $\tensor{{\accentset{\left(\Yboxdim{2.5pt}\yng(3),\Yboxdim{2.5pt}\yng(1)\right)}{C}}}{_{\, abc}}$ one has $\Yboxdim{5pt}\yng(3)\in \Lambda_3(2)$ and hence the Littlewood's restriction rules \ref{rem:littlewood_restriction_rule} apply: $\tensor{{\accentset{\left(\Yboxdim{2.5pt}\yng(3),\Yboxdim{2.5pt}\yng(1)\right)}{C}}}{_{\, abc}}$ is present in the decomposition. For the other two components $\,\,\tensor{{\accentset{\left(\Yboxdim{2.5pt}\yng(2,1)S,\Yboxdim{2.5pt}\yng(1)\right)}{C}}}{_{\, abc}}$ and $\,\,\tensor{{\accentset{\left(\Yboxdim{2.5pt}\yng(2,1)A,\Yboxdim{2.5pt}\yng(1)\right)}{C}}}{_{\, abc}}$, the Littlewood's restriction rule does not apply but the irreducible representation $V^{\Yboxdim{3.5pt}\yng(2,1)}$ is present with multiplicity two ($\Yboxdim{5pt}\yng(2,1)\in \mathcal{P}_3(2)$). Hence $\,\,\tensor{{\accentset{\left(\Yboxdim{2.5pt}\yng(2,1)S,\Yboxdim{2.5pt}\yng(1)\right)}{C}}}{_{\, abc}}$ and $\,\,\tensor{{\accentset{\left(\Yboxdim{2.5pt}\yng(2,1)A,\Yboxdim{2.5pt}\yng(1)\right)}{C}}}{_{\, abc}}$ are present in the decomposition.
\end{itemize}
\end{proof}
\vskip 4pt
\section{Irreducible decompositions of the Riemann tensor}\label{sec:Riemann_Decomposition}

In this section, unless otherwise specified we assume $\Dim\geqslant 4$. A Mathematica notebook going through all the steps presented here is available on gitHub \cite{xMAGRiemann}.

\subsection{$\GL(\Dim,\mathbb{R})$ isotypic decomposition}\label{subsec:CentralGLRiemann}

In metric-affine gravity the Riemann tensor associated with the independent connection $\nabla$ is antisymmetric in the first pair of indices, and has no other symmetries with respect to permutation of indices. Therefore we identify the space of metric-affine Riemann tensors $V_{\mathcal{R}}$ as the tensor product $V^{\Yboxdim{3pt}\yng(1,1)}\otimes V^{\Yboxdim{3pt}\yng(1)}\otimes V^{\Yboxdim{3pt}\yng(1)}$ of $\GL(\Dim,\mathbb{R})$ irreducible representations. Note that in metric-affine gravity the algebraic Bianchi identity is not an identity for the Riemann tensor $\tensor{\mathcal{R}}{_{abcd}}$ per se as it involves covariant derivatives of the torsion tensor\footnote{Expanding the algebraic Bianchi identity in terms of the distortion tensor using \eqref{eq:def_distortion} one recover the algebraic Bianchi identity of the Riemann tensor $R$ of (pseudo)-Riemannian geometry which is related to the Jacobi identity for the Lie bracket of vector fields.}.
Assuming $\Dim\geqslant4$, the Littlewood-Richardson rule give us the following direct sum decomposition of $V_{\mathcal{R}}$ into irreducible representations of $\GL(\Dim,\mathbb{R})$: 
\begin{equation}\label{eq:centralGL_decompositionVR}
	V_{\mathcal{R}}=\, V^{\Yboxdim{3pt}\yng(3,1)}\oplus \, V^{\Yboxdim{3pt}\yng(2,2)}\oplus \,2 \, V^{\Yboxdim{3pt}\yng(2,1,1)} \oplus V^{\Yboxdim{3pt}\yng(1,1,1,1)} \, .
\end{equation}
We define $\tensor{{\accentset{\left(\Yboxdim{2.5pt}\yng(3,1)\right)}{\mathcal{R}}}}{_{abcd}}$, $\tensor{{\accentset{\left(\Yboxdim{2.5pt}\yng(2,2)\right)}{\mathcal{R}}}}{_{abcd}}$, $\tensor{{\accentset{\left(\Yboxdim{2.5pt}\yng(2,1,1)\right)}{\mathcal{R}}}}{_{abcd}}$ and $\tensor{{\accentset{\left(\Yboxdim{2.5pt}\yng(1,1,1,1)\right)}{\mathcal{R}}}}{_{abcd}}$ as the image of the action of the central Young idempotents on the Riemann tensor, namely: 
\begin{equation}\label{eq:def_central_Riemann}
	\tensor{{\accentset{\left(\Yboxdim{2.5pt}\yng(3,1)\right)}{\mathcal{R}}}}{_{abcd}}:=\tensor{{(\mathcal{R}\cdot Z^{\Yboxdim{2.5pt}\yng(3,1)})}}{_{abcd}}\,,\hspace{0.5cm}
	\tensor{{\accentset{\left(\Yboxdim{2.5pt}\yng(2,2)\right)}{\mathcal{R}}}}{_{abcd}}:=\tensor{{(\mathcal{R}\cdot Z^{\Yboxdim{2.5pt}\yng(2,2)})}}{_{abcd}}\,,\hspace{0.5cm}
	\tensor{{\accentset{\left(\Yboxdim{2.5pt}\yng(2,1,1)\right)}{\mathcal{R}}}}{_{abcd}}:=\tensor{{(\mathcal{R}\cdot Z^{\Yboxdim{2.5pt}\yng(2,1,1)})}}{_{abcd}}\,,\hspace{0.5cm}
	\tensor{\accentset{\left(\Yboxdim{2.5pt}\yng(1,1,1,1)\right)}{\mathcal{R}}}{_{abcd}}:=\tensor{(\mathcal{R}\cdot Z^{\Yboxdim{2.5pt}\yng(1,1,1,1)})}{_{abcd}}\,.
\end{equation}
The above central Young idempotents are expressed in the conjugacy class sum basis of the center $\mathcal{Z}_4$ of $\C\Sn{4}$ in the equations \eqref{eq:centralYoung_4}. Explicit computations yield the following expressions in terms of symmetrization of indices of the Riemann tensor:
\begin{equation}\label{eq:central_Riemann_GL_explicit}
	\begin{array}{ll}
		\tensor{{\accentset{\left(\Yboxdim{2.5pt}\yng(3,1)\right)}{\mathcal{R}}}}{_{abcd}}=\dfrac{3}{4}\left(\tensor{{\mathcal{R}}}{_{a(bcd)}}+\tensor{{\mathcal{R}}}{_{(a|b|cd)}}\right),\hspace{0.3cm}
		&\tensor{{\accentset{\left(\Yboxdim{2.5pt}\yng(2,2)\right)}{\mathcal{R}}}}{_{abcd}}=\tensor{{\mathcal{R}}}{_{[ab]cd}}-\tensor{{\accentset{\left(\Yboxdim{2.5pt}\yng(1,1,1,1)\right)}{\mathcal{R}}}}{_{abcd}}-\tensor{{\accentset{\left(\Yboxdim{2.5pt}\yng(2,1,1)\right)}{\mathcal{R}}}}{_{abcd}}-\tensor{{\accentset{\left(\Yboxdim{2.5pt}\yng(3,1)\right)}{\mathcal{R}}}}{_{abcd}}\,,\\
		\tensor{{\accentset{\left(\Yboxdim{2.5pt}\yng(2,1,1)\right)}{\mathcal{R}}}}{_{abcd}}=\dfrac{3}{4}\left(\tensor{{\mathcal{R}}}{_{[abc]d}}+\tensor{{\mathcal{R}}}{_{[ab|c|d]}}+\tensor{{\mathcal{R}}}{_{[a|c|bd]}}-\tensor{{\mathcal{R}}}{_{[a|d|bc]}}\right)\,,\hspace{0.3cm}
		&\tensor{\accentset{\left(\Yboxdim{2.5pt}\yng(1,1,1,1)\right)}{\mathcal{R}}}{_{abcd}}=\tensor{{\mathcal{R}}}{_{[abcd]}}\,.
	\end{array}
\end{equation}
The only component of this decomposition which is not irreducible and contains two equivalent irreducible pieces is $\tensor{{\accentset{\left(\Yboxdim{2.5pt}\yng(2,1,1)\right)}{\mathcal{R}}}}{_{abcd}}$. The other components are already irreducibles and defined uniquely. They enjoy the following symmetries of indices and Young symmetries: 
\begin{equation}
	\begin{aligned}
		&\tensor{{\accentset{\left(\Yboxdim{2.5pt}\yng(3,1)\right)}{\mathcal{R}}}}{_{[ab](cd)}}=\tensor{{\accentset{\left(\Yboxdim{2.5pt}\yng(3,1)\right)}{\mathcal{R}}}}{_{abcd}}\,,\hspace{1cm}\tensor{{\accentset{\left(\Yboxdim{2.5pt}\yng(3,1)\right)}{\mathcal{R}}}}{_{[abc]d}}=\tensor{{\accentset{\left(\Yboxdim{2.5pt}\yng(3,1)\right)}{\mathcal{R}}}}{_{[ab|c|d]}}=0 \, ,\\
		&
		\tensor{{\accentset{\left(\Yboxdim{2.5pt}\yng(2,2)\right)}{\mathcal{R}}}}{_{[ab][cd]}}=\tensor{{\accentset{\left(\Yboxdim{2.5pt}\yng(2,2)\right)}{\mathcal{R}}}}{_{abcd}}\,,\hspace{1cm}\tensor{{\accentset{\left(\Yboxdim{2.5pt}\yng(2,2)\right)}{\mathcal{R}}}}{_{[abc]d}}=\tensor{{\accentset{\left(\Yboxdim{2.5pt}\yng(2,2)\right)}{\mathcal{R}}}}{_{[ab|c|d]}}=0\,.
	\end{aligned}
\end{equation}

\subsection{$\GL(\Dim,\mathbb{R})$ irreducible decomposition and projective invariance}\label{subsec:GLRiemann}
To define a unique orthogonal $\GL(\Dim,\mathbb{R})$ irreducible decomposition of $\tensor{{\accentset{\left(\Yboxdim{2.5pt}\yng(2,1,1)\right)}{\mathcal{R}}}}{_{abcd}}$ (recall the definition of orthogonal irreducible decomposition in \eqref{def:orthogonal_decomposition}), and hence of the Riemann tensor, we follow the same guideline as for the distortion tensor. That is, we require that one of the irreducible pieces in $\tensor{{\accentset{\left(\Yboxdim{2.5pt}\yng(2,1,1)\right)}{\mathcal{R}}}}{_{abcd}}$ be projective invariant. We are looking for $Y_{S},Y_{A}\in \mathbb{R}\mathfrak{S}_4$  such that $\accentset{\left(\Yboxdim{2.5pt}\yng(2,1,1)\right)}{\mathcal{R}}= \accentset{\left(\Yboxdim{2.5pt}\yng(2,1,1)\right)}{\mathcal{R}}\cdot Y_{S}+\accentset{\left(\Yboxdim{2.5pt}\yng(2,1,1)\right)}{\mathcal{R}}\cdot Y_{A}$, with $\delta_\xi (\accentset{\left(\Yboxdim{2.5pt}\yng(2,1,1)\right)}{\mathcal{R}}\cdot Y_{S})\neq 0$ and  $\delta_\xi (\accentset{\left(\Yboxdim{2.5pt}\yng(2,1,1)\right)}{\mathcal{R}}\cdot Y_{A})= 0$, and we want $Y_{S}$ and $Y_{A}$ to realize an orthogonal irreducible decomposition of $\accentset{\left(\Yboxdim{2.5pt}\yng(2,1,1)\right)}{\mathcal{R}}$. Under a change of the connection \eqref{eq:projective_transformation}, $\tensor{{\accentset{\left(\Yboxdim{2.5pt}\yng(2,1,1)\right)}{\mathcal{R}}}}{_{abcd}}$ transforms as 
\begin{equation}
\tensor{{\accentset{\left(\Yboxdim{2.5pt}\yng(2,1,1)\right)}{\mathcal{R}}}}{_{abcd}}\to  \tensor{{\accentset{\left(\Yboxdim{2.5pt}\yng(2,1,1)\right)}{\mathcal{R}}}}{_{abcd}}+\tensor{{\delta_\xi \accentset{\left(\Yboxdim{2.5pt}\yng(2,1,1)\right)}{R}}}{_{abcd}}\,,
\end{equation}
with 
\begin{equation}
\tensor{{\delta_\xi \accentset{\left(\Yboxdim{2.5pt}\yng(2,1,1)\right)}{R}}}{_{abcd}}=\tensor{g}{_c_d}\tensor{\accentset{(g)}{\nabla}}{_{[a}}\tensor{\xi}{_{c]}}+\dfrac{1}{2}\left(\tensor{g}{_{ac}}\tensor{\accentset{(g)}{\nabla}}{_{[b}}\tensor{\xi}{_{d]}}+\tensor{g}{_{ad}}\tensor{\accentset{(g)}{\nabla}}{_{[b}}\tensor{\xi}{_{c]}}-\tensor{g}{_{bc}}\tensor{\accentset{(g)}{\nabla}}{_{[a}}\tensor{\xi}{_{d]}}-\tensor{g}{_{bd}}\tensor{\accentset{(g)}{\nabla}}{_{[a}}\tensor{\xi}{_{c]}}\right).
\end{equation}
Notice that $\tensor{{\delta_\xi \accentset{\left(\Yboxdim{2.5pt}\yng(2,1,1)\right)}{R}}}{_{abcd}}$ enjoys the following symmetry of indices 
\begin{equation}
	\tensor{{\delta_\xi \accentset{\left(\Yboxdim{2.5pt}\yng(2,1,1)\right)}{R}}}{_{ab(cd)}}=\tensor{{\delta_\xi \accentset{\left(\Yboxdim{2.5pt}\yng(2,1,1)\right)}{R}}}{_{abcd}}\,,
\end{equation}
which gives use a natural candidate for the sought-after decomposition. We write
\begin{equation}\label{eq:irre_decomp_R211}
	\tensor{{\accentset{\left(\Yboxdim{2.5pt}\yng(2,1,1)\right)}{\mathcal{R}}}}{_{abcd}}=\tensor{{\accentset{\left(\Yboxdim{2.5pt}\yng(2,1,1) S\right)}{\mathcal{R}}}}{_{abcd}}+\tensor{{\accentset{\left(\Yboxdim{2.5pt}\yng(2,1,1)A\right)}{\mathcal{R}}}}{_{abcd}}\,, \hspace{1cm}\text{with}\hspace{1cm} \tensor{{\accentset{\left(\Yboxdim{2.5pt}\yng(2,1,1)S\right)}{\mathcal{R}}}}{_{abcd}}:=\tensor{{\accentset{\left(\Yboxdim{2.5pt}\yng(2,1,1)\right)}{\mathcal{R}}}}{_{ab(cd)}}\hspace{0.5cm}\text{and}\hspace{0.5cm} \tensor{{\accentset{\left(\Yboxdim{2.5pt}\yng(2,1,1)A\right)}{\mathcal{R}}}}{_{abcd}}:=\tensor{{\accentset{\left(\Yboxdim{2.5pt}\yng(2,1,1)\right)}{\mathcal{R}}}}{_{ab[cd]}}\,.
\end{equation}
The irreducible tensors $\tensor{{\accentset{\left(\Yboxdim{2.5pt}\yng(2,1,1)S\right)}{\mathcal{R}}}}{_{abcd}}$ and $\tensor{{\accentset{\left(\Yboxdim{2.5pt}\yng(2,1,1)A\right)}{\mathcal{R}}}}{_{abcd}}$ \eqref{eq:irre_decomp_R211} are clearly non zero and orthogonal.\medskip

This decomposition is realized by the action on $\accentset{\left(\Yboxdim{2.5pt}\yng(2,1,1)\right)}{\mathcal{R}}$ of the following operators:
\begin{equation}
	Y_{S}=s\, Z^{\Yboxdim{3.5pt}\yng(2)}\, s^{\shortminus 1}\,, \hspace{0.5cm} Y_{A}=s\, Z^{\Yboxdim{3.5pt}\yng(1,1)}\, s^{\shortminus 1}\,\hspace{1cm}
	\text{with}\hspace{1cm} s=\raisebox{-.4\height}{\includegraphics[scale=0.55]{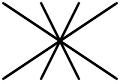}}\,.
\end{equation}
\begin{remark}
By direct computation one find that the irreducible tensors $\accentset{\left(\Yboxdim{2.5pt}\yng(2,1,1)S\right)}{\mathcal{R}}$ and $\accentset{\left(\Yboxdim{2.5pt}\yng(2,1,1)A\right)}{\mathcal{R}}$ can be obtained by the action on the Riemann tensor of the self adjoint operators $Y^{\,\Yboxdim{3.5pt}\yng(2,1,1)S}$ and $Y^{\,\Yboxdim{3.5pt}\yng(2,1,1)A}$ defined by:
\begin{equation}
	Y^{\,\Yboxdim{3.5pt}\yng(2,1,1)S}:=s \, Y^{\,\Scale[0.55]{\young(12,3,4)}}\, s^{-1}\, \hspace{1cm} 
	Y^{\,\Yboxdim{3.5pt}\yng(2,1,1)A}:=s \,\left(Y^{\,\Scale[0.55]{\young(13,2,4)}}+Y^{\,\Scale[0.55]{\young(14,2,3)}}\,\,\right)  \, s^{-1}\,, \hspace{1cm} \text{with}\hspace{1cm} s=\raisebox{-.4\height}{\includegraphics[scale=0.55]{fig/cycle14_23.pdf}}\,.
\end{equation}
The \textit{standard} Young seminormal idempotents are expressed in term of the central Young idempotents as (see formula \eqref{eq:Young_seminormal_idempotents_intro} and section \ref{sec:seminormalYoung} for more details):
\begin{equation}
	Y^{\,\Scale[0.5]{\young(12,3,4)}}=Z^{\Yboxdim{4pt}\yng(2)}Z^{\Yboxdim{4pt}\yng(2,1)}Z^{\Yboxdim{4pt}\yng(2,1,1)}\,,\hspace{1cm}Y^{\,\Scale[0.5]{\young(13,2,4)}}=Z^{\Yboxdim{4pt}\yng(1,1)}Z^{\Yboxdim{4pt}\yng(2,1)}Z^{\Yboxdim{4pt}\yng(2,1,1)}\,,\hspace{1cm} Y^{\,\Scale[0.5]{\young(14,2,3)}}=Z^{\Yboxdim{4pt}\yng(1,1)}Z^{\Yboxdim{4pt}\yng(1,1,1)}Z^{\Yboxdim{4pt}\yng(2,1,1)}\,.
\end{equation}
To sum up one has:
\begin{equation}
\accentset{\left(\Yboxdim{2.5pt}\yng(2,1,1)S\right)}{\mathcal{R}}=\mathcal{R}\cdot Y^{\,\Yboxdim{3.5pt}\yng(2,1,1)S}=\mathcal{R}\cdot Y_{S}\, Z^{\Yboxdim{3pt}\yng(2,1,1)}\,,\hspace{1cm} 
\accentset{\left(\Yboxdim{2.5pt}\yng(2,1,1)A\right)}{\mathcal{R}}=\mathcal{R}\cdot Y^{\,\Yboxdim{3.5pt}\yng(2,1,1)A}=\mathcal{R}\cdot Y_{A}\, Z^{\Yboxdim{3pt}\yng(2,1,1)}\,.
\end{equation}
\end{remark}
\medskip

\begin{proposition}
	The decomposition \eqref{eq:irre_decomp_R211} is the unique orthogonal irreducible decomposition of $ \,\tensor{{\accentset{\left(\Yboxdim{2.5pt}\yng(2,1,1)\right)}{\mathcal{R}}}}{_{abcd}} \,$ with respect to $\GL(\Dim,\mathbb{R})$ which maximizes the number of projective invariant tensors within its parts. 
\end{proposition}
The proof is similar to the one given for Proposition \eqref{prop:irreducible_distortion_GL}.

\begin{corollary}\label{cor:GLdecomposition_Riemann}
	The unique orthogonal irreducible decomposition of the Riemann tensor with respect to $\GL(\Dim,\mathbb{R})$ which maximizes the number of projective invariant tensors within its parts is :
	\begin{equation}
		\tensor{{\mathcal{R}}}{_{[ab]cd}}=\tensor{{\accentset{\left(\Yboxdim{2.5pt}\yng(3,1)\right)}{\mathcal{R}}}}{_{abcd}}+\tensor{{\accentset{\left(\Yboxdim{2.5pt}\yng(2,2)\right)}{\mathcal{R}}}}{_{abcd}}+\tensor{{\accentset{\left(\Yboxdim{2.5pt}\yng(2,1,1)S\right)}{\mathcal{R}}}}{_{abcd}}+\tensor{{\accentset{\left(\Yboxdim{2.5pt}\yng(2,1,1)A\right)}{\mathcal{R}}}}{_{abcd}}+\tensor{{\accentset{\left(\Yboxdim{2.5pt}\yng(1,1,1,1)\right)}{\mathcal{R}}}}{_{abcd}}
	\end{equation}
	with
	\begin{equation}
		\begin{array}{ll}
			\tensor{{\accentset{\left(\Yboxdim{2.5pt}\yng(3,1)\right)}{\mathcal{R}}}}{_{abcd}}=\dfrac{3}{4}\left(\tensor{{\mathcal{R}}}{_{a(bcd)}}+\tensor{{\mathcal{R}}}{_{(a|b|cd)}}\right) \,,\hspace{1cm}
			&\tensor{{\accentset{\left(\Yboxdim{2.5pt}\yng(2,1,1)S\right)}{\mathcal{R}}}}{_{abcd}}=\dfrac{3}{8}\left(\tensor{{\mathcal{R}}}{_{[abc]d}}+\tensor{{\mathcal{R}}}{_{[abd]c}}+\tensor{{\mathcal{R}}}{_{[ab|c|d]}}+\tensor{{\mathcal{R}}}{_{[ab|d|c]}}\right)\,, \hspace{1cm}\\
			\tensor{{\accentset{\left(\Yboxdim{2.5pt}\yng(2,1,1)A\right)}{\mathcal{R}}}}{_{abcd}}=\dfrac{1}{2}\left(\tensor{{\mathcal{R}}}{_{ab[cd]}}-\tensor{{\mathcal{R}}}{_{cd[ab]}}\right)\,,
			&\tensor{{\accentset{\left(\Yboxdim{2.5pt}\yng(1,1,1,1)\right)}{\mathcal{R}}}}{_{abcd}}=\tensor{{\mathcal{R}}}{_{[abcd]}}\,,\\
			\tensor{{\accentset{\left(\Yboxdim{2.5pt}\yng(2,2)\right)}{\mathcal{R}}}}{_{abcd}}=\tensor{{\mathcal{R}}}{_{[ab]cd}}-\tensor{{\accentset{\left(\Yboxdim{2.5pt}\yng(1,1,1,1)\right)}{\mathcal{R}}}}{_{abcd}}-\tensor{{\accentset{\left(\Yboxdim{2.5pt}\yng(2,1,1)\right)}{\mathcal{R}}}}{_{abcd}}-\tensor{{\accentset{\left(\Yboxdim{2.5pt}\yng(3,1)\right)}{\mathcal{R}}}}{_{abcd}}\,.
		\end{array}
	\end{equation}
	The projective invariant components of the decomposition are $\tensor{{\accentset{\left(\Yboxdim{2.5pt}\yng(2,2)\right)}{\mathcal{R}}}}{_{abcd}}$, $\tensor{{\accentset{\left(\Yboxdim{2.5pt}\yng(2,1,1)A\right)}{\mathcal{R}}}}{_{abcd}}$ and $\tensor{{\accentset{\left(\Yboxdim{2.5pt}\yng(1,1,1,1)\right)}{\mathcal{R}}}}{_{abcd}}$. 
\end{corollary}
\subsection{A unique $\Or(1,\Dim-1)$ irreducible decomposition}\label{subsec:ORiemann}
The branching rules \eqref{eq:branching_GL_O} applied to the irreducible representations of $\GL(\Dim,\mathbb{C})$  appearing in the Riemann tensor are:
\begin{equation}\label{eq:branching_Riemann}
	V^{\Yboxdim{4pt}\yng(3,1)}= D^{\Yboxdim{4pt}\yng(3,1)}\oplus D^{\Yboxdim{4pt}\yng(2)}\oplus D^{\Yboxdim{4pt}\yng(1,1)}\,, \hspace{0.5cm}
	V^{\Yboxdim{4pt}\yng(2,2)}= D^{\Yboxdim{4pt}\yng(2,2)}\oplus D^{\Yboxdim{4pt}\yng(2)}\oplus D^{\emptyset}\,, \hspace{0.5cm}
	V^{\Yboxdim{4pt}\yng(2,1,1)}= D^{\Yboxdim{4pt}\yng(2,1,1)}\oplus D^{\Yboxdim{4pt}\yng(1,1)}\,,\hspace{0.5cm} 
	V^{\Yboxdim{4pt}\yng(1,1,1,1)}= D^{\Yboxdim{4pt}\yng(1,1,1,1)}\,,
\end{equation}
and one has $V_\mathcal{R}=D^{\Yboxdim{3pt}\yng(3,1)}\oplus D^{\Yboxdim{3pt}\yng(2,2)}\oplus 2\, D^{\Yboxdim{3pt}\yng(2,1,1)}\oplus D^{\Yboxdim{3pt}\yng(1,1,1,1)} \oplus 2\, D^{\Yboxdim{3pt}\yng(2)}\oplus 3\, D^{\Yboxdim{3pt}\yng(1,1)}\oplus D^{\Yboxdim{3pt}\emptyset}$. Similarly to the case of the distortion tensor, the irreducible decomposition of the Riemann tensor with respect to $\Or(1,\Dim-1)$ is not unique because some irreducible representation of $\Or(1,\Dim-1)$ appear with multiplicities in $V_\mathcal{R}$. Nevertheless,
the branchings \eqref{eq:branching_Riemann} are multiplicity free and therefore the projection of a $\GL(\Dim,\mathbb{C})$ irreducible tensor onto an irreducible representation with respect to $\Or(1,\Dim-1)$ is defined uniquely.
\begin{theorem}\label{theo:irreducible_Riemann_O}
	Within the two-step decomposition: 
	\begin{itemize}
		\item[i)] The unique orthogonal irreducible decomposition of the Riemann tensor with respect to\\
		 $\Or(1,\Dim-1)$  which maximizes the number of projective invariant tensors within its parts is 
		\begin{equation}\label{eq:Irreducible_Riemann}
		\begin{array}{ll}
		\tensor{\mathcal{R}}{_{abcd}}=\underbrace{\tensor{{\accentset{\left(\Yboxdim{3pt}\yng(3,1)\right)}{\underline{\mathcal{R}}}}}{_{\, abcd}}+\tensor{{\accentset{\left(\Yboxdim{3pt}\yng(2,2)\right)}{\underline{\mathcal{R}}}}}{_{\, abcd}}+\tensor{{\accentset{\left(\Yboxdim{3pt}\yng(2,1,1)S\right)}{\underline{\mathcal{R}}}}}{_{\,abcd}}+\tensor{{\accentset{\left(\Yboxdim{3pt}\yng(2,1,1)A\right)}{\underline{\mathcal{R}}}}}{_{\,abcd}}+\tensor{{\accentset{\left(\Yboxdim{3pt}\yng(1,1,1,1)\right)}{\underline{\mathcal{R}}}}}{_{\, abcd}}}_{\text{traceless}}&\\[25pt]
		&\hspace{-5.5cm}+\,\,\underbrace{\tensor{{\accentset{\left(\Yboxdim{3pt}\yng(3,1),\,\Yboxdim{3pt}\yng(2)\right)}{\mathcal{R}}}}{_{\, abcd}}+\tensor{{\accentset{\left(\Yboxdim{3pt}\yng(2,2),\,\Yboxdim{3pt}\yng(2)\right)}{\mathcal{R}}}}{_{\, abcd}}+\tensor{{\accentset{\left(\Yboxdim{3pt}\yng(3,1),\,\Yboxdim{3pt}\yng(1,1)\right)}{\mathcal{R}}}}{_{\, abcd}}+\,\,\,\tensor{{\accentset{\left(\Yboxdim{3pt}\yng(2,1,1)A,\,\Yboxdim{3pt}\yng(1,1)\right)}{{\mathcal{R}}}}}{_{abcd}}+\,\,\,\tensor{{\accentset{\left(\Yboxdim{3pt}\yng(2,1,1)S,\,\Yboxdim{3pt}\yng(1,1)\right)}{{\mathcal{R}}}}}{_{abcd}}}_{\text{2-traceless}}+\underbrace{\tensor{{\accentset{\left(\Yboxdim{3pt}\yng(2,2),\,\emptyset\right)}{\mathcal{R}}}}{_{\, abcd}}}_{\text{full trace}}\,,
		\end{array}
		\end{equation} 
		where each irreducible components are given by $\,\tensor{{\accentset{\left(\Yboxdim{3pt}\mu\,, \lambda \right)}{{\mathcal{R}}}}}{_{abcd}}=\tensor{(\,\accentset{(\mu)}{\mathcal{R}}\,\cdot P^{\lambda}_4)}{_{abcd}}$ and $\tensor{{\accentset{\left(\Yboxdim{2.5pt}\yng(1,1,1,1)\right)}{\underline{\mathcal{R}}}}}{_{\,abcd}}=\tensor{{\accentset{\left(\Yboxdim{2.5pt}\yng(1,1,1)\right)}{\,\mathcal{R}\,}}}{_{abcd}}$.\medskip
		\item[ii)] Introducing the following traceless building block tensors\tablefootnote{The traces of the Riemann tensor $\tensor{\accentset{(1)}{\mathcal{R}}}{_{ab}}$, $\tensor{\accentset{(2)}{\mathcal{R}}}{_{ab}}$ and $\tensor{\accentset{(3)}{\mathcal{R}}}{_{ab}}$ are defined in \eqref{eq:defTracesRiemann}.} 
		\begin{equation}\label{eq:building_blocks_Riemann}
		\begin{array}{lll}
		\tensor{\accentset{(1)}{\mathcal B}}{_{ab}}=\tensor{\accentset{(1)}{\mathcal{R}}}{_{(ab)}}-\tensor{\accentset{(2)}{\mathcal{R}}}{_{(ab)}}\,, \hspace{0.5cm}&
		 \tensor{\accentset{(2)}{\mathcal B}}{_{ab}}=\tensor{\underline{\accentset{(1)}{\mathcal{R}}}}{_{(ab)}}+\tensor{\underline{\accentset{(2)}{\mathcal{R}}}}{_{(ab)}}\,,
		&\tensor{\accentset{(3)}{\mathcal B}}{_{ab}}=\tensor{\accentset{(1)}{\mathcal{R}}}{_{[ab]}}+\tensor{\accentset{(2)}{\mathcal{R}}}{_{[ab]}}\,,\\[6pt]
		\tensor{\accentset{(4)}{\mathcal B}}{_{ab}}=\tensor{\accentset{(1)}{\mathcal{R}}}{_{[ab]}}-\tensor{\accentset{(2)}{\mathcal{R}}}{_{[ab]}}-\tensor{\accentset{(3)}{\mathcal{R}}}{_{ab}}\,,& \tensor{\accentset{(5)}{\mathcal B}}{_{ab}}=\tensor{\accentset{(1)}{\mathcal{R}}}{_{[ab]}}-\tensor{\accentset{(2)}{\mathcal{R}}}{_{[ab]}}+\tensor{\accentset{(3)}{\mathcal{R}}}{_{ab}} \,,
		\end{array}
		\end{equation}
		the explicit expressions of the irreducible tensor components of the Riemann tensor are:\medskip  
					\begin{equation}\label{eq:O_decomposition_Riemann1}
		\hspace{-0.91cm}\begin{array}{ll}
			\tensor{{\accentset{\left(\Yboxdim{3pt}\yng(3,1),\,\Yboxdim{3pt}\yng(2)\right)}{\mathcal{R}}}}{_{\, abcd}}=\dfrac{\tensor{\accentset{(1)}{\mathcal B}}{_{[a|d|}}\tensor{g}{_{b]c}} +\tensor{\accentset{(1)}{\mathcal B}}{_{[a|c|}}\tensor{g}{_{b]d}}}{\Dim} , \hspace{0.7cm}
			&\tensor{{\accentset{\left(\Yboxdim{3pt}\yng(3,1),\,\Yboxdim{3pt}\yng(1,1)\right)}{\mathcal{R}}}}{_{\, abcd}}=\dfrac{ \tensor{\accentset{(5)}{\mathcal B}}{_{ab}}\tensor{g}{_{cd}}+\tensor{\accentset{(5)}{\mathcal B}}{_{[a|d|}}\tensor{g}{_{b]c}}+\tensor{\accentset{(5)}{\mathcal B}}{_{[a|c|}}\tensor{g}{_{b]d}}}{2\left(\, \Dim+2 \, \right)} , \\[10pt]
			\tensor{{\accentset{\left(\Yboxdim{3pt}\yng(2,2),\,\Yboxdim{3pt}\yng(2)\right)}{\mathcal{R}}}}{_{\, abcd}}=\dfrac{\tensor{\accentset{(2)}{\mathcal B}}{_{[a|c|}}\tensor{g}{_{b]d}} - \tensor{\accentset{(2)}{\mathcal B}}{_{[a|d|}}\tensor{g}{_{b]c}}}{\Dim-2},  \hspace{0.7cm}
			&\tensor{{\accentset{\left(\Yboxdim{3pt}\yng(2,2),\,\emptyset\right)}{\mathcal{R}}}}{_{\, abcd}}=\dfrac{\mathcal{R}\Bigl(\tensor{g}{_{ac}}\tensor{g}{_{bd}}-\tensor{g}{_{ad}}\tensor{g}{_{bc}}\Bigr)}{\Dim \left( \Dim-1  \right)} , \\[10pt]
			\tensor{{\accentset{\left(\Yboxdim{3pt}\yng(2,1,1)A,\,\Yboxdim{3pt}\yng(1,1)\right)}{{\mathcal{R}}}}}{_{abcd}}=\dfrac{\tensor{\accentset{(3)}{\mathcal B}}{_{[b|d|}}\tensor{g}{_{a]c}}-\tensor{\accentset{(3)}{\mathcal B}}{_{[b|c|}}\tensor{g}{_{a]d}}}{\Dim-2},\hspace{0.7cm}
			&\tensor{{\accentset{\left(\Yboxdim{3pt}\yng(2,1,1)S,\,\Yboxdim{3pt}\yng(1,1)\right)}{{\mathcal{R}}}}}{_{abcd}}=-\dfrac{\tensor{\accentset{(4)}{\mathcal B}}{_{ab}}\,\tensor{g}{_{cd}}-\tensor{\accentset{(4)}{\mathcal B}}{_{[a|d|}}\tensor{g}{_{b]c}}-\tensor{\accentset{(4)}{\mathcal B}}{_{[a|c|}}\tensor{g}{_{b]d}}}{2(\Dim-2)},
		\end{array}
	\end{equation}	
and for the totally traceless parts:
				\begin{equation}\label{eq:O_decomposition_Riemann2}
			\hspace{-0.91cm}\begin{array}{ll}
				\tensor{{\accentset{\left(\Yboxdim{3pt}\yng(3,1)\right)}{\underline{\mathcal{R}}}}}{_{\, abcd}}=\tensor{{\accentset{\left(\Yboxdim{3pt}\yng(3,1)\right)}{\mathcal{R}}}}{_{abcd}}-\tensor{{\accentset{\left(\Yboxdim{3pt}\yng(3,1),\,\Yboxdim{3pt}\yng(2)\right)}{\mathcal{R}}}}{_{\, abcd}}-\tensor{{\accentset{\left(\Yboxdim{3pt}\yng(3,1),\,\Yboxdim{3pt}\yng(1,1)\right)}{\mathcal{R}}}}{_{\, abcd}}\,,\hspace{0.5cm}
				&\tensor{{\accentset{\left(\Yboxdim{3pt}\yng(2,2)\right)}{\underline{\mathcal{R}}}}}{_{\, abcd}}=\tensor{{\accentset{\left(\Yboxdim{3pt}\yng(2,2)\right)}{\mathcal{R}}}}{_{abcd}}-\tensor{{\accentset{\left(\Yboxdim{3pt}\yng(2,2),\,\Yboxdim{3pt}\yng(2)\right)}{\mathcal{R}}}}{_{\, abcd}}-\tensor{{\accentset{\left(\Yboxdim{3pt}\yng(2,2),\,\emptyset\right)}{\mathcal{R}}}}{_{\, abcd}}\,,\\[10pt]
				\tensor{{\accentset{\left(\Yboxdim{3pt}\yng(2,1,1)S\right)}{\underline{\mathcal{R}}}}}{_{\, abcd}}=\tensor{{\accentset{\left(\Yboxdim{3pt}\yng(2,1,1)S\right)}{{\mathcal{R}}}}}{_{abcd}}-\,\,\,\tensor{{\accentset{\left(\Yboxdim{3pt}\yng(2,1,1)S,\,\Yboxdim{3pt}\yng(1,1)\right)}{{\mathcal{R}}}}}{_{abcd}}\,,\hspace{0.5cm}
				&\tensor{{\accentset{\left(\Yboxdim{3pt}\yng(2,1,1)A\right)}{\underline{\mathcal{R}}}}}{_{\, abcd}}=\,\,\tensor{{\accentset{\left(\Yboxdim{3pt}\yng(2,1,1)A\right)}{{\mathcal{R}}}}}{_{abcd}}-\,\,\,\tensor{{\accentset{\left(\Yboxdim{3pt}\yng(2,1,1)A,\,\Yboxdim{3pt}\yng(1,1)\right)}{{\mathcal{R}}}}}{_{abcd}}\,.
			\end{array}
		\end{equation}
	\item[iii)]The irreducible traceless part of the Riemann tensor, the full trace part  $\,\,\tensor{{\accentset{\left(\Yboxdim{3pt}\yng(2,2),\,\emptyset\right)}{\mathcal{R}}}}{_{\, abcd}}$, and the 2-traceless parts $\,\,\tensor{{\accentset{\left(\Yboxdim{3pt}\yng(3,1),\,\Yboxdim{3pt}\yng(2)\right)}{\mathcal{R}}}}{_{\, abcd}}$, $\,\,\tensor{{\accentset{\left(\Yboxdim{3pt}\yng(2,2),\,\Yboxdim{3pt}\yng(2)\right)}{\mathcal{R}}}}{_{\, abcd}}$, and $\,\,\,\,\tensor{{\accentset{\left(\Yboxdim{3pt}\yng(2,1,1)A,\,\Yboxdim{3pt}\yng(1,1)\right)}{{\mathcal{R}}}}}{_{abcd}}$ which are associated with the building blocks $\tensor{\accentset{(1)}{\mathcal B}}{_{ab}}$, $\tensor{\accentset{(2)}{\mathcal B}}{_{ab}}$ and $\tensor{\accentset{(3)}{\mathcal B}}{_{ab}}$, are the projective invariant tensors of the decomposition.
	\end{itemize} 
\end{theorem}
\begin{proof}
	The uniqueness of the decomposition is a result of Corollary \eqref{cor:GLdecomposition_Riemann} and of the multiplicity free property of the branching decompositions \eqref{eq:branching_Riemann}. From the Schur-Weyl duality \eqref{eq:central_decomposition_O} one has indeed $\,\tensor{{\accentset{\left(\Yboxdim{3pt}\mu\,, \lambda \right)}{{\mathcal{R}}}}}{_{abcd}}=\tensor{(\,\accentset{(\mu)}{\mathcal{R}}\,\cdot P^{\lambda}_4)}{_{abcd}}$ and the expressions of the projectors $P^{\Yboxdim{3.pt}\yng(2)}_4$, $P^{\Yboxdim{3.pt}\yng(1,1)}_4$ and $P^{\emptyset}_4$ used for the decomposition is given in equation \eqref{eq:proj_Riemann} in terms of the conjugacy class sum basis of the centralizer $\mathcal{C}_4$ of $\Sn{4}$ in $\Bn{4}(\Dim)$. The second point and the definition of the traceless building blocks $\tensor{\accentset{(k)}{\mathcal B}}{_{ab}}$ ($k=1,\ldots, 5$) are the results of direct calculations which are detailed in the Mathematica notebook \cite{xMAGRiemann}. For the last point, let $\nabla$ and $\hat\nabla$ be two projectively equivalent connections. Then there exist a vector field $\xi$ such that 
	\begin{equation}
		\tensor{\mathcal{\hat R}}{_{abcd}}=\tensor{\mathcal{R}}{_{abcd}} + \tensor{\delta_{\xi} \mathcal{R}}{_{abcd}},\hspace{0.5cm}\text{with $\tensor{\delta_{\xi} \mathcal{R}}{_{abcd}}=2\tensor{g}{_{cd}} \tensor{{\accentset{g}{\nabla}}}{_{[a}} \tensor{\xi}{_{b]}}$}\,.
	\end{equation}	
	Because $\tensor{\delta_{\xi} \mathcal{R}}{_{abcd}}$ is proportional to the metric, its traceless part is zero. As a consequence the irreducible traceless parts of the Riemann tensor are projective invariant. The component $\,\,\tensor{{\accentset{\left(\Yboxdim{3pt}\yng(2,2),\,\emptyset\right)}{\mathcal{R}}}}{_{\, abcd}}$ is expressed in term of the Ricci scalar which is projective invariant. From the transformation properties of the Ricci tensor, co-Ricci tensor and homothetic tensor \eqref{eq:proj_trans_Riccis} it is a matter of simple calculations to show that $\tensor{\accentset{(1)}{\mathcal B}}{_{ab}}$, $\tensor{\accentset{(2)}{\mathcal B}}{_{ab}}$ and $\tensor{\accentset{(3)}{\mathcal B}}{_{ab}}$ are the projective invariant building blocks of the decomposition.
\end{proof}
\begin{remark}
The Weyl tensor of metric-affine gravity is defined as the totally traceless part of the Riemann tensor. It is given by the sum of all traceless irreducible parts of the Riemann tensor \eqref{eq:irreducible_trace_decomposition}
\begin{equation}
	\begin{aligned}
		\tensor{\mathcal{W}}{_{abcd}}&=\tensor{{\accentset{\left(\Yboxdim{3pt}\yng(3,1)\right)}{\underline{\mathcal{R}}}}}{_{\, abcd}}+\tensor{{\accentset{\left(\Yboxdim{3pt}\yng(2,2)\right)}{\underline{\mathcal{R}}}}}{_{\, abcd}}+\tensor{{\accentset{\left(\Yboxdim{3pt}\yng(2,1,1)S\right)}{\underline{\mathcal{R}}}}}{_{\, abcd}}+\tensor{{\accentset{\left(\Yboxdim{3pt}\yng(2,1,1)A\right)}{\underline{\mathcal{R}}}}}{_{\, abcd}}+\tensor{{\accentset{\left(\Yboxdim{3pt}\yng(1,1,1,1)\right)}{\underline{\mathcal{R}}}}}{_{\, abcd}}\,,\\
		&=\tensor{\mathcal{R}}{_{abcd}}-\tensor{{\accentset{\left(\Yboxdim{3pt}\yng(3,1),\,\Yboxdim{3pt}\yng(2)\right)}{\mathcal{R}}}}{_{\, abcd}}-\tensor{{\accentset{\left(\Yboxdim{3pt}\yng(3,1),\,\Yboxdim{3pt}\yng(1,1)\right)}{\mathcal{R}}}}{_{\, abcd}}-\tensor{{\accentset{\left(\Yboxdim{3pt}\yng(2,2),\,\Yboxdim{3pt}\yng(2)\right)}{\mathcal{R}}}}{_{\, abcd}}-\tensor{{\accentset{\left(\Yboxdim{3pt}\yng(2,2),\,\emptyset\right)}{\mathcal{R}}}}{_{\, abcd}}-\tensor{{\accentset{\left(\Yboxdim{3pt}\yng(2,1,1)A,\,\Yboxdim{3pt}\yng(1,1)\right)}{{\mathcal{R}}}}}{_{abcd}}-\tensor{{\accentset{\left(\Yboxdim{3pt}\yng(2,1,1)S,\,\Yboxdim{3pt}\yng(1,1)\right)}{{\mathcal{R}}}}}{_{abcd}}\,.
	\end{aligned}
\end{equation}
As a consequence of the next Proposition the Weyl tensor of metric-affine gravity vanishes identically for $\Dim=2$, which reminiscent to the pseudo-Riemann framework. However, it does not vanish identically for $\Dim=3$ as in this case $\tensor{\mathcal{W}}{_{abcd}}=\tensor{{\accentset{\left(\Yboxdim{3pt}\yng(3,1)\right)}{\underline{\mathcal{R}}}}}{_{\, abcd}}\,$.
\end{remark}

In order to generalize Theorem \ref{theo:irreducible_Riemann_O} for small dimensions ($\Dim < 4$) we will make use of remark \ref{rem:littlewood_restriction_rule}, and of the dimensions of the irreducible representations of both $\GL(\Dim,\C)$ and $\Or(\Dim,\C)$. The general dimension formulas for the irreducible representations of $\GL(\Dim,\C)$ and $\Or(\Dim,\C)$ can be found in \cite{Samra_King_1979}. 
\begin{mathematicas}[\textit{SymmetricFunctions}]
DimOfIrrep[$\mu$, GeneralLinearGroup[$\Dim$]]\,\\
\textit{Returns the dimension of the irreducible representation $V^\mu$ of $\GL(\Dim)$\,.}\\
	
DimOfIrrep[$\lambda$, OrthogonalGroup[$\Dim$]]\,\\
\textit{Returns the dimension of the irreducible representation $D^\lambda$ of $\Or(\Dim)$\,.}
\end{mathematicas}

\begin{proposition}\label{prop:small_d_Riemann} For $\Dim < 4$ the irreducible decomposition of the Riemann tensor of Theorem \ref{theo:irreducible_Riemann_O} reduces to:
	\begin{itemize}
		\item[\it{i)}] For $\Dim=2$, 
	\begin{equation}\label{eq:Irreducible_Riemann_d2}
	\tensor{\mathcal{R}}{_a_b_c_d}=\underbrace{\,\,\tensor{{\accentset{\left(\Yboxdim{3pt}\yng(3,1),\,\Yboxdim{3pt}\yng(2)\right)}{\mathcal{R}}}}{_{\, abcd}}+\tensor{{\accentset{\left(\Yboxdim{3pt}\yng(3,1),\,\Yboxdim{3pt}\yng(1,1)\right)}{\mathcal{R}}}}{_{\, abcd}}}_{\text{2-traceless}}+\underbrace{\tensor{{\accentset{\left(\Yboxdim{3pt}\yng(2,2),\,\emptyset\right)}{\mathcal{R}}}}{_{\, abcd}}}_{\text{full trace}}\,.
	\end{equation}
		\item[\it{ii)}] For $\Dim=3$, 
	\begin{equation}\label{eq:Irreducible_Riemann_d3}
	\begin{array}{ll}
		\tensor{\mathcal{R}}{_{abcd}}=\underbrace{\tensor{{\accentset{\left(\Yboxdim{3pt}\yng(3,1)\right)}{\underline{\mathcal{R}}}}}{_{\, abcd}}}_{\text{traceless}}
		+\,\,\underbrace{\tensor{{\accentset{\left(\Yboxdim{3pt}\yng(3,1),\,\Yboxdim{3pt}\yng(2)\right)}{\mathcal{R}}}}{_{\, abcd}}+\tensor{{\accentset{\left(\Yboxdim{3pt}\yng(2,2),\,\Yboxdim{3pt}\yng(2)\right)}{\mathcal{R}}}}{_{\, abcd}}+\tensor{{\accentset{\left(\Yboxdim{3pt}\yng(3,1),\,\Yboxdim{3pt}\yng(1,1)\right)}{\mathcal{R}}}}{_{\, abcd}}+\,\,\,\tensor{{\accentset{\left(\Yboxdim{3pt}\yng(2,1,1)A,\,\Yboxdim{3pt}\yng(1,1)\right)}{{\mathcal{R}}}}}{_{abcd}}+\,\,\,\tensor{{\accentset{\left(\Yboxdim{3pt}\yng(2,1,1)S,\,\Yboxdim{3pt}\yng(1,1)\right)}{{\mathcal{R}}}}}{_{abcd}}}_{\text{2-traceless}}+\underbrace{\tensor{{\accentset{\left(\Yboxdim{3pt}\yng(2,2),\,\emptyset\right)}{\mathcal{R}}}}{_{\, abcd}}}_{\text{full trace}}.
	\end{array}
\end{equation} 
	\end{itemize}
\end{proposition}
\begin{proof}
First recall that for $\Dim=1$ the Riemann tensor is identically zero. 
\begin{itemize}
\item[\it{i)}] All traceless tensors in \eqref{eq:O_decomposition_Riemann2} are identically zero because the first two columns of each diagram is greater than 2.\medskip

The expressions for $\,\,\tensor{{\accentset{\left(\Yboxdim{3pt}\yng(2,2),\,\Yboxdim{3pt}\yng(2)\right)}{\mathcal{R}}}}{_{\, abcd}}\,\,$, $\,\,\tensor{{\accentset{\left(\Yboxdim{3pt}\yng(2,1,1)A,\,\Yboxdim{3pt}\yng(1,1)\right)}{{\mathcal{R}}}}}{_{abcd}}\,\,$, and $\,\,\tensor{{\accentset{\left(\Yboxdim{3pt}\yng(2,1,1)S,\,\Yboxdim{3pt}\yng(1,1)\right)}{{\mathcal{R}}}}}{_{abcd}}\,\,$ in \eqref{eq:O_decomposition_Riemann1} are singular, let us show that these tensors are actually identically zero. Because $\,\Yboxdim{5pt}\yng(2,1,1)\,\notin \mathcal{P}_4(2)$ the associated $\GL(\Dim,\R)$ irreducible representations are not present in the decomposition. Accordingly, $\,\,\tensor{{\accentset{\left(\Yboxdim{3pt}\yng(2,1,1)A,\,\Yboxdim{3pt}\yng(1,1)\right)}{{\mathcal{R}}}}}{_{abcd}}\,\,$ and $\,\,\tensor{{\accentset{\left(\Yboxdim{3pt}\yng(2,1,1)S,\,\Yboxdim{3pt}\yng(1,1)\right)}{{\mathcal{R}}}}}{_{abcd}}\,\,$ are identically zero. For $\tensor{{\accentset{\left(\Yboxdim{3pt}\yng(2,2),\,\Yboxdim{3pt}\yng(2)\right)}{\mathcal{R}}}}{_{\, abcd}}$, we note that $\Yboxdim{5pt}\yng(2,1,1)\,\notin \Lambda_4(2)$ so the Littlewood's restriction rules do not apply, yet $\Yboxdim{5pt}\yng(2)\,\in \Lambda_4(2)$ and we cannot conclude directly on the presence or absence of the irreducible representation. We resort to an analysis of the dimensions. One has $\dim(V^{\Yboxdim{3pt}\yng(2,2)})=1$ and $\dim(D^{\Yboxdim{3pt}\yng(2)})=2$, hence $\tensor{{\accentset{\left(\Yboxdim{3pt}\yng(2,2),\,\Yboxdim{3pt}\yng(2)\right)}{\mathcal{R}}}}{_{\, abcd}}$ is identically zero. The fact that $\dim(V^{\Yboxdim{3pt}\yng(2,2)})=1$ marks the presence of the scalar component $\tensor{{\accentset{\left(\Yboxdim{3pt}\yng(2,2),\,\emptyset\right)}{\mathcal{R}}}}{_{\, abcd}}$ in the decomposition.

It remains to decide on the presence or absence of $\,\,\,\tensor{{\accentset{\left(\Yboxdim{3pt}\yng(3,1),\,\Yboxdim{3pt}\yng(2)\right)}{\mathcal{R}}}}{_{\, abcd}}$ and $\,\,\,\tensor{{\accentset{\left(\Yboxdim{3pt}\yng(3,1),\,\Yboxdim{3pt}\yng(1,1)\right)}{\mathcal{R}}}}{_{\, abcd}}$ in $V_{\mathcal{R}}$. For these cases note that $\Yboxdim{3pt}\yng(3,1)\in \Par_{4}(2)$, but $\Yboxdim{3pt}\yng(3,1)\notin \Lambda_{4}(2)$ so that the Littlewood's restriction rules do not apply. Again, we resort to dimension counting arguments. One has $\dim(V^{\Yboxdim{3pt}\yng(3,1)})=3$, $\dim(D^{\Yboxdim{3pt}\yng(2)})=2$ and $\dim(D^{\Yboxdim{3pt}\yng(1,1)})=1$ so that $\dim(V^{\Yboxdim{3pt}\yng(3,1)})=\dim(D^{\Yboxdim{3pt}\yng(2)})+\dim(D^{\Yboxdim{3pt}\yng(1,1)})$. Hence both tensors are present in the decomposition.
\item[\it{ii)}]  For the traceless parts one has $\Yboxdim{5pt}\yng(3,1)\in \Lambda_4(3)$ while $\Yboxdim{5pt}\yng(2,2)\,, \Yboxdim{5pt}\yng(2,1,1)\,, \Yboxdim{5pt}\yng(1,1,1,1)\notin \Lambda_4(3)$. As a consequence $\,\,\tensor{{\accentset{\left(\Yboxdim{3pt}\yng(3,1)\right)}{\underline{\mathcal{R}}}}}{_{\, abcd}}$ is not zero while  $\,\,\tensor{{\accentset{\left(\Yboxdim{3pt}\yng(2,2)\right)}{\underline{\mathcal{R}}}}}{_{\, abcd}}$,$\,\,	\tensor{{\accentset{\left(\Yboxdim{3pt}\yng(2,1,1)S\right)}{\underline{\mathcal{R}}}}}{_{\, abcd}}$ and $\,\,	\tensor{{\accentset{\left(\Yboxdim{3pt}\yng(2,1,1)A\right)}{\underline{\mathcal{R}}}}}{_{\, abcd}}$ are identically zero.\medskip 

The 2-traceless tensors $\tensor{{\accentset{\left(\Yboxdim{3pt}\yng(3,1),\,\Yboxdim{3pt}\yng(2)\right)}{\mathcal{R}}}}{_{\, abcd}},$ $\tensor{{\accentset{\left(\Yboxdim{3pt}\yng(2,2),\,\Yboxdim{3pt}\yng(2)\right)}{\mathcal{R}}}}{_{\, abcd}}$, and $\tensor{{\accentset{\left(\Yboxdim{3pt}\yng(3,1),\,\Yboxdim{3pt}\yng(1,1)\right)}{\mathcal{R}}}}{_{\, abcd}}$ are present in the decomposition because they fall within the regime of the Littlewood's restriction rule. The two remaining 2-traceless tensors $\,\,\,\tensor{{\accentset{\left(\Yboxdim{3pt}\yng(2,1,1)A,\,\Yboxdim{3pt}\yng(1,1)\right)}{{\mathcal{R}}}}}{_{abcd}}$ and $\,\,\,\tensor{{\accentset{\left(\Yboxdim{3pt}\yng(2,1,1)S,\,\Yboxdim{3pt}\yng(1,1)\right)}{{\mathcal{R}}}}}{_{abcd}}$ do not fall within the regime of the Littlewood's restriction rules. From $\dim(V^{\Yboxdim{3pt}\yng(2,1,1)})=3$ and $\dim(D^{\Yboxdim{3pt}\yng(1,1)})=3$, we conclude that they are both present in the decomposition.
\end{itemize}
\end{proof}
\vskip 4pt

For completeness the irreducible decomposition of the Riemann tensor in the cases of zero torsion and of zero non-metricity are presented in the appendix \ref{app:Irreducible_Riemanns}.
\section{Quadratic curvature Lagrangians with distortion mass terms}\label{sec:Lagrangians}

When the irreducible decomposition of a tensor is multiplicity free (see remark \ref{rem:isotypic_component_irreducible}), finding all quadratic scalar invariants is straightforward. They are the scalar product of two copies of the same irreducible tensor, which we recall is defined uniquely in this case.\medskip 

Consider for example the irreducible decomposition of the metric Riemann tensor (also known as the Ricci decomposition) \cite{hawking_ellis_1973,lee2018introduction}: 
\begin{equation}\label{eq:irreducible_Riemann_metric}
	\tensor{R}{_{[ab][cd]}}=\accentset{\left(\Yboxdim{3pt}\yng(2,2)\right)}{\underline{R}}\tensor{\vphantom{\mathcal{R}}}{_{\,\,[ab][cd]}}+\,\,\accentset{\left(\Yboxdim{3pt}\yng(2,2)\,,\yng(2)\right)}{R}\tensor{\vphantom{\mathcal{R}}}{_{\,\,[ab][cd]}}+\,\,\accentset{\left(\Yboxdim{3pt}\yng(2,2)\,,\emptyset\right)}{R}\tensor{\vphantom{\mathcal{R}}}{_{\,\,[ab][cd]}},
\end{equation}
where $\accentset{\left(\Yboxdim{3pt}\yng(2,2)\right)}{\underline{R}}\tensor{\vphantom{\mathcal{R}}}{_{\,\,[ab][cd]}}$ is the metric Weyl tensor of pseudo-Riemannian geometry, while $\,\,\,\accentset{\left(\Yboxdim{3pt}\yng(2,2)\,,\yng(2)\right)}{R}\tensor{\vphantom{\mathcal{R}}}{_{\,\,[ab][cd]}}$ and $\,\,\,\accentset{\left(\Yboxdim{3pt}\yng(2,2)\,,\emptyset\right)}{R}\tensor{\vphantom{\mathcal{R}}}{_{\,\,[ab][cd]}}$ are the 2-traceless and full trace part of the Riemann tensor, associated with the traceless part of the Ricci tensor and the Ricci scalar respectively. In this case, the quadratic scalar invariant are $\langle\,\, \accentset{\left(\Yboxdim{3pt}\yng(2,2)\right)}{\underline{R}}\,\,,\, \,\accentset{\left(\Yboxdim{3pt}\yng(2,2)\right)}{\underline{R}}\,\,\rangle$, $\langle \hspace{0.2cm} \accentset{\left(\Yboxdim{3pt}\yng(2,2)\,,\yng(2)\right)}{R}\hspace{0.1cm},\, \hspace{0.2cm}\accentset{\left(\Yboxdim{3pt}\yng(2,2)\,,\yng(2)\right)}{R}\hspace{0.2cm}\,\rangle$, $\langle \hspace{0.2cm} \accentset{\left(\Yboxdim{3pt}\yng(2,2)\,,\,\emptyset\right)}{R}\hspace{0.1cm},\, \hspace{0.2cm}\accentset{\left(\Yboxdim{3pt}\yng(2,2)\,,\emptyset\right)}{R}\hspace{0.2cm}\rangle$, and as it is well known, there is at most three terms in the most general Lagrangian  quadratic in the Riemann tensor: the square of the Weyl tensor, of the Ricci tensor and of the Ricci scalar.\footnote{For $\Dim=2$ and $\Dim=3$, the Weyl tensor is identically zero. Additionally for $\Dim=2$ the traceless part of the Ricci tensor is zero \cite{hawking_ellis_1973,lee2018introduction}. For $\Dim=4$, the Gauss-Bonnet term is a boundary term \cite{Lanczos_1938}. This last fact generalizes to metric-affine gravity with torsion and with non-metricity containing the Weyl vector only (Weyl-Cartan spacetime) \cite{VONDERHEYDE_1975,hayashi1981gravity,Hehl_1991,Babourova_1997}. Note that the metric-affine generalization of Lovelock theories are projective invariant \cite{JANSSEN2019}.   We refer the reader to the thesis manuscript \cite[Part I, Chapter 4]{jimenez2022metric} and references therein for a general account on metric-affine Lovelock gravity.} When the connection has zero non-metricity, the decomposition of the Riemann tensor is also multiplicity-free (see appendix \ref{app:IrredRiemannT}) and the same reasoning applies for the construction quadratic curvature Lagrangians \cite{Hayashi_1980}.\medskip

It is the presence of non-metricity which is responsible for the non-uniqueness of the irreducible decomposition of the Riemann tensor (see appendix \ref{app:IrredRiemannT}). As a consequence, the same reasoning does not apply for the construction quadratic curvature Lagrangians. The construction of the quadratic Lagrangians in the Riemann curvature tensor $\tensor{\mathcal{R}}{_{abc}^d}$ and in the distortion tensor $\tensor{C}{^a_{bc}}$ are then guided by Proposition \eqref{prop:lagrangians}. The following computations are detailed in the Mathematica notebooks \cite{xMAGRiemann,xMAGdistortion}.\medskip


\paragraph{Quadratic curvature Lagrangians.} For convenience we recall here the definition of the traceless order two building blocks tensors of the irreducible decomposition \eqref{eq:Irreducible_Riemann} of a metric-affine Riemann tensor: 
		\begin{equation}\label{eq:building_blocks_Riemann_2}
	\begin{array}{lll}
		\tensor{\accentset{(1)}{\mathcal B}}{_{ab}}=\tensor{\accentset{(1)}{\mathcal{R}}}{_{(ab)}}-\tensor{\accentset{(2)}{\mathcal{R}}}{_{(ab)}}\,, \hspace{0.5cm}&
		\tensor{\accentset{(2)}{\mathcal B}}{_{ab}}=\tensor{\underline{\accentset{(1)}{\mathcal{R}}}}{_{(ab)}}+\tensor{\underline{\accentset{(2)}{\mathcal{R}}}}{_{(ab)}}\,,
		&\tensor{\accentset{(3)}{\mathcal B}}{_{ab}}=\tensor{\accentset{(1)}{\mathcal{R}}}{_{[ab]}}+\tensor{\accentset{(2)}{\mathcal{R}}}{_{[ab]}}\,,\\[6pt]
		\tensor{\accentset{(4)}{\mathcal B}}{_{ab}}=\tensor{\accentset{(1)}{\mathcal{R}}}{_{[ab]}}-\tensor{\accentset{(2)}{\mathcal{R}}}{_{[ab]}}-\tensor{\accentset{(3)}{\mathcal{R}}}{_{ab}}\,,& \tensor{\accentset{(5)}{\mathcal B}}{_{ab}}=\tensor{\accentset{(1)}{\mathcal{R}}}{_{[ab]}}-\tensor{\accentset{(2)}{\mathcal{R}}}{_{[ab]}}+\tensor{\accentset{(3)}{\mathcal{R}}}{_{ab}} \,.
	\end{array}
\end{equation}
The tensors $\tensor{\accentset{(1)}{\mathcal B}}{_{ab}}$ and $\tensor{\accentset{(2)}{\mathcal B}}{_{ab}}$ are symmetric with respect to permutation of indices while 
$\tensor{\accentset{(3)}{\mathcal B}}{_{ab}}$, $\tensor{\accentset{(4)}{\mathcal B}}{_{ab}}$ and $\tensor{\accentset{(5)}{\mathcal B}}{_{ab}}$ are antisymmetric. 
\begin{proposition}\label{prop:Lagrangian_Riemann}
	The general Lagrangian at most quadratic in the Riemann tensor obtained from the irreducible decomposition \ref{theo:irreducible_Riemann_O} is given by:
	\begin{equation}\label{eq:lagrangienR}
		\begin{aligned}
			\mathcal{L}_{\mathcal{R}}= \gamma_1\mathcal{R}& + \alpha_1 \,\, \accentset{\left(\Yboxdim{3pt}\yng(3,1)\right)}{\underline{\mathcal{R}}}\tensor{\vphantom{\mathcal{R}}}{^{\,abcd}}\,\, \tensor{{\accentset{\left(\Yboxdim{3pt}\yng(3,1)\right)}{\underline{\mathcal{R}}}}}{_{\, abcd}}+\alpha_2\,\, \accentset{\left(\Yboxdim{3pt}\yng(2,2)\right)}{\underline{\mathcal{R}}}\tensor{\vphantom{\mathcal{R}}}{^{\,abcd}}\,\, \tensor{{\accentset{\left(\Yboxdim{3pt}\yng(2,2)\right)}{\underline{\mathcal{R}}}}}{_{\, abcd}}+\alpha_3 \,\, \accentset{\left(\Yboxdim{3pt}\yng(2,1,1)S\right)}{\underline{\mathcal{R}}}\tensor{\vphantom{\mathcal{R}}}{^{\,abcd}}\,\,\,\, \tensor{{\accentset{\left(\Yboxdim{3pt}\yng(2,1,1)S\right)}{{\underline{\mathcal{R}}}}}}{_{\, abcd}}+\alpha_4 \,\, \accentset{\left(\Yboxdim{3pt}\yng(2,1,1)A\right)}{\underline{\mathcal{R}}}\tensor{\vphantom{\mathcal{R}}}{^{\,abcd}}\,\,\,\, \tensor{{\accentset{\left(\Yboxdim{3pt}\yng(2,1,1)A\right)}{{\underline{\mathcal{R}}}}}}{_{\, abcd}}\\
			&+\alpha_5\, \accentset{\left(\Yboxdim{3pt}\yng(1,1,1,1)\right)}{\underline{\mathcal{R}}}\tensor{\vphantom{\mathcal{R}}}{^{\,abcd}}\,\,\tensor{{\accentset{\left(\Yboxdim{3pt}\yng(1,1,1,1)\right)}{\underline{\mathcal{R}}}}}{_{\, abcd}}
			+\alpha_{34} \,\, \accentset{\left(\Yboxdim{3pt}\yng(2,1,1)S\right)}{\underline{\mathcal{R}}}\tensor{\vphantom{\mathcal{R}}}{^{\,abcd}}\,\,\,\,\tensor{{\accentset{\left(\Yboxdim{3pt}\yng(2,1,1)A\right)}{{\underline{\mathcal{R}}}}}}{_{\, acbd}}+ \beta_1 \, \tensor{\accentset{(1)}{\mathcal B}}{^{ab}} \,\tensor{\accentset{(1)}{\mathcal B}}{_{ab}}+\beta_2 \, \tensor{\accentset{(2)}{\mathcal B}}{^{ab}} \,\tensor{\accentset{(2)}{\mathcal B}}{_{ab}}+\beta_3 \, \tensor{\accentset{(3)}{\mathcal B}}{^{ab}} \,\tensor{\accentset{(3)}{\mathcal B}}{_{ab}}+\beta_4 \, \tensor{\accentset{(4)}{\mathcal B}}{^{ab}}\, \tensor{\accentset{(4)}{\mathcal B}}{_{ab}}\\
			&+\beta_5 \, \tensor{\accentset{(5)}{\mathcal B}}{^{ab}}\, \tensor{\accentset{(5)}{\mathcal B}}{_{ab}}+\beta_{12}\,\tensor{\accentset{(1)}{\mathcal B}}{^{ab}}\,\tensor{\accentset{(2)}{\mathcal B}}{_{ab}}+\beta_{34}\,\tensor{\accentset{(3)}{\mathcal B}}{^{ab}}\,\tensor{\accentset{(4)}{\mathcal B}}{_{ab}}+\,\beta_{35}\,\tensor{\accentset{(3)}{\mathcal B}}{^{ab}}\,\tensor{\accentset{(5)}{\mathcal B}}{_{ab}}+\beta_{45}\,\tensor{\accentset{(4)}{\mathcal B}}{^{ab}}\,\tensor{\accentset{(5)}{\mathcal B}}{_{ab}}+\,\gamma_2 \,\mathcal{R}^2\,.
		\end{aligned}
	\end{equation}
The conditions for the projective invariance of $\mathcal{L}_{\mathcal{R}}$ are
\begin{equation}
	\frac{\beta_5}{\beta_4}=\left(\frac{\Dim-2}{\Dim+2}\right)^2 \, \,, \hspace{1cm} \frac{\beta_{45}}{\beta_{4}}=2\left(\frac{\Dim-2}{\Dim+2 }\right)\,\,, \hspace{1cm} \frac{\beta_{34}}{\beta_{35}}=\left(\frac{\Dim+2}{\Dim-2}\right)\,.
\end{equation}
\end{proposition}
For the construction of this Lagrangian we have used the fact that each 2-traceless parts in the decomposition \eqref{eq:Irreducible_Riemann} is determined by exactly one of the building block tensors $\accentset{(i)}{\mathcal{B}}$. Indeed we have the following relations:
\begin{equation*}
	\begin{aligned}
		&\accentset{\left(\Yboxdim{3pt}\yng(3,1)\,,\Yboxdim{3pt}\yng(2)\right)}{\mathcal{R}}\tensor{\vphantom{\mathcal{R}}}{^{\hspace{0.3cm} abcd}}\hspace{0.3cm} \tensor{{\accentset{\left(\Yboxdim{3pt}\yng(3,1)\,,\Yboxdim{3pt}\yng(2)\right)}{\mathcal{R}}}}{_{\,\, abcd}}=\frac{\tensor{\accentset{(1)}{\mathcal{B}}}{^{ab}}\,\tensor{\accentset{(1)}{\mathcal{B}}}{_{ab}}}{\Dim}\, \,,\hspace{1cm}
		\accentset{\left(\Yboxdim{3pt}\yng(2,2)\,,\Yboxdim{3pt}\yng(2)\right)}{\mathcal{R}}\tensor{\vphantom{\mathcal{R}}}{^{\hspace{0.3cm} abcd}}\hspace{0.3cm}
		\tensor{{\accentset{\left(\Yboxdim{3pt}\yng(2,2)\,,\Yboxdim{3pt}\yng(2)\right)}{\mathcal{R}}}}{_{\,\, abcd}}=\frac{\tensor{\accentset{(2)}{\mathcal{B}}}{^{ab}}\,\tensor{\accentset{(2)}{\mathcal{B}}}{_{ab}}}{\Dim-2}\,,\hspace{1cm}
		\accentset{\left(\Yboxdim{3pt}\yng(2,1,1)A\,,\Yboxdim{3pt}\yng(1,1)\right)}{{\mathcal{R}}}\tensor{\vphantom{{\mathcal{R}}}}{^{\hspace{0.1cm} abcd}}\hspace{0.3cm} \tensor{{\accentset{\left(\Yboxdim{3pt}\yng(2,1,1)A\,,\Yboxdim{3pt}\yng(1,1)\right)}{{\mathcal{R}}}}}{_{\,\, abcd}}=\frac{\tensor{\accentset{(3)}{\mathcal{B}}}{^{ab}}\,\tensor{\accentset{(3)}{\mathcal{B}}}{_{ab}}}{\Dim-2}\,,\\
		&
		\accentset{\left(\Yboxdim{3pt}\yng(2,1,1)S\,,\Yboxdim{3pt}\yng(1,1)\right)}{{\mathcal{R}}}\tensor{\vphantom{{\mathcal{R}}}}{^{\hspace{0.1cm} abcd}}\hspace{0.3cm} \tensor{{\accentset{\left(\Yboxdim{3pt}\yng(2,1,1)S\,,\Yboxdim{3pt}\yng(1,1)\right)}{{\mathcal{R}}}}}{_{\,\, abcd}}=\dfrac{\tensor{\accentset{(4)}{\mathcal{B}}}{^{ab}}\,\tensor{\accentset{(4)}{\mathcal{B}}}{_{ab}}}{2(\Dim-2)}\,,\hspace{1cm}
		\accentset{\left(\Yboxdim{3pt}\yng(3,1)\,,\Yboxdim{3pt}\yng(1,1)\right)}{\mathcal{R}}\tensor{\vphantom{\mathcal{R}}}{^{\hspace{0.3cm} abcd}}\hspace{0.3cm} \tensor{{\accentset{\left(\Yboxdim{3pt}\yng(3,1)\,,\Yboxdim{3pt}\yng(1,1)\right)}{\mathcal{R}}}}{_{\,\, abcd}}=\dfrac{\tensor{\accentset{(5)}{\mathcal{B}}}{^{ab}}\,\tensor{\accentset{(5)}{\mathcal{B}}}{_{ab}}}{2(\Dim+2)}\, \,,\hspace{0.8cm}
		\accentset{\left(\Yboxdim{3pt}\yng(3,1)\,,\Yboxdim{3pt}\yng(2)\right)}{\mathcal{R}}\tensor{\vphantom{\mathcal{R}}}{^{\hspace{0.3cm} a}_c^c^b}\hspace{0.3cm}\accentset{\left(\Yboxdim{3pt}\yng(2,2)\,,\Yboxdim{3pt}\yng(2)\right)}{\mathcal{R}}\tensor{\vphantom{\mathcal{R}}}{_{\hspace{0.3cm} a}_c^c_b}=-\frac{1}{4}\tensor{\accentset{(1)}{\mathcal{B}}}{^{ab}}\,\tensor{\accentset{(2)}{\mathcal{B}}}{_{ab}}\,,\\
		&
		\accentset{\left(\Yboxdim{3pt}\yng(2,1,1)A\,,\Yboxdim{3pt}\yng(1,1)\right)}{\mathcal{R}}\tensor{\vphantom{\mathcal{R}}}{^{\hspace{0.2cm} a}_c^c^b}\hspace{0.2cm}\accentset{\left(\Yboxdim{3pt}\yng(2,1,1)S\,,\Yboxdim{3pt}\yng(1,1)\right)}{\mathcal{R}}\tensor{\vphantom{\mathcal{R}}}{_{\hspace{0.3cm} a}_c^c_b}=-\frac{1}{8}\tensor{\accentset{(3)}{\mathcal{B}}}{^{ab}}\,\tensor{\accentset{(4)}{\mathcal{B}}}{_{ab}}\,,\hspace{0.6cm}
		\accentset{\left(\Yboxdim{3pt}\yng(2,1,1)A\,,\Yboxdim{3pt}\yng(1,1)\right)}{\mathcal{R}}\tensor{\vphantom{\mathcal{R}}}{^{\hspace{0.3cm} a}_c^c^b}\hspace{0.3cm}\accentset{\left(\Yboxdim{3pt}\yng(3,1)\,,\Yboxdim{3pt}\yng(1,1)\right)}{\mathcal{R}}\tensor{\vphantom{\mathcal{R}}}{_{\hspace{0.1cm} a}_c^c_b}=-\frac{1}{8}\tensor{\accentset{(3)}{\mathcal{B}}}{^{ab}}\,\tensor{\accentset{(5)}{\mathcal{B}}}{_{ab}}\,,\hspace{0.6cm}
		\accentset{\left(\Yboxdim{3pt}\yng(2,1,1)S\,,\Yboxdim{3pt}\yng(1,1)\right)}{\mathcal{R}}\tensor{\vphantom{\mathcal{R}}}{^{\hspace{0.3cm} a}_c^c^b}\hspace{0.3cm}\accentset{\left(\Yboxdim{3pt}\yng(3,1)\,,\Yboxdim{3pt}\yng(1,1)\right)}{\mathcal{R}}\tensor{\vphantom{\mathcal{R}}}{_{\hspace{0.3cm} a}_c^c_b}=\frac{1}{16}\tensor{\accentset{(4)}{\mathcal{B}}}{^{ab}}\,\tensor{\accentset{(5)}{\mathcal{B}}}{_{ab}}\,.
	\end{aligned}
\end{equation*}

\begin{remark}
\begin{itemize}
\item[\it{i)}] For $\Dim=2$, one has (see Proposition \ref{prop:small_d_Riemann}) 
	\begin{equation}\label{eq:lagrangienR_2d}
		\mathcal{L}_{\mathcal{R}}= \gamma_1\mathcal{R}+ \beta_1 \, \tensor{\accentset{(1)}{\mathcal B}}{^{ab}} \,\tensor{\accentset{(1)}{\mathcal B}}{_{ab}}+\beta_5 \, \tensor{\accentset{(5)}{\mathcal B}}{^{ab}}\, \tensor{\accentset{(5)}{\mathcal B}}{_{ab}}\,+\,\gamma_2 \,\mathcal{R}^2\,,
\end{equation}
and the condition of projective invariance of the action is $\beta_5=0$\,.
\item[\it{ii)}] For $\Dim=3$, one has (see Proposition \ref{prop:small_d_Riemann}) 
	\begin{equation}\label{eq:lagrangienR_3d}
	\begin{aligned}
		\mathcal{L}_{\mathcal{R}}= &\gamma_1\mathcal{R} + \alpha_1 \,\, \accentset{\left(\Yboxdim{3pt}\yng(3,1)\right)}{\underline{\mathcal{R}}}\tensor{\vphantom{\mathcal{R}}}{^{\,abcd}}\,\, \tensor{{\accentset{\left(\Yboxdim{3pt}\yng(3,1)\right)}{\underline{\mathcal{R}}}}}{_{\, abcd}}+ \beta_1 \, \tensor{\accentset{(1)}{\mathcal B}}{^{ab}} \,\tensor{\accentset{(1)}{\mathcal B}}{_{ab}}+\beta_2 \, \tensor{\accentset{(2)}{\mathcal B}}{^{ab}} \,\tensor{\accentset{(2)}{\mathcal B}}{_{ab}}+\beta_3 \, \tensor{\accentset{(3)}{\mathcal B}}{^{ab}} \,\tensor{\accentset{(3)}{\mathcal B}}{_{ab}}+\beta_4 \, \tensor{\accentset{(4)}{\mathcal B}}{^{ab}}\, \tensor{\accentset{(4)}{\mathcal B}}{_{ab}}\\
		&+\beta_5 \, \tensor{\accentset{(5)}{\mathcal B}}{^{ab}}\, \tensor{\accentset{(5)}{\mathcal B}}{_{ab}}+\beta_{12}\,\tensor{\accentset{(1)}{\mathcal B}}{^{ab}}\,\tensor{\accentset{(2)}{\mathcal B}}{_{ab}}+\beta_{34}\,\tensor{\accentset{(3)}{\mathcal B}}{^{ab}}\,\tensor{\accentset{(4)}{\mathcal B}}{_{ab}}+\,\beta_{35}\,\tensor{\accentset{(3)}{\mathcal B}}{^{ab}}\,\tensor{\accentset{(5)}{\mathcal B}}{_{ab}}+\beta_{45}\,\tensor{\accentset{(4)}{\mathcal B}}{^{ab}}\,\tensor{\accentset{(5)}{\mathcal B}}{_{ab}}+\,\gamma_2 \,\mathcal{R}^2\,.
	\end{aligned}
\end{equation}
\end{itemize}
\end{remark}
\vskip 8pt

The conditions for which $\mathcal{L}_{\mathcal{R}}$ reduces to the Einstein-Hilbert action of General Relativity in the limit of zero distortion are
\begin{equation}
	\alpha_2=\beta_2=\gamma_2=0.
\end{equation}
These last conditions can be understood from the form of the irreducible decomposition of the metric Riemann tensor \eqref{eq:irreducible_Riemann_metric}. As a consequence the most general projective invariant Lagrangian which reduces to the Einstein-Hilbert Lagrangian in the limit of zero distortion has $11$ free parameters and may be written as: 

\begin{equation}\label{eq:lagrangienR_EH}
	\begin{aligned}
		\mathcal{L}_{\mathcal{R}}=\gamma_1\mathcal{R} & + \alpha_1 \,\, \accentset{\left(\Yboxdim{3pt}\yng(3,1)\right)}{\underline{\mathcal{R}}}\tensor{\vphantom{\mathcal{R}}}{^{\,abcd}}\,\, \tensor{{\accentset{\left(\Yboxdim{3pt}\yng(3,1)\right)}{\underline{\mathcal{R}}}}}{_{\, abcd}}+\alpha_3 \,\, \accentset{\left(\Yboxdim{3pt}\yng(2,1,1)S\right)}{\underline{\mathcal{R}}}\tensor{\vphantom{\mathcal{R}}}{^{\,abcd}}\,\,\,\, \tensor{{\accentset{\left(\Yboxdim{3pt}\yng(2,1,1)S\right)}{{\underline{\mathcal{R}}}}}}{_{\, abcd}}+\alpha_4 \,\, \accentset{\left(\Yboxdim{3pt}\yng(2,1,1)A\right)}{\underline{\mathcal{R}}}\tensor{\vphantom{\mathcal{R}}}{^{\,abcd}}\,\,\,\, \tensor{{\accentset{\left(\Yboxdim{3pt}\yng(2,1,1)A\right)}{{\underline{\mathcal{R}}}}}}{_{\, abcd}}
		+\alpha_5\, \accentset{\left(\Yboxdim{3pt}\yng(1,1,1,1)\right)}{\underline{\mathcal{R}}}\tensor{\vphantom{\mathcal{R}}}{^{\,abcd}}\,\,\tensor{{\accentset{\left(\Yboxdim{3pt}\yng(1,1,1,1)\right)}{\underline{\mathcal{R}}}}}{_{\, abcd}}\\
		&+\alpha_{34} \,\, \accentset{\left(\Yboxdim{3pt}\yng(2,1,1)S\right)}{\underline{\mathcal{R}}}\tensor{\vphantom{\mathcal{R}}}{^{\,abcd}}\,\,\,\,\tensor{{\accentset{\left(\Yboxdim{3pt}\yng(2,1,1)A\right)}{{\underline{\mathcal{R}}}}}}{_{\, acbd}}+\beta_4 \left(\, \tensor{\accentset{(4)}{\mathcal B}}{^{ab}}\, \tensor{\accentset{(4)}{\mathcal B}}{_{ab}}+ \, f_{45}(\Dim)\,\tensor{\accentset{(4)}{\mathcal B}}{^{ab}}\,\tensor{\accentset{(5)}{\mathcal B}}{_{ab}}+ f_{5}(\Dim)\tensor{\accentset{(5)}{\mathcal B}}{^{ab}}\, \tensor{\accentset{(5)}{\mathcal B}}{_{ab}}\right)\,\\
		&+ \,\beta_{35}\left(\tensor{\accentset{(3)}{\mathcal B}}{^{ab}}\,\tensor{\accentset{(5)}{\mathcal B}}{_{ab}}+f_{34}(\Dim)\,\tensor{\accentset{(3)}{\mathcal B}}{^{ab}}\,\tensor{\accentset{(4)}{\mathcal B}}{_{ab}}\right)+\beta_1 \, \tensor{\accentset{(1)}{\mathcal B}}{^{ab}} \,\tensor{\accentset{(1)}{\mathcal B}}{_{ab}}+\beta_3 \, \tensor{\accentset{(3)}{\mathcal B}}{^{ab}} \,\tensor{\accentset{(3)}{\mathcal B}}{_{ab}}+\beta_{12}\,\tensor{\accentset{(1)}{\mathcal B}}{^{ab}}\,\tensor{\accentset{(2)}{\mathcal B}}{_{ab}},\\
	\end{aligned}
\end{equation}
where 
\begin{equation}
f_{5}(\Dim)=\left(\frac{\Dim-2}{\Dim+2}\right)^2\,, \hspace{0.5cm} f_{45}(\Dim)=2\left(\frac{\Dim-2}{\Dim+2 }\right)\,, \hspace{0.5cm} f_{34}(\Dim)=\left(\frac{\Dim+2}{\Dim-2}\right)\,.
\end{equation}

\vskip 6pt
\paragraph{Quadratic distortion Lagrangians.} 
\begin{proposition}\label{prop:Lagrangian_distortion}
	The most general Lagrangian quadratic in the distortion tensor constructed from the decomposition \ref{theo:irreducible_distortion_O} is 
	\begin{equation}\label{eq:lagrangienC}
		\begin{aligned}
			\mathcal{L}_C=&c_1\,\, \accentset{\left(\Yboxdim{3pt}\yng(3)\right)}{\underline{C}}\tensor{\vphantom{C}}{^{\,\,abc}}\,\,\tensor{{\accentset{\left(\Yboxdim{3pt}\yng(3)\right)}{\underline{C}}}}{_{\,abc}}+c_2\,\, \accentset{\left(\Yboxdim{3pt}\yng(2,1)S\right)}{\underline{C}}\tensor{\vphantom{C}}{^{\,\,abc}}\,\,\tensor{{\accentset{\left(\Yboxdim{3pt}\yng(2,1)S\right)}{\underline{C}}}}{_{\,abc}}+c_3\,\, \accentset{\left(\Yboxdim{3pt}\yng(2,1)A\right)}{\underline{C}}\tensor{\vphantom{C}}{^{\,\,abc}}\,\,\tensor{{\accentset{\left(\Yboxdim{3pt}\yng(2,1)A\right)}{\underline{C}}}}{_{\,abc}}+c_4\,\, \accentset{\left(\Yboxdim{3pt}\yng(1,1,1)\right)}{\underline{C}}\tensor{\vphantom{C}}{^{\,\,abc}}\,\,\tensor{{\accentset{\left(\Yboxdim{3pt}\yng(1,1,1)\right)}{\underline{C}}}}{_{\,abc}}+ c_{23}\,\,\accentset{\left(\Yboxdim{3pt}\yng(2,1)S\right)}{\underline{C}}\tensor{\vphantom{C}}{^{\,\,abc}}\,\,\tensor{{\accentset{\left(\Yboxdim{3pt}\yng(2,1)A\right)}{\underline{C}}}}{_{\,acb}}\,\\
			&+b_1 \,\accentset{(1)}{B}\tensor{\vphantom{B}}{^{\, a}}\,\tensor{{\accentset{(1)}{B}}}{_a}+b_2 \,\accentset{(2)}{B}\tensor{\vphantom{B}}{^{\, a}}\,\tensor{{\accentset{(2)}{B}}}{_a}+b_3 \,\accentset{(3)}{B}\tensor{\vphantom{B}}{^{\, a}}\,\tensor{{\accentset{(3)}{B}}}{_a}+b_{12} \,\accentset{(1)}{B}\tensor{\vphantom{B}}{^{\, a}}\,\tensor{{\accentset{(2)}{B}}}{_a}+b_{13} \,\accentset{(1)}{B}\tensor{\vphantom{B}}{^{\, a}}\,\tensor{{\accentset{(3)}{B}}}{_a}+b_{23} \,\accentset{(2)}{B}\tensor{\vphantom{B}}{^{\, a}}\,\tensor{{\accentset{(3)}{B}}}{_a}\,.
		\end{aligned}
	\end{equation}
The conditions for the projective invariance of \eqref{eq:lagrangienC} are 
\begin{equation}
	b_{1}=4\,\left(\frac{\Dim-1}{\Dim+2}\right)^2 b_{2}\,, \hspace{1cm}  b_{12}=4\,\left(\frac{\Dim-1}{\Dim+2}\right) b_2\,, \hspace{1cm} b_{13}=2\,\left(\frac{\Dim-1}{\Dim+2}\right) b_{23}\,.
\end{equation}
\end{proposition}
We recall that the definition of the trace building block tensors $\accentset{(i)}{B}$ are given in \eqref{eq:building_blocks_distortion}. Again we took advantage of the fact that each vector trace parts in the decomposition \eqref{eq:Irreducible_distortion} is determined by exactly one of the building block tensors $\accentset{(i)}{B}$. Indeed we have the following relations:
\begin{equation*}
	\begin{aligned}
		&\accentset{\left(\Yboxdim{3pt}\yng(3), \Yboxdim{3pt}\yng(1)\right)}{\underline{C}}\tensor{\vphantom{C}}{^{\,\,\,\,ab}_b}\,\,\,\,\tensor{{\accentset{\left(\Yboxdim{3pt}\yng(3),\Yboxdim{3pt}\yng(1)\right)}{\underline{C}}}}{_{\,\,\,abc}}=\frac{\tensor{\accentset{(1)}{B}}{_{\,a}}\,\tensor{\accentset{(1)}{B}}{^{\,a}}}{3(\Dim+2)}\,, \hspace{1.3cm}
		\accentset{\left(\Yboxdim{3pt}\yng(2,1)S, \Yboxdim{3pt}\yng(1)\right)}{\underline{C}}\tensor{\vphantom{C}}{^{\,\,\,\,abc}}\,\,\,\,\tensor{{\accentset{\left(\Yboxdim{3pt}\yng(2,1)S,\Yboxdim{3pt}\yng(1)\right)}{\underline{C}}}}{_{\,\,\,abc}}=\frac{\tensor{\accentset{(2)}{B}}{_{\,a}}\,\tensor{\accentset{(2)}{B}}{^{\,a}}}{6(\Dim-1)}\,, \hspace{1.3cm}
		\accentset{\left(\Yboxdim{3pt}\yng(2,1)A, \Yboxdim{3pt}\yng(1)\right)}{\underline{C}}\tensor{\vphantom{C}}{^{\,\,\,\,abc}}\,\,\,\,\tensor{{\accentset{\left(\Yboxdim{3pt}\yng(2,1)A,\Yboxdim{3pt}\yng(1)\right)}{\underline{C}}}}{_{\,\,\,abc}}=\frac{\tensor{\accentset{(3)}{B}}{_{\,a}}\,\tensor{\accentset{(3)}{B}}{^{\,a}}}{2(\Dim-1)}\,, \hspace{1.3cm}\\
		&\accentset{\left(\Yboxdim{3pt}\yng(3), \Yboxdim{3pt}\yng(1)\right)}{\underline{C}}\tensor{\vphantom{C}}{^{\,\,\,\,ab}_b}\,\,\,\,\accentset{\left(\Yboxdim{3pt}\yng(2,1)S,\Yboxdim{3pt}\yng(1)\right)}{\underline{C}}\tensor{\vphantom{C}}{_{\,\,\,\,a}^b_b}=\frac{1}{18}\,\tensor{\accentset{(1)}{B}}{_{\,a}}\,\tensor{\accentset{(2)}{B}}{^{\,a}}\,, \hspace{1.1cm}
		\accentset{\left(\Yboxdim{3pt}\yng(3), \Yboxdim{3pt}\yng(1)\right)}{\underline{C}}\tensor{\vphantom{C}}{^{\,\,\,\,ab}_b}\,\,\,\,\accentset{\left(\Yboxdim{3pt}\yng(2,1)A,\Yboxdim{3pt}\yng(1)\right)}{\underline{C}}\tensor{\vphantom{C}}{_{\,\,\,\,a}^b_b}=\frac{1}{6}\,\tensor{\accentset{(1)}{B}}{_{\,a}}\,\tensor{\accentset{(3)}{B}}{^{\,a}}\,, \hspace{1.2cm}
		\accentset{\left(\Yboxdim{3pt}\yng(2,1)S, \Yboxdim{3pt}\yng(1)\right)}{\underline{C}}\tensor{\vphantom{C}}{^{\,\,\,\,ab}_b}\,\,\,\,\accentset{\left(\Yboxdim{3pt}\yng(2,1)A,\Yboxdim{3pt}\yng(1)\right)}{\underline{C}}\tensor{\vphantom{C}}{_{\,\,\,\,a}^b_b}=\frac{1}{12}\,\tensor{\accentset{(2)}{B}}{_{\,a}}\,\tensor{\accentset{(3)}{B}}{^{\,a}}\,. \hspace{1cm}\\
	\end{aligned}
\end{equation*}

\paragraph{Perspectives and future work.} \hphantom{A}\medskip
\begin{itemize}
\item[\it{i)}] One of the first things to do would be to provide the corresponding irreducible decompositions and quadratic Lagrangians in the language differential forms which is widely use in metric-affine gravity. As it was done in \cite{McCrea_1992}, one should also extend these irreducible decompositions to the Bianchi identities of metric-affine gravity.
\item[\it{ii)}] One should try to generalize the result of \cite{percacci2020new} by analyzing the physical viability (absence of ghosts and tachyons) for the sum of the projective invariant Lagrangians \eqref{eq:lagrangienR_EH}-\eqref{eq:lagrangienC} when perturbing around Minkowski and $\textup{(A)dS}$ backgrounds with a zero vaccum distortion tensor. For this, it would be convenient to consider a volume-preserving connection \cite{thomas1926asymmetric,Hehl1981}:
\begin{equation}\label{eq:normal_projective_connection}
	\tensor{\overline{\Gamma}}{^\sigma_\alpha_\beta}:=\tensor{\Gamma}{^\sigma_\alpha_\beta}-\tensor{\delta}{^\sigma_\beta}\tensor{\accentset{(2)}{C}}{_\alpha}=\tensor{\Gamma}{^\sigma_\alpha_\beta}+\dfrac{1}{2}\tensor{\delta}{^\sigma_\beta} Q_\alpha \,.
\end{equation}
The connection $\tensor{\overline{\Gamma}}{^{\,\sigma}_\alpha_\beta}$ is indeed a volume-preserving connection: $\tensor{\overline{\Gamma}}{^{\,\sigma}_\alpha_\sigma}=\Bigl \{ \tensor{}{^{\sigma}_{\alpha\sigma}}\Bigr \}$ which implies
\begin{equation}
	\overline{\nabla}_\alpha (\sqrt{|g|})=\accentset{(g)}{\nabla}_\alpha (\sqrt{|g|})=0\,.
\end{equation}
With this gauge fixing of the projective symmetry, the homothetic tensor $\tensor{\accentset{(3)}{\mathcal{R}}}{_{ab}}$ is identically zero (recall \eqref{eq:homothetic_tensor}). As a consequence, one has $\tensor{\accentset{(4)}{\mathcal B}}{_{ab}}=\tensor{\accentset{(5)}{\mathcal B}}{_{ab}}$ (see \eqref{eq:building_blocks_Riemann}) and the expression of the Lagrangian \eqref{eq:lagrangienR_EH} simplifies significantly.\smallskip

While performing linear analysis of the theory, particular attention should be paid to the totally traceless symmetric part of the non-metricity, which is identified as a spin-$3$ field \cite{Nicolas_2006}: although the free massless spin-3 field can be embedded in linearized metric-affine gravity, the associated linearized gauge symmetry is absent at the non-linear level (see, e.g., \cite{percacci2020new}). Note that there exists exact solutions in particular non-linear metric-affine gravity theories with propagating totally symmetric traceless non-metricity \cite{Nicolas_2006} which is interesting from the perspective of bypassing higher-spin no-go theorems \cite{bekaert2012higher,bekaert2020nonlinear}.
\item[\it{iii)}] It would also be interesting to investigate the Cartan geometry associated with the projective structure of metric-affine gravity \eqref{def:proj_tor}, and to analyze the constructed Lagrangians in the light of the dressing field method \cite{berghofer2021gauge2,franccois2014reduction}. 
\end{itemize}

	\chapter{The projection operators for $\GL(\Dim,\mathbb{C})$ irreducible decomposition of tensors}\label{chap:projectors_GL}
\vspace{0.5 cm}

\section{Introduction}\label{sec:Introduction_chapter3}

In this chapter we aim at constructing the projectors $Z^\mu\in \C\sn$ which perform the isotypic decomposition of $V^{\otimes n}$ with respect to $\GL(\Dim,\C)$:
\begin{equation}\label{eq:isotypic_decomposition_Schur_Weyl_2}
	V^{\otimes n}=\bigoplus_{\mu\in \Par_n(\Dim)}(V^{\otimes n})\cdot Z^{\mu},\hspace{0.5cm} \text{such that}\hspace{0.5cm} (V^{\otimes n})\cdot Z^{\mu}=\left(V^{\mu}\right)^{\oplus m_\mu}\,.
\end{equation}
Recall the definition of $\Par_n(\Dim)$ in \eqref{eq:P_d}.\, Any element $Z^\mu$ acts by the identity on $V^\mu$ and annihilates any irreducible module\footnote{In the sequel we find it more convenient to use the terminology of modules instead of representations. In the context of algebras the two notions are equivalent (see appendix \ref{subsec:definitions} for more details).}(representation) of $\GL(\Dim,\C)$ not isomorphic to $V^\mu$.\medskip

By Maschke's Theorem, any module over $\C\sn$ is completely reducible. In particular the tensor product space $V^{\otimes n}$ is a completely reducible $\C \mathfrak{S}_n$-module: 
\begin{equation}
	V^{\otimes n}\cong\bigoplus_{\mu\in \Par_n(\Dim)} \left(L^{\mu}\right)^{\oplus g_{\mu}}\,.
\end{equation}
As a consequence of the Schur-Weyl duality for $\GL(\Dim,\mathbb{C})$ and $\mathfrak{S}_n$
\begin{equation}\label{eq:schur_weyl_iso}
	\left(V^{\mu}\right)^{\oplus m_\mu}=\left(L^{\mu}\right)^{\oplus g_\mu}\,,
\end{equation}
where $g_{\mu}=\dim(V^\mu)$ and $m_{\mu}=\dim(L^\mu)$.
The last equation implies that projecting to an isotypic component of $V^{\otimes n}$ with respect to $\GL(\Dim,\C)$ action amounts to projecting to the corresponding isotypic component with respect to $\C\sn$ action. Hence, each element $Z^\mu$ acts by the identity on the irreducible module $L^\mu$ and annihilates any irreducible $\C\sn$-module not isomorphic to $L^\mu$. Such element is very well known to mathematicians: they are the central idempotents of $\C\sn$. In this manuscript we call them the central Young idempotents.\medskip

The central Young idempotents form a complete set of central pairwise orthogonal idempotents:\medskip 
\vskip 4pt
\begin{equation}\label{eq:properties_central_idempotents}
	\begin{aligned}
		&Z^\mu \ v= v \ Z^\mu \quad \text{for any } \, v\in \C\sn\,, \quad &(\textit{central}\,) \\
		&Z^\rho Z^\mu=0  \quad \text{for any } \rho\neq\lambda\,, \quad &(\textit{pairwise orthogonal}\,) \\
		&Z^\mu Z^\mu=Z^\mu\,, \quad &(\textit{idempotent \text{\slash} projector}\,)\\
		&\id_{\C\sn}=\sum_{\mu\vdash n}Z^\mu\,. \quad &(\textit{partition of unity}\,)\\
	\end{aligned}
\end{equation} \medskip
\vskip 4pt
The central Young idempotents can be constructed via the projection character formula \cite{Janusz_1966} presented in section \ref{subsec:irreduciblecharacter}. In this chapter we propose an alternative approach based on a Lagrange-interpolation-type formula (see \cite[Section 2]{doty2019canonical} for a general account on this type of formula). The first ingredient of the construction is a particular element $T_n$ of the center of $\C\sn$ which is diagonal on irreducible $\C\sn$-module $L^\mu$, with eigenvalue given by the content of the Young diagram $\mu$ (see \eqref{eq:content_YD} for the definition of the content of a Young diagram). The second ingredient is an induction map $\mathcal{L}: \C\Sn{n-1}\to \C\sn$ defined in section \ref{subsec:inductiveformula}. The important property of $\mathcal{L}$ can be formulated as follows (Proposition \ref{prop:line_induction_sn}): for any $\,\nu\vdash n-1 $ and $\mu\vdash n$, 
\begin{equation}
	\mathcal{L}(Z^{\nu})Z^{\mu}=\left\{
	\begin{array}{ll}
		z_{\mu\backslash\nu} \, Z^{\mu}\, \quad \text{ with $\,z_{\mu\backslash\nu}=n \, \dfrac{\dim(L^\nu)}{\dim(L^\mu)}$}\,, &\text{if} \quad \mu\in \mathcal{A}_{\nu}\,,\\
		\quad 0\,\, \quad &\text{otherwise,}
	\end{array}
	\right.
\end{equation}
where $\mathcal{A}_\nu$ (respectively $\mathcal{R}_\nu$) is the set of Young diagram with one box added to $\nu$ (respectively one box removed from $\nu$). As a consequence of the properties of $T_n$ and $\mathcal{L}$ one arrives at the main result of this chapter (see Theorem \ref{Theorem:centralYoung}): for any $\mu\vdash n$ and $\nu \in \mathcal{R}_{\mu}$ the element 

\begin{equation}\label{eq:gen_central_idempotent_Sn_Intro}
	\mathcal{Z}^{\mu}(\nu)=\frac{\mathcal{L}(Z^{\nu})}{z_{\mu\backslash\nu}}\prod_{\begin{array}{c}
			{\scriptstyle \rho\in \mathcal{A}_\nu}\\
			{\scriptstyle \rho\neq\mu}
	\end{array}}\dfrac{c_\rho-T_n}{c_\rho-c_\mu}\,, 
\end{equation}  
coincides with the central Young idempotent $Z^\mu$ of $\C\sn$.
\vskip 6pt

We also aim at constructing the projectors which perform the irreducible decomposition of tensors with respect to $\GL(\Dim,\C)$ \eqref{eq:Irreducible_decomposition_GL}. As already mentioned in the previous chapter, these projectors are the Young seminormal idempotents $Y^{\,\tab}$ which are labeled by the standard tableaux $\tab$. It appears that there exists a bridge connecting the isotypic decomposition of a tensor to its irreducible decomposition. This bridge is embodied by the following relation
\begin{equation}\label{eq:Young_seminormal_idempotents_intro}
	Y^{\,\tab}=\prod_{\mu \in \tab} Z^{\mu} \,,
\end{equation}
where the standard tableau $\tab$ is understood as a path in the \textit{Bratteli diagram} for the chain of embedded algebras 
\begin{equation}\label{eq:chain_embedding_sn_0}
	\C\cong\C\mathfrak{S}_0\subseteq\C\mathfrak{S}_1\subset \ldots \subset \C\mathfrak{S}_{n-1}\subset \C\mathfrak{S}_{n}\,.
\end{equation}
In \eqref{eq:Young_seminormal_idempotents_intro} the natural embedding $\C\Sn{k}\hookrightarrow\C\Sn{k+1}$ is realized by the insertion of a vertical line at the right end of permutation diagrams. This formula originates from the Okounkov-Vershik approach to the representation theory of symmetric groups \cite{Okounkov1996,OV_2} (see also \cite{ceccherini2010representation,block2018representation}), and was demonstrated explicitly in \cite[Prop. 4.7]{doty2019canonical}. The formula \eqref{eq:Young_seminormal_idempotents_intro} will be explained in some details in the last section.

\section{Ideals and regular representation of $\mathbb{C}\mathfrak{S}_n$}\label{sec:ideals}
\paragraph{Ideals of $\C\sn$.} A right (resp. left) ideal of an algebra $A$ is a subalgebra $I\subseteq A$ such that $I\,a\subseteq I$ (resp. $a\, I\subseteq I$) for all $a\in A$. In other word, an ideal of an algebra $A$ is a subalgebra which is also an invariant subspace with respect to the action of the algebra on itself. Hence, ideals are sub-modules (sub-representations) of the module (representation) corresponding to the action of the algebra on itself. This module is called the right (resp. left) regular module, and the decomposition of an algebra into \textit{minimal} ideals is the same as the decomposition of the right (resp. left) regular module into irreducible modules.\medskip

The symmetric group algebra splits into a direct sum of (two sided\footnote{A two sided ideal is both a right and a left ideal.}) ideals 
\begin{equation}
\C\sn=\bigoplus_{\mu\vdash n} I^{\mu}\,,\hspace{0.5cm}\text{with}\hspace{0.5cm} I^\mu=Z^\mu \C\sn\,,
\end{equation}
and one says that the ideal $I^\mu$ is generated by the central idempotent $Z^\mu$. This decomposition of $\C\sn$ is not \text{maximal} in the sense that it can be decomposed into smaller ideals called \text{minimal} ideals. These minimal ideals are generated by primitive idempotents of $\C\sn$ (for example the Young symmetrizers). A special class of such idempotents are the one which are pairwise orthogonal and hence form a partition of unity. As already mentioned in the previous chapter and in the introduction above, these idempotents are the Young seminormal idempotents, and one has\footnote{We choose to work with right ideals but we could just as well work with left ideals.}: 
\begin{equation}
\C\sn=\bigoplus_{\tab\in \Tab(n)} I^{\,\tab}\,,\hspace{0.5cm}\text{with}\hspace{0.5cm} I^{\,\tab}=Y^{\,\tab} \C\sn\,,
\end{equation}
where $\Tab(n)$ is the set of standard tableaux with $n$ boxes. One may say that the right ideal $I^{\,\tab}$ is generated by the left action of $Y^{\,\tab}$ on $\C\sn$. An important remark is that the set of Young seminormal idempotents parametrized by standard tableaux of shape $\mu$ generate a partition of unity of the ideal $I^\mu$:
\begin{equation}\label{eq:partition_Zmu}
Z^\mu=\sum_{\tab\in \textup{Tab}(\mu)} Y^{\,\tab}\,.
\end{equation}
\medskip

\paragraph{The right regular module.} A module over an algebra $A$ (that is a $A$-module) is semisimple if it is a direct sum of irreducible module, and $A$ is said to be semisimple if all its $A$-module are semisimple. Besides, the right regular module over a semisimple algebra contains all its pairwise inequivalent irreducible modules. \medskip

By Maschke's Theorem $\C\sn$ is a semisimple algebra and writing $L^\mu$ \footnote{For our purpose we may think of $L^\mu$ as the irreducible tensor representations of $\C\sn$ appearing in $V^{\otimes n}$ hence we keep the same name.} for a representative in the class of isomorphic irreducible $\C\sn$-module one has:
\begin{equation}\label{eq:irreducible_decomp_Maschke}
	\C\sn\cong\, \bigoplus_{\mu\,\vdash n} \left(L^\mu\right)^{\oplus m_\mu} \hspace{1cm} \text{(isomorphism of vector spaces)}\,.
\end{equation}
According to the Wedderburn-Artin theorem \ref{theo:Wedderburn} one has
\begin{equation}
	\mathbb{C}\mathfrak{S}_n=\bigoplus_{\mu\,\vdash\,n} Z^{\mu}\mathbb{C}\mathfrak{S}_n\cong\bigoplus_{\mu\vdash n} \textup{End}(L^\mu) \hspace{1cm} \text{(isomorphism of algebras)}\,,
\end{equation}
where $Z^\mu\mathbb{C}\mathfrak{S}_n$ is the ideal isomorphic to $\textup{End}(L^\mu)$. As a consequence, one has that the multiplicity $m_\mu$ in $\eqref{eq:irreducible_decomp_Maschke}$ is equal to the dimension of $L^\mu$, and hence the multiplicities $m_\mu$ in \eqref{eq:schur_weyl_iso} and \eqref{eq:irreducible_decomp_Maschke} are the same. Still, we will write $\mathrm{d}_\mu$ for the dimension of an irreducible $\C\sn$-module $L^\mu$. 
\section{The center  $\mathcal{Z}_n$ of $\mathbb{C}\mathfrak{S}_n$}\label{sec:centerSn}

In this section we present selected results which can be found in many textbooks related to the symmetric group. For the most part they can be generalized to the centralizer $\cn$ of $\sn$ in the Brauer algebra which is the subject of chapter \ref{chap:cn}. 
\paragraph{Integer partitions, Young diagrams and unary bracelets.} For convenience we recall that an integer partition $\mu=\left(\mu_1,\mu_2,\ldots,\mu_l\right)$ of $n$ denoted $\mu\vdash n$ is a sequence of non increasing positive integers such that $|\mu|=\sum_{i=1}^{l}\mu_i=n$. To each such integer partition we identify a \textit{Young diagram}: an array of $n$ boxes arranged in $l$ left justified rows. For example, the partitions of $n=4$ are: 
\begin{equation*}
	\left(4\right)=\Yboxdim{8pt}\yng(4)\hspace{0.8cm} \left(3,1\right)=\Yboxdim{8pt}\yng(3,1)\hspace{0.8cm} \left(2,2\right)=\Yboxdim{8pt}\yng(2,2)\hspace{0.8cm} \left(2,1,1\right)=\Yboxdim{8pt}\yng(2,1,1)\hspace{0.8cm}\left(1,1,1,1\right)=\Yboxdim{8pt}\yng(1,1,1,1)
\end{equation*}
To each box of a Young diagram at a position $(i,j)$ we associate its {\it content} $c(i,j) = j - i$. In this respect we define the \textit{content} of a Young diagram $\mu$ as the sum:
\begin{equation}\label{eq:content_YD}
	c_\mu = \sum_{(i,j)\in \mu } c(i,j)\,.
\end{equation}
\begin{example}
	The contents of the Young diagrams with $3$ and $4$ boxes are 
	\begin{equation}
		c_{\Yboxdim{3pt}\yng(3)}=3\,,\hspace{0.3cm} c_{\Yboxdim{3pt}\yng(2,1)}=0\,,\hspace{0.3cm} c_{\Yboxdim{3pt}\yng(1,1,1)}=-3\,;\hspace{0.6cm}
		c_{\Yboxdim{3pt}\yng(4)}=6\,,\hspace{0.3cm} c_{\Yboxdim{3pt}\yng(3,1)}=2\,,\hspace{0.3cm}c_{\Yboxdim{3pt}\yng(2,2)}=0\,,\hspace{0.3cm}c_{\Yboxdim{3pt}\yng(2,1,1)}=-2\,,\hspace{0.3cm}c_{\Yboxdim{3pt}\yng(1,1,1,1)}=-6\,.
	\end{equation}
\end{example}
We will use the condensed notation $\mu=(\mu_1^{m_1},\mu_2^{m_2},\ldots,\mu_k^{m_k})$ when $m_j$ parts of $\mu$ are equal to $\mu_j$ $(j=1,\ldots ,k)$. For example $(2,1,1)=(2,1^2)$ and $(3,3,3,2,2,1)=(3^3,2^2,1)$. Note also that integer partitions can be seen as concatenation of unary bracelets, where each $\mu_i$ in $\mu=\left(\mu_1,\mu_2,\ldots,\mu_l\right)$ correspond to a bracelet with $\mu_i$ beads with only one color. For example one has the correspondence
\begin{equation}\label{eq:example_partition_bracelets}
(6,4,2)=\left\{\raisebox{-.4\height}{\includegraphics[scale=0.7]{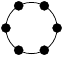}}\,,\,\raisebox{-.4\height}{\includegraphics[scale=0.7]{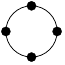}}\,,\,\raisebox{-.4\height}{\includegraphics[scale=0.7]{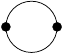}}\right\}, \hspace{0.2cm} (3,2,1^2)=\left\{\raisebox{-.4\height}{\includegraphics[scale=0.7]{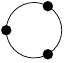}}\,,\,\raisebox{-.4\height}{\includegraphics[scale=0.7]{fig/unary_bracelets_2.pdf}}\,,\,\raisebox{-.4\height}{\includegraphics[scale=0.7]{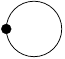}}\,,\,\raisebox{-.4\height}{\includegraphics[scale=0.7]{fig/unary_bracelets_1.pdf}}\right\}.
\end{equation}
As we will see in chapter \ref{chap:cn} this last representation of integer partitions as concatenation of unary bracelets generalize to particular ternary bracelets which parametrize the conjugacy classes in the Brauer algebra $B_n(\Dim)$.

\subsection{Conjugacy classes} \label{subsec:conjugacyclassSn}
Two permutations $s_1,\, s_2 \in \mathfrak{S}_n$ are said to be \textit{conjugate}, which is denoted $s_1\sim s_2$, if there exist a permutation $\sigma \in \mathfrak{S}_n$ such that $s_1=\sigma\, s_2 \, \sigma^{\shortminus 1}$. To each permutation $s\in\mathfrak{S}_n$ we associate its \textit{cycle type} $\CT(s)$ which is given by an integer partition where each integer correspond the length of a cycle present in the decomposition of $s$ into a product of disjoint cycles. 
\begin{example}
For the following permutations of $\mathfrak{S}_6$
\begin{equation}
	s_1=\raisebox{-.4\height}{\includegraphics[scale=0.81]{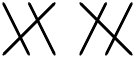}}=(123)(465)\,, \hspace{1cm} s_2=\raisebox{-.4\height}{\includegraphics[scale=0.81]{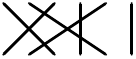}}=(13)(25)\,,
\end{equation}
one has
\begin{equation}
\CT(s_1)=(3,3)=\Yboxdim{6pt}\yng(3,3)=\left\{\raisebox{-.25\height}{\includegraphics[scale=0.7]{fig/unary_bracelets_3.pdf}}\,,\,\raisebox{-.25\height}{\includegraphics[scale=0.7]{fig/unary_bracelets_3.pdf}}\right\}\,,
\end{equation}
and
\begin{equation}
\CT(s_2)=(2,2,1,1)=\Yboxdim{6pt}\yng(2,2,1,1)=\left\{\raisebox{-.25\height}{\includegraphics[scale=0.7]{fig/unary_bracelets_2.pdf}}\,,\,\raisebox{-.25\height}{\includegraphics[scale=0.7]{fig/unary_bracelets_2.pdf}}\,,\,\raisebox{-.25\height}{\includegraphics[scale=0.7]{fig/unary_bracelets_1.pdf}}\,,\,\raisebox{-.25\height}{\includegraphics[scale=0.7]{fig/unary_bracelets_1.pdf}}\right\}\,.
\end{equation}
\end{example}
\vskip 6pt
The importance of the cycle type of a permutation results from the following theorem:
\begin{theorem}\label{theo:classes_partition}
Two permutations are \textit{conjugate} if and only if they have the same cycle type $\mu$. 
\end{theorem}
Conjugacy classes are then parametrized by the cycle type of its elements, and  we write $C_\mu$ for the conjugacy classes of permutations with cycle type $\mu$.\medskip
\begin{remark} 
For any $s\in \sn$ such that $\CT(s)=\mu$, the conjugacy class of $s$ is the orbit of $s$ under the action of conjugation by the elements of $\sn$: $C_s:=\lbrace \sigma \,s \,\sigma^{-1} \st \sigma \in \sn  \rbrace= C_\mu$. 
\end{remark}
We denote by $\Stab_{\sn}(s)$ the stabilizer of a permutation $s\in\sn$ with respect to conjugation, that is, the set of elements $\sigma\in \sn$ such that $\sigma\,s \sigma^{-1}=s$. From the orbit-stabilizer theorem one obtains the cardinality of a conjugacy class: 
\begin{equation}
|C_\mu|=\frac{n!}{|\Stab_{\sn}(s)|}\,, \hspace{0.5cm} \text{where} \hspace{0.5cm} \CT(s)=\mu\,.
\end{equation}
\vskip 4pt
For any integer partition $\mu=(\mu_1^{m_1},\mu_2^{m_2},\ldots,\mu_k^{m_k})$, we define its \textit{stability index} $\stab(\mu)$ as follows
\begin{equation}\label{eq:stability_index_symmetric_groups}
	\stab(\mu)=\prod_{i=1}^{k} \stab(\mu_i)^{m_i} m_i!\,,\hspace{0.5cm} \text{with} \hspace{0.5cm} \stab(\mu_i)=\mu_i\,.
\end{equation}
\begin{lemma}
For any $s\in\sn$ with $\CT(s)=\mu$ holds $|\Stab_{\sn}(s)|=\stab(\mu)$. 
\end{lemma}
\begin{remark}
The stability index $\stab(\mu)$ is often denoted $z_\mu$ in the literature related with symmetric functions and characters of the symmetric group (see for example \cite{goupil2000central}). This coefficient generalizes nicely to conjugacy classes of the Brauer algebra \eqref{eq:stab_bracelets}.\medskip
\end{remark}
\subsection{Conjugacy class sum.}\label{subsection:conjugacyclass_sum}

\paragraph{A basis of $\mathcal{Z}_n$.} The \textit{normal conjugacy class sum} (or simply \textit{conjugacy class sum}) $K_\mu$ is the formal sum of the elements of $C_\mu$ 
\begin{equation}
K_\mu= \sum_{\sigma \,\in C_\mu} \sigma\,.
\end{equation}
We also define the $\textit{averaged}$ conjugacy class sum $\bar{K}_\mu$ as 
\begin{equation}
\bar{K}_\mu=\sum_{\sigma \in \sn } \sigma s \sigma^{\shortminus 1}= \stab(\mu) K_\mu\,, \hspace{0.5cm} \text{for any $s\in \sn$ such that $CT(s)=\mu$}.
\end{equation}
\begin{lemma}
The set of conjugacy class sums $\lbrace K_\mu \st \mu\vdash n \rbrace$ forms a basis of the center $\mathcal{Z}_n$ of $\mathbb{C}\mathfrak{S}_n$. 
\end{lemma}

\begin{mathematica}[\textit{BrauerAlgebra}]
	ClassSum\,\\
	\textit{\textup{ClassSum} is the head for the conjugacy class sum of the symmetric group algebra and of the Brauer algebra. Usage: \textup{ClassSum[$\xi$]} where $\xi$ is a ternary bracelet with head \textup{Bracelets} or an integer partition (for the conjugacy class sums of the symmetric group algebra).}
\end{mathematica}
\paragraph{Product of conjugacy class sums.} The center of $\C\sn$ being an algebra, the multiplication of two conjugacy class sums yield a linear combination of conjugacy class sums. Let $\mu$ and $\zeta$ be two integer partitions of $n$. The product of two class sums $K_\mu$ and $K_\zeta$ in $\mathcal{Z}_n$ is given by 
\begin{equation}\label{eq:product_classes_Sn}
	K_\mu K_\zeta= \sum_{\xi\vdash n} \tensor{C}{^\xi_\mu_\zeta} K_{\xi}\,, 
\end{equation}
where $\tensor{C}{^\xi_\mu_\zeta}$ are the \textit{connection coefficients} or \textit{structure constants} of $\mathcal{Z}_n$ \cite{goupil1998factoring,goulden2000transitive,corteel2004content}.
\newpage
\begin{remark} The coefficient $\tensor{C}{^\xi_\mu_\zeta}$  corresponds to the number of ways of expressing any permutation $s_\xi\in C_\xi$ as a product of permutations $s_\mu s_\zeta$ where $s_\mu\in C_\mu$ and $s_\zeta\in C_\zeta$ (up to a constant depending on $\stab(\xi)$, $\stab(\mu)$ and $\stab(\zeta)$). 
\end{remark}

There are mainly two approach for the computation of $\tensor{C}{^\xi_\mu_\zeta}$. The first one involves the irreducible character $\chi^\rho$ associated with irreducible representation $L^\rho$ (see the next section for the definitions)~\cite{jackson_1987,jackson_1988,goupil1990products}: 
\begin{equation}\label{eq:struc_constants_Sn}
	\tensor{C}{^\xi_\mu_\zeta}=\frac{n!}{\stab(\mu)\stab(\zeta)}\sum_{\rho\vdash n}\frac{1}{\mathrm{d}_\rho}\,\chi^{\rho}_{\mu}\,\chi^{\rho}_{\zeta}\,\chi^{\rho}_{\xi}\,, \hspace{0.5cm} \text{with} \hspace{0.5cm} \text{$\xi\,,\mu\,,\zeta \vdash n$\,.}
\end{equation}
This formula can be easily derived from the character projection formula \eqref{eq:character_central_idempotent_Sn} and \eqref{eq:centralYoung_Characters} presented below. 

In the second approach, multiplication by a conjugacy class sum is described by the action of particular differential operator defined on the space of power symmetric functions \cite{goulden1994differential,goulden1997transitive,goupil2005katriel}. 

\begin{mathematica}[\textit{BrauerAlgebra}]
	ConjugacyClassProduct[ClassSum[$\mu$],ClassSum[$\zeta$]]\,\\
	\textit{Returns the product of $K_\mu$ with $K_\zeta$ where $\mu$ and $\zeta$ are integer partitions of $n$\,.}
\end{mathematica}

\section{Central idempotents in $\mathbb{C}\mathfrak{S}_n$}\label{sec:central_idempotent_Sn}
\subsection{Central idempotents and irreducible characters}\label{subsec:irreduciblecharacter}
Apart from the conjugacy class sums $K_\mu$ there exist another basis of the center of $\C\sn$ which is given by the central Young idempotents $Z^\mu$. In section \ref{sec:ideals} we have seen that the latter basis is intimately related to the representation theory of $\C\sn$ via its regular module. The representation theory of groups can be treated to large extend from the point of view of character theory,\footnote{See for example \cite{ceccherini2010representation} where irreducible characters play a major role in the discussions of the representation theory of finite groups. See also \cite{Huppert_1998} for a standard textbook on character theory of finite groups.} and as we have already mentioned in the introduction irreducible characters can be used to construct the central Young idempotents. The privileged status of characters stems from the fact that they remain constant on conjugacy classes.  \medskip

The character $\chi^{\rho}\,:\, G\to \C$ of a representation $\rho\,:\, G \to \GL(V)$ of a group $G$ on a finite dimensional vector space $V$ over $\C$ is the trace of the representation $\rho$:
\begin{equation}
\chi^\rho(g)=\text{Tr}(\rho(g))\,\hspace{1cm}\text{for all } g\in G \,.
\end{equation}
This definition is extended naturally to any group algebra.
\paragraph{The character formula.} Let $L^\mu$ be an irreducible $\C\mathfrak{S}_n$-module. The \textit{irreducible character} of $L^\mu$ is the function $\chi^{\mu}\,:\, \C\mathfrak{S}_n\to \C$ given by
	\begin{equation}
		\chi^{\mu}(z)=\text{Tr}(\mu(z)), \quad \text{for all } z\in \C\sn,
	\end{equation}   
where with a slight abuse of notation we also used $\mu$ to denote the irreducible representation homomorphism. 
We denote by $\chi^{\mu}_\rho$ the character of the irreducible module $L^\mu$ evaluated at any permutation with cycle type $\rho$. With this definition one has 
\begin{equation}
\chi^{\mu}(K_{\rho})=\frac{n!}{\stab(\mu)}\chi^\mu_\rho\,.
\end{equation}
\begin{theorem}[The character formula \cite{Janusz_1966}]
Let $L^\mu$ be an irreducible $\C\mathfrak{S}_n$-module. In terms of the irreducible characters $\chi^{\mu}$ the central Young idempotent $Z^\mu$ is given by the projection formula: 
\begin{equation}\label{eq:character_central_idempotent_Sn}
	Z^\mu=\frac{\mathrm{d}_\mu}{n!}\sum_{\rho\vdash n} \chi^{\mu}_\rho \, K_\rho\,,
\end{equation}
where $\mathrm{d}_\mu$ is the dimension of $L^\mu$ (it is also the multiplicity of $L^\mu$ in the right regular module).
\end{theorem}

\begin{remark}\label{rem:characters}
\begin{itemize}
	\item[\textit{i)}] There exist an efficient algorithm for computing the irreducible characters of $\C\sn$ called the Murnaghan-Nakayama rule \cite{murnaghan1937representations,nakayama1940some1,nakayama1940some2} (see also \cite[Section 1.3]{GOUPIL1994102} and \cite[Chapter 3]{ceccherini2010representation}).

	\item[\it ii)] The previous formula is a change of basis in $\mathcal{Z}_n$, from the conjugacy class sum basis to the central Young basis. The inverse formula is given by
	\begin{equation}\label{eq:centralYoung_Characters}
		K_\mu=\frac{n!}{\stab(\mu)}\sum_{\rho\vdash n} \frac{\chi^{\rho}_\mu}{\mathrm{d}_\rho}\,Z^\rho\,.
	\end{equation}
	\item[\it iii)] Any element $z$ of the center can be expanded in the basis $\lbrace Z^\mu \st \mu\vdash n \rbrace $:
	\begin{equation}
	z=\sum_{\rho\vdash n} z_\rho Z^\rho \hspace{0.5cm} \text{with} \hspace{0.5cm} z_\rho\in \C \,.
	\end{equation}
	From the orthogonality property of the central Young idempotents (recall \eqref{eq:properties_central_idempotents}) one has 
	\begin{equation}
	Z^\nu z = z_\nu Z^\nu\,.
	\end{equation}
	Hence, the central Young projectors are eigenvectors of the conjugacy class sums:  $Z^\nu K_\mu = \omega^{\nu}_{\mu} Z^\nu$  where the eigenvalues 
	\begin{equation}\label{eq:central_characters}
	\omega^{\nu}_{\mu}=\dfrac{n!}{\mathrm{d}_\nu\,\stab(\mu)} \chi^{\nu}_\mu
	\end{equation}
	are called \textit{the central characters} \cite{lassalle2007explicit,katriel1996explicit,goupil2000central,vershik1981asymptotic}.  Equivalently one may say that the conjugacy class sum $K_\mu$ acts\tablefootnote{We consider right regular module so in our mind $K_\mu$ acts on $L^{\mu}$ by the right. Nevertheless $K_\mu$ is central so left action of $K_\mu$ is the same as its right action.} by the scalar $\omega^{\nu}_{\mu}$ on irreducible $\C\sn$-module $L^\nu$. 
	\item[\it iv)] The dimension $\mathrm{d}_\mu$ of an irreducible $\C\sn$-module $L^\mu$ is the number of standard tableaux of shape $\mu$ which can be computed via a simple combinatorial formula called the Hook length formula \cite[chapter 4]{fulton1997young}.
\end{itemize}
\end{remark}
\vskip 6pt

We recall the orthogonality relations for the irreducible character, as the first one will be used for the proof of Proposition \ref{prop:line_induction_sn}.

\begin{theorem}[{First orthogonality relation for the irreducible characters \cite[Theo. 3.5]{Huppert_1998}}]
	Let $L^\mu$ and $L^\nu$ be two irreducible representations of $\C\sn$. Then 
	\begin{equation}\label{eq:first_ortho_Characters}
		\sum_{\rho\vdash n}\frac{1}{\mathrm{st(\rho)}} \chi^\mu_\rho\,\chi^\nu_\rho=\left\{
		\begin{array}{ll}
			1 \quad &\text{if} \quad L^\mu\sim L^\nu\,,\\
			0\,\,,  \quad &\text{otherwise.}
		\end{array}
		\right.
	\end{equation}
\end{theorem}
\vspace{1.5cm}

\begin{theorem}[{Second orthogonality relation for the irreducible characters \cite[Theo. 3.10]{Huppert_1998}}]
	Let $\mu$ and $\nu$ be two integer partitions of $n$. Then 
	\begin{equation}\label{eq:second_ortho_Characters}
		\frac{1}{\mathrm{st(\mu)}}\sum_{\rho\vdash n} \chi^\rho_\mu\,\chi^\rho_\nu=\left\{
		\begin{array}{ll}
			1 \quad &\text{if} \quad \mu=\nu\,,\\
			0\,\,,  \quad &\text{otherwise.}
		\end{array}
		\right.
	\end{equation}
\end{theorem}

\begin{mathematicas}[\textit{SymmetricFunctions}]
CharacterSymmetricGroup[\,$\mu$\,,\,$\rho$\,]\,\\
\textit{Returns the irreducible character $\chi^\mu$ evaluated at a permutation $s\in C_\rho$\,.}\\	
	
DimOfIrrepSn[$\mu$]\\
\textit{Returns the dimension of the irreducible representation $L^\mu$.}
\end{mathematicas}

\begin{mathematica}[\textit{BrauerAlgebra}]
	CentralYoungProjector[$\mu$]\\ 
	\textit{Returns $Z^\mu$ in the conjugacy class sum basis of $\mathcal{Z}_n$.}
\end{mathematica}

\subsection{The class $T_n$: principal building block of the construction}
\paragraph{Jucys-Murphy elements.} Let $s_{ij}$ denote the transposition $(i,j)\in \sn$:
\begin{equation}\label{eq:sij}
	s_{ij}=\raisebox{-.3\height}{\includegraphics[scale=0.65]{fig/Tn.pdf}}\,.
\end{equation}
We introduce the following set of pairwise-commuting elements of $\C\sn$, known as the {\it Jucys-Murphy elements} $j_k$ ($k = 1,\dots, n$) \cite{Jucys,Murphy}: 
\begin{equation}\label{eq:JM_sn}
	\quad j_{k} = \sum_{j=1}^{k-1} s_{jk}\quad \text{for}\quad k\geqslant 2\,,
\end{equation}
and $j_1=0$.\medskip

For any standard Young tableau $\tab$ with $n$ boxes of a given shape $\mu$ we denote by $c_{\tab}(k)$ ($1 \leqslant k \leqslant n $), the content (recall \eqref{eq:content_YD}) of the boxe in $\mu$ which contains the integer $k$ in the standard tableau $\tab$.

\begin{example} Consider the following two standard tableaux with $4$ boxes:
\begin{equation}
\tab_1=\Scale[0.8]{\young(13,24)}\,, \hspace{1cm} \tab_2=\Scale[0.8]{\young(123,4)}\,.
\end{equation}
One has 
\begin{equation}
	\begin{aligned}
&c_{\tab_1}(1)=0 \,,\hspace{0.3cm} c_{\tab_1}(2)=-1 \,,\hspace{0.3cm} c_{\tab_1}(3)=1 \,,\hspace{0.3cm} c_{\tab_1}(4)=0 \,,\\
&c_{\tab_2}(1)=0 \,,\hspace{0.3cm} c_{\tab_2}(2)=1 \,,\hspace{0.63cm} c_{\tab_2}(3)=2 \,,\hspace{0.3cm} c_{\tab_1}(4)=-1 \,.
	\end{aligned}
\end{equation}
\end{example}

\begin{theorem}[Theo. 3.3 \cite{Garsia2020}]\label{theo:JM_sn}
For any standard Young tableau $\tab$ with $n$ boxes we have for $2\leqslant k\leqslant n$,
\begin{equation}
Y^{\,\tab} j_k = c_\tab(k)Y^{\,\tab}\,, \hspace{0.5cm} \text{besides} \hspace{0.5cm} Y^{\,\tab} j_k =  j_k Y^{\,\tab}\,,
\end{equation}
where we recall that $Y^{\,\tab}$ are the Young seminormal idempotents.
\end{theorem}
As a direct consequence of this Theorem (recall \eqref{eq:partition_Zmu}) one has
\begin{equation}\label{eq:JM_simple_module}
\text{For any } \mu\vdash n\,, \hspace{1.5cm} Z^{\mu} j_k = \sum_{\tab\in \textup{Tab}(\mu)} c_\tab(k) Y^{\,\tab}\,.
\end{equation}
Hence (as expected from Schur's Lemma) the Jucys-Murphy elements are not proportional to unity on simple $\C\sn$-modules $L^\mu$.
\begin{remark}
\begin{itemize}
\item[\it i)] The center of $\C\sn$ is generated by particular symmetric polynomial in $j_k$ (see for example \cite[Sec. 8]{lassalle2007explicit} and \cite{corteel2004content}).
\item[\it ii)] The Jucys-Murphy elements generate the maximal commutative subalgebra of $\C\sn$ called the Gelfand-Tsetlin algebra \cite[Prop. 2.1]{Okounkov1996}.
\item[\it iii)] The notion of Jucys-Murphy elements can be generalized to any multiplicity-free family of finite dimensional algebras \cite[Def. 3.1]{doty2019canonical}\,. Their generalization to the Brauer algebra is given in section \ref{subsec:JM_Xn}.
\end{itemize}
\end{remark}

\paragraph{The class sum $T_n$.} The principal building block of the construction presented in this section is a particular element $T_n$ of the center $\mathcal{Z}_n$. It is define as the sum of transpositions in $\mathfrak{S}_n$: 
\begin{equation}\label{eq:Tn_def}
	T_n=\sum_{1\leqslant i \, < \, j\leqslant n}s_{ij}\,,
\end{equation}
\newpage
\begin{example}For $n=3$ and $n=4$ one has: 
\begin{equation}
	T_3=\raisebox{-.4\height}{\includegraphics[scale=0.4]{fig/s3a.pdf}}+\raisebox{-.4\height}{\includegraphics[scale=0.4]{fig/s3c.pdf}}+\raisebox{-.4\height}{\includegraphics[scale=0.4]{fig/s3b.pdf}}\,,
\end{equation}
\begin{equation}
	T_4=\raisebox{-.4\height}{\includegraphics[scale=0.55]{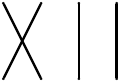}}+\raisebox{-.4\height}{\includegraphics[scale=0.55]{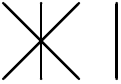}}+\raisebox{-.4\height}{\includegraphics[scale=0.55]{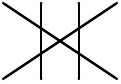}}+\raisebox{-.4\height}{\includegraphics[scale=0.55]{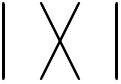}}+\raisebox{-.4\height}{\includegraphics[scale=0.55]{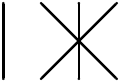}}+\raisebox{-.4\height}{\includegraphics[scale=0.55]{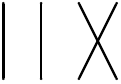}}\,.
\end{equation}
\end{example}
\vskip 4pt

\begin{remark}\label{rem:def_Tn}\hphantom{0}\medskip
\begin{itemize}
\item[\it i)] $T_n$ is the conjugacy class sum parametrized by the integer partition $(2,1^{n-2})$:
\begin{equation}
	T_{n}=K_{(2,1^{n-2})}\,.
\end{equation}
\item[\it ii)] $T_n$ is the sum of Jucys-Murphy elements 
\begin{equation}\label{eq:Tn_jm_sum}
T_n=\sum_{k=1}^{n} \, j_k
\end{equation}
\item[\it iii)] Conversely the Jucys-Murphy element $j_k$ is the difference between $T_k$ and $T_{k-1}$:
\begin{equation}
j_k=T_k-T_{k-1}
\end{equation}
\end{itemize}
\end{remark}

As a direct consequence of \eqref{eq:Tn_jm_sum} and Theorem \ref{theo:JM_sn}, on has the following well known fact. 
\begin{lemma}\label{lem:ev_Tn}
	The element $T_n$ is proportional to the identity on irreducible $\C\sn$-module $L^\mu$ with eigenvalue $\omega^\mu_{(2,1^{n-2})}$ given by the content of the Young diagram $\mu$:
	\begin{equation}
		Z^\mu \, T_n =c_\mu\, Z^\mu \, \hspace{0.5cm} \text{with} \hspace{0.5cm} \mu\vdash n.
	\end{equation}
\end{lemma}
\vskip 6pt
Let the multiset $ \spec(T_n)=(c_\mu \st \mu \vdash n)$ denote the \textit{spectrum} of $T_n$. Following \cite{doty2019canonical} we say that a spectrum is \textit{simple} if it has no repeated entries.\medskip 
\begin{lemma}\label{lem:simple_ev_Tn}
The spectrum of $T_n$ is simple if and only if $n<6$ or $n=7$.
\end{lemma}
See appendix \ref{subsec:proof_simple_ev_Tn} for the proof.

\paragraph{Lagrange-interpolation-type formula.} For convenience we recall the definition of Lagrange polynomial in appendix \ref{subsec:definitions}. A general account on the construction of central idempotent in multiplicity-free families of algebras via Lagrange-interpolation-type formula can found in \cite{doty2019canonical}[Section $2$]. Therein the authors give the following formula which is presented here as a proposition. 
\begin{proposition}
The central idempotents of $\mathbb{C}\mathfrak{S}_n$ for $n<6$ and $n=7$ is given by
	\begin{equation}\label{eq:central_young_res}
		Z^{\mu}=\prod_{\begin{array}{c}
				{\scriptstyle \rho\,\vdash\, n}\\
				{\scriptstyle \rho\neq\mu }
		\end{array}}\dfrac{c_\rho-T_n}{c_\rho-c_\mu}\,.
	\end{equation} 
\end{proposition}
\begin{proof}
From Lemma \eqref{lem:ev_Tn}, when $\spec(T_n)$ is simple one has 
\begin{equation}
	\prod_{\begin{array}{c}
			{\scriptstyle \rho\,\vdash\, n}\\
			{\scriptstyle \rho\neq\mu }
	\end{array}}\dfrac{c_\rho-T_n}{c_\rho-c_\mu}\, Z^\nu=\delta_{\mu\nu} Z^\nu, \hspace{1cm} \text{for any $\nu \vdash n$\,},
\end{equation}
and the result follows.
\end{proof}
Let us stress that when $n=6$ and for $n>7$, \eqref{eq:central_young_res} can not be applied uniformly. Indeed, in these cases $\spec(T_n)$ is not simple, and as a result the denominator in \eqref{eq:central_young_res} may be zero.

\begin{example}[Central Young idempotents for $n=3$]\label{ex:CYI_n3}
The contents of the Young diagrams with 3 boxes are 
\begin{equation}
	c_{\Yboxdim{3pt}\yng(3)}=3\,,\hspace{1cm} c_{\Yboxdim{3pt}\yng(2,1)}=0\,,\hspace{1cm} c_{\Yboxdim{3pt}\yng(1,1,1)}=-3\,,
\end{equation}
By direct computation, the central Young idempotents of $\C\Sn{3}$ take the form: 
\begin{equation}\label{eq:factorized_Z3}
Z^{\Yboxdim{3pt}\yng(3)}=\frac{1}{18}\, T_3\,\left(3+T_3\right)\,,\hspace{0.4cm} Z^{\Yboxdim{3pt}\yng(2,1)}=\frac{1}{9}\left(3-T_3\right)\left( 3+T_3 \right)\,,\hspace{0.4cm}
Z^{\Yboxdim{3pt}\yng(1,1,1)}=-\frac{1}{18}T_3\, \left(3-T_3\right)\,.
\end{equation}
\end{example}
\begin{example}[Central Young idempotents for $n=4$]\label{ex:CYI_n4}
The contents of the Young diagrams with $4$ boxes are: 
\begin{equation}
	c_{\Yboxdim{3pt}\yng(4)}=6\,,\hspace{1cm}c_{\Yboxdim{3pt}\yng(3,1)}=2\,,\hspace{1cm}c_{\Yboxdim{3pt}\yng(2,2)}=0\,,\hspace{1cm}c_{\Yboxdim{3pt}\yng(2,1,1)}=-2\,,\hspace{1cm}c_{\Yboxdim{3pt}\yng(1,1,1,1)}=-6\,.
\end{equation}
By direct computation, the central Young idempotents of $\C\Sn{4}$ take the form: 
\begin{equation}\label{eq:factorized_Z4}
\begin{array}{ll}
		Z^{\Yboxdim{3pt}\yng(4)}=\dfrac{1}{2304}\, T_4\left(6+T_4\right)\left(2+T_4\right)\left(T_4-2\right)\,,
		&Z^{\Yboxdim{3pt}\yng(1,1,1,1)}=-\dfrac{1}{2304}\, T_4\left(6-T_4\right)\left(2+T_4\right)\left(T_4-2\right),\\[10pt]
		Z^{\Yboxdim{3pt}\yng(3,1)}=\dfrac{1}{256}\, T_4\left(6-T_4\right)\left(2+T_4\right)\left(6+T_4\right)\,,
		&Z^{\Yboxdim{3pt}\yng(2,1,1)}=-\dfrac{1}{256}\, T_4\left(6-T_4\right)\left(2-T_4\right)\left(6+T_4\right)\,,\\[10pt]
		Z^{\Yboxdim{3pt}\yng(2,2)}=\dfrac{1}{144}\, \left(6-T_4\right)\left(2-T_4\right)\left(2+T_4\right)\left(6+T_4\right)\,.
	\end{array}
\end{equation} 
\end{example}
\paragraph{Product of $T_n$ with other class sums.} For the purpose of application it is desirable to be able to expand efficiently formula \eqref{eq:central_young_res} in the conjugacy class sums basis of $\mathcal{Z}_n$. This can be done using the differential operator described in \cite[Section 2]{goulden1997transitive}. In chapter \ref{chap:cn} we propose a quite similar technique which also applies to products of $T_n$ with conjugacy class sums of the Brauer algebra. The expanded expressions of the central Young idempotents of examples \ref{ex:CYI_n3} - \ref{ex:CYI_n4} in the conjugacy class sum basis of $\mathcal{Z}_n$ are presented in section \ref{sec:applications}.\medskip
 
For completeness let us apply formula \eqref{eq:struc_constants_Sn} to multiplication by $T_n$, that is $\mu=(2,1^{n-2})$ in \eqref{eq:struc_constants_Sn}. First note that, 
\begin{equation}
\stab(2,1^{n-2})=2\,(n-2)!
\end{equation}
and from $\omega^{\rho}_{(2,1^{n-2})}=c_\rho$ (see Lemma \eqref{lem:ev_Tn}) and the definition of central character \eqref{eq:central_characters} one has 
\begin{equation}
\chi^{\rho}_{(2,1^{n-2})}=\dfrac{2 \,\mathrm{d}_\rho \, c_\rho}{n(n-1)}\,.
\end{equation}
This yield the following expression for the connection coefficient $C^{\xi}_{(2,1^{n-2})\,\zeta}$ in terms of irreducible characters 
\begin{equation}
C^{\xi}_{(2,1^{n-2})\,\zeta}=\frac{1}{\stab(\zeta)}\sum_{\rho\vdash n} c_\rho \, \chi^{\rho}_{\zeta}\,\chi^{\rho}_{\xi}.
\end{equation}

\subsection{The inductive formula for $Z^\mu$}\label{subsec:inductiveformula}
The content of the present section are original results. 
\paragraph{The averaged line induction map $\mathcal{L}$.}  For the purpose of generalizing formula \eqref{eq:central_young_res} to arbitrary $n$ we introduce the line induction map $\mathcal{L}\,:\, \mathbb{C}\mathfrak{S}_{n-1}\to \mathbb{C}\mathfrak{S}_n$ defined on a single permutation diagram $s\in \mathfrak{S}_{n-1}$ as the sum of all possible diagrams obtained from $s$ by the insertion of a vertical line. For example, 
\begin{equation}
	\mathcal{L}(\raisebox{-.4\height}{\includegraphics[scale=0.4]{fig/s3a.pdf}})=2\,\raisebox{-.4\height}{\includegraphics[scale=0.55]{fig/T4a.pdf}}+\raisebox{-.4\height}{\includegraphics[scale=0.55]{fig/T4b.pdf}}+\raisebox{-.4\height}{\includegraphics[scale=0.55]{fig/T4d.pdf}}\,.
\end{equation}
Let us define the map $\mathfrak{l}_i$ which adds a vertical line at node $i$ of a permutation diagram of $\sn$ as:
\begin{equation}\label{eq:ci}
	\begin{aligned}
		\mathfrak{l}_{i}\,:\, \Sn{n-1}&\to \sn \\
		s&\mapsto c_{i}\, s \,c_{i}^{-1}\,
	\end{aligned}
\end{equation}
where $c_i\in \sn$ is the cycle $\left(i,i+1,\ldots, n\right)$:
\begin{equation}
c_i=\raisebox{-.3\height}{\includegraphics[scale=0.3]{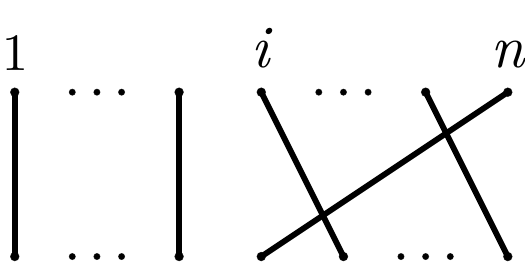}}\,,
\end{equation}
 and $c_n=\id$. In terms of $\mathfrak{l}_i$ the action of the line induction map $\mathcal{L}$ on a single permutation is given by 
\begin{equation}\label{eq:def_L_li}
\mathcal{L}(s)=\sum_{1\leqslant i\leqslant n} \mathfrak{l}_i(s)\,.
\end{equation}
The map $\mathcal{L}$ is extended to the group algebra by linearity.\medskip 

The following two results describe the action of $\mathcal{L}$ on the center of $\C\Sn{n-1}$. We denote by $\nu_{[1]}$ the Young diagram obtained from $\nu$ by adding a box below its last row. The action of $\mathcal{L}$ on conjugacy class sums of $\C\Sn{n-1}$ is described by the following lemma whose proof is given in the appendix \ref{subsec:proof_induction_center_Sn}.

\begin{lemma}\label{lem:induction_center_Sn}
The line induction $\mathcal{L}$ is a map from $\mathcal{Z}_{n-1}$ to $\mathcal{Z}_{n}$, and in particular for any $\nu \vdash n-1 $ one has 
\begin{equation}
\mathcal{L}(K_\nu)=\frac{\stab(\nu_{[1]})}{\stab(\nu)}K_{\nu_{[1]}}\,.
\end{equation}
\end{lemma}
Let $\mathcal{A}_\nu$ be the set of Young diagrams obtained by adding a box to $\nu$. The action of $\mathcal{L}$ on the orthogonal basis of $Z_{n-1}$ is described by the following proposition.
\begin{proposition}\label{prop:line_induction_sn}
	Let $\,\nu\vdash n-1 $ and $\mu\vdash n$  , then 
	\begin{equation}
		\mathcal{L}(Z^{\nu})Z^{\mu}=\left\{
		\begin{array}{ll}
			z_{\mu\backslash\nu} \, Z^{\mu}\, \quad \text{ with $\,z_{\mu\backslash\nu}=n\,\dfrac{\,\mathrm{d}_\nu}{\mathrm{d}_\mu}$}\,, &\text{if} \quad \mu\in \mathcal{A}_{\nu}\,,\\
			\quad 0\,\, \quad &\text{otherwise.}
		\end{array}
		\right.
	\end{equation}
Hence, 
\begin{equation}
\mathcal{L}(Z^\nu)=\sum_{\rho\in \mathcal{A_\nu}} z_{\rho\backslash\nu} Z^\rho\,.
\end{equation}
\end{proposition}
See appendix \ref{subsec:proof_line_induction_sn} for the proof.\medskip

Let $\mathcal{R}_\nu$ be the set of Young diagrams obtained by removing a box from $\nu$.
\begin{lemma}\label{lem:addable_removable_simple}
	For any $\nu\vdash n$ the multisets $(c_\rho\,, \rho\in \mathcal{A}_\nu)$ and $(c_\rho\,, \rho\in \mathcal{R}_\nu)$ are simple.
\end{lemma}
\begin{proof}
	For any pair of diagrams  $\bar\nu,\tilde\nu \in \mathcal{A}_\nu$ such that $\bar\nu$ and $\tilde{\nu}$ are obtained by adding a box at row $i$ and $j$ respectively with $i<j$ one has $c_{\bar\nu}-c_{\tilde\nu}=\nu_i-\nu_j +j-i>0$. Hence all pairs of diagrams in $\mathcal{A}_\nu$ have different contents. Similarly, for any pair of diagrams  $\bar\nu,\tilde\nu \in \mathcal{R}_\nu$ such that $\bar\nu$ and $\tilde{\nu}$ are obtained by removing a box at row $i$ and $j$ respectively with $i<j$ one has $c_{\bar\nu}-c_{\tilde\nu}=\nu_j-\nu_i +i-j<0$. Hence all pairs of diagrams in $\mathcal{R}_\nu$ have different contents. 
\end{proof}	

\paragraph{The formula.} With Lemma \ref{lem:ev_Tn}, Lemma \ref{lem:addable_removable_simple} and Proposition \ref{prop:line_induction_sn} at hand, we are now in position to introduce the element $\mathcal{Z}^\mu(\nu)\in \mathcal{Z}_n$ defined for any $\nu \in \mathcal{R}_{\mu}$ by
\begin{equation}\label{eq:gen_central_idempotent_Sn}
	\mathcal{Z}^{\mu}(\nu)=\frac{\mathcal{L}(Z^{\nu})}{z_{\mu\backslash\nu}}\prod_{\begin{array}{c}
			{\scriptstyle \rho\in \mathcal{A}_\nu}\\
			{\scriptstyle \rho\neq\mu}
	\end{array}}\dfrac{c_\rho-T_n}{c_\rho-c_\mu}\,.
\end{equation}   
Due to Lemma \ref{lem:addable_removable_simple} it is clear that $\mathcal{Z}^{\mu}(\nu)$ never diverges. The main result of this chapter is described by the following theorem.
\begin{theorem}\label{Theorem:centralYoung}
For any $\mu\vdash n$ and $\nu \in \mathcal{R}_{\mu}$ the element $\mathcal{Z}^\mu(\nu)\in \mathcal{Z}_n$ coincides with the central Young projector $Z^\mu$.
\end{theorem}
\begin{proof}
Let $\nu\in \mathcal{R}_\mu $. For any $\beta\in \mathcal{A}_\nu$ one as $\nu \in \mathcal{R}_\beta$ and from Lemma \eqref{lem:ev_Tn} and Proposition  \eqref{prop:line_induction_sn}
\begin{equation}
\mathcal{Z}^{\mu}(\nu)Z^\beta=\delta_{\beta\mu} Z^\mu,
\end{equation}
whereas, for any $\beta\notin \mathcal{A}_\nu$, one has $\nu \notin \mathcal{R}_\beta$ and from Proposition \eqref{prop:line_induction_sn}
\begin{equation}
	\mathcal{Z}^{\mu}(\nu)Z^\beta=0,
\end{equation}
and the result follows.
\end{proof}
\section{Construction of the Young seminormal idempotents}\label{sec:seminormalYoung}

\subsection{Bratteli diagrams and branching rules $\C\sn\downarrow\C\Sn{n-1}$.}
The recent approach to the representation theory of $\mathbb{C}\mathfrak{S}_n$ by Okunkov and Vershik \cite{Okounkov1996} is structured around the existence of the natural chain of embedding
\begin{equation}\label{eq:chain_embedding_sn}
	\mathfrak{S}_1\subset \ldots \subset \mathfrak{S}_{n-1}\subset \mathfrak{S}_{n}\,.
\end{equation}
The chain of embedding \eqref{eq:chain_embedding_sn} is by convention realized at each step by the homomorphism $\iota\,:\, \Sn{n-1} \hookrightarrow \sn$ defined by the insertion of a vertical line at the right end of any permutation diagram. We follow this convention and $\Sn{n-1}$ is identify as a subgroup of $\sn$ without further mention of the homomorphism $\iota$. Note that each step $\mathfrak{S}_{k-1}\hookrightarrow \mathfrak{S}_{k}$ of the embedding chain \eqref{eq:chain_embedding_sn} can be realized with the map $\mathfrak{l}_i\,:\,\mathfrak{S}_{k-1} \hookrightarrow \mathfrak{S}_{k}$ \eqref{eq:ci} for any integer $i$ from 1 to $k$.\medskip

\paragraph{Bratteli diagrams.} In order to introduce the notion of Bratteli diagram in all generality, consider a semisimple algebra $\mathcal{A}_n$ over $\C$ such that there exist a unital preserving embedding chain of algebras:
\begin{equation}\label{eq:chain_embedding_An}
	\mathcal{A}_{0}\subset\mathcal{A}_{1}\subset \ldots \subset \mathcal{A}_{n-1}\subset \mathcal{A}_n\,,
\end{equation}
where $\mathcal{A}_{0}\cong\C$. Following \cite{doty2019canonical}, we denote by $\text{Irr}(\mathcal{A}_k)$ the set of all classes of isomorphic irreducible $\mathcal{A}_k$-modules and by $W^\rho$ a representative of the class $\rho\in \text{Irr}(\mathcal{A}_k)$. 
\begin{definition}\label{def:bratteli_diagram}
The Bratteli diagram for the chain \eqref{eq:chain_embedding_An} (or equivalently for the family of algebras $\lbrace\mathcal{A}_k \,\st\, 0\geqslant k\geqslant n\, \rbrace$) is the multigraph whose vertices and edges are such that:\footnote{A Bratteli diagram is often defined for an infinite chain $\mathcal{A}_{0}\subset\mathcal{A}_{1}\subset \ldots \,$ of semisimple algebras. Our definition is similar to the definition in \cite[p. $75$]{ceccherini2010representation} where the chain of algebras is truncated at a given level $n$.} 
\begin{itemize}
	\item The vertices are the disjoint union $\displaystyle\coprod_{k=0}^{n}\Irr(\mathcal{A}_k)=\displaystyle\bigcup_{k=0}^{n}\lbrace \left(\lambda, k\right)\,\st\, \lambda \in \Irr(\mathcal{A}_k) \rbrace$: on each level $k$ place one vertex for every element of $\Irr(\mathcal{A}_k)$.
	\item We put $\tensor{m}{^{\mu}_{\rho}}$ edges $\rho\rightarrow \mu$ from the vertex $\rho\in \Irr(\mathcal{A}_{k})$ to the vertex $\mu\in\Irr(\mathcal{A}_{k-1})$ if $W^\mu$ appears with multiplicity $\tensor{m}{^{\mu}_{\rho}}$ in the restriction of $W^{\rho}$ into irreducible $\mathcal{A}_{k-1}$-module.
\end{itemize}
\end{definition}
\begin{remark}
\begin{itemize}
\item[\it i)] In the first chapter we use the terminology ‘‘ Bratteli diagram associated with $\mathcal{A}_n$" to refer to the Bratteli diagram for the chain \eqref{eq:chain_embedding_An}. This terminology may not be conventional.
\item[\it ii)] When the multiplicities $\tensor{m}{^{\mu}_{\rho}}$ are all bounded by $1$ the family of algebras $\lbrace\mathcal{A}_k\,\st 0\leqslant k\leqslant n \rbrace$ is called multiplicity-free \cite[Def. $1.1$]{doty2019canonical}. In this case the corresponding Bratteli diagram is a directed graph.
\end{itemize}
\end{remark}


\paragraph{Branching rules $\C\sn\downarrow\C\Sn{n-1}$.}
We recall that $\mathcal{R}_{\mu}$ is the set of Young diagrams obtained from $\mu$ by removing a box. The following theorem is a well known result which is proved for example in \cite{Okounkov1996}.
\begin{theorem} Let $L^\mu$ be an irreducible $\C\sn$-module. Under the restriction to the subalgebra $\C\Sn{n-1}\subset\C\Sn{n}$, $L^\mu$ decomposes as the following direct sum of irreducible $\C\Sn{n-1}$-module
	\begin{equation}\label{eq:branching_restriction_Sn}
		L^{\mu}\cong\bigoplus_{\rho\in\mathcal{R}_{\mu}} L^{\rho} \,, \hspace{1cm} \textup{(\textit{upon} $\C\sn\downarrow\C\Sn{n-1}$)}\,.
	\end{equation}  
\end{theorem}

This theorem allows one to construct the Bratteli diagram for the chain 
\begin{equation}\label{eq:chain_embedding_Csn}
	\C\cong\C\mathfrak{S}_0\subseteq\C\mathfrak{S}_1\subset \ldots \subset \C\mathfrak{S}_{n-1}\subset \C\mathfrak{S}_{n}\,,
\end{equation}
which is call Young's lattice, and implies that the decomposition of any irreducible $\C\Sn{n}$-module into irreducible $\C\Sn{n-1}$-module is multiplicity-free.
\begin{example}[The Bratteli diagrams (Young's lattices) for $\lbrace\C\mathfrak{S}_k\rbrace_{k\leqslant 3}$ and for $\lbrace\C\mathfrak{S}_k\rbrace_{k\leqslant 4}$]
\begin{equation*}
	\includegraphics[scale=0.7]{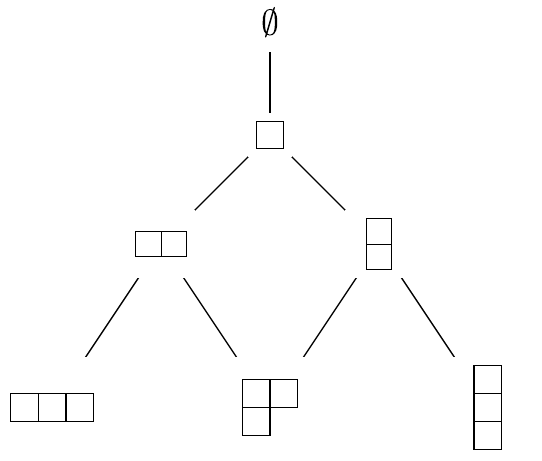}	\hspace{2.2cm}	\includegraphics[scale=0.86]{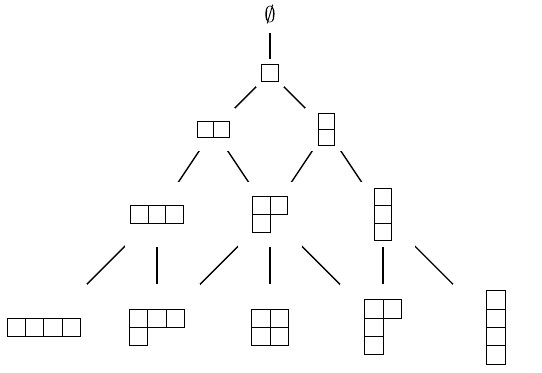}
\end{equation*}
\noindent Each vertex of the lattice at a given horizontal level $k$ is a Young diagram $\mu$ which labels the class of isomorphic irreducible $\C\mathfrak{S}_k$-modules. A Young diagram $\mu$ at a level $k$ is connected to the diagrams $\nu\in \mathcal{R}_{\mu}$ at level $k-1$.
\end{example}
\begin{mathematicas}[\textit{BrauerAlgebra}]
	BratteliDiagramSn[\,$n$\,]\,\\
	\textit{Returns the Bratteli diagram for the chain $\C\mathfrak{S}_1\subset \ldots \subset \C\mathfrak{S}_{n-1}\subset \C\mathfrak{S}_{n}$ as a list of paths.}\\
	
	BratteliDiagramSn[\,$n$\,,\,Output$\rightarrow$Graph\,]\,\\
	\textit{Returns the Bratteli diagram for the chain $\C\mathfrak{S}_0\subset\C\mathfrak{S}_1\subset \ldots \subset \C\mathfrak{S}_{n-1}\subset \C\mathfrak{S}_{n}$\, with head \textup{Graph}.}
\end{mathematicas}

\subsection{Standard tableaux and the Young Seminormal idempotents}\label{subsec:standard_tableau_seminormal_Young}

For any $\mu\vdash n$, let $\Tab(\mu)$ denote the set of \textit{paths} in the Young lattice of $\sn$ starting from $\emptyset$ and terminating at $\mu$. An element of $\Tab(\mu)$ as the form 
\begin{equation}
\tab=\lbrace \emptyset\,,\, \mu_1\,,\, \ldots \,,\, \mu_n\rbrace\,,
\end{equation}
where $\mu_n=\mu$.
As a consequence of the multiplicity-free property of the branching rule \eqref{eq:branching_restriction_Sn} the paths in the Young's lattice parametrize the basis of irreducible $\C\sn$-modules. This can be seen from the construction of the so-called Gelfand-Tsetlin basis for irreducible $\C\sn$-modules~\cite{doty2019canonical},  and the paths in $\Tab(\mu)$ are in to one correspondence with the standard tableaux of $\sn$.\smallskip


It is demonstrated in \cite{doty2019canonical} that the Young seminormal idempotents can be constructed from the formula
\begin{equation}\label{eq:Young_seminormal_idempotents}
	Y^{\,\tab}=\prod_{\mu \in \tab} Z^{\mu} \,,
\end{equation}
where the standard tableau $\tab$ is understood as a path in the Young's lattice of $\sn$. In the following two examples we shall just unpack the above formula for the construction of Young seminormal idempotents of $\C\Sn{3}$ and $\C\Sn{4}$. 

\begin{example}[The Young seminormal idempotents of $\C\mathfrak{S}_3$]
	
	\begin{figure}[H]
		\centering
		\includegraphics[scale=0.75]{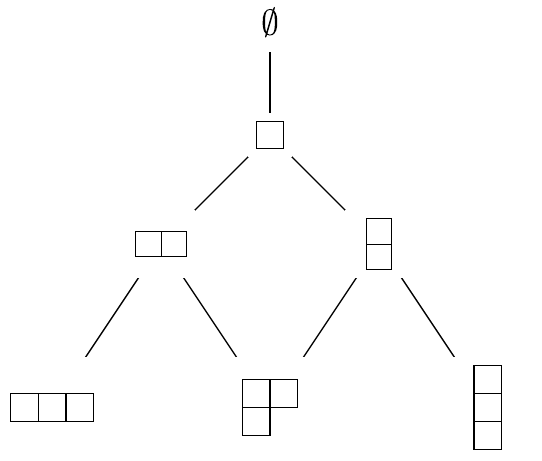}
		\caption{Bratteli diagram for $\lbrace\C\mathfrak{S}_k\rbrace_{k\leqslant 3}$.}
		\label{fig_proj:bratteli_s3}
	\end{figure}
	\noindent The number of elements in an equivalence class $\mu$ is given by the number of paths starting at $\emptyset$ and ending at $\mu$ and correspond to the number of standard tableau of shape $\mu$. In this example the paths are:
	\begin{equation*}
		\tab_1=\lbrace\,\emptyset\,,\, \Yboxdim{7pt}\yng(1),\,\Yboxdim{7pt}\yng(2),\,\Yboxdim{7pt}\yng(3) \,\rbrace\,, \hspace{0.3cm}
		\tab_2=\lbrace \,\emptyset\,,\,\Yboxdim{7pt}\yng(1),\,\Yboxdim{7pt}\yng(2),\,\Yboxdim{7pt}\yng(2,1)\, \rbrace\,, \hspace{0.3cm}
		\tab_3=\lbrace \,\emptyset\,,\,\Yboxdim{7pt}\yng(1),\,\Yboxdim{7pt}\yng(1,1),\,\Yboxdim{7pt}\yng(2,1) \,\rbrace\,, \hspace{0.3cm}
		\tab_4=\lbrace\,\emptyset\,,\, \Yboxdim{7pt}\yng(1),\,\Yboxdim{7pt}\yng(1,1),\,\Yboxdim{7pt}\yng(1,1,1) \,\rbrace\,. 
	\end{equation*}
	There is two paths ending at the vertex $\Yboxdim{6pt}\yng(2,1)$ which indicates that the dimension of the representation $\Yboxdim{6pt}\yng(2,1)$ is $2$.
	These paths are in one to one correspondence with the standard Young tableaux: 
	\begin{equation}\label{eq:Standard_tableau3}
		\tab_1=\Scale[0.8]{\young(123)}\,, \quad \tab_2=\Scale[0.8]{\young(12,3)}\,, \quad \tab_3 = \Scale[0.8]{\young(13,2)}\,, \quad \tab_4= \Scale[0.8]{\young(1,2,3)}\,.
	\end{equation} 
	Applying formula \eqref{eq:Young_seminormal_idempotents} with $Z^\emptyset=1$, the Young seminormal idempotents of $\C\mathfrak{S}_3$ are given by:
	\begin{equation}
		\begin{aligned}
			Y^{\,\tab_1}=Z^{\Yboxdim{3pt}\yng(3)}\, ,  \quad Y^{\,\tab_2}=Z^{\Yboxdim{3pt}\yng(2)}Z^{\Yboxdim{3pt}\yng(2,1)}
			\, , \quad Y^{\,\tab_3}=Z^{\Yboxdim{3pt}\yng(1,1)}Z^{\Yboxdim{3pt}\yng(2,1)}\, , \quad Y^{\,\tab_4}=Z^{\Yboxdim{3pt}\yng(1,1,1)}\,,
		\end{aligned}
	\end{equation}
	where we have use the facts that $Z^{\Yboxdim{3pt}\yng(1)}=1$, and
	\begin{equation}\label{eq:property_CY}
		Z^{(k-1)}Z^{(k)}=Z^{(k)}\, , \hspace{1cm} Z^{(1^{k-1})}Z^{(1^{k})}=Z^{(1^k)}\,.
	\end{equation}
\end{example}

\begin{example}[The Young seminormal idempotents of $\C\mathfrak{S}_4$]
	\begin{figure}[H]
	\centering
		\includegraphics[scale=0.98]{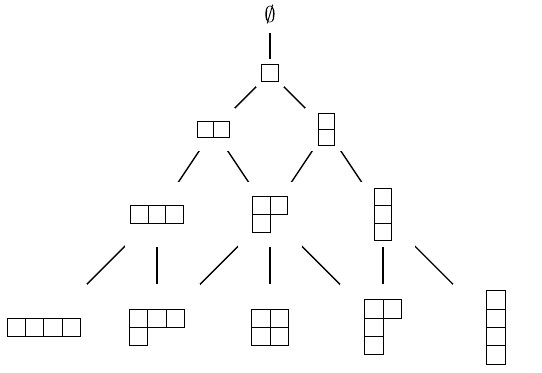}
	\caption{Bratteli diagram for $\lbrace\C\mathfrak{S}_k\rbrace_{k\leqslant 4}$.}
	\label{fig_proj:bratteli_s4}
	\end{figure}
From the paths of the above Bratteli diagram one obtains the following standard tableaux:
\begin{equation*}
	\arraycolsep=10pt
	\begin{array}{lllll}
		\tab_1=\Scale[0.8]{\young(1234)}\, &\tab_2=\Scale[0.8]{\young(123,4)}\, &\tab_3=\Scale[0.8]{\young(124,3)}\, &\tab_4=\Scale[0.8]{\young(134,2)}\,, &\tab_5=\Scale[0.8]{\young(12,34)}\,,\\[10pt]
		\tab_6=\Scale[0.8]{\young(13,24)}\,, &\tab_7=\Scale[0.8]{\young(12,3,4)}, &\tab_8=\Scale[0.8]{\young(13,2,4)}\, &\tab_9=\Scale[0.8]{\young(14,2,3)}\,, &\tab_{10}=\Scale[0.8]{\young(1,2,3,4)}\,.
	\end{array}
\end{equation*}
Applying formula \eqref{eq:Young_seminormal_idempotents}, the associated Young seminormal idempotents of $\C\Sn{4}$ are given by:
\begin{equation}\label{eq:seminormalYoung4}
	\arraycolsep=10pt
	\begin{array}{llll}
		Y^{\,\tab_1}=Z^{\Yboxdim{3pt}\yng(4)}\,, 
		&Y^{\,\tab_2}=Z^{\Yboxdim{3pt}\yng(3)}Z^{\Yboxdim{3pt}\yng(3,1)}\,, 
		&Y^{\,\tab_3}=Z^{\Yboxdim{3pt}\yng(2)}Z^{\Yboxdim{3pt}\yng(2,1)}Z^{\Yboxdim{3pt}\yng(3,1)}\, &Y^{\,\tab_4}=Z^{\Yboxdim{3pt}\yng(1,1)}Z^{\Yboxdim{3pt}\yng(2,1)}Z^{\Yboxdim{3pt}\yng(3,1)}\,,\\[10pt]
		Y^{\,\tab_5}=Z^{\Yboxdim{3pt}\yng(2)}Z^{\Yboxdim{3pt}\yng(2,1)}Z^{\Yboxdim{3pt}\yng(2,2)}\,,
		&Y^{\,\tab_6}=Z^{\Yboxdim{3pt}\yng(1,1)}Z^{\Yboxdim{3pt}\yng(2,1)}Z^{\Yboxdim{3pt}\yng(2,2)}\,, &Y^{\,\tab_7}=Z^{\Yboxdim{3pt}\yng(2)}Z^{\Yboxdim{3pt}\yng(2,1)}Z^{\Yboxdim{3pt}\yng(2,1,1)}\,, &Y^{\,\tab_8}=Z^{\Yboxdim{3pt}\yng(1,1)}Z^{\Yboxdim{3pt}\yng(2,1)}Z^{\Yboxdim{3pt}\yng(2,1,1)}\,,\\[10pt]
		Y^{\,\tab_9}=Z^{\Yboxdim{3pt}\yng(1,1,1)}Z^{\Yboxdim{3pt}\yng(2,1,1)}\,, 
		&Y^{\,\tab_{10}}=Z^{\Yboxdim{3pt}\yng(1,1,1,1)}\,.
	\end{array}
\end{equation}
\end{example}


\begin{mathematicas}[\textit{Package BrauerAlgebra}]	
	SemiNormalYoungUnit[\,$\tab$\,]\,\\
	\textit{Returns the Young Seminormal idempotent associated with $\tab$\,.}\medskip
	
	YoungSymmetrizer[\,$\tab$\,]\,\\
	\textit{Returns the Young symmetrizer associated with $\tab$\,.}
\end{mathematicas}

Let us mention that their exists alternative formulas for the construction of the Young seminormal idempotents. They can either be constructed from the Jucys-Murphy elements  \cite{jucys1966young,Murphy,molev2006fusion} or from the Young symmetrizers \cite{Thrall_seminormalYoung_1941}. Yet, for the purpose of the irreducible decomposition of a tensor under the action of $\GL(\Dim)$ the formula \eqref{eq:Young_seminormal_idempotents} is very attractive as it offers control and transparency for the decomposition. It allows one the possibility to go slowly, first perform the unique isotypic decomposition and then, if needed, ‘reduce' the isotypic decomposition to an irreducible decomposition by following the path on the Bratteli diagram. This procedure was followed for the irreducible decomposition of the distortion and Riemann tensor with respect to $\GL(\Dim,\R)$ in the sections \ref{sec:Distortion_Decomposition} and \ref{sec:Riemann_Decomposition} of chapter \ref{chap:Irreducible_MAG}. \smallskip 

We conclude this chapter with the following lemma which guarantees that an irreducible decomposition of $V^{\otimes n}$ obtained from the Young seminormal idempotents will be orthogonal with respect the scalar product \eqref{eq:scalar_product} defined in chapter \ref{chap:Irreducible_MAG}. Recall that for any $s\in \mathfrak{S}_n$, the operation $(\cdot)^*$ of flipping a permutation diagram with respect to the middle horizontal line defined in the first chapter correspond to the inverse operation in $\mathfrak{S}_n$: $s^*=s^{\shortminus 1}$. This flip operation is extended to the algebra by linearity and any element $z\in \mathbb{C}\mathfrak{S}_n $ such that 
\begin{equation}\label{eq:flip_inv_Sn_1}
	z=z^{*}
\end{equation}
is said to be \textit{flip invariant}.

\begin{lemma}\label{lem:flip_inv_SN}
	The Young seminormal idempotents are flip invariant elements of $\mathbb{C}\mathfrak{S}_n$.
\end{lemma}
\begin{proof}
	Choose any conjugacy class sum $K^{\rho}$ with $\rho\vdash k$ and $1 \leqslant k\leqslant n$. By construction one clearly has $K^{\rho}=(K^{\rho})^{*}$. Hence any element of center of $\mathbb{C}\mathfrak{S}_k$ is flip invariant. In particular any element $Z^{\rho}$ with $\rho\vdash k$ and $1 \leqslant k\leqslant n$ is flip invariant. Also recall that the central Young projectors $Z^\rho$ with $\rho\vdash k$ forms a basis of the center $\mathcal{Z}_k$ of $\mathbb{C}\mathfrak{S}_k$. 
	The Gelfand-Tsetlin algebra of $\mathbb{C}\mathfrak{S}_n$, that is the algebra generated by the centers 
	\begin{equation*}
		\mathcal{Z}_1,\,\mathcal{Z}_2,\ldots,\mathcal{Z}_{n-1},\,\mathcal{Z}_n \,,
	\end{equation*}
	is a maximal commutative subalgebra of $\mathbb{C}\mathfrak{S}_n$~\cite[Proposition 1.1]{Okounkov1996}. Finally, because the flip operation is an anti-involution, the product of commuting flip invariant elements is a flip invariant element. From the formula \eqref{eq:Young_seminormal_idempotents} the Young seminormal idempotents are therefore flip invariant.
\end{proof}

	\chapter{The projection operators for $\Or(\Dim,\mathbb{C})$ irreducible decomposition of tensors}\label{chap:projectors_O}
\vspace{0.5 cm}

\section{Introduction}\label{sec:Introduction_chapter4}

Consider the Brauer algebra $B_n(\Dim)$ over the complex field $\mathbb{C}$ with parameter $\Dim$. Unless otherwise stated, $\Dim$ is an integer corresponding to the dimension of a vector space $V$ with a symmetric non-degenerate metric, where the orthogonal group $\Or(\Dim,\C)$ acts. For the sake of brevity, in what follows we will omit the parameter $\Dim$ when referring to the Brauer algebra.\medskip

We aim at constructing the elements $P^{\lambda}_n \in \bn$ which perform the isotypic decomposition of $V^{\otimes n}$ with respect to $\Or(\Dim,\C)$
\begin{equation}\label{eq:isotypic_decomposition_O_2}
	V^{\otimes n}=\bigoplus_{\lambda\in \Lambda_n(\Dim)}(V^{\otimes n})\cdot P^{\lambda}_n\,,
\end{equation}
such that
\begin{equation}\label{eq:properties_isotypic_proj}
	\hspace{0.5cm} (V^{\otimes n})\cdot P^{\lambda}_n=\left(D^{\lambda}\right)^{\oplus m_\lambda}\,,\hspace{1cm}(V^{\otimes n})\cdot (P^{\lambda}_n P^{\lambda}_n)=(V^{\otimes n})\cdot P^{\lambda}_n\,.
\end{equation}
Recall the definition of $\Lambda_n(\Dim)$ in \eqref{eq:Lambda_d}. The elements $P^{\lambda}_n$ act by the identity on $D^\lambda$ and annihilates any irreducible module over $\Or(\Dim,\C)$ not isomorphic to $D^\lambda$.\medskip

As already mentioned in chapter \ref{chap:Irreducible_MAG}, $\bn$ acts on $V^{\otimes n}$ such that the latter decomposes as a direct sum of irreducible $\bn$-modules:
\begin{equation}\label{eq:decomposition_V_bn_2}
	V^{\otimes n}\cong\bigoplus_{\lambda\vdash \Lambda_n(\Dim)} \left(M^{\lambda}_n\right)^{\oplus g_{\lambda}}\,,
\end{equation}
where $g_\lambda=\dim(D^\lambda)$. The main idea is that by Schur-Weyl duality one has 
\begin{equation}
\left(D^{\lambda}\right)^{\oplus m_\lambda}=\left(M^{\lambda}_n\right)^{\oplus g_\lambda}\,,
\end{equation}
which means that projecting to an isotypic component of $V^{\otimes n}$ under $\Or(\Dim,\C)$ action amounts to projecting to the corresponding isotypic component of $V^{\otimes n}$ under $\bn$ action. Hence,
the sought elements $P^{\lambda}_n$ act by the identity on $M^\lambda_n$ and annihilate any irreducible $\bn$-module not isomorphic to $M^\lambda_n$.\medskip

Let us stress that by definition, the image of each $P^{\lambda}_n$ in $\End(V^{\otimes n})$ is a projector, while they are not necessarily projectors as elements of $\bn$. Because we are interested essentially in their action on $V^{\otimes n}$ we will still refer to $P^\lambda_n$ as projectors. Besides, when $\Dim\geqslant n-1$, $\bn$ is semisimple \cite{Brown_semisimplicity_1956,Rui_Br_semisimple} and each $P^{\lambda}_n$ is a central idempotent in $\bn$ \cite{doty2019canonical}. In this case, they enjoy the following properties:\medskip 

\begin{equation}\label{eq:properties_central_idempotents}
	\begin{aligned}
		&P^{\lambda}_{n} \ v= v \ P^{\lambda}_{n} \quad \text{for any } \, v\in B_n\,, \quad &(\textit{central}\,) \\
		&P^{\beta}_{n}P^{\lambda}_{n}=0  \quad \text{for any } \beta\neq\lambda\,, \quad &(\textit{pairwise orthogonal}\,) \\
		&P^{\lambda}_{n}P^{\lambda}_{n}=P^{\lambda}_{n}\,, \quad &(\textit{idempotent \text{\slash} projector}\,)\\
		&\id_n=\sum_{\lambda\in \Lambda_n}P_{n}^\lambda\,. \quad &(\textit{partition of unity}\,)\\
	\end{aligned}
\end{equation}
where the set  $\Lambda_n$ corresponds to $\Lambda_n(\Dim)$ where the restrictions on the number of rows in the first two columns of the Young diagrams are dropped. \medskip

The construction proposed here resembles the one presented in the previous chapter for the central idempotent of $\mathbb{C}\mathfrak{S}_n$. It relies on a modified Lagrange-interpolation-type formula involving a particular element $A_{n}\in B_n$ which commutes with the symmetric group (that is, $A_{n}\in  \cn$, see chapter \ref{chap:cn} for a study of $\cn$). Remarkably $A_n$  is diagonalizable on each irreducible components $M^\lambda_n$, with eigenvalues parametrized by pairs of Young diagrams and given by: 
\medskip
\begin{equation}
	a_{\mu\backslash\lambda}=(\Dim-1)f_\lambda\, + c_\mu-c_\lambda\,, \hspace{0.5cm} \text{with}\hspace{0.5cm} f_\lambda=\frac{n-|\lambda|}{2}\,.
\end{equation}
The other ingredient of the construction is an induction map $\mathcal{A}\,:\, \Bn{n-2}\to\bn$ defined in section \ref{subsection:average_arc}.
The main results of this chapter (Theorem \eqref{theo:non_inductive_traceless_projectors} and Theorem \eqref{theo:isotypic_projectors_2}) can be presented as:
\begin{itemize}
	\item[\textit{i)}]If $f_\lambda=0$, that is $\lambda\vdash n$\,,
	\begin{equation}\label{eq:non_inductive_traceless_projectors_0}
		\hspace{-2.5cm}\text{then}\hspace{2cm}P^{\lambda}_{n}=Z^\lambda \prod_{\begin{array}{c}
				{\scriptstyle \beta \in \overbar{\Lambda}_{\lambda}(\Dim)}\\
		\end{array}}\left(1 - \frac{A_n}{a_{\lambda\backslash\beta}}\,\right)\,.
	\end{equation}
	\item[\textit{ii)}]If $f_\lambda\geqslant 1$\,,
	\begin{equation}\label{eq:main_res_f_traceless_central_idempotent}
		\hspace{0.5cm}\text{then}\hspace{2cm}	P_n^{\lambda}=\sum_{\mu\,\in \overbar{\mathrm{M}}_{n,\lambda}(\Dim)}P_n^{\mu\backslash\lambda}\,,\hspace{0.5cm} \text{with} \hspace{0.5cm} P_n^{\mu\backslash\lambda}=\,\frac{\mathcal{A}(P^{\lambda}_{n-2})\, Z^\mu}{a_{\mu\backslash\lambda}}\,. \hspace{2cm} 
	\end{equation}
\end{itemize}
The set of Young diagrams $ \overbar{\Lambda}_{\lambda}(\Dim)$ and $\overbar{\mathrm{M}}_{n,\lambda}(\Dim)$ are defined in \eqref{eq:subsets_M_L}.\medskip

As direct consequence of \eqref{eq:irreducible_trace_decomposition} we construct the elements $P_n^{(f)}$ which by their action on $V^{\otimes n}$ realize the trace decomposition of Weyl \eqref{eq:trace_decomposition}:
\begin{equation}\label{eq:proj_trace_decomposition}
P_n^{(f)}=\displaystyle{\sum_{\begin{array}{c}
			{\scriptstyle \lambda\in \Lambda_n(\Dim)}\\
			{\scriptstyle |\lambda|=n-2f}
\end{array}}} P^{\lambda}_{n}\,.
\end{equation}
\section{The traceless projectors $\lbrace P^{\lambda}_n\,,\lambda\vdash n \rbrace$}\label{sec:central_traceless}

\subsection{The element $X_n$ and Jucys-Murphy elements}\label{subsec:JM_Xn}

In what follows we make use of the results of \cite{Nazarov}. Consider the following set of pairwise-commuting elements known as the {\it Jucys-Murphy elements} $J_k$ ($k = 1,\dots, n$): 
\begin{equation}\label{eq:JM_Brauer}
\quad J_{k} = \sum_{j=1}^{k-1} (s_{jk} - d_{jk})\quad \text{for}\quad k\geqslant 2\,,
\end{equation}
with 
\begin{equation}\label{eq:sij_dij}
s_{ij}=\raisebox{-.3\height}{\includegraphics[scale=0.65]{fig/Tn.pdf}}\,,\hspace{0.5cm}	d_{ij}=\raisebox{-.3\height}{\includegraphics[scale=0.7]{fig/An.pdf}}\,,
\end{equation}
and we define $J_1=0$. 
In \eqref{eq:JM_Brauer} we follow the definition in \cite{doty2019canonical} where an unnecessary shift of $(\Dim-1)/2$ as been removed compared to the definition in \cite{Nazarov}. The following proposition is contained in Corollary $2.4$ and in the proof of Theorem $2.6$ in \cite{Nazarov}.
\begin{proposition}
The element 
\begin{equation}\label{eq:def_Xn}
	X_n = \sum_{k = 1}^{n} J_k\,,
\end{equation}
is central in $B_n$  and acts on an irreducible $B_n$-module $M_n^\lambda$ as the scalar 
\begin{equation}
	x_\lambda=(1-\Dim)f_\lambda+c_\lambda\,, \hspace{0.5cm} \text{with}\hspace{0.5cm} f_\lambda=\dfrac{n-|\lambda|}{2}\,.
\end{equation}
\end{proposition}
We recall that $c_\lambda$ is the content \eqref{eq:content_YD} of the Young diagram $\lambda$.
\begin{remark}
	\begin{itemize}
\item[\it i)] For the purpose of applications to tensor calculus one make use of Schur-Weyl duality, so that the previous lemma implies: 
		\begin{equation}
			\text{for any $T\in V^{\otimes n}$ one has}	\hspace{0.5cm} T\cdot (P^\lambda_n X_n)=x_\lambda \left(T \cdot P^\lambda_n \right) .
		\end{equation}
		Equivalently,
		\begin{equation}
			\text{for any $T\in D^\lambda$ one has}	\hspace{0.5cm} T\cdot X_n=x_\lambda T\,.
		\end{equation}
\item[\it ii)] In the semisimple regime of Brauer the previous Lemma implies
\begin{equation}
	X_n P^\lambda_n= x_\lambda P^\lambda_n\,, \hspace{1cm} \text{for all $\lambda\in \Lambda_n(\Dim)$}.
\end{equation}
The same relation holds when replacing central idempotents by primitive idempotents. \medskip
\end{itemize}
\end{remark} 

\subsection{Restriction of $\bn$ to $\C\sn$ and the element $A_n$ of $\bn$} 

\medskip
\paragraph{Branching rules $\bn\downarrow\C\sn $.}
Recall that any irreducible tensor representation of $\GL(\Dim)$  decomposes into a direct sum of irreducible tensor representations of $\Or(\Dim,\C)$\cite{Koike_Terada_banching,Enright_Willenbring_branching,Kwon_branching,Jang_Kwon_branching}:
\begin{equation}\label{eq:branching_GL_O_2}
	V^{\mu} \cong \bigoplus_{\begin{array}{c}
			{\scriptstyle \text{even }\nu \, \subset \mu}\\
	\end{array}}\big(D^{\lambda}\big)^{\oplus \tensor{\overline{C}}{^{\,\mu}_\lambda_\nu}(\Dim)}\quad\text{(upon $\GL(\Dim,\mathbb{C})\downarrow \Or(\Dim,\mathbb{C})$)}\,.
\end{equation}
One the other hand, any irreducible $B_{n}$-modules decompose into a direct sum of irreducible $\mathbb{C}\Sn{n}$-modules upon restriction to the subalgebra $\mathbb{C}\Sn{n} \subset B_{n}$. The branching rules $\bn\downarrow\C\sn $ for the irreducible $B_{n}$-modules $M^{\lambda}_n$ appearing in $V^{\otimes n}$ can be obtained via comparing the two multiplicity-free decompositions (see appendix \ref{subsec:Double Centralizer Theorem} for more detail on Schur-Weyl duality)

\begin{equation}
V^{\otimes n}\cong\bigoplus_{\begin{array}{c}
		{\scriptstyle \mu \, \in \mathcal{P}_n(\Dim)}\\
\end{array}}V^{\mu}\otimes L^\mu\,,
\hspace{1cm}V^{\otimes n}\cong\bigoplus_{\begin{array}{c}
		{\scriptstyle \lambda \, \in \Lambda_n(\Dim)}\\
\end{array}}D^{\lambda}\otimes M^\lambda_n\,,
\end{equation}
of the same space $V^{\otimes n}$. This was done in \cite[Lemma 4.2]{Gavarini_LitRich} and one has 
\begin{equation}\label{eq:branching_Bn_Sn}
M^{\lambda}_n \cong \bigoplus_{\begin{array}{c}
				{\scriptstyle \text{even }\nu \, \subset \mu}\\
		\end{array}} \big(L^{\mu}\big)^{\oplus \tensor{\overline{C}}{^{\,\mu}_\lambda_\nu}(\Dim)}\quad\text{(upon $B_{n}(\Dim)\downarrow \mathbb{C}\Sn{n}$)}\,.
\end{equation}
Obtaining the coefficient $\tensor{\overline{C}}{^{\,\mu}_\lambda_\nu}(\Dim)$ when $\Dim<n-1$ is a difficult problem. It was solved recently in \cite[Theorem 4.17 and Remark 4.19]{Jang_Kwon_branching} with a sophisticated combinatorial approach.\medskip

We recall that when $\mu\in\Lambda_{n}(\Dim)$ that is $\mu\vdash n$ and $\mu_1^{\prime}+\mu_2^{\prime}\leqslant\Dim$ one has 
\begin{equation}\label{eq:Littlewood_restriction_rule}
	\tensor{\overline{C}}{^{\,\mu}_\lambda_\nu}(\Dim)=  \tensor{C}{^{\,\mu}_\lambda_\nu}\,, \hspace{1cm} \text{(the Littlewood-Richardson coefficients)}
\end{equation}
and \eqref{eq:branching_GL_O_2} is referred to as the \textit{Littlewood's restriction rules}. Two combinatorial methods for the computation of the Littlewood-Richardson coefficients are presented in the appendix \ref{subsec:Littlewood_Richardson_rules}.
\begin{remark}\label{rem:stable_case}
	The regime $\Dim\geqslant n$ is called the \textit{stable} case. There, all integer partitions $\mu\vdash n$ are both in $\Par_n(\Dim)$ and in $\Lambda_n(\Dim)$ so that the Littlewood's restriction rule applies. 
\end{remark}

\noindent The following lemma is contained in Lemma $3.2.$ of \cite{bulgakova2022construction}.
\begin{lemma}\label{lem:Branching_0_GL_Bn_semisimple}
	Let $\Dim=n-1$ \textup{(}which correspond to the limiting case of the semisimple regime of $\bn$\textup{)}.
	\begin{equation}
		\hspace{-0.5cm}\text{Then,\hspace{0.5cm} for any  $\mu\vdash n $ and $\lambda\neq\mu$}\hspace{0.5cm}\text{holds}\hspace{0.5cm}\tensor{\overline{C}}{^{\,\mu}_\lambda_\nu}(\Dim)=  \tensor{C}{^{\,\mu}_\lambda_\nu}\,\,.
	\end{equation}
\end{lemma}

Let $\lambda\in \Lambda_n(\Dim)$ with $|\lambda|<n$. We denote by $\mathrm{M}_{n,\lambda}(\Dim)$ the set of Young diagrams $\mu$ parametrizing the pairwise inequivalent irreducible $\C\sn$-modules which appear in the irreducible $\bn$-module $M^\lambda_n$ upon restriction to $\C\sn$:
\begin{equation}\label{eq:set_Mnlambda}
	\mathrm{M}_{n,\lambda}(\Dim)=\lbrace \mu \in \Par_n(\Dim)\, \st \, \tensor{\overline{C}}{^{\,\mu}_\nu_\lambda}(\Dim)\neq 0 \text{ \, for some $\nu\neq \emptyset$ even} \rbrace\,.
\end{equation}

Conversely, for $L^\mu$ a given simple $\C\sn$-module, we define $\Lambda_\mu(\Dim)\subset \Lambda_n(\Dim)$ as the set of young diagrams $\lambda$ parametrizing the pairwise inequivalent simple $\bn$-modules $M^\lambda_n$ with $|\lambda|<n$ where $L^\mu$ is present upon restriction to $\C\sn$: 
\begin{equation}
	\Lambda_\mu(\Dim)=\lbrace \lambda \in \Lambda_{|\mu|}(\Dim)  \,\st\, \tensor{\overline{C}}{^{\,\mu}_\nu_\lambda}(\Dim)\neq 0 \text{ \, for some $\nu\neq \emptyset$ even} \rbrace\,.
\end{equation}


\paragraph{$A_n$ : the principal building block of the construction.}

Recalling that $\mathcal{C}_n$ denotes the centralizer of $\sn$ in $\bn$ we define the element $A_{n}\in \cn \subset B_n$ as the sum of flip invariant one arc diagrams: 
\begin{equation}\label{eq:master_class}
	A_n =\sum_{1\leqslant i<j\leqslant n} \, d_{ij}\,\,,
\end{equation}
where we recall that $d_{ij}$ is given by
\begin{equation}\label{eq:dij}
	d_{ij}=\raisebox{-.3\height}{\includegraphics[scale=0.7]{fig/An.pdf}}\,.
\end{equation}
From the definition \eqref{eq:def_Xn} of the central element $X_n$ one has
\begin{equation}\label{eq:X_nA_n}
		X_n = T_n - A_n \,,
\end{equation}
where $T_n$ were introduced in chapter \ref{chap:projectors_GL} (see \eqref{eq:Tn_def}).
\begin{lemma}[{\cite[Lem. $3.3$]{bulgakova2022construction}}]\label{lem:A_block_diagonal}
	Let $M^\lambda_n$ be a simple $B_n$-module present in $V^{\otimes n}$, and let a simple $\mathbb{C}\Sn{n}$-module $L^\mu$ occurs in the decomposition of $M^\lambda_n$ into irreducible summands upon restriction to $\mathbb{C}\Sn{n}$. Then 
	\begin{equation}\label{eq:eigenvalue_A}
		\text{for any}\quad v\in L^{\mu}\,, \quad v\cdot A_n = a_{\mu\backslash\lambda}\,v 
		\hspace{0.3cm}\text{with}\quad a_{\mu\backslash\lambda} = (\Dim-1)f_\lambda\, + c(\mu)-c(\lambda)\,.
	\end{equation}
\end{lemma}
\begin{proof}
 The relation \eqref{eq:X_nA_n}, together with the fact that $X_n$ and $T_n$ are both proportional to identity on $L^{\mu}\subset M^{\lambda}_n$, implies that $A_n$ is proportional to identity on $L^{\mu}$ as well. The eigenvalue in the assertion is a direct consequence of \eqref{eq:X_nA_n} for $X_n$ (respectively, $T_n$) restricted to $M^{\lambda}_n$ (respectively, to $L^{\mu}$).
\end{proof}
\begin{remark}
\begin{itemize}
\item[\it{i)}] The element $A_n$ acts by zero on traceless modules $M^{\lambda}_n$ which are such that $f_\lambda=0$. Indeed, for such integer partitions $\lambda$ one has $a_{\lambda\backslash\lambda}=0$.
\item[\it{ii)}] From the point of view of tensor calculus Lemma \eqref{lem:A_block_diagonal} implies, via Schur-Weyl duality, that for any $T^{(\mu)}\in V^\mu$ such that
$D^\lambda$ occurs in the direct sum decomposition of $V^\mu$ into irreducible components under restriction to $\Or(\Dim,\C)$
one has
\begin{equation}
	T^{(\mu)}\cdot (P^\lambda_n A_n)=a_{\lambda\backslash\mu} \,\left(T^{(\mu)}\cdot P^\lambda_n\right)\,.
\end{equation}
\item[\it{iii)}] In the semisimple regime of $\bn$ one may write
\begin{equation}
	P^\lambda_n\, A_n= (T_n- x_\lambda)P^\lambda_n\,.
\end{equation}
Multiplication by the central Young projector $Z^\mu\in \mathcal{Z}_n$ on both sides of the previous equation yields
\begin{equation}
	 P^{\mu\backslash\lambda}_n\, A_n = a_{\mu\backslash\lambda}\, P^{\mu\backslash\lambda}_n\, \hspace{0.5cm} \text{with} \hspace{0.5cm}  a_{\mu\backslash\lambda}=c_\mu -x_\lambda\,,
\end{equation} 
where $P^{\mu\backslash\lambda}_n:= P^\lambda_n Z^\mu $. Theses relations hold when replacing central idempotent by primitive idempotents. \medskip
\end{itemize}
\end{remark}
\begin{lemma}[{\cite[Lem. $3.1$]{bulgakova2022construction}}]\label{lem:diag_ev_A_2}
	The action of $A_n$ on $V^{\otimes n}$ is diagonalizable. The subspace $\mathrm{Ker}\,A_n \subset V^{\otimes n}$ is exactly the space of traceless tensors, while non-zero eigenvalues of $A_n$ are positive integers. 
\end{lemma}
\noindent See \cite{bulgakova2022construction} for the proof.\medskip 

We recall that the non-negative integer $\tensor{\overline{C}}{^{\,\mu}_\lambda}(\Dim)=\sum_{ \text{even }\nu\,\subset \, \mu} \tensor{\overline{C}}{^{\,\mu}_\lambda_\nu}(\Dim)$ gives the multiplicity of the irreducible representation $L^\mu$ in $M^\lambda_n$  (respectively $D^\lambda$ in $V^\mu$), which is often denoted $[\, M^\lambda_n\,:\, L^\mu\,]$.
The following proposition may be seen as a reformulation of Proposition $3.4.$ in \cite{bulgakova2022construction}.
\begin{proposition}[{\cite[Prop. $3.4$]{bulgakova2022construction}}]\label{prop:condition_A_branching_GL_O_2}
	Let $\lambda\in \Lambda_n(\Dim)$ such that $f_\lambda>0$ ($\, |\lambda|\neq n$). For any $\mu\in \Par_n(\Dim)$ such that $\tensor{C}{^\mu_\lambda}\neq 0$,
	\begin{itemize}
		\item[i)] If $\,\Dim\geqslant n-1$,  \hspace{0.5cm} then \hspace{0.3cm} $\tensor{\overline{C}}{^{\,\mu}_\lambda}(\Dim)=\tensor{C}{^{\,\mu}_\lambda}$, \hspace{0.3cm} \text{and} \hspace{0.3cm} $a_{\mu\backslash\lambda}$ as defined in \eqref{lem:A_block_diagonal} is a positive eigenvalue of $A_n$ on $V^{\otimes n}$.
		\item[ii)] If $\,\Dim<n-1$, \hspace{0.5cm} then  \hspace{0.3cm} $a_{\mu\backslash\lambda}\leqslant 0\,\hspace{0.3cm} \text{implies} \hspace{0.3cm}\tensor{\overline{C}}{^{\,\mu}_\lambda}(\Dim)=0 \hspace{0.2cm}(\text{$L^\mu$ is not present in $M^\lambda_n$})\,.$
	\end{itemize}
\end{proposition}
\begin{proof}
	For the first point, if $\,\Dim\geqslant n$ we are in the stable regime so $\tensor{\overline{C}}{^{\,\mu}_\lambda}(\Dim)=\tensor{C}{^{\,\mu}_\lambda}$, hence $L^\mu \subset M^\lambda_n$ (resp. $D^\lambda\subset V^\mu$). If $\,\Dim = n-1$, by Lemma \ref{lem:Branching_0_GL_Bn_semisimple} one has $\tensor{\overline{C}}{^{\,\mu}_\lambda}(\Dim)=\tensor{C}{^{\,\mu}_\lambda}$ and hence $L^\mu\subset M^\lambda_n$ (resp. $D^\lambda\subset V^\mu$). By Lemma \ref{lem:diag_ev_A_2}, $a_{\mu\backslash\lambda}$ is a positive eigenvalue of $A_n$ on $V^{\otimes n}$.\medskip 
	
	\noindent For the second point, note that if $\tensor{\overline{C}}{^{\,\mu}_\lambda}(\Dim)\neq 0$ then by Lemma \ref{lem:diag_ev_A_2} $a_{\mu\backslash\lambda}>0$ is an eigenvalue of $A_n$ on $V^{\otimes n}$.
\end{proof}

In order to construct the elements $P^\lambda_n$ it will be sufficient to consider the set of Young diagrams $\overbar{\mathrm{M}}_{n,\lambda}(\Dim)$ and $\overbar{\Lambda}_{\mu}(\Dim)$ defined as 
\begin{equation}\label{eq:subsets_M_L}\
	\begin{aligned}
		\overbar{\mathrm{M}}_{n,\lambda}(\Dim)&=\lbrace \mu \in \Par_n(\Dim) \,\st\, \tensor{C}{^\mu_\nu_\lambda}\neq 0 \text{ \, for some $\nu$ even, and $a_{\mu\backslash\lambda}\geqslant1$} \rbrace\,, \\[8pt]
		\overbar{\Lambda}_{\mu}(\Dim)&=\lbrace \lambda \in  \Lambda_{|\mu|}(\Dim) \,\st\, \tensor{C}{^\mu_\nu_\lambda}\neq 0 \text{ \, for some $\nu$ even, and $a_{\mu\backslash\lambda}\geqslant1$} \rbrace\,.
	\end{aligned}
\end{equation}
These two sets can be constructed efficiently (at least for a computer) from the first and second formulations (\textit{jeu de taquin}) of the Littlewood-Richardson rule presented in section \ref{subsec:Littlewood_Richardson_rules} of appendix \ref{app:maths}. As a result, we will avoid the computation of the generalized Littlewood-Richardson coefficients $\tensor{\overline{C}}{^{\,\mu}_\lambda_\nu}(\Dim)$ which is a rather tedious task \cite{Jang_Kwon_branching}. As a direct consequence of the Proposition \ref{prop:condition_A_branching_GL_O_2} one has the following lemma. 
\begin{lemma}\label{lem:lemma_sets} In the semisimple regime of $B_n$ ($\Dim\geqslant n-1$) one has
	\begin{equation}\label{eq:sets_semisimple}
		\mathrm{M}_{n,\lambda}(\Dim)=\overbar{\mathrm{M}}_{n,\lambda}(\Dim)\,, \hspace{0.5cm} \text{and} \hspace{0.5cm} \mathrm{\Lambda}_\mu(\Dim)=\overbar{\Lambda}_{\mu}(\Dim)\,,
	\end{equation}
while for $\Dim < n-1$ 
\begin{equation}\label{eq:sets_not_semisimple}
	\mathrm{M}_{n,\lambda}(\Dim)\subseteq \overbar{\mathrm{M}}_{n,\lambda}(\Dim)\,, \hspace{0.5cm} \text{and} \hspace{0.5cm}	\mathrm{\Lambda}_\mu(\Dim)\subseteq 	\overbar{\Lambda}_{\mu}(\Dim)\,.
\end{equation}
\end{lemma} 

\begin{remark}
\begin{itemize}
\item[\it i)] To our knowledge, whether or not \eqref{eq:sets_semisimple} holds for $\Dim < n-1$ is an open problem. 
\item[\it ii)] Let $\beta\in \Lambda_n(\Dim)$ with $|\beta|<n$, and $v\in M^{\beta}_n$. Then $v$ decomposes as 
\begin{equation}
	v=\sum_{\rho\in \overbar{\mathrm{M}}_{n,\lambda}}v^{\rho}\,, \hspace{0.5cm} \text{with} \hspace{0.5cm} v^{\rho}=v\cdot Z^{\rho}\,,
\end{equation}
where for $\Dim < n-1$ some of the $v^{\rho}$ in the sum may be identically zero. Indeed, $L^\rho$ \lp respectively  $D^\beta$ \rp\,  may not present in $M^\beta_n$ \lp respectively $V^\rho$\rp \,, whereas $a_{\rho\backslash\beta}$ may still be a positive integer.
\end{itemize}
\end{remark}

In the following examples we demonstrate how Proposition \ref{prop:condition_A_branching_GL_O_2} can be used to conclude on the presence (or absence) of certain irreducible module $L^\mu$ in  $M^\lambda_n$ (respectively $D^\lambda$ in  $V^\mu$) for small $\Dim$, without having to compute the coefficients $\tensor{\overline{C}}{^{\,\mu}_\lambda}(\Dim)$. 
\begin{example}\label{ex:Branching_O_GL} The first three examples are relevant for the irreducible decomposition of the distortion and Riemann tensor for small $\Dim$ presented in chapter \ref{chap:Irreducible_MAG}. In particular, these examples give an alternative to the dimensional analysis given in the Propositions \ref{prop:small_d_distortion} and \ref{prop:small_d_Riemann}.\medskip 
	
	\noindent Consider the case $n=3$ with $\Dim=2$ (semisimple regime of $\bn$).
	\begin{itemize}
		\item[i)] Take the integer partitions $\mu=(2,1)$ and $\lambda=(1)$. One has $\mu\notin \Lambda_3(2)$ and $\lambda\in \Lambda_3(2)$. One can check that $\tensor{C}{^\mu_\lambda}=1$ so $\tensor{\overline{C}}{^{\,\mu}_\lambda}=1$ and $L^\mu$ is present in $M^\lambda_n$.  In this case one has $a_{\mu\backslash\lambda}=1$.
	\end{itemize}
	Consider the case $n=4$ with $\Dim=2$.	
	\begin{itemize}
		\item[ii)] Take the integer partitions $\mu=(2,2)$ and $\lambda=(2)$. One has $\mu\notin \Lambda_4(2)$ and $\lambda\in \Lambda_4(2)$. One can check that $\tensor{C}{^\mu_\lambda}=1$. In this case $a_{\mu\backslash\lambda}=0$, and hence $\tensor{\overline{C}}{^{\,\mu}_\lambda}=0$: the irreducible module $L^\mu$ is absent from $M^\lambda_n$.
	\end{itemize}
	Consider the case $n=4$ with $\Dim=3$  (semisimple regime of $\bn$).
	\begin{itemize}
		\item[iii)] Take the integer partitions $\mu=(2,1,1)$ and $\lambda=(1,1)$. One has $\mu\notin \Lambda_4(3)$ and $\lambda\in \Lambda_4(3)$. One can check that $\tensor{C}{^\mu_\lambda}=1$ so $\tensor{\overline{C}}{^{\,\mu}_\lambda}=1$ and $L^\mu$ is present in $M^\lambda_n$. In this case $a_{\mu\backslash\lambda}=1$.
	\end{itemize}
	Consider the case $n=6$ with $\Dim=2$.
	\begin{itemize}
		\item[iv)] Take the integer partitions $\mu=(2,2,1,1)$ and $\lambda=(1,1)$. One has $\mu\notin \Lambda_6(2)$ and $\lambda\in \Lambda_6(2)$. One can check that $\tensor{C}{^\mu_\lambda}=1$. In this case $a_{\mu\backslash\lambda}=-2$ and hence $\tensor{\overline{C}}{^{\,\mu}_\lambda}=0$: the irreducible module $L^\mu$ is absent from $M^\lambda_n$.
	\end{itemize}
\end{example}

\smallskip

\paragraph{Spectrum of $A_n$.} We denote by $\mathrm{spec}_{\,\mu}(A_n)$ the set of positive eigenvalues of the operator $A_n$ on the irreducible $\C\sn$-module $L^{\mu}$ which occurs in the irreducible $\bn$-modules $M^{\lambda}_n$ ($\lambda\in \Lambda_n(\Dim)$). One has
\begin{equation}\label{eq:spec_An_mu}
	\mathrm{spec}_{\,\mu}(A_n)= \Big\{a_{\mu\backslash\lambda}\,\st \hspace{0.2cm} \lambda \in \Lambda_\mu(\Dim)\, \Big\}\,.
\end{equation}
Similarly, we denote by $\mathrm{spec}^{\lambda}(A_n)$ the set of positive eigenvalues of the operator $A_n$ on irreducible $\C\sn$-modules $L^{\mu}$ which occur in a fixed irreducible $\bn$-module $M^{\lambda}_n$. One has
\begin{equation}\label{eq:spec_An_lambda}
	\mathrm{spec}^{\,\lambda}(A_n)= \Big\{a_{\mu\backslash\lambda} \,\st \hspace{0.2cm} \mu \in \mathrm{M}_{n,\lambda}(\Dim)\, \Big\}\,.
\end{equation}
With the previous definitions, the set of all positive eigenvalues $\mathrm{spec}(A_n)$ of the operator $A_n$ on $V^{\otimes n}$ is:
\begin{equation}\label{eq:spec_An_full}
	\mathrm{spec}(A_n)=\bigcup_{\mu\,\in \,\mathcal{P}_n(\Dim)}\mathrm{spec}_{\,\mu}(A_n)=\bigcup_{\lambda\,\in\,\Lambda_n(\Dim)}\mathrm{spec}^{\,\lambda}(A_n)\,.
\end{equation}

Again, in order to avoid the computation of the generalized Littlewood-Richardson coefficients $\tensor{\overline{C}}{^{\,\mu}_\lambda_\nu}(\Dim)$ we introduce the extended spectrum of $A_n$, denoted $\overbar{\mathrm{spec}}(A_n)$. First, we denote by $\overbar{\mathrm{spec}}_{\mu}(A_n)$ and $\overbar{\mathrm{spec}}^{\lambda}(A_n)$ the following sets of positive integers: 
\begin{equation}\label{eq:spec_An_mu_lambda_mod}
	\overbar{\mathrm{spec}}_{\,\mu}(A_n)= \Big\{a_{\mu\backslash\lambda} \,\st \hspace{0.2cm} \lambda \in \overbar{\Lambda}_{\mu}(\Dim)\, \Big\}\,,\hspace{0.5cm} \overbar{\mathrm{spec}}^{\,\lambda}(A_n)= \Big\{a_{\mu\backslash\lambda} \,\st \hspace{0.2cm}  \mu \in \overbar{\mathrm{M}}_{n,\lambda}(\Dim)\, \Big\}\,.
\end{equation}
The extended spectrum $\overbar{\mathrm{spec}}(A_n)$ is defined as
\begin{equation}\label{eq:spec_An_full_mod}
	\overbar{\mathrm{spec}}(A_n):=\bigcup_{\mu\in\mathcal{P}_n(\Dim)}	\overbar{\mathrm{spec}}_{\,\mu}(A_n)=\bigcup_{\lambda\in\Lambda_n(\Dim)}\overbar{\mathrm{spec}}^{\,\lambda}(A_n)\,.
\end{equation}
As a direct consequence of Proposition \ref{prop:condition_A_branching_GL_O_2} one has the following lemma. 
\begin{lemma}\label{lem:lemma_spec} In the semisimple regime of $B_n$ ($\Dim\geqslant n-1$) one has
	\begin{equation}\label{eq:spec_semisimple}
		\mathrm{spec}^{\,\lambda}(A_n)=\,\overbar{\mathrm{spec}}^{\,\lambda}(A_n)\,, \hspace{0.5cm} \text{and} \hspace{0.5cm} \mathrm{spec}_{\,\mu}(A_n)=\,\overbar{\mathrm{spec}}_{\,\mu}(A_n)\,,
	\end{equation}
	while for $\Dim < n-1$ one has:
	\begin{equation}\label{eq:spec_not_semisimple}
		\mathrm{spec}^{\,\lambda}(A_n)\subseteq \,\overbar{\mathrm{spec}}^{\,\lambda}(A_n)\,, \hspace{0.5cm} \text{and} \hspace{0.5cm}	\mathrm{spec}_{\,\mu}(A_n)\subseteq \,\overbar{\mathrm{spec}}_{\,\mu}(A_n)\,.
	\end{equation}
\end{lemma} 

\newpage
\begin{remark}
The set of positive integers $\overbar{\mathrm{spec}}_{\,\mu}(A_n)$ and $\overbar{\mathrm{spec}}(A_n)$ correspond respectively to $\mathrm{spec}_{\,\mu}^{\stimes}(A_n)$ and $\mathrm{spec}^{\stimes}(A_n)$ in the notation of the article \cite{bulgakova2022construction}.
\end{remark}

\subsection{The formula and applications.}\label{subsec:traceless_projectors_app}
 
For any $\lambda\in \Lambda_n(\Dim)$ with $f_\lambda=0$ and $\Dim\geqslant1$ define the element 
\begin{equation}\label{eq:non_inductive_traceless_projectors}
	\mathcal{P}^{\lambda}_{n}=Z^\lambda \prod_{\begin{array}{c}
			{\scriptstyle \beta \in \overbar{\Lambda}_\lambda(\Dim)}\\
	\end{array}}\left(1 - \frac{A_n}{a_{\lambda\backslash\beta}}\,\right)\,.
\end{equation}
\begin{remark}
By definition of the set $\overbar{\Lambda}_\lambda(\Dim)$ the integers $a_{\lambda\backslash\beta}$ present in the factors of \eqref{eq:non_inductive_traceless_projectors} are positive and hence \eqref{eq:non_inductive_traceless_projectors} is well defined.
\end{remark}

\begin{theorem}[{\cite[Theo. 4.3]{bulgakova2022construction}}]\label{theo:non_inductive_traceless_projectors} For any $\lambda\vdash n$ such that $\lambda\in \Lambda_n(\Dim)$ and $\Dim\geqslant1$, the elements $\mathcal{P}^\lambda_n$ defined above are the traceless projectors which realize the projection of tensors to the isotypic components $(D^{\lambda})^{\oplus m_\lambda}$:
	\begin{equation}
	P^\lambda_n=\mathcal{P}^\lambda_n\,.
	\end{equation} 
\end{theorem}
\begin{proof}
Consider the decomposition of $V^{\otimes n}$ into a direct sum of irreducible
	$\bn$-modules $M^\rho_n$. We need to show that $\mathcal{P}^\lambda_n$ annihilates any irreducible module not isomorphic to $M^\lambda_n$ and acts by the identity on $M^\lambda_n$.\medskip  

Let $\rho$ be such that $f_\rho>0$, and take $v\in M^{\rho}_n$. Upon restriction of $ M^{\rho}_n$ to $\C\sn$, $v$ decomposes as 
	\begin{equation}
		v=\sum_{\mu\,\in \,\overbar{\mathrm{M}}_{n,\rho}(\Dim)}v^{\mu}\,, \hspace{0.5cm} \text{with} \hspace{0.5cm} v^{\mu}=v\cdot Z^{\mu}\,.
	\end{equation}
For $\mu\neq\lambda$ one has $v^\mu\cdot \mathcal{P}^{\lambda}_{n}=0$ due to the orthogonality property of the central Young idempotents. In the case $\mu=\lambda$,  $\mathcal{P}^{\lambda}_{n}$ will also annihilates $v^{\mu}$ due to the presence of the factor $1-\dfrac{A_n}{a_{\lambda\backslash\rho}}$ in the product. Hence, $v\cdot \mathcal{P}^{\lambda}_{n}=0$.\medskip 
	
	Let $\rho$ be such that $f_\rho=0$, and take $v\in M^{\rho}_n$. For any $\rho\neq \lambda$ one has $v\cdot Z^{\lambda}=0$, so $v\cdot \mathcal{P}^{\lambda}_{n}=0$. When $\rho=\lambda$, because $v\cdot A_n=0$ and $v\cdot Z^{\lambda}=v$, one has $v\cdot\mathcal{P}^{\lambda}_{n}=v$. Hence, $\mathcal{P}^\lambda_n$ acts by the identity on $M^\lambda_n$ and annihilates any irreducible module not isomorphic to $M^\lambda_n$.\medskip 
	
	
\end{proof}
\vskip 6pt

\begin{corollary}\label{cor:non_inductive_traceless_projectors}
The projection operators on traceless isotypic components parametrized by $\lambda$ with $f_\lambda=0$ may also be written as 
\begin{equation}
P^\lambda_n=Z^\lambda \prod_{\begin{array}{c}
		{\scriptstyle \alpha \, \in \,\overbar{\mathrm{spec}}_{\,\lambda}(A_n)}\\
\end{array}}\left(1 - \frac{1}{\alpha}A_n\,\right)\,.
\end{equation}
\end{corollary}


In the examples below we use the following notations. For a pair $\lambda\subset\mu $ define the {\it skew-shape Young diagram} $\mu\backslash \lambda$ as a set-theoretical difference of the corresponding Young diagrams. We set by definition $|\mu\backslash \lambda| = |\mu| - |\lambda|$. For example, 
\begin{equation}\label{eq:skew-shape_example}
	\text{given}\quad\mu=\Yboxdim{9pt}\yng(4,2,2,1) \quad\text{and}\quad \lambda=\Yboxdim{9pt}\yng(2,1)\,,\quad\text{one has}\quad \mu\backslash \lambda =\Yboxdim{9pt}\young(\times\times ~~,\times ~,~~,~)
\end{equation}
We denote by $\mu\slashdiv\nu$ the set of diagrams $\lambda$ such that $\tensor{C}{^{\mu}_{\lambda\nu}} \neq 0$ (see appendix \ref{subsec:Littlewood_Richardson_rules}). We also assume $\Dim\geqslant n-1$ so that $\tensor{\overline{C}}{^{\,\mu}_\lambda_\nu}(\Dim)=  \tensor{C}{^{\,\mu}_\lambda_\nu}$ and $\mathrm{spec}_{\,\mu}(A_n)=\,\overbar{\mathrm{spec}}_{\,\mu}(A_n)$ (see Lemma \ref{lem:lemma_spec}).

\begin{example}[Totally symmetric $O(\Dim,\C)$ irreducible projectors]\hphantom{ex}\medskip
	
For the fixed partition $\lambda =(n)$ one constructs $\mathrm{spec}_{\lambda}(A_n)$ for the skew-shape diagrams $\lambda\backslash\beta$ for all $\beta \in (n)\slashdiv (2f)$, $f = 1,\dots,\lfloor{\frac{n}{2}\rfloor}$. This leads to
\begin{displaymath}
	\ytableausetup{mathmode,boxsize=1.5em,centertableaux}
	\lambda\backslash\beta=
	\begin{ytableau}{\scriptstyle}
		\times & \none[\scriptstyle{\cdots}]& \times & {\scriptstyle f} &\none[\scriptstyle{\cdots}] & {\scriptstyle n-1}
	\end{ytableau}
	\quad \Rightarrow \quad \mathrm{spec}_{(n)}(A_n) = \left\{\big(\Dim + 2\ (n - f - 1)\big)\,f\;:\; f = 1,\dots, \lfloor \tfrac{n}{2}\rfloor\right\}\,.
\end{displaymath}
Hence the projector \eqref{eq:non_inductive_traceless_projectors} takes the form
\begin{equation}\label{eq:projector_symmetric0}
	P_n^{(n)} = Z^{(n)}\,\displaystyle{\prod_{f=1}^{\lfloor \tfrac{n}{2}\rfloor}} \left(1 - \frac{A_n}{\big(\Dim+2\ (n-f-1)\big)f}\right)\,.
\end{equation}
\end{example}
\vskip 4 pt
\begin{example}[Maximally-antisymmetric hook $O(\Dim,\C)$ isotypic projectors]\hphantom{ex}\medskip

For the partition $\lambda=(2,1^{n-2})$ one constructs $\mathrm{spec}_{\lambda}(A_n)$ for the skew shape diagrams $\lambda\backslash\beta$ with $\beta \in (2,1^{n-2})\slashdiv (2)$, which leads to the only possibility: \\
	\begin{equation}
		\ytableausetup{mathmode,boxsize=1.5em,centertableaux}
		\lambda\backslash\beta=
		\begin{ytableau}{\scriptstyle}
			\times & {\scriptstyle 1} \\
			\times\\
			\none[\svdots]\\
			\times\\
			{\scriptstyle 2-n}\\
		\end{ytableau} 
		\hspace{1cm} \Rightarrow \hspace{1cm} \mathrm{spec}_{(2,1^{n-2})}(A_n) = \big\{\Dim - n + 2 \, \big\}\,,
	\end{equation}
so
\begin{equation}
P_n^{\lambda}=Z^{\lambda}\left(1- \dfrac{A_n}{\Dim-n+2}\right)
\end{equation}
Note that in order for the $\GL(\Dim,\C)$-module in question to be present in $V^{\otimes n}$, one assumes $\Dim \geqslant n-1$, so the denominator is non-singular. 
\end{example}
\newpage
\begin{example}[Arbitrary hook $O(\Dim,\C)$ isotypic projectors ]\hphantom{aa}\medskip
For the partition $\lambda=(m,1^{n-m})$ with $m\geqslant 2$ one constructs $\mathrm{spec}_{\lambda}(A_n)$ by reconstructing the skew shape diagrams $\lambda\backslash\beta$ \lp with $|\lambda\backslash\beta| = 2f$\rp using the reverse jeu de taquin. The only starting tableaux are
\begin{displaymath}
	\ytableausetup{mathmode,boxsize=0.9em,centertableaux}
	\begin{ytableau}{\scriptstyle}
		\one &\one & \none[\scriptstyle{\cdots}]& \one & \times &\none[\scriptstyle{\cdots}]& \times \\
		\times\\
		\none[\svdots]\\
		\times\\
		\times\\
	\end{ytableau}
\hspace{1cm} \text{\lp with $2f$ entries `1', for all $f = 1,\dots, \lfloor \tfrac{m}{2}\rfloor$\rp,}
\end{displaymath}
which leads to the following possibilities:
\begin{displaymath}
	\ytableausetup{mathmode,boxsize=0.9em,centertableaux}
	\lambda\backslash\beta=
	\begin{ytableau}{\scriptstyle}
		\times &\times & \none[\scriptstyle{\cdots}]& \times & &\none[\scriptstyle{\cdots}]&\\
		\times\\
		\none[\svdots]\\
		\times\\
		\\
	\end{ytableau}
\hspace{1cm}\text{\lp for all $f = 1,\dots, \lfloor \tfrac{m}{2}\rfloor$\rp}
\end{displaymath}
and,
\begin{displaymath}
	\ytableausetup{mathmode,boxsize=0.9em,centertableaux}	
	\lambda\backslash\beta=\begin{ytableau}{\scriptstyle}
	\times & \none[\scriptstyle{\cdots}]& \times & &\none[\scriptstyle{\cdots}]&\\
	\times\\
	\times\\
	\none[\svdots]\\
	\times\\
\end{ytableau} 
\hspace{1cm}\text{\lp for all $f = 1,\dots, \lfloor \tfrac{m-1}{2}\rfloor$\rp}
\end{displaymath}
so,
\begin{equation}
\mathrm{spec}_{(\lambda)}(A_n) =\left\{f\,\big(\Dim+2\ \left(m-f\right)\big)-n  \right\} \  \cup \  \left\{f\,\big(\Dim+2\ \left(m-1-f\right)\big) \right\}\,.
\end{equation}
Then, the projector \eqref{eq:non_inductive_traceless_projectors}, is given by 
\begin{equation}\label{eq:projector_hook1}
	P_n^{\lambda} = Z^{\lambda}\,\displaystyle{\prod_{f=1}^{\lfloor \tfrac{m}{2}\rfloor}} \left(1 - \frac{A_n}{f(\Dim+2\ (m-f))-n}\,\right)\displaystyle{\prod_{f=1}^{\lfloor \tfrac{m-1}{2}\rfloor}} \left(1 - \frac{A_n}{f(\Dim+2\ (m-1-f))}\,\right)\,.
\end{equation}
\end{example}
\begin{remark}\label{rem:above_central}
These examples are specific cases where the Littlewood-Richardson rule (see appendix \ref{subsec:Littlewood_Richardson_rules}) simplifies in such a way that a more explicit formulation of \eqref{eq:non_inductive_traceless_projectors} can be obtained. In particular, in these restricted cases the product runs over the integers $f$ which implicitly parameterizes the irreducible modules $M^\lambda_n$ as $|\lambda|=n-2f$.
\end{remark}

\begin{mathematica}[\textit{BrauerAlgebra}]
	CentralIdempotent[\,$n$\,,\,$\lambda$\,,\,$\Dim$\,]\,\\
	\textit{Returns the element in $\Bn{n}(\Dim)$ realizing the $\lambda$-isotypic projection of tensors with respect to the action of $\Or(\Dim,\C)$ on $V^{\otimes}$. When $\Bn{n}(\Dim)$ is semisimple this element is central.}
\end{mathematica}

Below we give two examples of the projection operators obtained with the Mathematica package~\textit{BrauerAlgebra}. The expressions are given in their expanded form in the conjugacy class sum basis of the centralizer $\cn$ of $\sn$ in $B_n$. More precisely, the elements $K_{\xi}$ which appear in the equations below are conjugacy class sums of $\bn$ where the list $\xi$ are the ternary bracelets which parameterize the conjugacy classes of $\bn$ (see Theorem \ref{theo:classes-bracelets}). More details are given in chapter \ref{chap:cn}.
\begin{example}[Isotypic traceless projection operators for $n=3$ and $n=4$.]
\begin{itemize}
\item[\it i)] $n=3$:
\begin{equation}
\begin{aligned}
P^{(3)}_3&=\raisebox{-0.43\height}{\includegraphics[scale=1.]{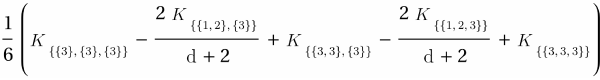}}\\
P^{(2,1)}_3&=\raisebox{-0.43\height}{\includegraphics[scale=0.8]{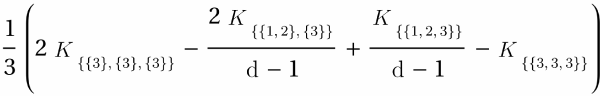}}\\
Z^{(1^3)}=P^{(1^3)}_3&=\raisebox{-0.3\height}{\includegraphics[scale=0.55]{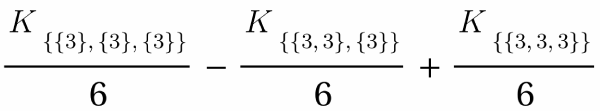}}
\end{aligned}
\end{equation}
\item[\it ii)] $n=4$:

\begin{equation}\label{eq:proj_4TL}
\hspace{-1.5cm}	\begin{aligned}
P^{(4)}_4&=\raisebox{-0.74\height}{\includegraphics[scale=1.2]{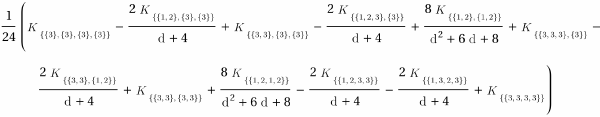}}\\
P^{(3,1)}_4&=\raisebox{-0.74\height}{\includegraphics[scale=1.2]{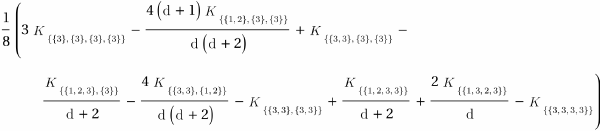}}\\
P^{(2,2)}_4&=\raisebox{-0.74\height}{\includegraphics[scale=1.2]{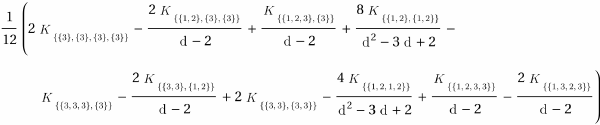}}\\
P^{(2,1^2)}_4&=\raisebox{-0.74\height}{\includegraphics[scale=1.]{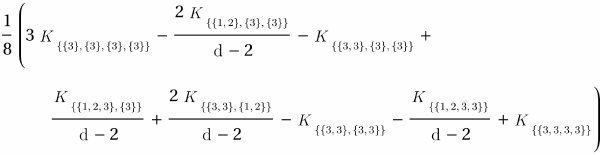}}\\
Z^{(1^4)}=P^{(1^4)}_4&=\raisebox{-0.43\height}{\includegraphics[scale=1.2]{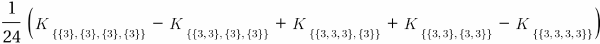}}
	\end{aligned}
\end{equation}
\end{itemize}
\end{example}

\begin{example}[More traceless projection operators obtained with Mathematica]
\begin{equation}
	P^{(3,2)}_5=\raisebox{-0.87\height}{\includegraphics[scale=1.4]{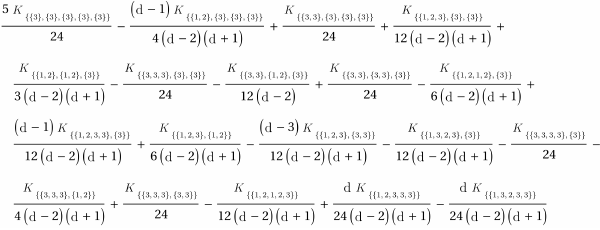}}\,
\end{equation}
\begin{equation}
	P^{(2,2,2)}_6=\raisebox{-0.92\height}{\includegraphics[scale=1.4]{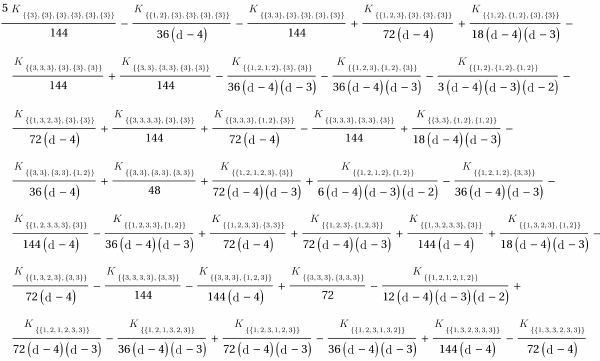}}\,
\end{equation}
\end{example}

\subsection{The traceless projector $P_n$}\label{subsec:universal_traceless_projectors}
Restricting our attention on the set of $\Or(\Dim,\C)$-modules $D^\lambda$ such that $\lambda \in\Lambda_n(\Dim)$ and $f_\lambda=0$ (traceless tensors), we construct the following element in $\bn$:
\begin{equation}\label{eq:factorized_traceless_expanded}
P_n=\sum_{\lambda\in\Lambda_n(\Dim) } P^\lambda_n.\,
\end{equation}
The $\mathfrak{r}$-image of $P_n$, $\mathfrak{r}(P_n)\in \End(V^{\otimes n})$ is the universal traceless projector presented in \cite{bulgakova2022construction} (see \cite[Cor. 4.5]{bulgakova2022construction}). The universal traceless projection operator also admits the following form \cite[Theo. 4.1]{bulgakova2022construction}
\begin{equation}\label{eq:factorized_traceless_projector}
	P_n=\prod_{\alpha\in\,\,\overbar{\mathrm{spec}}(A_n) } \left(1-\frac{1}{\alpha} A_n\right)\,.
\end{equation}
In the following two examples we assume $\Dim\geqslant n-1$ so that $\tensor{\overline{C}}{^{\,\mu}_\lambda_\nu}(\Dim)=  \tensor{C}{^{\,\mu}_\lambda_\nu}$ and $\mathrm{spec}_{\,\mu}(A_n)=\,\overbar{\mathrm{spec}}_{\,\mu}(A_n)$ (see Lemma \ref{lem:lemma_spec}).

\begin{example}[\textbf{The traceless projector for $n=3$}]\label{ex:traceless_proj_n3} 
To arrange the set of eigenvalues let us make use of the fact that 
\begin{equation}
	\def\arraystretch{0.7}
	\mathrm{spec}(A_3) = \mathrm{spec}_{(3)}(A_n) \cup \mathrm{spec}_{(2,1)}(A_n)\,,
\end{equation}
so one uses \eqref{eq:spec_An_full} defined in section \ref{subsec:traceless_projectors_app}. Applying the Littlewood-Richardson rule (see appendix \ref{subsec:Littlewood_Richardson_rules})  one finds:
\begin{equation}\label{eigenvalues_A3}
		\mathrm{spec}_{(3)}(A_3)=\lbrace \Dim+2\rbrace, \qquad \mathrm{spec}_{(2,1)}(A_3) =\lbrace \,\Dim-1 \rbrace\,.
\end{equation}
\noindent The operator \eqref{eq:factorized_traceless_projector} takes the form
\begin{align}
	P_3 &=\left(1-\frac{A_3}{\Dim-1}\,\right)\left(1-\frac{A_3}{\Dim+2}\,\right) \label{eq:factorized_traceless_projector_Br_3}\\
	&=1-\frac{2\Dim+1}{(\Dim-1)(\Dim+2)}\,A_3+\frac{1}{(\Dim-1)(\Dim+2)}\,(A_3)^2\,.
\end{align}
Upon expanding the previous expression in the diagrammatic basis of $\bn$ one finds:
\begin{equation}\label{eq:projector_Br_3}
	\begin{aligned}
P_3=&\Scale[0.8]{\raisebox{-.4\height}{\includegraphics[scale=0.6]{fig/id3.pdf}}}
		\,-\,\frac{\Dim+1}{(\Dim-1)(\Dim+2)}\,\Bigl(\  
			\Scale[0.8]{\raisebox{-.4\height}{\includegraphics[scale=0.6]{fig/d3a.pdf}}+\raisebox{-.4\height}{\includegraphics[scale=0.6]{fig/d3b.pdf}}+\raisebox{-.4\height}{\includegraphics[scale=0.6]{fig/d3c.pdf}}\  }
			\Bigr)\,\\[10pt]
			&+  \frac{1}{(\Dim-1)(\Dim+2)}\, \ \Bigl(
			\Scale[0.8]{ \raisebox{-.4\height}{\includegraphics[scale=0.6]{fig/d3d.pdf}}+\raisebox{-.4\height}{\includegraphics[scale=0.6]{fig/d3e.pdf}}+\raisebox{-.4\height}{\includegraphics[scale=0.6]{fig/d3f.pdf}}+\raisebox{-.4\height}{\includegraphics[scale=0.6]{fig/d3g.pdf}}+\raisebox{-.4\height}{\includegraphics[scale=0.6]{fig/d3h.pdf}}+\raisebox{-.4\height}{\includegraphics[scale=0.6]{fig/d3i.pdf}}}
			\Bigr)\,\,.
	\end{aligned}
	\end{equation}
The latter expansion can be very tedious and time consuming (even for a computer) due to the large number of products of diagrams involved, specially for $n>3$. A technique which allows one to bypass diagrammatic computations is presented in chapter \ref{chap:cn}. In particular, see section \ref{sec:applications} for the application of the technique.
\end{example}
\begin{example}[\textbf{The traceless projector for $n=4$.}]\label{ex:traceless_proj_n4} 
For the eigenvalues of $A_4$ we have:
\begin{equation}
	\def\arraystretch{0.7}
	\mathrm{spec}(A_4) = \mathrm{spec}_{(4)}(A_n) \cup \mathrm{spec}_{(3,1)}(A_n) \cup \mathrm{spec}_{(2,2)}(A_n) \cup \mathrm{spec}_{(2,1,1)}(A_n)\,,
\end{equation}
Applying the Littlewood-Richardson rule (appendix \ref{subsec:Littlewood_Richardson_rules})  one finds:
\begin{equation}\label{eigenvalues_A4}
	\begin{aligned}
		\mathrm{spec}_{(4)}(A_4) &=\lbrace{\Dim+4\, ,\, 2(\Dim+2)}\rbrace, \qquad \mathrm{spec}_{(3,1)}(A_4) =\lbrace{\Dim\,,\,\Dim+2}\rbrace \qquad\\
		\mathrm{spec}_{(2,2)}(A_4) &=\lbrace{\Dim-2\,,\, 2(\Dim-1)}\rbrace, \qquad 
		\mathrm{spec}_{(2,1,1)}(A_4) =\lbrace{\Dim-2}\rbrace\,,
	\end{aligned}
\end{equation}
so that 
\begin{equation}
	\begin{aligned}
		\mathrm{spec}(A_4)=\lbrace\, \Dim+4\,,2(\Dim+2) \,, \Dim\,,\,\Dim+2\,,\Dim-2\,, 2(\Dim-1)\rbrace\,.
	\end{aligned}
\end{equation}

A this level one notes that some elements in the above sets vanish for particular non-zero integer values $\Dim\in \{1,2\}$, which mark exactly the non-semisimple regime of $B_{n}(\Dim)$ \cite{Rui_Br_semisimple}. Also for $\Dim = 1$, $\mathrm{spec}_{(2,2)}(A_4)$ and $\mathrm{spec}_{(2,1,1)}(A_4)$ contain a negative element to be excluded from the consideration according to Lemma \ref{lem:diag_ev_A_2}. Nevertheless we proceed by keeping $\Dim$ generic, such that all the above elements enter $\mathrm{spec}(A_4)$.

\noindent The operator \eqref{eq:factorized_traceless_projector} takes the following factorized form
\begin{equation}\label{eq:factorized_traceless_projector_Br_4}
	\hspace{-0.5cm}	P_4=\scalemath{0.91}{\left(1-\frac{A_4}{\Dim+4}\,\right)\left(1-\frac{A_4}{2(\Dim+2)}\,\right)\left(1-\frac{A_4}{\Dim}\,\right)\left(1-\frac{A_4}{\Dim+2}\,\right) \left(1-\frac{A_4}{\Dim-2}\,\right)\left(1-\frac{A_4}{2(\Dim-1)}\,\right)\,.}
\end{equation}
Here we do not present the expanded diagrammatic form of the projector as it is too cumbersome. In section \ref{sec:applications} we give the expression of the expanded form in the basis of $\mathcal{C}_4$ (the centralizer of $\Sn{4}$ in $\Bn{4})$).  
\end{example}
\begin{mathematica}[\textit{BrauerAlgebra}]
	CentralIdempotent[\,$n$\,,\,0\,,\,$\Dim$\,]\,\\
	\textit{Returns the element in $\Bn{n}(\Dim)$ realizing the 1-traceless projection of tensors. When $\Bn{n}(\Dim)$ is semisimple this element is central and idempotent.}
\end{mathematica}

\section{The (f+1)-traceless and full-trace projectors $\lbrace P^{\lambda}_n \,, \lambda\vdash n-2f\rbrace$.}\label{sec:central_h_traceless}

\subsection{The averaged arc induction map $\mathcal{A}$}\label{subsection:average_arc}
The averaged arc induction linear map $\mathcal{A}\,:\,B_{n-2}\to B_n$ is defined on a single diagram $b\in B_{n-2}$ as the sum of all possible diagrams obtained from $b$ by the insertion of a pair of arcs, symmetric with respect to the horizontal middle line. For example, 
\begin{equation}
	\begin{aligned}
		&\mathcal{A}(\,\,\raisebox{-.4\height}{\includegraphics[scale=0.2]{fig/id2.pdf}}\,\,)=\,\raisebox{-.4\height}{\includegraphics[scale=0.55]{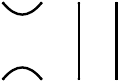}}+\raisebox{-.4\height}{\includegraphics[scale=0.55]{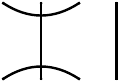}}+\raisebox{-.4\height}{\includegraphics[scale=0.55]{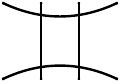}}+\raisebox{-.4\height}{\includegraphics[scale=0.55]{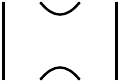}}+\raisebox{-.4\height}{\includegraphics[scale=0.55]{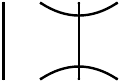}}+\raisebox{-.4\height}{\includegraphics[scale=0.55]{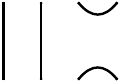}}\,, \\[6pt]
		&\mathcal{A}(\,\,\raisebox{-.4\height}{\includegraphics[scale=0.2]{fig/s12.pdf}}\,\,)=\,\raisebox{-.4\height}{\includegraphics[scale=0.55]{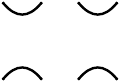}}+\raisebox{-.4\height}{\includegraphics[scale=0.55]{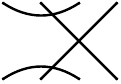}}+\raisebox{-.4\height}{\includegraphics[scale=0.55]{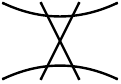}}+\raisebox{-.4\height}{\includegraphics[scale=0.55]{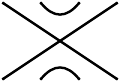}}+\raisebox{-.4\height}{\includegraphics[scale=0.55]{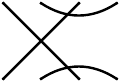}}+\,\raisebox{-.4\height}{\includegraphics[scale=0.55]{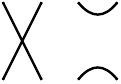}}\,,\\[6pt]
		&\mathcal{A}(\,\,\raisebox{-.4\height}{\includegraphics[scale=0.2]{fig/b2.pdf}}\,\,)=\,2\,\raisebox{-.4\height}{\includegraphics[scale=0.55]{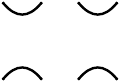}}+2\,\raisebox{-.4\height}{\includegraphics[scale=0.55]{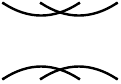}}+2\,\raisebox{-.4\height}{\includegraphics[scale=0.55]{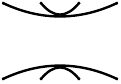}}\,.
	\end{aligned}
\end{equation}
For any pair of integer $(i,j)$ such that $1\leqslant i<j \leqslant n$, we define the linear map $\mathfrak{a}_{ij}$ which adds a pair of symmetric arcs to any diagram $b\in B_{n-2}\subset B_n$: one of the arc joins the vertices $i$ and $j$ of the upper row of $b$, while the second arc is symmetric to the first arc with respect to the middle horizontal line of $b$. In terms of $\mathfrak{a}_{ij}$ the action of the averaged arc induction map $\mathcal{A}$ is given by 
\begin{equation}\label{eq:def_A_aij}
	\begin{aligned}
		\mathcal{A}\,:\, B_{n-2}&\to B_n \\
		b\,&\mapsto\sum_{1\leqslant i<j \leqslant n} \mathfrak{a}_{ij} (b).    
	\end{aligned}
\end{equation}
We denote by $\cn$ the centralizer of $\sn$ in $\bn$.
\begin{remark}\label{remark_arc_induction}\hphantom{h}
\begin{itemize}
	\item[\it{i)}] $\mathcal{A}$ sends the identity in $B_{n-2}$ to $A_n$.
	\item[\it{ii)}] $\mathcal{A}$ maps elements of the centralizer $\Cn{n-2}$ to elements of $\cn$.
	\item[\it{iii)}] When the particular expression of $x\in B_{n-2}$ is irrelevant for the discussion we may represent it as a rectangle with $n-2$ vertices
	\begin{equation}\label{eq:Rect_pic}
		x\mapsto \raisebox{-.35\height}{\includegraphics[scale=0.45]{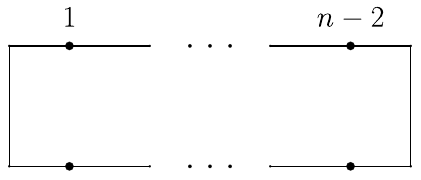}} \,,
	\end{equation}
	and the action of $\mathcal{A}$ on $x$ can be pictured as  
	\begin{equation}\label{eq:AP}
		\mathcal{A}(x)\mapsto \sum_{1\leqslant i<j \leqslant n}\raisebox{-.35\height}{\includegraphics[scale=0.7]{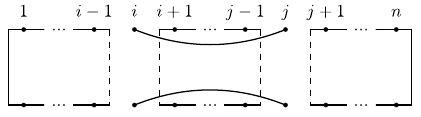}} .
	\end{equation}
\end{itemize}
\end{remark}


\subsection{Action of $\mathcal{A}(P^\lambda_{n-2})$ on irreducible tensors}\label{subsection:action_A_on_irreducible_tensors}
\paragraph{Diagrammatic presentation of irreducible tensors.} 

Let $T^{\beta}_n\in D^{\beta}\subset V^{\otimes n}$ denote an irreducible tensor with respect to $\Or(\Dim,\C)$. As already mentioned in section \ref{subsection:Schur_Weyl_O} of chapter \ref{chap:Irreducible_MAG} such a tensor belongs \textit{schematically} to the space $g^{\otimes f} \otimes D^{\beta}$ where $f=\frac{n-|\beta|}{2}$ and $D^{\beta}$ is the space of traceless tensors. 
Hence, still \textit{schematically}, $T^{\beta}_n$ as the following form
\begin{equation}
	g^{\otimes f} \otimes T^{(\beta)}_{n-2f}\,
\end{equation}
where $T^{(\beta)}_{n-2f}\in D^{\beta}\subset V^{\beta} $ is a traceless irreducible tensor of order $n-2f$. Moreover one has
\begin{equation}
	T^{(\beta)}_{n-2f}\cdot Z^{\lambda}=\delta{_{\beta\lambda}}\, T^{(\beta)}_{n-2f}\,,
\end{equation} 
where $\delta$ is the Kronecker delta symbol and $Z^{\lambda}$ is a Young central idempotent. In more details, the components of $T^{\beta}_n$ are sums of expressions with $f$ product of the metric times the components of a traceless tensor of order $|\beta|$, and $T^{\beta}_n$ is a $(f+1)$-traceless tensor. With respect to the action of Brauer diagrams on irreducible tensors $T^{\beta}_n$, the latter can be illustrated by a sum of one-row diagrams on $n$ nodes (each symbolising a factor $V$ in $V^{\otimes n}$), with distinguished $f$ pairs of nodes joined by an arc (symbolising insertions of the metric $g$), and the remaining $n-2f$ nodes representing a traceless tensor in $D^{\beta}$. For example, for $|\beta|=n-2$, that is $f=1$, a general irreducible 2-traceless tensor $T^{\beta}_{n}$ is represented as
\begin{equation}
	T^{\beta}_{n}= \sum_{1\leqslant k <l \leqslant n}  \raisebox{-.5\height}{\includegraphics[scale=0.7]{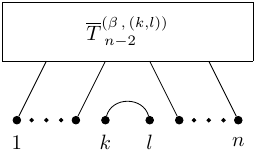}}\,.
\end{equation}
Recall the structure the irreducible decomposition \eqref{theo:irreducible_Riemann_O} of the Riemann tensor presented in chapter \ref{chap:Irreducible_MAG}.\medskip

Then Brauer diagrams act on the \textit{irreducible tensor diagrams} by identifying the upper nodes of the former and the nodes of the latter, followed by straightening the lines and omitting closed loops. The resulting irreducible tensor diagram is multiplied by $\Dim^{\ell}$ if $\ell$ closed loops were omitted \phantomsection\cite{Brauer,Brown_Br,HW_discriminants_Br_algebra,Hartmann_Paget_modules_Br}. In particular the action of $\mathcal{A}(P^\lambda_{n-2})$ on $T^{\beta}_{n}$ can be represented as follow :\medskip

\begin{equation}
	T^{\beta}_{n}\cdot \mathcal{A}(P^\lambda_{n-2})= \sum_{\begin{array}{c}
			{\scriptstyle 1\leqslant k <l \leqslant n}\\
			{\scriptstyle 1\leqslant i <j \leqslant n}
	\end{array}} \,\,\, \raisebox{-.5\height}{\includegraphics[scale=0.7]{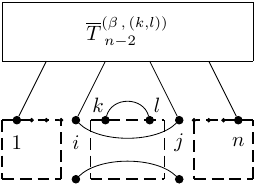}}\,.
\end{equation}

\begin{remark}
When translated to the action of $\bn$ on itself, that is to the right regular $\bn$-module, this diagrammatic construction of the action of $\bn$ on irreducible tensors gives rise to the diagrammatic basis in the standard $\bn$-modules $\Delta^{\lambda}_n$ (see \cite{HW_discriminants_Br_algebra,Hartmann_Paget_modules_Br} and \cite{Graham_Lehrer_cellular} in relation to cell modules). 
\end{remark}

The elements $P^\lambda_{n-2}$ belong to $\Cn{n-2}$ while $\mathcal{A}(P^\lambda_{n-2})$ belong to $\cn$. Hence, in order to analyze the action of  $\mathcal{A}(P^\lambda_{n-2})$ on $D^\beta$ it is sufficient to focus attention on the following particular element of $D^\beta$:
\begin{equation}\label{eq:std_module_farcs}
	T^{(\beta,f)}= \raisebox{-.5\height}{\includegraphics[scale=0.8]{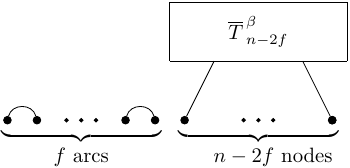}}\,,
\end{equation}
where $\overline{T}^\beta_{n-2f}$ is a traceless tensor. From the definition of the product of Brauer diagrams \eqref{eq:multiplication_Bn} it follows that the number of arcs in a diagram can not be reduced via multiplication by other diagrams. We denote by $J_{f} \subset B_{n}$ the span of all diagrams with at least $f$ arcs.\footnote{By saying that a diagram has $f$ arcs we mean that it has $f$ arcs in its upper (or equivalently, lower) row.}\medskip

The action of $\mathcal{A}(P^{\lambda}_{n-2})$ on $T^{(\beta,f)}$ is represented as follows
\begin{equation}
	T^{(\beta,f)}\cdot \mathcal{A}(P^{\lambda}_{n-2})=\sum_{1\leqslant i<j \leqslant n}\,\,\raisebox{-.5\height}{\includegraphics[scale=0.7]{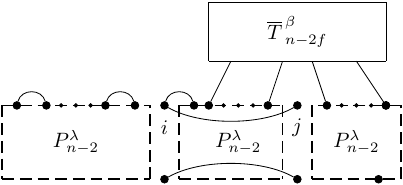}}\,.
\end{equation}
Let $f_\lambda=\frac{n-|\lambda|}{2}$, such that $P^{\lambda}_{n-2}$ is an $f_\lambda$-traceless operator annihilated by $J_{f_\lambda}$ and $\mathcal{A}(P^{\lambda}_{n-2})$ is annihilated by $J_{f_\lambda+1}$.\medskip
\begin{flushleft}
	\begin{tabular}{ll}
		\textbf{Fact 0:} & \hspace{2cm} If $\,f\neq f_\lambda\,$  then $\, T^{(\beta,f)}\cdot\mathcal{A}(P^{\lambda}_{n-2})=0$.
	\end{tabular}
\end{flushleft}
Indeed, recall that $T^{(\beta,f)}$ is an $(f+1)$-traceless tensor.  If  $f<f_\lambda\,$, then $T^{(\beta,f)}\cdot\mathcal{A}(P^{\lambda}_{n-2})=0$, because  $\mathcal{A}(P^{\lambda}_{n-2})\in J_{f_\lambda}$ and $T^{(\beta,f)}$ is annihilated by $J_{f+1}$. If $f>f_\lambda$, then $T^{(\beta,f)}\cdot\mathcal{A}(P^{\lambda}_{n-2})=0$, because $\mathcal{A}(P^{\lambda}_{n-2})$ is annihilated by $J_{f_\lambda+1}$.\medskip


To prepare for the proof of the next Lemma, let us demonstrate:\medskip
\begin{flushleft}
\begin{tabular}{ll}
\textbf{Fact 1:} & \hspace{2cm} If $\,|\beta|=|\lambda|\,$ and $\,\beta\neq \lambda\,$ then $\,T^{(\beta,1)}\cdot \mathcal{A}(P^\lambda_{n-2})=0$.
\end{tabular}
\end{flushleft}

\noindent The next lemma is a mere generalization of \textbf{Fact 1} to arbitrary $f\geqslant2$. 

%
%
%

\begin{mdframed}[style=mystyle,frametitle=Proof of Fact 1.]\label{ex:preLemma_Arc_induction_1}
	\begin{equation}
		T^{(\beta,1)}\cdot \mathcal{A}(P^\lambda_{n-2})=\sum_{1\leqslant i <j \leqslant n} \,\,\, \raisebox{-.42\height}{\includegraphics[scale=0.6]{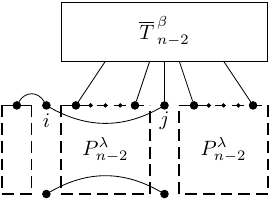}}
	\end{equation}
	There are three possible configurations for the position of the arc $(ij)$:
	\begin{itemize}
		\item[$1)$] The arc $(ij)$ is  directly connected to $\overline{T}^\beta_{n-2}$:
		\begin{equation}
			T^{(\beta,1)}\cdot \mathfrak{a}_{ij}(P^\lambda_{n-2})=\raisebox{-.38\height}{\includegraphics[scale=0.6]{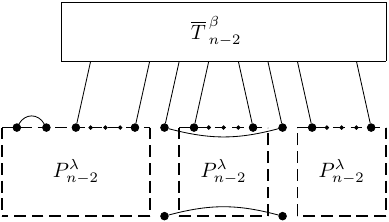}}\,.
		\end{equation}
		The tensor $\overline{T}^\beta_{n-2}$ is traceless hence $T^{(\beta,1)}\cdot \mathfrak{a}_{ij}(P^\lambda_{n-2})=0$. 
		\item[$2)$] The arc $(ij)$ is connected to the only avaible  arc of $T^{(\beta,1)}$ (which means $i=1$, $j=2$): 
		\begin{equation}
			T^{(\beta,1)}\cdot \mathfrak{a}_{ij}(P^\lambda_{n-2})=\raisebox{-.43\height}{\includegraphics[scale=0.6]{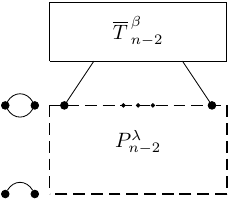}}\,.
		\end{equation}
		All the nodes of $P^\lambda_{n-2}$ are directly connected to $\overline{T}^\beta_{n-2}$. Because $\overline{T}^\beta_{n-2}\in D^{\beta}$ and by construction
		$D^{\beta}\cdot P^\lambda_{n-2}=0$, we have $T^{(\beta,1)}\cdot \mathfrak{a}_{ij}(P^\lambda_{n-2})=0$.
		\item[$3)$] One node of the arc $(ij)$ is connected to a node of the arc of $T^{(\beta,1)}$ and the other one is connected to $\overline{T}^\beta_{n-2}$ (which means $i=1$ or $i=2$, and  $j>2$):
		\begin{equation}\label{eq:f_1_action_AP}
			T^{(\beta,1)}\cdot \mathfrak{a}_{ij}(P^\lambda_{n-2})=\raisebox{-.43\height}{\includegraphics[scale=0.8]{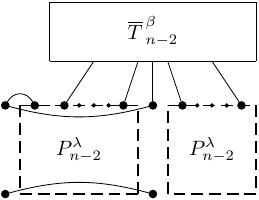}}
		\end{equation}
		To show that \eqref{eq:f_1_action_AP} equals zero, let us analyze the structure of $P^{\lambda}_{n-2}$ in more detail. First note that all diagrams entering $P^{\lambda}_{n-2}$ which are in $J_1$ annihilate $\overline{T}^\beta_{n-2}$. Indeed the upper arc(s) of such diagrams will either be connected directly to $\overline{T}^\beta_{n-2}$ \textup{(}which yield zero contribution according to first case considered above\textup{)}, or connected to a node of $\overline{T}^\beta_{n-2}$ and to the only remaining arc node of $T^{(\beta,1)}$. The latter configuration result in an arc contraction of $\overline{T}^\beta_{n-2}$ as depicted in the figure below
		\begin{equation}\label{eq:arc_contraction_zero_Fact1}
			\raisebox{-.42\height}{\includegraphics[scale=0.7]{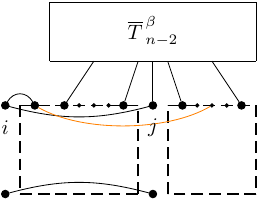}}\,.
		\end{equation}
		Hence it is enough to focus on the diagrams within $P^{\lambda}_{n-2}$ which contain no arcs, $i.e.$ permutation diagrams. In particular, note that
		\begin{equation}
			P^{\lambda}_{n-2}=Z^\lambda+\ldots
		\end{equation}
		where ellipses denote elements in $J_1$.\medskip
		
		Hence one has 
		\begin{equation}
			T^{(\beta,1)}\cdot \mathfrak{a}_{ij}(P^\lambda_{n-2})=T^{(\beta,1)}\cdot \mathfrak{a}_{ij}(Z^\lambda)=	\raisebox{-.48\height}{\includegraphics[scale=0.7]{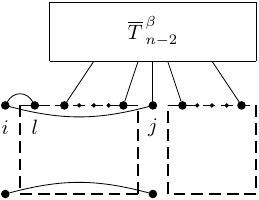}}
		\end{equation}
		The previous configuration can be cast into the following equivalent forms
		\begin{equation}
			\raisebox{-.48\height}{\includegraphics[scale=0.65]{fig/case_2_Arc_indcution_proof_2_traceless_l_2floor.pdf}}=\frac{1}{\Dim}\hspace{0.3cm}	\raisebox{-.48\height}{\includegraphics[scale=0.65]{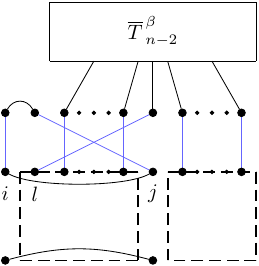}}=\frac{1}{\Dim}\hspace{0.3cm}\raisebox{-.48\height}{\includegraphics[scale=0.65]{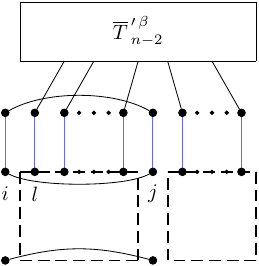}}\,
		\end{equation}
		The above sequence of diagrams translates to
		\begin{equation}
			T^{(\beta,1)}\cdot \mathfrak{a}_{ij}(Z^{\lambda})=\dfrac{1}{\Dim}\,T^{(\beta,1)}\cdot \left(s_{lj}\, \mathfrak{a}_{ij}(Z^{\lambda}) \right)=\dfrac{1}{\Dim}\,T^{\,\prime\,(\beta,1)}\cdot \mathfrak{a}_{ij}(Z^{\lambda})\,,
		\end{equation}
		where 
		\begin{equation}
			T^{\,\prime\,(\beta,1)}=T^{(\beta,1)}\cdot s_{lj}=\raisebox{-.5\height}{\includegraphics[scale=0.7]{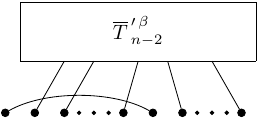}}\,\,\,\in D^\beta\,.
		\end{equation}
		\vskip 6pt
		
		The central Young projector $Z^{\lambda}$ acts directly on $\overline{T}^{\,\prime\,\beta}_{n-2}$, which gives zero.\medskip
		
		Finally we conclude that 
		\begin{equation}
			T^{(\beta,1)}\cdot\mathcal{A}(P^{\lambda}_{n-2})=0\,.
		\end{equation}
	\end{itemize}
\end{mdframed}
\hphantom{hh}\medskip
The following important lemma generalizes the previous fact. 
\begin{lemma}\label{lem:arcinduction_tens} Let $\lambda\,,\beta \in \Lambda_n(\Dim)$, and let $T\in D^\beta \subset V^{\otimes n}$,
\begin{equation}
\text{If} \hspace{1cm} \beta\neq \lambda \hspace{1cm} \text{then} \hspace{1cm} T^{(\beta,f)}\cdot\mathcal{A}(P^{\lambda}_{n-2})=0\,.
\end{equation}
In other words, $\mathcal{A}(P^{\lambda}_{n-2})$ preserves the isotypic component $(D^\lambda)^{\oplus m_\lambda}$.
\end{lemma}
The proof is given in the appendix \ref{subsec:proof_Lemma_arc_induc_tens}.
\begin{proposition}\label{prop:arc_induction} For any $\lambda\, \in \Lambda_n(\Dim)$ and $T\in V^{\otimes n}$
	\begin{equation}
		\begin{aligned}
T\cdot \mathcal{A}(P^{\lambda}_{n-2})&=T\cdot\left( A_n \,P^{\lambda}_n \right)\,,\\
&=T\cdot\left(\,P^{\lambda}_n \, A_n  \right)\,.
		\end{aligned}
	\end{equation}
\end{proposition}	
\vskip 6pt
\begin{proof}
We have the following partition of unity in $\End(V^{\otimes n-2})$, and in $\End(V^{\otimes n})$
\begin{equation}
	\mathfrak{r}(\id_{n-2})=\sum_{\beta\in \Lambda_{n-2}(\Dim)}\mathfrak{r}(P_{n-2}^\beta)\,,\hspace{1cm}\text{and} \hspace{1cm}\mathfrak{r}(\id_n)=\sum_{\beta\in \Lambda_n(\Dim)}\mathfrak{r}(P_{n}^\beta)\,.
\end{equation}
Acting with the arc induction map on the first partition of unity and multiplying by $\mathfrak{r}(A_n)$ the second partition of unity we have: 
\begin{equation}
	\mathfrak{r}(A_n)=\sum_{\beta\in \Lambda_{n-2}(\Dim)}\mathfrak{r}(\mathcal{A}(P_{n-2}^\beta))\,,\hspace{1cm}\text{and} \hspace{1cm}\mathfrak{r}(A_n)=\sum_{\beta\in \Lambda_{n-2}(\Dim)}\mathfrak{r}(P_{n}^\beta A_n)\,,
\end{equation}
where we have use the fact that $\mathcal{A}(\id_{n-2})=A_n$, and that the kernel of $A_n$ is exactly the space of traceless tensors (see Lemma \eqref{lem:diag_ev_A_2}). Hence we have 
\begin{equation}
	\sum_{\beta\in \Lambda_{n-2}(\Dim)}\mathfrak{r}(\mathcal{A}(P_{n-2}^\beta)\,=\sum_{\beta\in \Lambda_{n-2}(\Dim)}\mathfrak{r}(P_{n}^\beta A_n)\,.
\end{equation}
Multiplying the previous equation by $\mathfrak{r}({P}_{n}^\lambda)$ with $\lambda\in \Lambda_{n-2}(\Dim)$ we get
\begin{equation}
	\sum_{\beta\in \Lambda_{n-2}(\Dim)}\,\,\mathfrak{r}(\mathcal{A}(P_{n-2}^\beta))\mathfrak{r}(P_{n}^\lambda)\,=\mathfrak{r}(P_{n}^\lambda A_n)\,,
\end{equation}
From Lemma \eqref{lem:arcinduction_tens} one has  $\mathfrak{r}(\mathcal{A}(P_{n-2}^\beta))\in\End((D^\lambda)^{\oplus m_\lambda})$, hence $\mathfrak{r}(\mathcal{A}(P_{n-2}^\lambda))=\mathfrak{r}(P_{n}^\lambda A_n)$ and the result follows.\medskip

\end{proof}

As a direct consequence of the previous proposition and of Lemma \eqref{lem:A_block_diagonal} we have: 
\begin{lemma}\label{lem:eigenvalue_Arc_induction}
\begin{itemize}
\item[i)] For any $n>2$ and $\beta\, \in \Lambda_n(\Dim)$, let $D^\beta$ be an irreducible $\Or(\Dim,\C)$-module that occurs in the decomposition of the irreducible $\GL(\Dim,\C)$-module $V^\mu$ into irreducible summands upon restriction to $\Or(\Dim,\C)$.
\begin{equation}\label{eq:Action_Arc_Induction_O}
	\text{Then for any}\quad T\in D^{\beta}\,, \quad T\cdot\mathcal{A}(P^{\lambda}_{n-2})=\left\{
	\begin{array}{ll}
		a_{\mu\backslash\lambda}\,T\,, &\text{if} \quad \beta=\lambda\,,\\
		\quad 0\, \quad &\text{otherwise.}
	\end{array}
	\right.
	\hspace{0.3cm}
\end{equation}
\item[ii)] For any $n>2$ and $\beta\, \in \Lambda_n(\Dim)$, let $L^\mu$ be an irreducible $\C\sn$-module that occurs in the decomposition of the irreducible $\bn$-module $M^\beta_n$ into irreducible summands upon restriction to $\C\sn$.
\begin{equation}\label{eq:Action_Arc_Induction_B}
	\text{Then for any}\quad v\in L^{\mu}\,, \quad v\cdot\mathcal{A}(P^{\lambda}_{n-2})=\left\{
	\begin{array}{ll}
		a_{\mu\backslash\lambda}\,v\,, &\text{if} \quad \beta=\lambda\,,\\
		\quad 0\, \quad &\text{otherwise.}
	\end{array}
	\right.
	\hspace{0.3cm}
\end{equation}
\end{itemize}
\end{lemma}


\subsection{An inductive formula for $\lbrace P^{\lambda}_n \,, \lambda\vdash n-2f\rbrace$}
For any $n>2$, and integer partition $\lambda$ such that $f_\lambda\geqslant 1$ we define $\mathcal{P}^{\lambda}_{n}$ as 
\begin{equation}\label{eq:isotypic_projectors_2}
\mathcal{P}_n^{\lambda}=\sum_{\mu\,\in \overbar{\mathrm{M}}_{n,\lambda}(\Dim)}P_n^{\mu\backslash\lambda}\,,\hspace{1cm} \text{with} \hspace{1cm}	P_n^{\mu\backslash\lambda}=\,\frac{\mathcal{A}(P^{\lambda}_{n-2})\, Z^\mu}{a_{\mu\backslash\lambda}}\,.
\end{equation}
We recall that the set of integer partitions $\overbar{\mathrm{M}}_{n,\lambda}(\Dim)$ is defined in \eqref{eq:subsets_M_L}.
\begin{theorem}\label{theo:isotypic_projectors_2}
For any $n>2$, and integer partition $\lambda$ such that $f_\lambda\geqslant 1$, the elements $\mathcal{P}^{\lambda}_{n}$ are the $(f_\lambda+1)$-traceless projectors which realize the projection of tensors to the isotypic components $(D^{\lambda})^{\oplus m_\lambda}$:
\begin{equation}
	P^\lambda_n=\mathcal{P}^{\lambda}_{n}\,.
\end{equation}
\end{theorem}
%

\begin{proof}
	Take $v\in M^\beta_n$. First, let $\beta$ be such that $\beta\neq \lambda$. Then $\mathcal{P}^{\lambda}_{n}(v)=0$, as a direct consequence of Lemma \ref{lem:eigenvalue_Arc_induction}. Second, let $\beta$ be such that $\beta=\lambda$. Upon restriction of $M^{\lambda}_{n}$ to $\mathbb{C}\Sn{n}$, $v$ decomposes as 
	\begin{equation}
		v=\sum_{\rho\,\in \,\overbar{\mathrm{M}}_{n,\lambda}}v^{\rho}\,, \hspace{0.5cm} \text{with} \hspace{0.5cm} v^{\rho}=v\cdot Z^{\rho}\,,
	\end{equation}
Due to Lemma \ref{lem:eigenvalue_Arc_induction} and to the orthogonality property of the Young central idempotents one has $ v^{\rho}\cdot P_n^{\mu\backslash\lambda}=\delta_{\mu\rho} v^\rho$. As a consequence $\mathcal{P}_n^{\lambda}$ annihilates all irreducible modules not isomorphic to $M^\lambda_n$ and acts by the identity on $M^\lambda_n$.
\end{proof}
\begin{remark}
Formulas \eqref{eq:non_inductive_traceless_projectors} for $f=0$ and \eqref{eq:isotypic_projectors_2} for $f\geqslant 1$, and the following initial data are sufficient to construct the set of elements in $\bn$ realizing the isotypic decomposition of $V^{\otimes n}$.
\begin{equation}\label{initialization_central_Bn}
	P^{(1^n)}_{n}=Z^{(1^n)}\,,\hspace{0.5cm} P^{(2)}_{2}=\frac{1}{\Dim}\,\,\raisebox{-.4\height}{\includegraphics[scale=0.2]{fig/b2.pdf}}\,.
\end{equation}
\end{remark}
\vskip 4pt

The following proposition provides an alternative formula for the isotypic projections operators which uses exclusively the conjugacy class sum $A_n$ and its spectrum. With this formula one can fully leverage the algorithm for optimizing the multiplication of the element $A_n$, which is discussed in the next chapter. From a computer algebra perspective, this formula has proven to be the most efficient and, as a result, it is the one implemented in the \textit{BrauerAlgebra} package.
\begin{proposition}
For any $n>2$, and integer partition $\lambda$ such that $f_\lambda\geqslant 1$, the $(f_\lambda+1)$-traceless projectors which realize the projection of tensors to the isotypic components $(D^{\lambda})^{\oplus m_\lambda}$ admit the following form: 
\begin{equation}\label{eq:branching_f_traceless_central_idempotent_class_A}
	P_n^{\lambda}=\sum_{\scriptsize\alpha\,\in\, \overbar{\mathrm{spec}}^{\,\lambda}(A_n)}\, \, P^{(\lambda\,,\,\alpha)}\,, \hspace{0.4cm} \textit{with} \hspace{0.4cm} P^{(\lambda\,,\,\alpha)}=\frac{\mathcal{A}(P^{\lambda}_{n-2})}{\alpha}\prod_{\begin{array}{c}
			{\scriptstyle \beta\,\in\, \overbar{\mathrm{spec}}^{\,\lambda}(A_n)\,}\\
			{\scriptstyle \beta \,\neq\,\alpha }
	\end{array}}\left(\,\frac{\beta - A_n}{\beta-\alpha}\,\right)\,.
\end{equation}
\end{proposition}
\begin{proof}
Take $v\in M^\rho_n$. If  $\rho\neq \lambda$ one has $\mathcal{P}^{\lambda}_{n}(v)=0$ as a direct consequence of Lemma \ref{lem:eigenvalue_Arc_induction}. Let $\rho$ be such that $\rho=\lambda$ and let $\mathcal{L}_\alpha$ denote the direct sum of irreducible $\C\sn$-modules $L^{\mu}\subset M^{\lambda}_n$ which are such that $a_{\mu\backslash \lambda}=\alpha$. Upon restriction of $M^{\lambda}_{n}$ to $\mathbb{C}\Sn{n}$, $v$ decomposes as 
\begin{equation}
	v=\sum_{\scriptsize\gamma\,\in\, \overbar{\mathrm{spec}}^{\lambda}(A_n)}v_{\gamma}\,, \hspace{0.5cm} \text{with} \hspace{0.5cm}  \mathcal{L}_\gamma\ni v_\gamma=\sum_{\begin{array}{c}
			{\scriptstyle \mu\,\in \overbar{\mathrm{M}}_{n,\lambda}(\Dim)\,}\\
			{\scriptstyle a_{\mu\backslash\lambda}=\,\gamma }
	\end{array}} v\cdot Z^\mu\,.
\end{equation}
When $B_n$ is not semisimple ($\Dim<n-1$) it may happen that $v_{\gamma}$ or that some of the vectors $v^\mu=v\cdot Z^\mu$ are identically zero (recall Lemma \ref{lem:lemma_sets} and Lemma \ref{lem:lemma_spec}). By construction (see Lemma \ref{lem:A_block_diagonal} and Lemma \ref{lem:eigenvalue_Arc_induction}), one has $v_{\gamma}\cdot P^{(\lambda\,,\,\alpha)}=\delta_{\gamma\alpha}\, v_{\gamma}$ and hence $v\cdot P^{(\lambda\,,\,\alpha)}=v_{\alpha}$. As a direct consequence one has
$v\cdot \left(\displaystyle{\sum_{\scriptsize\alpha\,\in\, \overbar{\mathrm{spec}}^{\lambda}(A_n)}\, \, P^{(\lambda\,,\,\alpha)}}\right)=v$.
\end{proof}
The following examples are obtained with the Mathematica package \textit{BrauerAlgebra} using the command CentralIdempotent[$n$, $\lambda$, $\Dim$] described below Remark \ref{rem:above_central}.
\begin{example}[Projection operators utilized in sections \ref{sec:Distortion_Decomposition} and \ref{sec:Riemann_Decomposition}.]
\begin{itemize}
\item[\it i)] Projection operator utilized for the irreducible decomposition of the distortion tensor
\begin{equation}\label{eq:proj_Distortion}
	P^{(1)}_3=\raisebox{-0.41\height}{\includegraphics[scale=0.55]{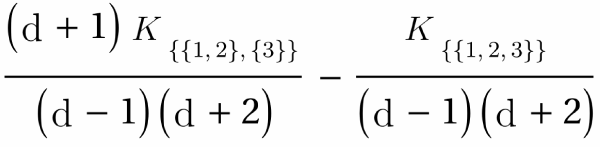}}\,
\end{equation}
\item[\it ii)] Projection operators utilized for the irreducible decomposition of the Riemann tensor
\begin{equation}\label{eq:proj_Riemann}
\hspace{-0.5cm}	\begin{aligned}
	P^{(2)}_4&=\raisebox{-0.73\height}{\includegraphics[scale=1.2]{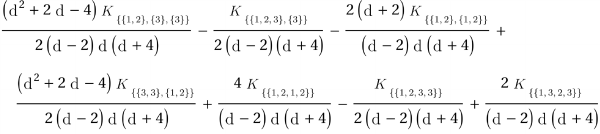}}\,\\
	P^{(11)}_4&=\raisebox{-0.42\height}{\includegraphics[scale=1.1]{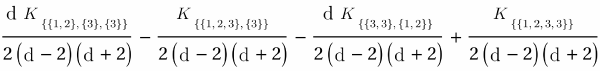}}\,\\
	P^{\emptyset}_4&=\raisebox{-0.41\height}{\includegraphics[scale=0.6]{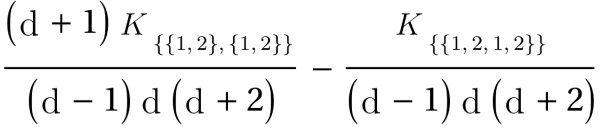}}\,
	\end{aligned}
\end{equation}
\end{itemize}
\end{example}

\newpage
\begin{example}[More projection operators obtained with Mathematica]
	\begin{equation*}
		P^{(2,1)}_5=\raisebox{-0.87\height}{\includegraphics[scale=1.4]{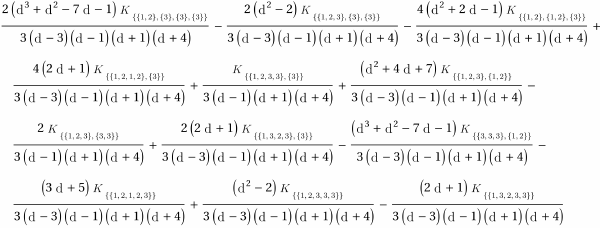}}\,
	\end{equation*}
	\begin{equation*}
		P^{(1,1)}_8=\raisebox{-0.94\height}{\includegraphics[scale=1.4]{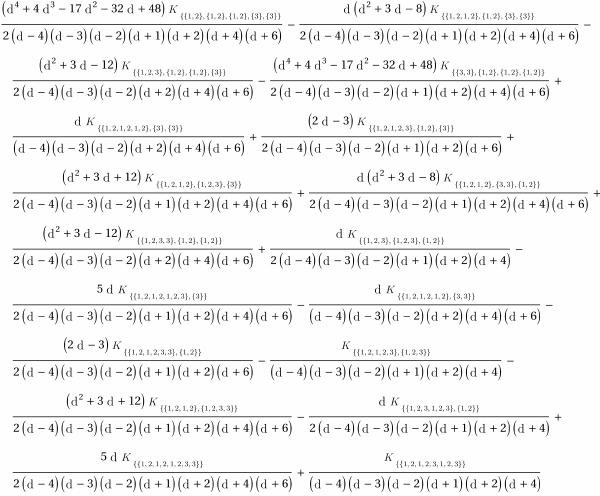}}\,
	\end{equation*}
\end{example}

\paragraph{Scalar invariants and the full trace projector.} The trace decomposition of tensors is obtained via the projector 
\begin{equation}\label{eq:proj_trace_decomposition_bulk}
	P_n^{(f)}=\displaystyle{\sum_{\begin{array}{c}
				{\scriptstyle \lambda\in \Lambda_n(\Dim)}\\
				{\scriptstyle |\lambda|=n-2f}
	\end{array}}} P^{\lambda}_{n}\,.
\end{equation}
\begin{mathematica}[\textit{BrauerAlgebra}]
	CentralIdempotent[\,$n$\,,\,$f$\,,\,$\Dim$\,]\,\\
	\textit{For $f<\lfloor \frac{n}{2}\rfloor$, returns the element in $\Bn{n}(\Dim)$ realizing the $(f+1)$-traceless projection for tensors of order $n$. For $f=\lfloor \frac{n}{2}\rfloor$, $\,$returns the element in $\Bn{n}(\Dim)$ realizing the full trace projection for tensors of order $n$. When $\Bn{n}(\Dim)$ is semisimple this element is central and idempotent.}
\end{mathematica}
Here we restrict our attention to the semisimple regime of $\bn$ and to the case $n$ even. The full trace projector is obtained for $f=\frac{n}{2}$ and one has
\begin{equation}
P_n^{(\lfloor\frac{n}{2}\rfloor)}=P_n^{\emptyset}\,.
\end{equation}
The Littlewood-Richardson coefficients entering $\mathrm{M}_{n,\emptyset}$ are such that
\begin{equation}\label{eq:multi_scalar_invariant}
\tensor{C}{^\mu_\nu_\emptyset}=\delta_{\mu\nu}\tensor{C}{^\mu_\emptyset}=\delta_{\mu\nu}
\end{equation}
where $\nu\vdash n$ is an even partition and we recall that $\tensor{C}{^\mu_\emptyset}$ denotes is the multiplicity of $D^{\emptyset}$ in $V^{\mu}$. As a consequence one has $\mathrm{M}_{n,\emptyset}=\lbrace \mu \vdash n \st \mu \text{ is even}\rbrace$, and the expression for the full trace projector is 
\begin{equation}
P_n^{\emptyset}=\sum_{\text{even } \mu \vdash n}P_n^{\mu\backslash\emptyset}\,.
\end{equation}
Despite not being expert on the following subject, let us just point out that the computation of the full trace projector seems to provide an alternative to the Weingarten calculus \cite{Collins2006,collins2009some,MATSUMOTO,collins2022weingarten} for the orthogonal groups. In appendix \ref{app:Traceprojector} we give the full trace projector $P^{(\emptyset)}_{12}$ to demonstrate that it reproduces the result given at the end of \cite{collins2009some}. We also give the full trace projector $P^{(\emptyset)}_{14}$ and $P^{(\emptyset)}_{16}$ to demonstrate the efficiency of the inductive formula.
\begin{remark}
\begin{itemize}
\item[\it i)] The number of scalar invariants $m_{\emptyset}$ of a tensor of order $n$ with no symmetries of indices is given on the one hand by the number of diagrams in $\Bn{m}$ with $m=\frac{n}{2}$ (for $\Dim\geqslant n$), and, on the other hand by the number of paths which end by the empty partition at level $n$ in the Bratteli diagram for the chain $\Bn{0}\subset\Bn{1}\subset \ldots \subset\Bn{n}$ where $\Bn{0}\cong \C$ (recall the Definition \ref{def:bratteli_diagram} for Bratteli diagrams).\smallskip

According to \eqref{eq:multi_scalar_invariant} $m_{\emptyset}$ is also given by 
\begin{equation}
	m_{\emptyset}=\sum_{\text{even }\mu\in \mathcal{P}_n(\Dim)} m_\mu\,.
\end{equation}
The advantage of the above formula is that it can be generalized to spaces of tensors with symmetries in which some irreducible representations $V^{\mu}$ of $V^{\otimes n}$ may be absent or present with smaller multiplicities.
\item[\it ii)] The computation of the number of the scalar invariants of the $p^{th}$ tensor product of a given tensor of order $n$ (with or without symmetries of indices) is a complicated problem which is related to the notion of \textit{plethysm} of Schur functions. For a discussion on this subject applied to the computation of the number of scalar invariants of the $p^{th}$ tensor product of the metric Riemann tensor see \cite{S_A_Fulling_1992,Lachaume2019}.
\end{itemize}
\begin{mathematica}[\textit{SymmetricFunctions}]
	Plethysm[$f[\rho]$, $h[\beta]$]\,\\
	\textit{Returns the plethysm $f_\rho[h_\beta]$, where $f$ and $h$ are symmetric functions while $\rho$ and $\beta$ are}\\
	\textit{integer partitions.}
\end{mathematica}
\end{remark}

%

\paragraph{Final remarks.} 

As for the symmetric group, there is a chain of natural embeddings
\begin{equation}\label{eq:chain_bn}
	\C\cong B_{0}\subseteq B_{1} \subset B_{2} \subset \dots \subset B_{n}
\end{equation}
realized at each step by the homomorphism $\iota:B_{n-1}\to B_{n}$ defined by the insertion of a vertical line at the right end of any diagram. As a consequence of the multiplicity free property of the branching rules $B_{n}\downarrow B_{n-1}$ \cite{Wenzl_structure-Br_1988} the approach of Okunkov and Vershik to the representation theory of $\mathbb{C}\mathfrak{S}_n$ can be generalized to $\bn$~\cite{doty2019canonical}. As a result, the projectors $P^{\tab}_n$ \eqref{eq:irreducible_decomposition_O} which realize the irreducible decomposition of tensors with respect to $\Or(\Dim,\C)$ can be obtained following the same scheme as for the primitive pairwise orthogonal idempotents of $\C\sn$:
\begin{equation}\label{eq:irreducible_projector_O}
	P^{\,\tab}_n=\prod_{k=0}^{n} P_{k}^{\lambda_k} \,,
\end{equation}
where $\lambda_k$ is the $k^{th}$ Young diagram in the path $\tab$ (also referred to as \textit{up-down tableaux} in \cite{doty2019canonical} or \textit{oscillating tableaux} in \cite{Isaev_2020}) in the Bratteli diagram for the chain \eqref{eq:chain_bn} and $P_{0}^{\emptyset}=0$.\footnote{The element $P^{\,\tab}_n$ is identically zero if the path $\tab$ contains at least one Young diagram which is not in $\Lambda_n(\Dim)$ \cite{Nazarov}.} Note that their exists an alternative construction for $P^{\tab}_n$ using Jucys-Murphy elements (see for example \cite[eq.$\,3.31$]{Isaev_2020}  and \cite{Nazarov}).

\begin{mathematicas}[\textit{BrauerAlgebra}]
	BratteliDiagramBn[\,$n$\,]\,\\
	\textit{Returns the Bratteli diagram for the chain $\Bn{1}\subset \ldots \subset\Bn{n}$ as a List of paths.}\\
	
	BratteliDiagramBn[\,$n$\,,\,Output$\rightarrow$Graph\,]\,\\
	\textit{Returns the Bratteli diagram for the chain $\Bn{0}\subset\Bn{1}\subset \ldots \subset\Bn{n}$ as a Graph\,.}\\
	
	BratteliPathBn[$n$, $\lambda$]\,\\
	\textit{Returns the set of paths $\tab$ in the Bratteli diagram for the chain $\Bn{1}\subset \ldots \subset\Bn{n}$ which end with the integer partition $\lambda$.}
\end{mathematicas}
\begin{example}[The primitive pairwise orthogonal idempotents of $\Bn{3}$]
	\begin{figure}[H]
	\centering
		\includegraphics[scale=0.45]{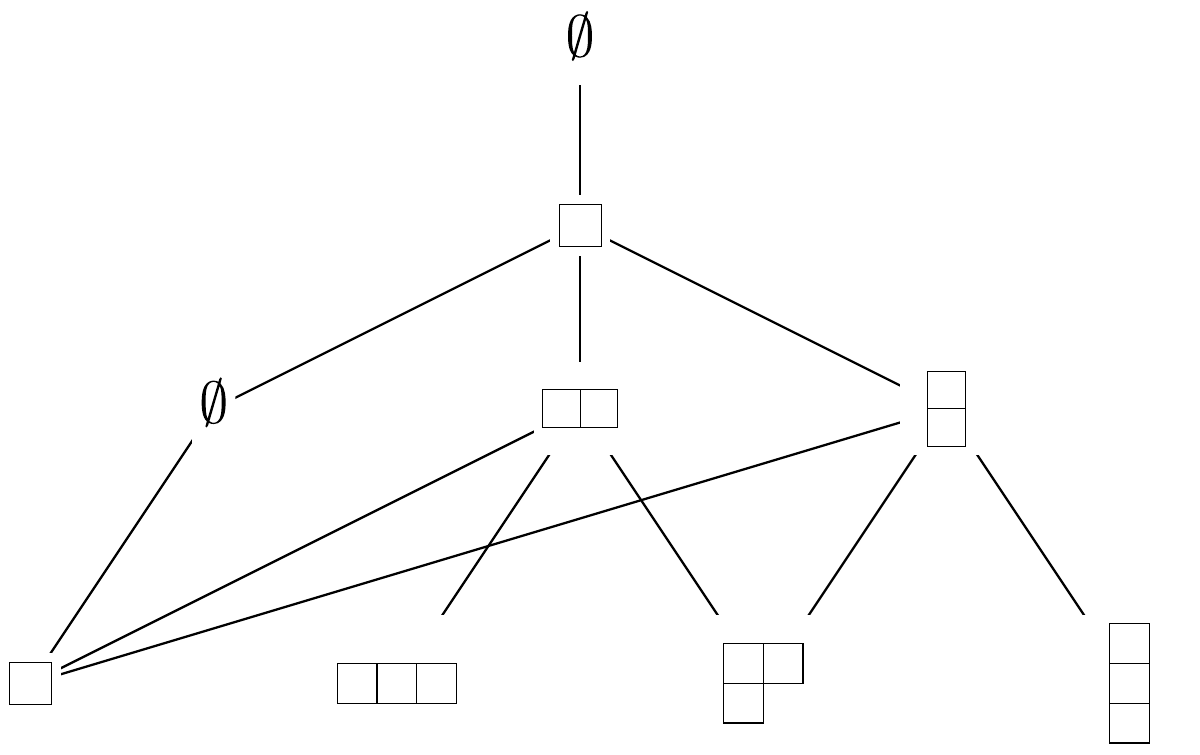}
	\caption{The Bratteli diagram for the chain $\Bn{0}\subset\Bn{1}\subset\Bn{2}\subset\Bn{3}$.}
	\label{fig_proj:bratteli_B3}
\end{figure}
	\noindent The \textit{oscillating tableaux} starting by $\emptyset$ at level zero and ending by an integer partition $\lambda$ at level $3$ are: 
	\begin{equation*}
		\begin{array}{llll}
		\tab_1=\lbrace\,\emptyset\,,\, \Yboxdim{7pt}\yng(1),\,\emptyset,\,\Yboxdim{7pt}\yng(1) \,\rbrace\,, \hspace{0.3cm}&
		\tab_2=\lbrace\,\emptyset\,,\, \Yboxdim{7pt}\yng(1),\,\Yboxdim{7pt}\yng(2),\,\Yboxdim{7pt}\yng(1)\,\rbrace\,, \hspace{0.3cm}&
		\tab_3=\lbrace\,\emptyset\,,\, \Yboxdim{7pt}\yng(1),\,\Yboxdim{7pt}\yng(1,1),\,\Yboxdim{7pt}\yng(1)\,\rbrace\,, \hspace{0.3cm}& \\
		\tab_4=\lbrace \,\emptyset\,,\,\Yboxdim{7pt}\yng(1),\,\Yboxdim{7pt}\yng(2),\,\Yboxdim{7pt}\yng(3) \,\rbrace\,, \hspace{0.3cm}&
		\tab_5=\lbrace \,\emptyset\,,\,\Yboxdim{7pt}\yng(1),\,\Yboxdim{7pt}\yng(2),\,\Yboxdim{7pt}\yng(2,1) \,\rbrace\,, \hspace{0.3cm}&
		\tab_6=\lbrace \,\emptyset\,,\,\Yboxdim{7pt}\yng(1),\,\Yboxdim{7pt}\yng(2),\,\Yboxdim{7pt}\yng(2,1) \,\rbrace\,, \hspace{0.3cm}&
		\tab_7=\lbrace\,\emptyset\,,\, \Yboxdim{7pt}\yng(1),\,\Yboxdim{7pt}\yng(1,1),\,\Yboxdim{7pt}\yng(1,1,1) \,\rbrace\,. 
		\end{array}
	\end{equation*}
	Applying formula \eqref{eq:irreducible_projector_O} with $P^\emptyset_0=P^{\Yboxdim{3pt}\yng(1)}_1=1$, the pairwise orthogonal idempotents of $\Bn{3}$ for $\Dim\geqslant 2$ (semisimple regime of $B_3$) are given by:
	\begin{equation*}
	\begin{array}{llll}
		P^{\tab_1}_3=P^{\emptyset}_2\,P^{(1)}_3\,, \hspace{0.3cm}&
		P^{\tab_2}_3=P^{(2)}_2\,P^{(1)}_3\,, \hspace{0.3cm}&
		P^{\tab_3}_3=P^{(1^2)}_2\,P^{(1)}_3\,, \hspace{0.3cm}& \\[10pt]
		P^{\tab_4}_3=P^{(2)}_2\,P^{(3)}_3\,,  \hspace{0.3cm}&
		P^{\tab_5}_3=P^{(2)}_2\,P^{(2,1)}_3\,, \hspace{0.3cm}&
		P^{\tab_6}_3=P^{(1^2)}_2\,P^{(2,1)}_3\,, \hspace{0.3cm}&
		P^{\tab_7}_3=P^{(1^2)}_2\,P^{(1^3)}_3\,\,.
	\end{array}
	\end{equation*}
For $\Dim=2$, because $\Yboxdim{5pt}\yng(2,1)\,,\Yboxdim{5pt}\yng(1,1,1)\notin\Lambda_{3}(2)$ one has $\mathfrak{r}(P^{\tab_5}_3)=\mathfrak{r}(P^{\tab_6}_3)=\mathfrak{r}(P^{\tab_7}_3)=0$, while for $\Dim=1$, because $\Yboxdim{5pt}\yng(1),\Yboxdim{5pt}\yng(2,1)\,,\Yboxdim{5pt}\yng(1,1,1)\notin\Lambda_{3}(1)$  one has $\mathfrak{r}(P^{\tab_1}_3)=\mathfrak{r}(P^{\tab_2}_3)=\mathfrak{r}(P^{\tab_3}_3)=\mathfrak{r}(P^{\tab_5}_3)=\mathfrak{r}(P^{\tab_6}_3)=\mathfrak{r}(P^{\tab_7}_3)=0$. Note that this is coherent with the irreducible decomposition of the distortion tensor for small dimensions obtained in chapter \ref{chap:Irreducible_MAG} (see Proposition \ref{prop:small_d_distortion}).
\end{example}

\chapter{The centralizer $\cn$ of $\sn$ in $B_n$}\label{chap:cn}
\section{Introduction}

The isotypic projection operators $Z^\mu\in \sn$ and $P^{\lambda}_n \in \bn$ constructed respectively in chapter \ref{chap:projectors_GL} and \ref{chap:projectors_O} commute with the symmetric group $\sn$, and as such belong to subalgebra $\cn\subset \bn$ of all elements in $\bn$ which commute with $\sn$. The Young central idempotents and the traceless isotypic projection operators are given by:\smallskip
\begin{equation}\label{eq:ctYoung_traceless5}
Z^{\mu}=\frac{\mathcal{L}(Z^{\nu})}{z_{\mu\backslash\nu}}\prod_{\begin{array}{c}
		{\scriptstyle \rho\in \mathcal{A}_\nu}\\
		{\scriptstyle \rho\neq\mu}
\end{array}}\dfrac{c_\rho-T_n}{c_\rho-c_\mu}\,,
\hspace{2cm}
P^{\lambda}_{n}=Z^\lambda \prod_{\begin{array}{c}
		{\scriptstyle \beta \in \overbar{\Lambda}_\lambda(\Dim)}\\
\end{array}}\left(1 - \frac{A_n}{a_{\lambda\backslash\beta}}\,\right)\,,
\end{equation}  
while for the $(f+1)$-traceless isotypic projection operators one has: 
\begin{equation}\label{eq:branching_f_traceless_central_idempotent_class_A25}
P_n^{\lambda}=\sum_{\scriptsize\alpha\,\in\, \overbar{\mathrm{spec}}^{\,\lambda}(A_n)}\, \, P^{(\lambda\,,\,\alpha)}\,, \hspace{0.4cm} \textit{with} \hspace{0.4cm} P^{(\lambda\,,\,\alpha)}=\frac{\mathcal{A}(P^{\lambda}_{n-2})}{\alpha}\prod_{\begin{array}{c}
		{\scriptstyle \beta\,\in\, \overbar{\mathrm{spec}}^{\,\lambda}(A_n)\,}\\
		{\scriptstyle \beta \,\neq\,\alpha }
\end{array}}\left(\,\frac{\beta - A_n}{\beta-\alpha}\,\right)\,.
\end{equation}

For the purpose of applications, the main computational difficulty resides in expanding the above factorized formulas. Doing so at the level of the diagrams constituting $T_n$ and $A_n$ appears to be very inefficient as the number of diagrams and hence the computational time grows drastically with $n$. As any element in the subalgebra $\cn$ can be decomposed over a given basis, it is reasonable to expect that the multiplication of two elements of $\cn$ could be done without having to resort to diagram-wise computations. In the case of $\C\sn$, we have seen in chapter \ref{chap:projectors_GL} that the product of two conjugacy class sum are given by the connection coefficient of $\C\sn$ which we recall can be obtained from the irreducible characters \eqref{eq:struc_constants_Sn}. Unfortunately, such a formula is not available in the context of the Brauer algebra, which led us to propose an alternative approach.

In the first section, we present the parametrization of the diagrams in $\bn$ which is used in the package $\textit{BrauerAlgebra}$ and we give a generalization of the cycle notation for permutations to the elements of $\bn$. We then describe the parametrization of the conjugacy classes of $\bn$ in terms of particular ternary bracelets. In the second section we describe the conjugacy class sum basis of $\cn$ and propose a technique which allows one to avoid diagram-wise computations in the expansion of the formulas \eqref{eq:ctYoung_traceless5} and \eqref{eq:branching_f_traceless_central_idempotent_class_A25}. With this technique the products of $A_n,\, T_n$ and $X_n$ with the conjugacy class sums of $\cn$ are obtained from particular second order differential operators on ternary bracelets. Finally, in the last section we give a detailed application of the technique, demonstrating its use in expanding the Young central idempotents for both $\Sn{3}$ and $\Sn{4}$, along with the traceless projector in $\Bn{4}$.

\vskip 2 pt

\section{Parameterization of $\cn$}\label{sec:parametrization_cn}
\subsection{Brauer diagrams in more details}\label{subsec:diagrams}
\paragraph{Brauer diagrams: pairings of $2n$ elements.} Denote by $\dbn$ the diagrammatic basis of $\bn$. Consider $2n$ vertices arranged in two rows: the upper row with $n$ vertices labeled by the integer $\{1,\dots,n\}$ with \textit{type up} and the bottom row with $n$ vertices labeled by the integers $\{\underline{1},\dots,\underline{n}\}$ with \textit{type down}. The elements of $\mathcal{B}_n$ are parameterized by pairings of the set  $\{1,\dots,n\}\cup\{\underline{1},\dots,\underline{n}\}$. 

\begin{example}
	For the permutations $s_1=(123)(465)$ and $s_2=(13)(25)$:
	\begin{equation}
		\begin{aligned}
			&s_1=\raisebox{-.4\height}{\includegraphics[scale=0.81]{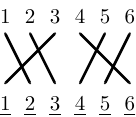}}=\blist{\blist{1,\down{2}},\blist{2,\down{3}},\blist{3,\down{1}},\blist{4,\down{6}},\blist{5,\down{4}},\blist{6,\down{5}}}\,,\\[6pt]
			&s_2=\raisebox{-.4\height}{\includegraphics[scale=0.81]{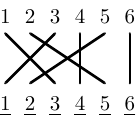}}=\blist{\blist{1,\down{3}},\blist{2,\down{5}},\blist{3,\down{1}},\blist{4,\down{4}},\blist{5,\down{2}},\blist{6,\down{6}}}\,.
		\end{aligned}
	\end{equation}
	For the diagrams in \eqref{eq:multiplication_Bn}:
	\begin{equation}
		\begin{aligned}
			&b_1=\raisebox{-.4\height}{\includegraphics[scale=0.75]{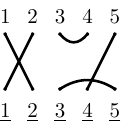}}=\blist{\blist{1,\down{2}},\blist{2,\down{1}},\blist{3,4},\blist{5,\down{4}},\blist{\down{3},\down{5}}}\,,\\[6pt]
			&b_2=\raisebox{-.4\height}{\includegraphics[scale=0.75]{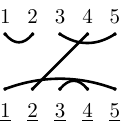}}=\blist{\blist{1,2},\blist{3,5},\blist{4,\down{2}},\blist{\down{1},\down{5}},\blist{\down{3},\down{4}}}\,.
		\end{aligned}
	\end{equation}
\end{example}
Two lists of pairing of $\{1,\dots,n\}\cup\{\underline{1},\dots,\underline{n}\}$ which differ by the order of the labels within a pair or by the position of the pairs within the list are equivalent as they represent the same pairing and hence the same diagram.
\paragraph{Cycle notation for Brauer diagrams.} To obtain the decomposition of a diagram $b\in\dbn$ into a product of $N$ disjoint cycles we will construct a set $\{C_k \st k=1,\ldots, N\}$ of 2 rows matrices, where the first row of $C_k$ correspond to the $k^{th}$ cycle of $b$ in its decomposition. For this, label the upper vertices with  $\{1,\dots,n\}$ and the bottom vertices with  $\{\underline{1},\dots,\underline{n}\}$ and perform the following iterative procedure.\medskip

\textit{
\text{Start with $k=1$:}
	\begin{itemize}
	\item[\textit{i)}] Construct the set $X_k$ of pairs of labels of connected vertices of $b$ which are not of the form $\{i,\underline{i}\}$, that is vertical passing lines in $b$ are discarded from the set $X_k$.
	\item[\it ii)] Construct $C_k$ as follows
	\begin{itemize}
		\item[$\bullet$] The upper left most element of $C_k$ is equal to the minimum integer, with type up, among the elements of $X_k$, and the bottom right element is equal to to the minimum integer, with type down, among the elements of $X_k$.
		\item[$\bullet$] Each column of $C_k$ is filled with entries of one element of $\overbar{X}_k\subseteq X_k$. 
		\item[$\bullet$] The elements at position $(2,j)$ and $(1,j+1)$ of $C_k$, viewed as integers are equal. 
	\end{itemize}
	If $\overbar{X}_k\subset X_k$ define $X_{k+1}=X_{k} \backslash \overbar{X}_{k}$ and repeated step \textit{ii)} with $k=k+1$ until $\overbar{X}_k= X_k$ and then stop.
\end{itemize}
}

\begin{example}
	For the permutation diagram $s_1=\raisebox{-.4\height}{\includegraphics[scale=0.7]{fig/cycle123_465_indexed.pdf}}$:
	\begin{equation}
		C_1=\begin{pmatrix} 
			1 & 2 & 3 \\
			\down{2} & \down{3} & \down{1}\\
		\end{pmatrix}\,,
		\hspace{1cm}
		C_2=\begin{pmatrix} 
			4 & 6 & 5 \\
			\down{6} & \down{5} & \down{4}\\
		\end{pmatrix}\,,
		\hspace{0.5cm}\text{and we have} \hspace{0.5cm} s_1=(1\,2\,3)(4\,6\,5)\,.
	\end{equation}
	For the permutation diagram $s_2=\raisebox{-.4\height}{\includegraphics[scale=0.7]{fig/cycle13_25_indexed.pdf}}$:
	\begin{equation}
		C_1=\begin{pmatrix} 
			1 & 3 \\
			\down{3} & \down{1} \\
		\end{pmatrix}\,,
		\hspace{1cm}
		C_2=\begin{pmatrix} 
			2 & 5 \\
			\down{5} & \down{2}\\
		\end{pmatrix}\,,
		\hspace{0.5cm}\text{and we  have} \hspace{0.5cm} s_1=(1\,3)(2\,5)\,.
	\end{equation}
	For the diagram $b_1=\raisebox{-.4\height}{\includegraphics[scale=0.7]{fig/d1_indexed.pdf}}$:
	\begin{equation}
		C_1=\begin{pmatrix} 
			1 & 2\\
			\down{2} & \down{1} \\
		\end{pmatrix}\,,
		\hspace{1cm}
		C_2=\begin{pmatrix} 
			3 & \down{4} & \down{5} \\
			4& 5 & \down{3}\\
		\end{pmatrix}\,,
		\hspace{0.5cm}\text{and we have} \hspace{0.5cm} b_1=(3\,\down{4}\,\down{5})(1\,2)\,.
	\end{equation}
	For the diagram $b_2=\raisebox{-.4\height}{\includegraphics[scale=0.7]{fig/d2_indexed.pdf}}$:
	\begin{equation}
		C_1=\begin{pmatrix} 
			1 & \down{2} &\down{4}&3&\down{5}\\
			2 & 4 & \down{3}&5&\down{1}\\
		\end{pmatrix}\,,
		\hspace{0.5cm}\text{and we have} \hspace{0.5cm} b_2=(1\,\down{2}\,\down{4}\,3\,\down{5})\,.
	\end{equation}
	For the diagram $b_3=\raisebox{-.4\height}{\includegraphics[scale=0.7]{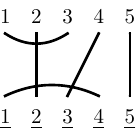}}$:
	\begin{equation}
	C_1=\begin{pmatrix} 
		1 & \down{3} &\down{4}\\
		3 & 4 & \down{1} \\
	\end{pmatrix}\,,
	\hspace{0.5cm}\text{and we have} \hspace{0.5cm} b_3=(1\,\down{3}\,\down{4})\,.
	\end{equation}
\end{example}

\begin{remark}
\begin{itemize}
	\item[\it i)] By construction a cycle always start with an integer with type up. 
	\item[\it ii)] To the best of our knowledge the cycle notation for Brauer diagrams is not presented in the literature.
\end{itemize}
\end{remark}

\begin{mathematicas}[\textit{BrauerAlgebra}]
	DownInteger\,\\
	\textit{\textup{DownInteger} is the head for the integer parameterizing the bottom row of a Brauer diagram. Usage: \textup{DownInteger[$k$]} where $k$ is an integer.}\\
	
	BrauerCycles[b]\,\\
	\textit{\textup{BrauerCycles[$\lbrace$$cycle_1$, $cycle_2$, \ldots $\rbrace$]} represents an element of the Brauer algebra in the cycle notation. This function generalizes the built-in \textup{Cycles} command of Mathematica for permutations.}\\
	
	BrauerElementsCycles[$n$]\,\\
	\textit{Returns a list of all basis element of $\bn$ with head \textup{BrauerCycles}.}
\end{mathematicas}

We recall that Mathematicas commands related to the implementation of Brauer diagrams were already given in the first chapter below equation \eqref{math:brauerdiagram}.

\subsection{Conjugacy classes and ternary bracelets}\label{sec:classes-bracelets_map}
\paragraph{Conjugacy classes.}Two diagrams $b,b^{\prime}\in \dbn$, are said to be {\it conjugate}, which is denoted $b\sim b^{\prime}$, if there exists an element $s\in\Sn{n}$ such that $b^{\prime} = s b s^{\shortminus1}$. The conjugacy class $C_{b}$ of a diagram $b\in \dbn$ is the orbit of $b$ under the conjugate action of $\sn$ on $\dbn$ 
\begin{equation}
C_{b}=\{s b s^{\shortminus1} \st s\in \sn \}\,.
\end{equation}
Also, let $\Stab_{\sn}(b)$ denote the stabilizer of a diagram $b\in \dbn $, {\it i.e.} the set of elements $s\in\Sn{n}$ such that $sbs^{\shortminus 1} = b$. From the orbit-stabilizer theorem one has 
\begin{equation}
|C_{b}|=\frac{n!}{|\Stab_{\sn}(b)|}\,.
\end{equation}
\paragraph{Cycle type of a diagram.} Assign to each cycle of a diagram (in its cycle notation) the absolute value of the difference between the number of type up integers and the number of type down integers. This sequence of numbers determines an integer partition allowing zeros. The \textit{cycle type} of a diagram $b\in \dbn$ denoted $\CT(b)$ is the aforementioned integer partition, which naturally generalize the notion introduced in chapter \ref{chap:projectors_GL}. This definition is equivalent to the one introduced by Arun Ram for the evalutation of irreducible characters of the Brauer algebra \cite{Ram_characters} (see also \cite{SHALILE_classes}). For the diagram $b_1,\, b_2$ and $b_3$ above one has 
\begin{equation}
\CT(b_1)=(2,1)\,,\hspace{1cm} \CT(b_2)=(1)\,,\hspace{1cm} \CT(b_3)=(1^3)\,.
\end{equation}
\begin{remark}
$\CT$ is constant on conjugacy classes.
\end{remark}

Contrary to the case of $\sn$, the cycle type of diagrams in $\dbn$, which are integer partitions, do not parametrize the conjugacy classes. Instead, the latter are labeled by a combinatorial object of a different nature: ternary bracelets. For permutations diagram in $\dbn$ we will see that these bracelets reduce to the unary bracelets presented in chapter \ref{chap:projectors_GL} which are in one to one correspondence with the integer partitions.

\paragraph{Ternary bracelets.} In order to parametrise the conjugacy class of $\dbn$  we consider an equivalent reformulation of the one described in \cite{SHALILE_classes,Shalile_Br-center} (see also \cite{KMP_central_idempotents}). Namely, the conjugacy classes are in one-to-one correspondence with a particular subset of the so-called ternary bracelets. A {\it ternary bracelet} is an equivalence class of non-empty words over the ternary alphabet $\mathcal{A} = \{\bb{n},\bb{s},\bb{p}\}$ related by cyclic permutations and inversions (dihedral symmetry), {\it i.e.} can be viewed as a word with its letters written along a closed loop without specifying the direction of reading. In the sequel, we will write $[w]$ to denote a bracelet containing a representative $w$ (a word over $\mathcal{A}$). For the reverse of $w$ we will write $I(w)$, so according to the definition of bracelets one has $[w] = [I(w)]$. The length of a bracelet is defined as the length of any among its representatives, which is written as $\vert w\vert$.
\vskip 2 pt

Denote $\mathfrak{b}(\mathcal{A})$ the set of non-empty ternary bracelets with the same number of occurrences of the letters $\bb{n}$ and $\bb{s}$ (which is allowed to be $0$), with the additional requirement that for any representative, if there is a pair of letters $\bb{n}$ (respectively, $\bb{s}$), there is necessarily the letter $\bb{s}$ (respectively, $\bb{n}$) in between. For example, $[\bb{s}],[\bb{n}\bb{n}\bb{s}\bb{s}]\notin \mathfrak{b}(\mathcal{A})$ and 
\begin{equation}
	\def\arraystretch{1.4}
	\begin{array}{l}
		[\bb{n}\bb{s}\bb{p}\bb{p}] = [\bb{s}\bb{p}\bb{p}\bb{n}] = [\bb{n}\bb{p}\bb{p}\bb{s}] = [\bb{p}\bb{p}\bb{n}\bb{s}] = [\bb{s}\bb{n}\bb{p}\bb{p}] = [\bb{p}\bb{n}\bb{s}\bb{p}] = [\bb{p}\bb{s}\bb{n}\bb{p}] \in \mathfrak{b}(\mathcal{A})\,,\\
		\text{with}\quad\vert\bb{n}\bb{s}\bb{p}\bb{p}\vert = 4\,.
	\end{array}
\end{equation}
\vskip 2 pt

Consider the polynomial algebra $\mathbb{C}[\mathfrak{b}(\mathcal{A})]$, {\it i.e.} the $\mathbb{C}$-span over the basis $\mathfrak{B}$ of monomials $[w_1]\dots [w_r]$ (with all $[w_j]\in \mathfrak{b}(\mathcal{A}))$ and $1$. We will write $[w]^m = \underbrace{[w]\dots [w]}_{m}$ for brevity, and also occasionally $\bb{a}^{m} = \underbrace{\bb{a}\dots\bb{a}}_{m}$ for a letter $\bb{a}\in\mathcal{A}$. The degree of a monomial is defined as the sum of the lengths of the bracelets, $\deg\left([w_{1}]\dots [w_{r}]\right) = \vert w_{1}\vert + \dots + \vert w_{r}\vert$ (where by definition one puts $\deg(1) = 0$), and we denote by $\mathbb{C}[\mathfrak{b}(\mathcal{A})]_{n} \subset \mathbb{C}[\mathfrak{b}(\mathcal{A})]$ the subset of polynomials of a given degree $n$ with basis of monomials denoted $\mathfrak{B}_n$.
\begin{example}
\begin{equation*}
\begin{aligned}
&\mathfrak{B}_2=\left\{[\bb p]^2,[\bb{pp}],[\bb{ns}]\right\}\,, \hspace{0.3cm} \mathfrak{B}_3=\left\{[\bb p]^3,[\bb{pp}][\bb{p}],[\bb{ppp}],[\bb{ns}][\bb{p}],[\bb{nsp}]\right\},\\[6pt]
&\mathfrak{B}_4=\left\{[\bb p]^4,[\bb{pp}][\bb{p}]^2,[\bb{pp}]^2,\,[\bb{ppp}][\bb{p}],[\bb{pppp}]\,,[\bb{ns}][\bb{p}]^2,[\bb{nsp}][\bb{p}]\,,[\bb{nspp}]\,,[\bb{npsp}]\,,
[\bb{ns}]^2,[\bb{nsns}]\right\}.
\end{aligned}
\end{equation*}
\end{example}
\paragraph{Generalized cycle type.} For any diagram $b\in \dbn$ define its \textit{generalized cycle type} $\GCT(b)\in  \mathbb{C}[\mathfrak{b}(\mathcal{A})]_n$ via the following procedure:
\begin{itemize}
	\item[{\it 1)}] label each line of the diagram $b$ by a letter from $\mathcal{A}$: the arcs in the upper (respectively, lower) row by $\bb{n}$ (respectively, $\bb{s}$), the passing lines by $\bb{p}$;
	\item[{\it 2)}] identify the upper nodes with the corresponding lower nodes and straighten the obtained loops, which results in a set of bracelets $[w_{1}],\dots,[w_{r}]\in \mathfrak{b}(\mathcal{A})$;
	\item[{\it 3)}] set $\GCT(b) = \,[w_1] \dots [w_r]\in \mathbb{C}[\mathfrak{b}(\mathcal{A})]$.
\end{itemize}
\noindent For the diagram $b=\raisebox{-.45\height}{\includegraphics[scale=0.4]{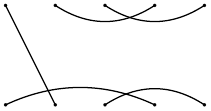}}$ one has the following sequence of transformations:

\begin{equation*}
	\raisebox{-.45\height}{\includegraphics[scale=0.6]{fig/b1.pdf}}\mapsto \raisebox{-.45\height}{\includegraphics[scale=0.7]{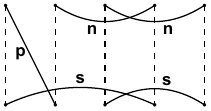}}\mapsto \Bigg\lbrace{\raisebox{-.35\height}{\includegraphics[scale=0.14]{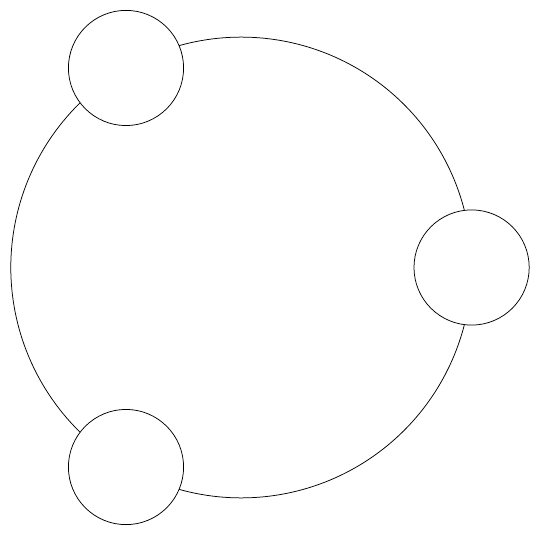}}\put(-30.5,17){${\scriptstyle \bb{n}}$}\put(-6,4){${\scriptstyle\bb{s}}$}\put(-31,-9.8){${\scriptstyle \bb{p}}$}\, , \, \raisebox{-.45\height}{\includegraphics[scale=0.14]{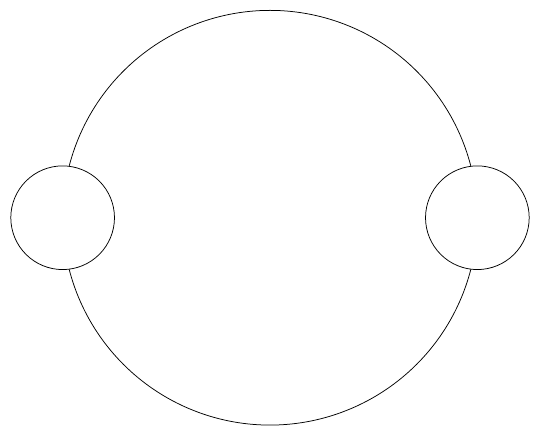}}\put(-35,0){${\scriptstyle \bb{n}}$}\put(-6.5,0){${\scriptstyle \bb{s}}$}\Bigg\rbrace}\mapsto \,[\bb{nsp}][\bb{ns}]
\end{equation*}
\begin{example}
Assigning the colors $\{\textcolor{black}{\bullet}\,,\textcolor{red}{\bullet}\,,\textcolor{orangebracelet}{\bullet} \}$ to the letters $\{\bb{p},\bb{n},\bb{s}\}$ respectively, one has for the diagram of section \ref{subsec:diagrams}:
\begin{equation*}
\begin{array}{ll}
\GCT(s_1)=[\bb{ppp}][\bb{ppp}]=\left\{\raisebox{-.4\height}{\includegraphics[scale=0.7]{fig/unary_bracelets_3.pdf}}\,,\raisebox{-.4\height}{\includegraphics[scale=0.65]{fig/unary_bracelets_3.pdf}}\right\}\,, &\hspace{-0.6cm}\GCT(s_2)=[\bb{pp}]^2[\bb{p}]^2=\left\{\raisebox{-.4\height}{\includegraphics[scale=0.7]{fig/unary_bracelets_2.pdf}}\,,\raisebox{-.4\height}{\includegraphics[scale=0.7]{fig/unary_bracelets_2.pdf}}\,,\raisebox{-.4\height}{\includegraphics[scale=0.7]{fig/unary_bracelets_1.pdf}}\,,\raisebox{-.4\height}{\includegraphics[scale=0.7]{fig/unary_bracelets_1.pdf}}\right\}\,,\\[10pt]
\GCT(b_1)=[\bb{n}\bb{s}\bb{p}][\bb{pp}]=\left\{\raisebox{-.4\height}{\includegraphics[scale=0.7]{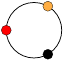}}\,,\raisebox{-.4\height}{\includegraphics[scale=0.7]{fig/unary_bracelets_2.pdf}}\right\}\,,&\hspace{-0.6cm}\GCT(b_2)=[\bb{nsnsp}]=\left\{\raisebox{-.4\height}{\includegraphics[scale=0.7]{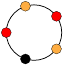}}\,\right\}\,,\\[10pt]
\GCT(b_3)=[\bb{nsp}][\bb{p}]^2=\left\{\raisebox{-.4\height}{\includegraphics[scale=0.7]{fig/nsp.pdf}}\,,\raisebox{-.4\height}{\includegraphics[scale=0.7]{fig/unary_bracelets_1.pdf}}\,,\raisebox{-.4\height}{\includegraphics[scale=0.7]{fig/unary_bracelets_1.pdf}}\right\}\,.
\end{array}
\end{equation*}
\end{example}
\noindent Note that for a permutation $s\in\Sn{n}$ one has $\GCT(s) = [\bb{p}^{\lambda_1}] \dots [\bb{p}^{\lambda_l}]$. Without loss of generality, $\lambda_1 \geqslant \dots \geqslant \lambda_l$, so one arrives at a partition of $n$. In other words for permutations the generalized cycle type reduces to the cycle type as the concatenation of unary bracelets.

\begin{theorem}[{\cite[Theo. 2.6 ]{SHALILE_classes}}]\label{theo:classes-bracelets}
	Two diagrams $b$ and $b'$ are conjugate if and only if they have the same generalized cycle type. 
\end{theorem}
We then  denote by $C_\xi$ for the conjugacy class associated to any diagram $b\in \dbn$ with generalized cycle type $\xi$. Also one has that the cardinality of the basis $\mathfrak{B}_n$ of $\mathbb{C}[\mathfrak{b}(\mathcal{A})]_n$ is the number of conjugacy class of $\dbn$. 
\begin{remark}
\begin{itemize}
\item[\it i)] This theorem is a generalization of Theorem \ref{theo:classes_partition}\,.
\item[\it ii)] In the mathematica packages, the letters $\bb{n}$, $\bb{s}$ and $\bb{p}$ are replaced by the letters $1, 2$ and $3$ respectively. 
\end{itemize}
\end{remark}

\paragraph{Generalized cycle type from the cycle notation.} The generalized cycle type of a diagram $b\in\dbn$ can be easily determined starting from its cycle notation. Associate a letter from the set $\mathcal{A}$ to each pair of consecutive integers in the cycle (the first and last elements of the cycle considered consecutive). If the elements of a pair share the same integer type, assign it to the letter $\bb p$. Otherwise, assign each pair to the letter $\bb n$ or $\bb s$.  The first pair of consecutive integers with different types, reading the cycle from left to right, is mapped to $\bb n$. Recall that there can not be two consecutive letters $\bb{n}$ or $\bb{s}$ in a bracelet. 
\begin{example}
Again, for the diagrams $b_1$, $b_2$, and $b_3$ one has:
\begin{equation*}
\hspace{-0.2cm}\GCT((3\,\down{4}\,\down{5})(1\,2))=[\bb{nps}][\bb{pp}], \hspace{0.2cm} \GCT((1\,\down{2}\,\down{4}\,3\,\down{5}))=[\bb{npsns}], \hspace{0.2cm} \GCT((1\,\down{3}\,\down{4}))=[\bb{nps}][\bb{p}]^2,
\end{equation*}
which coincide with the result of the previous example.
\end{example}

\begin{mathematicas}[\textit{BrauerAlgebra}]
	Bracelets\,\\
	\textit{\textup{Bracelets} is the head for the ternary bracelets parameterizing the conjugacy classes of the Brauer algebra. Usage: \textup{Bracelets[$\lbrace bracelet_1$, $bracelet_2$, $\ldots$$\rbrace$]}}\\
	
	BrauerBracelets[n]\,\\
	\textit{Returns a list of all the ternary bracelets parameterizing the conjugacy classes of the Brauer algebra $\bn$.}
\end{mathematicas}

\paragraph{Cardinality $\vert \Stab_{\Sn{n}}(b)\vert$ via symmetries of $\GCT(b)$.} 

Consider the group of cyclic permutations $\mathbb{Z}_{\ell}$ which acts on the words of the length $\ell$. Recall that $I(w)$ denotes the inversion of a word $w$. Given any word $w$ of length $\ell$, define its {\it turnover stabilizer}:
\begin{equation}\label{eq:stabilizer_word}
	S(w) = \left\{h\in \mathbb{Z}_{\ell}\;\; : \;\; \text{either}\;\; h(w) = w\;\; \text{or}\;\; h(w) = I(w)\right\}\,.
\end{equation}
For any bracelet $[w]\in\mathfrak{b}(\mathcal{A})$ define its \textit{stability index} as 
\begin{equation}
\mathrm{st}\big([w]\big) = \vert S(w)\vert
\end{equation}
Let $w$ and $w'$ be two words in a ternary bracelet\footnote{Recall that ternary bracelets are equivalent class of words on $\mathcal{A}$}. Then there exist $s\in \sn$ such that $w=s(w')$, and for any $h\in S(w)$ such that $h(w)=w$, one has $s^{-1}hs(w')=w'$. Hence, the cardinality $\vert S(w)\vert$ does not depend on a particular choice of a representative, the stability index is well defined.
Note that both conditions in \eqref{eq:stabilizer_word} are simultaneously satisfied by an element $h\in\mathbb{Z}_{\ell}$ only if $w = I(w)$, {\it i.e.} the word is inversion-invariant which in $\mathfrak{b}(\mathcal{A})$ happens only for bracelet $[\bb{p}\dots\bb{p}]$ associated to a permutation cycle. In this case $\big\vert S(\bb{p}^{\ell})\big\vert = \ell$ as already noted in section \ref{subsec:conjugacyclassSn}. 
\vskip 2 pt
For a monomial $\xi=[\xi_1]^{m_1}\dots [\xi_r]^{m_r}$ in $\mathbb{C}[\mathfrak{b}(\mathcal{A})]$ define its {\it stability index} as follows:
\begin{equation}\label{eq:stab_bracelets}
	\mathrm{st}\big(\xi \big) = \prod_{j=1}^{r} \mathrm{st}\big([\xi_j]\big)^{m_j}\,m_{j}!\,
\end{equation}
This definition is the generalization of \eqref{eq:stability_index_symmetric_groups} to the conjugacy class of $\bn$ and the following Lemma gives a convenient method for computing the cardinality of $\Stab_{\Sn{n}}(b)$ and $C_{b}$.
\begin{lemma}\label{lem:sym} 
	For any diagram $b\in  \dbn$ with $\GCT(b)=\xi$ holds $\vert \Stab_{\Sn{n}}(b)\vert = \mathrm{st}\big(\xi\big)$. Hence 
\begin{equation}
	|C_{b}|=\frac{n!}{\mathrm{st}\big(\xi\big)}\,.
\end{equation}
\end{lemma}
For the proof see appendix \eqref{app:Lemma_sym}.

\begin{mathematica}[\textit{BrauerAlgebra}]
	StabilityIndex[$\xi$]\,\\
	\textit{Returns the stability index $\stab (\xi)$ associated with the ternary bracelets $\xi$ with head \textup{Bracelets}.}
\end{mathematica}

\section{Conjugacy class sum in $\bn$}
\subsection{A basis for $\cn$}
The \textit{normal conjugacy class sum} (or simply \textit{conjugacy class sum}) of a diagram $b\in \dbn$ with $\GCT(b)=\xi$, denoted $K_\xi$, is the formal sum of the elements of $C_\xi$
\begin{equation}
	K_\xi= \sum_{b \,\in C_\xi} b\,.
\end{equation}
We also define the $\textit{averaged}$ conjugacy class sum $\bar{K}_\xi$ as 
\begin{equation}\label{eq:averaged_class_sum}
	\bar{K}_\xi=\sum_{s \in \sn } s \, b \, s^{\shortminus 1}\,, \hspace{0.5cm} \text{for any $b\in \dbn$ such that $\GCT(b)=\xi$}.
\end{equation}
The cardinalities of the stabilizers of two conjugate diagrams coincide, so by Lemma \ref{lem:sym}
\begin{equation}\label{eq:normalised_classes}
	\bar K_\xi= \stab(\xi)\,K_\xi \,.
\end{equation}
\begin{lemma}
The set of conjugacy class sums $\lbrace K_\xi \st \xi\in  \mathfrak{B}_n\rbrace$ forms a basis of the centralizer $\mathcal{C}_n$ of $\mathfrak{S}_n$ in $\bn$:
	\begin{equation}\label{eq:centraliser}
		\text{for any } s\in\Sn{n}\;\; \text{holds}\;\; us = su\quad\Leftrightarrow\quad  u = \sum_{\xi\in \mathfrak{B}_n }c_\xi \, K_{\xi}\,.  
	\end{equation}
\end{lemma}
\begin{proof}
	The proof of the above implication from right to left is straightforward, while for the opposite implication note that the condition $us = su$ for any $s\in \Sn{n}$ implies
	\begin{equation}
		u=\dfrac{1}{n!}\sum_{s\in \sn}s\, u\, s^{\shortminus 1}.
	\end{equation}
	Decomposing u over the diagram basis of $\bn$,
	\begin{equation}
		u = \sum_{b\in \dbn} c_{b}\,b\,, \hspace{0.5cm} \text{yields} \hspace{0.5cm} u=\dfrac{1}{n!}\sum_{b\in \dbn} c_b \bar{K}_{\GCT(b)}\,.
	\end{equation}
\end{proof}
\begin{example}
One has the following basis in $\mathcal{C}_{3}$ parametrized by monomials in $\mathfrak{B}_{3}$:
\begin{equation}\label{eq:classes_Br_3}
\begin{aligned}
K_{[\bb{p}]^3}  &= \raisebox{-.4\height}{\includegraphics[scale=0.32]{fig/id3.pdf}}\,, \hspace{0.5cm}		
K_{[\bb{pp}][\bb{p}]}= \, \raisebox{-.4\height}{\includegraphics[scale=0.32]{fig/s3a.pdf}}+\raisebox{-.4\height}{\includegraphics[scale=0.32]{fig/s3b.pdf}}+\raisebox{-.4\height}{\includegraphics[scale=0.32]{fig/s3c.pdf}}\,, \hspace{0.5cm}	
K_{[\bb{ppp}]}=\, \raisebox{-.4\height}{\includegraphics[scale=0.32]{fig/s3d.pdf}}+\raisebox{-.4\height}{\includegraphics[scale=0.32]{fig/s3e.pdf}}\,,
\\[6pt]
K_{[\bb{ns}][\bb{p}]}&=  \raisebox{-.4\height}{\includegraphics[scale=0.32]{fig/d3a.pdf}}+\raisebox{-.4\height}{\includegraphics[scale=0.32]{fig/d3b.pdf}}+\raisebox{-.4\height}{\includegraphics[scale=0.32]{fig/d3c.pdf}}\,,\hspace{0,5cm}
K_{[\bb{nsp}]} = \raisebox{-.4\height}{\includegraphics[scale=0.32]{fig/d3d.pdf}}+\raisebox{-.4\height}{\includegraphics[scale=0.32]{fig/d3e.pdf}}+\raisebox{-.4\height}{\includegraphics[scale=0.32]{fig/d3f.pdf}}+\raisebox{-.4\height}{\includegraphics[scale=0.32]{fig/d3g.pdf}}+\raisebox{-.4\height}{\includegraphics[scale=0.32]{fig/d3h.pdf}}+\raisebox{-.4\height}{\includegraphics[scale=0.32]{fig/d3i.pdf}}\,.
\end{aligned}
\end{equation}
\end{example}

\begin{remark}
The elements $T_n$, $A_n$ which enters the formulas \eqref{eq:ctYoung_traceless5}-\eqref{eq:branching_f_traceless_central_idempotent_class_A25} and the element $X_n$ are all in $\cn$. In chapter \ref{chap:projectors_GL} we have already seen that $T_n$ is the conjugacy class sum 
\begin{equation}
	T_n=K_{[\bb{pp}][\bb p]^{n-2}}\,.
\end{equation}
For any $1\leqslant i<j\leqslant n$ one has $\GCT(d_{ij})=[\bb{ns}][\bb{p}]^{n-2}$, so
\begin{equation}\label{class_A}
	A_n=K_{[\bb{ns}][\bb{p}]^{n-2}}\,,\hspace{1cm} X_n=K_{[\bb{pp}][\bb p]^{n-2}}-K_{[\bb{ns}][\bb{p}]^{n-2}}\,.
\end{equation}
\end{remark}

\subsection{The dimension of $\cn$} An efficient algorithm for counting the number of conjugacy class of $\dbn$, denoted $d_n$, is described implicitly in \cite{WILLENBRING2001_Stable_Branche_Brauer}. In there the author is interested in the space of $O(\Dim,\C)$-invariant homogeneous polynomials of degree $n$ of $\Dim\times \Dim$ matrices, denoted $\mathcal{P}^n(M_{\Dim})^{O(\Dim,\C))}$. In particular it is shown that $\dim(\mathcal{P}^n(M_{\Dim})^{O(\Dim,\C))})$ is the coefficient of $q^n$ in the expansion of
\begin{equation}\label{eq:produc_dn}
	\prod_{k=1}^{\infty}\left(\frac{1}{1-q^k}\right)^{c_k}\,,
\end{equation} 
where $c_k$ is the number of $k$ vertex cyclic graphs with directed edges counted up to dihedral symmetry. These combinatorial objects called \textit{directed cycles} are in one to one correspondence with the ternary bracelets described above \cite{KIM2021292}. A directed cycle with $
f$ \textit{sink vertices} $\raisebox{-.2\height}{\includegraphics[scale=1]{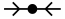}}$, $f$ \textit{source vertices} $\raisebox{-.2\height}{\includegraphics[scale=1]{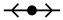}}$, and $k-2h$ \textit{flow vertices} $\raisebox{-.2\height}{\includegraphics[scale=1]{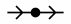}}$, correspond to a ternary bracelet with respectively $f$ letters $\bb{n}$, $f$ letters $\bb{s}$ and $k-2f$ letters $\bb{p}$. In \cite{KIM2021292} this correspondence is used to show that $\dim(\mathcal{P}^n(M_{\Dim})^{O(\Dim,\C))})=\dim(\cn)$, that is, the aforementioned coefficient in the product expansion of  \eqref{eq:produc_dn} is $d_n$.\medskip

Hence an algorithm computing $c_k$ leads to an algorithm for computing $d_n$ by expanding the product \eqref{eq:produc_dn}. For completeness, we give the algorithm presented in section $6.1.$ of \cite{WILLENBRING2001_Stable_Branche_Brauer}. Let,
\begin{equation}
	C_{f}(x)=\sum_{K=0}^{\infty} c_{k,f} \, x^{K}\,, \hspace{1cm} \text{$K=k-2f$}\,,
\end{equation}
be the generating function in which the coefficient $c_{k,f}$ is the number of directed cycles with $f$ letters $\bb{n}$, $f$ letters $\bb{s}$ and $k-2f$ letters $\bb{p}$. It can be shown that 
\begin{equation}
	C_0(x)=\frac{1}{1-x}\,, \hspace{1cm} C_f(x)=\frac{1}{2(1-x^2)^f}+\frac{1}{2f}\sum_{d|f}\frac{\phi(d)}{(1-x^d)^{\frac{2f}{d}}}\,, \hspace{0.5cm} \text{for $f\geqslant 1$}\,,
\end{equation}
where $\phi(d)$ is the Euler phi function. Then the coefficient $c_k$ is given by 
\begin{equation}
	c_k=\sum_{h=0}^{\lfloor k/2 \rfloor} c_{k,f}\,.
\end{equation}
It seems that finding a more explicit formula for the coefficients $c_k$ is an open problem.

\subsection{Connection coefficients}
Let $\mu$ and $\zeta$ be two bracelets in $\mathfrak{B}_n$. The product of two class sums $K_\mu$ and $K_\zeta$ in $\cn$ is given by 
\begin{equation}\label{eq:struc_constants_Bn}
K_\mu K_\zeta= \sum_{\xi\in\mathfrak{B}_n} \tensor{C}{^\xi_\mu_\zeta} K_{\xi}\,, 
\end{equation}
where $\tensor{C}{^\xi_\mu_\zeta}$ are the \textit{connection coefficients} or \textit{structure constants} of $\cn$\footnote{Here we extend the terminology used for the conjugacy class sums of $\C\sn$ \cite{goupil1998factoring,goulden2000transitive,corteel2004content}.}. We recall that in the restricted framework of the center $\mathcal{Z}_n\subset \cn$ of $\C\sn$, the coefficient $ \tensor{C}{^\xi_\mu_\zeta}$ can be computed via the irreducible character of $\sn$ \eqref{eq:struc_constants_Sn}.\medskip

For the Brauer algebra, the irreducible characters are not directly related to the conjugacy class sums, and a formula which would generalize $\eqref{eq:struc_constants_Sn}$ to $\bn$ is not present in the literature.\medskip



\begin{mathematica}[\textit{BrauerAlgebra}]
	ConjugacyClassProduct[ClassSum[$\xi$], ClassSum[$\zeta$]]\, \\
	\textit{Returns the product of $K_\xi$ with $K_\zeta$ where $\xi$ and $\zeta$ are ternary bracelets with head \textup{Bracelets}. Note: the head ClassSum is optional.}
\end{mathematica}

\section{$A_n$, $T_n$, and $X_n$ as second-order differential operators on $\mathbb{C}[\mathfrak{b}(\mathcal{A})]$}\label{sec:differential_operators}
For our purpose it appears to be more convenient to consider the action of $A_n$, $T_n$ and $X_n$ on the averaged basis $\lbrace \bar K_\xi \st \xi\in  \mathfrak{B}_n\rbrace$ of $\cn$ which is expressed in terms of linear operators $\Delta_a$, $\Delta_t$ and $\Delta_x=\Delta_t-\Delta_a$ on $\mathbb{C}[\mathfrak{b}(\mathcal{A})]_n$ as follows:
\begin{equation}\label{eq:A_T_X_oeprator_def}
	A_n \,\bar K_{\zeta}= \sum_{\xi \in \mathfrak{B}_n}  \left(\Delta_a(\zeta)\right)_\xi \bar K_{\xi}\,,\hspace{1cm}
	T_n \,\bar K_{\zeta} = \sum_{\xi \in \mathfrak{B}_n}  \left(\Delta_t(\zeta)\right)_\xi \bar K_{\xi}\,.
\end{equation}
The above formulas serve as definitions which allow one to construct $\Delta_a,\,\Delta_t$ and $\Delta_x$ by evaluating the products of Brauer diagrams in the left-hand-sides of \eqref{eq:A_T_X_oeprator_def}. Note that in terms of the components of $\Delta_a$ and $\Delta_t$ the associated connection coefficients are giving by 
\begin{equation}
	\tensor{C}{^\xi_a_\zeta}=\frac{\stab(\xi)}{\stab(\zeta)}\left(\Delta_a(\zeta)\right)_\xi\,,\hspace{2cm}\tensor{C}{^\xi_t_\zeta}=\frac{\stab(\xi)}{\stab(\zeta)}\left(\Delta_t(\zeta)\right)_\xi\,.
\end{equation}
\begin{example}
	By direct diagrammatic computation one finds the action of $A_3$ on $\lbrace \bar K_\xi \st \xi\in  \mathfrak{B}_3\rbrace$, which in turn fixes the action of $\Delta_a$ on $\mathbb{C}[\mathfrak{b}(\mathcal{A})]_3$. For example: 
	\begin{equation}\label{eq:delta_a_direct}
		\def\arraystretch{1.3}
		\begin{array}{rlcl}
			&A_3\, \bar{K}_{[\bb{p}]^3}  = 3\,\bar{K}_{[\bb{ns}][\bb{p}]}\,, & \multirow{1}{*}{$\Rightarrow$} & \Delta_a\big([\bb{p}]^3\big) =  3\,[\bb{ns}][\bb{p}]\,,\\
		\end{array}
	\end{equation}
	Similarly the action of $T_3$ fixes $\Delta_t$:
	\begin{equation}\label{eq:delta_t_direct}
		\def\arraystretch{1.3}
		\begin{array}{rlcl}
			&T_3\, \bar{K}_{[\bb{p}]^3}  = 3\,\bar{K}_{[\bb{pp}][\bb{p}]}\,, & \multirow{1}{*}{$\Rightarrow$} & \Delta_t\big([\bb{p}]^3\big) =  3\,[\bb{pp}][\bb{p}]\,,\\
		\end{array}
	\end{equation}
	And from $\Delta_x=\Delta_t-\Delta_a$ one has:
	\begin{equation}\label{eq:delta_x_direct}
		\begin{aligned}
			&\Delta_x\big([\bb{p}]^3\big) =  3\left([\bb{pp}][\bb{p}]-[\bb{ns}][\bb{p}]\right)\,.
		\end{aligned}
	\end{equation}
\end{example}

Our next goal consists in describing the action of $\Delta_a$, $\Delta_t$ and $\Delta_x$ directly in terms of bracelets, which will allow us to treat \eqref{eq:A_T_X_oeprator_def} the other way around and to read off the action of the corresponding conjugacy classes $A_n$, $T_n$ and $X_n$ on $\cn$ without addressing to diagram computations.
\subsection{Derivation of bracelets} 

Consider the polynomial algebra generated by bracelets over the extended alphabet $\bar{\mathcal{A}} = \mathcal{A}\cup \dot{\mathcal{A}} \cup \ddot{\mathcal{A}}$, where $\dot{\mathcal{A}} = \{\db{s},\db{p}\}$, $\ddot{\mathcal{A}} = \{\dd{s}\}$. Introduce the (linear) derivation map $\partial$ which acts via Leibniz rule: on any monomial $[w_1]\dots [w_k]$ as
\begin{equation}
	\partial \big([w_1]\dots [w_p]\big) = \sum_{j=1}^{p} [w_1]\dots\partial\big([w_j]\big)\dots [w_p]\,,
\end{equation}
and on each bracelet $[w] = [\bb{a}_1\dots \bb{a}_{\ell}]$ ($\bb{a}_j\in \bar{\mathcal{A}}$ for all $j = 1,\dots,\ell$) as
\begin{equation}
	\partial[\bb{a}_1\dots \bb{a}_{\ell}] = \sum_{j=1}^{\ell} [\bb{a}_1\dots \partial (\bb{a}_j)\dots \bb{a}_{\ell}]\,.
\end{equation}
To fix $\partial$, we set
\begin{equation}
	\def\arraystretch{1.4}
	\begin{array}{l}
		\partial(\bb{s}) = \db{s}\,,\;\;\partial(\bb{p}) = \db{p}\,,\;\; \partial(\db{s}) = \dd{s}\,,\quad \text{and}\quad  \partial(\bb{n}) = \partial(\db{p}) = \partial(\dd{s}) = \partial(1) = 0\\
		\text{(any bracelet where $0$ occurs is put to $0$)}\,.
	\end{array}
\end{equation}
Define the set $\mathfrak{b}(\bar{\mathcal{A}})\supset \mathfrak{b}(\mathcal{A})$ by extending $\mathfrak{b}(\mathcal{A})$ by all possibilities to substitute the undotted letters $\bb{s}$, $\bb{p}$ at some positions by their dotted counterparts in $\bar{\mathcal{A}}$. For example, the bracelet $[\bb{nsp}]$ gives rise to the following bracelets $[\bb{n}\db{s}\bb{p}], [\bb{n}\bb{s}\db{p}], [\bb{n}\db{s}\db{p}], [\bb{n}\dd{s}\bb{p}], [\bb{n}\dd{s}\db{p}]\in \mathfrak{b}(\bar{\mathcal{A}})$. As a result, $\mathbb{C}[\mathfrak{b}(\bar{\mathcal{A}})]$ contains all images of $\mathbb{C}[\mathfrak{b}(\mathcal{A})]$ upon consecutive application of $\partial$. The algebra $\mathbb{C}[\mathfrak{b}(\bar{\mathcal{A}})]$ is bi-graded: 
\begin{equation}\label{eq:bracelets_bi-grading}
	\mathbb{C}[\mathfrak{b}(\bar{\mathcal{A}})] = \bigoplus_{p\geqslant 0}\bigoplus_{q\leqslant p} \mathbb{C}[\mathfrak{b}(\bar{\mathcal{A}})]^{(q)}_{p}\,,
\end{equation}
where a monomial $[w_1] \dots [w_r]\in \mathbb{C}[\mathfrak{b}(\bar{\mathcal{A}})]^{(q)}_{p}$ has the total length $p$ ({\it i.e.} $\vert w_1\vert+\dots+\vert w_r\vert = p$) and carries the total amount $q$ of dots above the letters. For small $q$ the degree $(q)$ will be indicated by $q$ times the symbol $\prime$, and $\mathbb{C}[\mathfrak{b}(\bar{\mathcal{A}})]^{(0)}_{n} = \mathbb{C}[\mathfrak{b}(\mathcal{A})]_{n}$. In the sequel, omitting one of the bi-degrees of a component in \eqref{eq:bracelets_bi-grading} will imply the direct sum over all possible values of the omitted component.  For example, $[\db{p}], [\bb{n}\db{s}], [\bb{n}\db{s}][\bb{p}]\in \mathbb{C}[\mathfrak{b}(\bar{\mathcal{A})}]^{\prime}$, $[\bb{n}\dd{s}], [\bb{n}\dd{s}\bb{p}], [\bb{n}\db{s}][\db{p}\bb{p}]\in \mathbb{C}[\mathfrak{b}(\bar{\mathcal{A}})]^{\prime\prime}$ and $[\bb{n}\db{s}\bb{p}], [\bb{n}\dd{s}\bb{p}], [\bb{n}\db{s}][\db{p}] \in \mathbb{C}[\mathfrak{b}(\bar{\mathcal{A}})]_3$. The next lemma is a simple consequence of the definition of $\partial$ and the structure of $\mathfrak{b}(\bar{\mathcal{A}})$.

\subsection{Contraction operations} 

To introduce the final ingredient for the construction of $\Delta_a$ and $\Delta_t$, we consider the following $\mathbb{C}[\mathfrak{b}(\mathcal{A})]$-linear operations:
\begin{equation}
	\tau_{a} : \mathbb{C}[\mathfrak{b}(\bar{\mathcal{A}})]^{\prime\prime} \to \mathbb{C}[\mathfrak{b}(\mathcal{A})]\,,\hspace{1cm}\tau_{t} : \mathbb{C}[\mathfrak{b}(\bar{\mathcal{A}})]^{\prime\prime} \to \mathbb{C}[\mathfrak{b}(\mathcal{A})],
\end{equation}
which is defined via a set of rules given below. To formulate them, we accept a number of notations: {\it i)} we will write, for example, $[\bb{a}v]$ or $[\bb{a}u\bb{b}v]$ to specify particular letters $\bb{a},\bb{b}\in \bar{\mathcal{A}}$ in a bracelet, with the subwords $u,v$ being either empty or containing only letters from $\mathcal{A}$, {\it ii)} we will write $\vert w\vert_{\bb{a}}$ for the number of occurrences of the letter $\bb{a}$ in $w$ {\it iii)} we will say that $w$ is {\it fit} if it is either empty or $[w]\in\mathfrak{b}(\mathcal{A})$, and if each occurrence of $\bb{s}$ (if any) is followed by an occurrence of $\bb{n}$ at some position on the right. In the following expressions, the (sub)words $\db{p}u,\db{p}v$ on the left-hand-sides are assumed to have $u,v$ fit unless else is specified.
\vspace{-0.4cm}
\paragraph{Contraction operation for class $A_n$.} 
\begin{align}
&\tau_{a} : [\dd{s}u] \mapsto 2\Dim\,[\bb{s}u] \,,  \hfill\label{eq:trace_rulea_1} \\[8pt]  
&\tau_{a} : [\db{s}u\db{s}v] \mapsto 2\,\big([\bb{s}u\bb{s}\,I(v)] + [\bb{s}u][\bb{s}v]\big)\,, \hspace{1.4cm}  &&\tau_{a} : [\db{s}u][\db{s}v] \mapsto 2\,\big([\bb{s}u\bb{s}v] + [\bb{s}u\bb{s}\,I(v)]\big)\,,\label{eq:trace_rulea_2}\\[8pt]
&\tau_{a} : [\db{p}u\db{s}v] \mapsto [\bb{p}u\bb{s}\,I(v)] + [\bb{p}u][\bb{s}v]\,,  &&\tau_{a} : [\db{p}u][\db{s}v] \mapsto [\bb{p}u\bb{s}v] + [\bb{p}u\bb{s}\,I(v)]\,,\label{eq:trace_rulea_3}\\[8pt]
&\tau_{a} : [\db{p}u\db{p}v] \mapsto \left\{
\begin{array}{ll}
	[\bb{n}u\bb{s}I(v)]\,, & \text{if $u,v$ are fit}\, \\
	\left[\bb{n}u][\bb{s}v\right]\,, & \text{if $|u|_{\bb{s}} > |u|_{\bb{n}}$}
\end{array}
\right.\,,  &&\tau_{a} : [\db{p}u][\db{p}v] \mapsto [\bb{n}u\bb{s}I(v)]\,.\label{eq:trace_rulea_4}
\end{align}
\paragraph{Contraction operation for class $T_n$.} 
\begin{align}
	&\tau_{t} : [\dd{s}u] \mapsto 2\,[\bb{s}u] \,,  \hfill  \label{eq:trace_rulet_1}\\[8pt]
	&\tau_{t} : [\db{s}u\db{s}v] \mapsto \tau_{a}([\db{s}u\db{s}v]) \,, \hspace{1.4cm}  &&\tau_{t} : [\db{s}u][\db{s}v] \mapsto \tau_{a}([\db{s}u][\db{s}v]),\label{eq:trace_rulet_2}\\[8pt]
	&\tau_{t} : [\db{p}u\db{s}v] \mapsto \tau_{a}([\db{p}u\db{s}v])\,,  &&\tau_{t} : [\db{p}u][\db{s}v] \mapsto  \tau_{a}([\db{p}u][\db{s}v])\,,\label{eq:trace_rulet_3}\\[8pt]
	&\tau_{t} : [\db{p}u\db{p}v] \mapsto \left\{
	\begin{array}{ll}
			\left[\bb{p}u][\bb{p}v\right]\,, & \text{if $u,v$ are fit}\, \\
				\left[\bb{p}u\bb{p}I(v)\right]\,, & \text{if $|u|_{\bb{s}} > |u|_{\bb{n}}$}
		\end{array}
	\right.\,,  &&\tau_{t} : [\db{p}u][\db{p}v] \mapsto [\bb{p}u\bb{p}v]\,.\label{eq:trace_rulet_4}
\end{align}
\vskip 4pt

\begin{lemma}\label{lem:trace_rules}
	The rules \eqref{eq:trace_rulea_1}-\eqref{eq:trace_rulea_4} and \eqref{eq:trace_rulet_1}-\eqref{eq:trace_rulet_4} are correct and define $\tau_{a}$ and $\tau_{t}$ unambiguously. Namely, {\it i)} the monomials entering the left-hand-sides of the rules form a $\mathbb{C}[\mathfrak{b}(\mathcal{A})]$-basis\footnote{In other words, any element in $\mathbb{C}[\mathfrak{b}(\bar{\mathcal{A}})]^{\prime\prime}$ is given by a unique combination of the left-hand-sides of \eqref{eq:trace_rulea_1}-\eqref{eq:trace_rulea_4}  with the coefficients in $\mathbb{C}[\mathfrak{b}(\mathcal{A})]$.} in $\mathbb{C}[\mathfrak{b}(\bar{\mathcal{A}})]^{\prime\prime}$,
	{\it ii)} if there are several representatives of a kind, the rules nevertheless lead to the same result.
\end{lemma}
See appendix \ref{app:proof_lem_rules} for the proof.
\begin{remark} The contraction operation for the action of $T_n$ on the conjugacy class sum of $\C\sn$ is defined by the rules
\begin{equation}\label{eq:trace_ruletn_sn}
\tau_{t} : [\db{p}u\db{p}v] \mapsto [\bb{p}u][\bb{p}v]\,, \hspace{1cm} \tau_{t} : [\db{p}u][\db{p}v] \mapsto [\bb{p}u\bb{p}v]\,,
\end{equation}
where the first rule correspond to the contribution of transpositions that split two cycles and the second rule to the contribution of transpositions that joins two cycles.
\end{remark}
\vspace{-0.4cm}
\paragraph{Contraction operation for class $X_n$.} Because $\Delta_x=\Delta_t-\Delta_a$, we define the contraction rules for class $X_n$ as $\tau_{x}:=\tau_{t}-\tau_{a}$. They are given by
\begin{align}
	&\tau_{x} : [\dd{s}u] \mapsto 2\left(1-\Dim\right)\,[\bb{s}u] \,,  \hfill  \label{eq:trace_rulex_1}\\[8pt]
	&\tau_{x} : [\db{s}u\db{s}v] \mapsto 0\,, \hspace{1.4cm}  &&\tau_{x} : [\db{s}u][\db{s}v] \mapsto 0\,,\label{eq:trace_rulex_2}\\[8pt]
	&\tau_{x} : [\db{p}u\db{s}v] \mapsto 0\,,  &&\tau_{x} : [\db{p}u][\db{s}v] \mapsto 0\,,\label{eq:trace_rulex_3}\\[8pt]
	&\tau_{x} : [\db{p}u\db{p}v] \mapsto \left\{
	\begin{array}{ll}
		\left[\bb{p}u][\bb{p} v\right]-[\bb{n}u\bb{s}I(v)]\, & \text{if $u,v$ are fit}\,  \\
		\left[\bb{p}u \bb{p}I(v)\right]-[\bb{n}u][\bb{s}v]\, & \text{if $|u|_{\bb{s}} > |u|_{\bb{n}}$}
	\end{array}
	\right.,  &&\tau_{x} : [\db{p}u][\db{p}v] \mapsto [\bb{p}u\bb{p}v]-[\bb{n}u\bb{s}I(v)]\,.\label{eq:trace_rulex_4}
\end{align}

We are now in position to express the operators $\Delta_a$, $\Delta_t$ and $\Delta_x$ defined in \eqref{eq:A_T_X_oeprator_def} as a second-order differential operator on bracelets (see Appendix \ref{app:proof_Laplace} for proof).
\begin{theorem}\label{thm:Laplace}
	The operators $\Delta_a$\,, $\Delta_t$ and $\Delta_x$ defined in \eqref{eq:A_T_X_oeprator_def} are given by
	\begin{equation}\label{eq:Laplace}
		\Delta_a = \frac{1}{2}\,\tau_{a}\circ \partial^{2}\,,\hspace{1cm}\Delta_t = \frac{1}{2}\,\tau_t\circ \partial^{2}\,,\hspace{1cm}\Delta_x = \frac{1}{2}\,\tau_x\circ \partial^{2}\,.
	\end{equation}
\end{theorem}

For example, consider the algebra $\mathcal{C}_{3}\subset B_{3}$. As a matter of consistency let us check that $A_3  \bar{K}_{[\bb{p}]^3} = 6\,A_3 = 3\, \bar{K}_{[\bb{ns}][\bb{p}]}$. Indeed,
\begin{equation}\label{eq:Laplace_exp_Br_3}
	\partial^2 [\bb{p}]^3 = 6\,[\db{p}]^2[\bb{p}]\quad \Rightarrow \quad \Delta_a\left([\bb{p}]^3\right) = 3\,[\bb{n}\bb{s}][\bb{p}]\,.
\end{equation}

\vskip 2 pt

\section{Applications}\label{sec:applications}

We sum up the proposed approach by constructing the expanded forms of central Young idempotents of examples \ref{ex:CYI_n3} - \ref{ex:CYI_n4}, and of the traceless projector of examples \ref{ex:traceless_proj_n3} - \ref{ex:traceless_proj_n4}. 

\subsection{The central Young idempotents for $n=3$ and $n=4$}\label{sec:expansion_central_Youngs}

\paragraph{Central Young idempotents for $n=3$.} For convenience we reproduce the factorized expressions of the central Young idempotents for $n=3$ obtained in example \ref{ex:CYI_n3}:
\begin{equation}\label{eq:factorized_Z3_chap5}
	Z^{\Yboxdim{3pt}\yng(3)}=\frac{1}{18}\, T_3\,\left(3+T_3\right)\,,\hspace{0.4cm} Z^{\Yboxdim{3pt}\yng(2,1)}=\frac{1}{9}\left(3-T_3\right)\left( 3+T_3 \right)\,,\hspace{0.4cm}
	Z^{\Yboxdim{3pt}\yng(1,1,1)}=-\frac{1}{18}T_3\, \left(3-T_3\right)\,.
\end{equation}
In order to expand these formulas it is sufficient to determine the action of $T_3=K_{[\bb{pp}][\bb{p}]}=K_{(2,1)}$ on itself.\medskip 

The second derivative of $[\bb{pp}][\bb{p}]$ is
\begin{equation}
\partial^2 [\bb{pp}][\bb{p}] = 2\, [\bb{\dot{p}\dot{p}}][\bb{p}] + 4\, [\bb{\dot{p}p}][\bb{\dot{p}}]\,.
\end{equation}
Applying the contraction rules $\tau_t$ restricted to the product of classes of $\C\sn$ \eqref{eq:trace_ruletn_sn}, and Theorem \ref{thm:Laplace} yields
\begin{equation}\label{eq:trace_ruletn_sn2_app}
	\tau_{t} : [\db{p}u\db{p}v] \mapsto [\bb{p}u][\bb{p}v]\,, \hspace{1cm} \tau_{t} : [\db{p}u][\db{p}v] \mapsto [\bb{p}u\bb{p}v]\,,
\end{equation}
and  yields:
\begin{equation}
	T_3\,\bar{K}_{[\bb{pp}][\bb{p}]}= \, \bar{K}_{[\bb{p}]^3} + 2\, \bar{K}_{[\bb{p}\bb{p}\bb{p}]}\,.
\end{equation}
\vskip 6pt
From, 
\begin{equation}
\stab([\bb{pp}][\bb{p}])=2\,,\hspace{0.5cm}
\stab([\bb{p}]^3)=6\,,\hspace{0.5cm}
\stab([\bb{p}\bb{p}\bb{p}])=3\,.
\end{equation}
one reads off the action of $T_3$ on itself:
\begin{equation}\label{eq:rulesT3app}
(T_3)^2=T_3 K_{[\bb{pp}][\bb{p}]}=3\, K_{[\bb{p}]^3} +\, 3\, K_{[\bb{p}\bb{p}\bb{p}]}\,.
\end{equation}
Expanding equations \eqref{eq:factorized_Z3_chap5} with the rule \eqref{eq:rulesT3app} leads to 
\begin{equation}\label{eq:centralYoung_3}
\begin{array}{ll}
Z^{\Yboxdim{3pt}\yng(3)}=\dfrac{1}{6}\left(K_{[\bb{p}]^3}+K_{[\bb{p}\bb{p}][\bb{p}]}+K_{[\bb{p}\bb{p}\bb{p}]}\right)\,,\hspace{1cm}
&Z^{\Yboxdim{3pt}\yng(1,1,1)}=\dfrac{1}{6}\left(K_{[\bb{p}]^3}-K_{[\bb{p}\bb{p}][\bb{p}]}+K_{[\bb{p}\bb{p}\bb{p}]}\right)\,,\\[6pt]
Z^{\Yboxdim{3pt}\yng(2,1)}=\dfrac{1}{3}\left(2\, K_{[\bb{p}]^3}-K_{[\bb{p}\bb{p}\bb{p}]}\right)\,.
\end{array}
\end{equation}

\paragraph{Central Young idempotents for $n=4$.} For convenience we reproduce the factorized expressions of the central Young idempotents for $n=4$ obtained in example \ref{ex:CYI_n4}:
\begin{equation}\label{eq:factorized_Z4_chap5}
	\begin{array}{ll}
		Z^{\Yboxdim{3pt}\yng(4)}=\dfrac{1}{2304}\, T_4\left(6+T_4\right)\left(2+T_4\right)\left(T_4-2\right)\,,
		&Z^{\Yboxdim{3pt}\yng(1,1,1,1)}=-\dfrac{1}{2304}\, T_4\left(6-T_4\right)\left(2+T_4\right)\left(T_4-2\right),\\[10pt]
		Z^{\Yboxdim{3pt}\yng(3,1)}=\dfrac{1}{256}\, T_4\left(6-T_4\right)\left(2+T_4\right)\left(6+T_4\right)\,,
		&Z^{\Yboxdim{3pt}\yng(2,1,1)}=-\dfrac{1}{256}\, T_4\left(6-T_4\right)\left(2-T_4\right)\left(6+T_4\right)\,,\\[10pt]
		Z^{\Yboxdim{3pt}\yng(2,2)}=\dfrac{1}{144}\, \left(6-T_4\right)\left(2-T_4\right)\left(2+T_4\right)\left(6+T_4\right)\,.
	\end{array}
\end{equation} 
In order to expand these formulas one needs to determine the action of $T_4=K_{[\bb{pp}][\bb{p}]^2}=K_{(2,1^2)}$ on all other classes.\medskip 

The second derivative of the unary bracelets in $\mathcal{B}_4$ are:
\begin{equation}
\begin{array}{ll}
\partial^2 \left(\,[\bb{pp}][\bb{p}]^2\,\right) = 2\, [\bb{\dot{p}\dot{p}}][\bb{p}][\bb{p}] + 8\, [\bb{\dot{p}p}][\bb{\dot{p}}][\bb{p}]+ 8\, [\bb{pp}][\bb{\dot{p}}][\bb{\dot{p}}]\,,
&\partial^2 \left(\,[\bb{p}\bb{p}]^2\,\right) = 4\, [\bb{\dot{p}\dot{p}p}][\bb{p}] + 8\, [\bb{\dot{p}pp}][\bb{\dot{p}}]\,,\\[8pt]
\partial^2 \left(\,[\bb{p}\bb{p}\bb{p}][\bb{p}]\right) = 6\left([\bb{\dot{p}\dot{p}}\bb{p}][\bb{p}]+[\bb{\dot{p}p}\bb{p}][\bb{\dot{p}}]\right)\,,
&\partial^2 \left(\,[\bb{p}\bb{p}\bb{p}\bb{p}]\right) = 8\, [\bb{\dot{p}\dot{p}pp}]+ 4\, [\bb{\dot{p}p\dot{p}p}]\,.
\end{array}
\end{equation}
Applying the contraction rules $\tau_t$ \eqref{eq:trace_ruletn_sn2_app} and Theorem \ref{thm:Laplace} yields
\begin{equation}
	\begin{array}{ll}
	T_4\,\bar{K}_{[\bb{pp}][\bb{p}]^2}= \bar{K}_{[\bb{p}]^4} + \bar{K}_{[\bb{pp}]^2}+ 4\, \bar{K}_{[\bb{p}\bb{p}\bb{p}][\bb{p}]}\,,
	&T_4\,\bar{K}_{[\bb{pp}]^2}= 2 \bar{K}_{[\bb{pp}][\bb{p}]^2}+ 4\,\bar{K}_{[\bb{p}\bb{p}\bb{p}\bb{p}]}\,,\\[6pt]
	T_4\,\bar{K}_{[\bb{ppp}][\bb{p}]}= 3\bar{K}_{[\bb{pp}][\bb{p}]^2}+ 3\,\bar{K}_{[\bb{p}\bb{p}\bb{p}\bb{p}]}\,,
	&T_4\,\bar{K}_{[\bb{pppp}]}= 2 \bar{K}_{[\bb{pp}]^2}+ 4\,\bar{K}_{[\bb{p}\bb{p}\bb{p}][\bb{p}]}\,.
	\end{array}
\end{equation}
\vskip 6pt
\noindent The stability index of the classes are 
\begin{equation}
\stab([\bb{p}]^4)=24\,,\hspace{0.3cm}
\stab([\bb{pp}][\bb{p}]^2)=4\,,\hspace{0.3cm}
\stab([\bb{pp}]^2)=8\,,\hspace{0.3cm}
\stab([\bb{ppp}][\bb{p}])=3\,,\hspace{0.3cm}
\stab([\bb{pppp}])=4\,,
\end{equation}
so that the action of $T_4$ and the normalized conjugacy class sums of $\Bn{4}$ is
\begin{equation}\label{eq:rulesT4}
	\begin{array}{ll}
		T_4\,K_{[\bb{pp}][\bb{p}]^2}= 6\, K_{[\bb{p}]^4} + 2\, K_{[\bb{pp}]^2}+ 3\, K_{[\bb{p}\bb{p}\bb{p}][\bb{p}]}\,,
		&T_4\,K_{[\bb{pp}]^2}= K_{[\bb{pp}][\bb{p}]^2}+ 2\,K_{[\bb{p}\bb{p}\bb{p}\bb{p}]}\,,\\[6pt]
		T_4\,K_{[\bb{ppp}][\bb{p}]}= 4\, K_{[\bb{pp}][\bb{p}]^2}+ 4\, K_{[\bb{p}\bb{p}\bb{p}\bb{p}]}\,,
		&T_4\,K_{[\bb{pppp}]}= 4\, K_{[\bb{pp}]^2}+ 3\, K_{[\bb{p}\bb{p}\bb{p}][\bb{p}]}\,.
	\end{array}
\end{equation}
Expanding equations \eqref{eq:factorized_Z4_chap5} sequentially with the rules $\eqref{eq:rulesT4}$ leads to 
\begin{equation}\label{eq:centralYoung_4}
	\begin{aligned}
		Z^{\Yboxdim{3pt}\yng(4)}&=\frac{1}{24}\left(K_{[\bb{p}]^4}+K_{[\bb{p}\bb{p}][\bb{p}]^2}+K_{[\bb{p}\bb{p}]^2}+K_{[\bb{p}\bb{p}\bb{p}][\bb{p}]}+K_{[\bb{p}\bb{p}\bb{p}\bb{p}]}\right)\,,\\[6pt]
		Z^{\Yboxdim{3pt}\yng(3,1)}&=\frac{1}{8}\left(3\, K_{[\bb{p}]^4}+K_{[\bb{p}\bb{p}][\bb{p}]^2}-K_{[\bb{p}\bb{p}]^2}-K_{[\bb{p}\bb{p}\bb{p}\bb{p}]}\right)\,,\\[6pt]
		Z^{\Yboxdim{3pt}\yng(2,2)}&=\frac{1}{12}\left(2\, K_{[\bb{p}]^4}+2\, K_{[\bb{p}\bb{p}]^2}-K_{[\bb{p}\bb{p}\bb{p}][\bb{p}]}\right)\,,\\[6pt]
		Z^{\Yboxdim{3pt}\yng(2,1,1)}&=\frac{1}{8}\left(3\, K_{[\bb{p}]^4}-K_{[\bb{p}\bb{p}][\bb{p}]^2}-K_{[\bb{p}\bb{p}]^2}+K_{[\bb{p}\bb{p}\bb{p}\bb{p}]}\right)\,,\\[6pt]
		Z^{\Yboxdim{3pt}\yng(1,1,1,1)}&=\frac{1}{24}\left(K_{[\bb{p}]^4}-K_{[\bb{p}\bb{p}][\bb{p}]^2}+K_{[\bb{p}\bb{p}]^2}+K_{[\bb{p}\bb{p}\bb{p}][\bb{p}]}-K_{[\bb{p}\bb{p}\bb{p}\bb{p}]}\right)\,.
	\end{aligned}
\end{equation}

\subsection{The traceless projectors for $n=3$ and $n=4$\,}\label{sec:expansion_traceless_projector}

Throughout this section we assume $\Dim\geqslant n-1$, such that the algebra $B_{n}$ is semisimple.
\vspace{-0.4cm}
\paragraph{Traceless projector for $n=3$.} 
The eigenvalues of $A_3$ are given in the example \ref{ex:traceless_proj_n3}, which lead to the following polynomial expression in $A_3$ for the traceless projector $P_3$:\medskip

\begin{equation}\label{eq:traceless_P3_Poly}
	P_3=\left(1-\frac{A_3}{\Dim-1}\,\right)\left(1-\frac{A_3}{\Dim+2}\,\right)\,.
\end{equation}
\vskip 6pt

In order to expand this formula it is sufficient to determine the action of $A_3=K_{[\bb{ns}][\bb{p}]}$ on itself.\medskip 

The second derivative of $[\bb{ns}][\bb{p}]$ is
\begin{equation}
	\partial^2 [\bb{ns}][\bb{p}] = 2\, [\bb{n}\db{\bb{s}}][\dot{\bb{p}}] + 4\, [\bb{n}\dd{\bb{s}}][\bb{p}]\,.
\end{equation}
Applying the contraction rules $\tau_a$ \eqref{eq:trace_rulea_1}-\eqref{eq:trace_rulea_3} and Theorem \ref{thm:Laplace} yields:
\begin{equation}
	A_3\,\bar{K}_{[\bb{ns}][\bb{p}]}= \Dim \, \bar{K}_{[\bb{ns}][\bb{p}]} + 2\,\bar{K}_{[\bb{nsp}]}\,.
\end{equation}
\vskip 6pt
From, 
\begin{equation}
	\stab([\bb{ns}][\bb{p}])=2\,,\hspace{0.5cm}
	\stab([\bb{n}\bb{s}\bb{p}])=1\,.
\end{equation}
one reads off the action of $A_3$ on itself:
\begin{equation}\label{eq:rulesA3app}
	(A_3)^2=A_3 K_{[\bb{ns}][\bb{p}]}=\,\Dim\, K_{[\bb{ns}][\bb{p}]} +\, K_{[\bb{n}\bb{s}\bb{p}]}\,.
\end{equation}
Expanding \eqref{eq:traceless_P3_Poly} with the rule $\eqref{eq:rulesA3app}$ leads to 
\begin{equation}\label{eq:traceless_3_class}
P_3=K_{[\bb{p}][\bb{p}][\bb{p}]}-\frac{(\Dim +1)}{(\Dim-1)(\Dim+2)}\, K_{[\bb{ns}][\bb{p}]}+\frac{1}{(\Dim-1)(\Dim+2)}\, K_{[\bb{ns}][\bb{p}]}\,,
\end{equation}
which correspond indeed to the diagrammatic expression \eqref{eq:projector_Br_3}.
\paragraph{Traceless projector for $n=4$.} 
The eigenvalues of $A_4$ were given in the example \ref{ex:traceless_proj_n4}, which lead to the following polynomial expression in $A_4$ for the traceless projector $P_4$:\medskip

\begin{equation}\label{eq:traceless_P4_Poly}
	\hspace{-0.5cm}	P_4=\scalemath{0.98}{\left(1-\frac{A_4}{\Dim+4}\,\right)\left(1-\frac{A_4}{2(\Dim+2)}\,\right)\left(1-\frac{A_4}{\Dim}\,\right)\left(1-\frac{A_4}{\Dim+2}\,\right) \left(1-\frac{A_4}{\Dim-2}\,\right)\left(1-\frac{A_4}{2(\Dim-1)}\,\right)\,.}
\end{equation}
\vskip 6pt

Here, $P_4$ is polynomial of order $6$ in $A_n$, and in order to expand it we need to determine the rules describing the action of $A_4$ on all conjugacy class sums with at least one arc. The second derivative of the elements in $\mathbb{C}[\mathfrak{b}(\mathcal{A})]$ relevant for the construction of $P_4$ are: 
\begin{equation*}
\begin{aligned}
&\partial^2 \left[\bb{n}\bb{s}\right]\left[\bb{p}\right]\left[\bb{p}\right] = \left[\bb{n}\dd{s}\right]\left[\bb{p}\right]\left[\bb{p}\right] + 4\,\left[\bb{n}\db{s}\right]\left[\db{p}\right]\left[\bb{p}\right]+ 2\,\left[\bb{n}\bb{s}\right]\left[\db{p}\right]\left[\db{p}\right]\,,\\
&\partial^2\left[\bb{n}\bb{s}\right]\left[\bb{p}\bb{p}\right] = \left[\bb{n}\dd{s}\right]\left[\bb{p}\bb{p}\right] + 4\,\left[\bb{n}\db{s}\right]\left[\db{p}\bb{p}\right]+2\left[\bb{n}\bb{s}\right]\left[\db{p}\db{p}\right]\,,\\
&\partial^2\left[\bb{n}\bb{s}\bb{p}\right]\left[\bb{p}\right] = \left[\bb{n}\dd{s}\bb{p}\right]\left[\bb{p}\right] + 2\,(\,\left[\bb{n}\db{s}\db{p}\right]\left[\bb{p}\right]+\left[\bb{n}\db{s}\bb{p}\right]\left[\db{p}\right]+\left[\bb{n}\bb{s}\db{p}\right]\left[\db{p}\right]\,) \,,\\ 
&\partial^2\left[\bb{n}\bb{s}\bb{p}\bb{p}\right] = \left[\bb{n}\dd{s}\bb{p}\bb{p}\right] + 2\,(\left[\bb{n}\db{s}\db{p}\bb{p}\right]+\left[\bb{n}\db{s}\bb{p}\db{p}\right]+\left[\bb{n}\bb{s}\db{p}\db{p}\right])\,,\\
&\partial^2\left[\bb{n}\bb{p}\bb{s}\bb{p}\right] = \left[\bb{n}\bb{p}\dd{s}\bb{p}\right] + 4\, \left[\bb{n}\bb{p}\db{s}\db{p}\right] + 2\,\left[\bb{n}\db{p}\bb{s}\db{p}\right]\,,\\
&\partial^2\left[\bb{n}\bb{s}\right]\left[\bb{n}\bb{s}\right] =2\,(\left[\bb{n}\dd{s}\right]\left[\bb{n}\bb{s}\right]+\left[\bb{n}\db{s}\right]\left[\bb{n}\db{s}\right])\,,\\
&\partial^2\left[\bb{n}\bb{s}\bb{n}\bb{s}\right] = 2\,(\left[\bb{n}\dd{s}\bb{n}\bb{s}\right] +\left[\bb{n}\db{s}\bb{n}\db{s}\right])\,.\\
\end{aligned}
\end{equation*}
Applying $\tau_a$ \eqref{eq:trace_rulea_1}-\eqref{eq:trace_rulea_4} on the previous expressions and using Theorem \ref{thm:Laplace} yields: 

\begin{equation}\label{eq:Laplace_Br_4}
	\begin{aligned}
		& A_n \, \bar{K}_{\left[\bb{n}\bb{s}\right]\left[\bb{p}\right]\left[\bb{p}\right]} = \Dim\,\bar{K}_{\left[\bb{n}\bb{s}\right]\left[\bb{p}\right]\left[\bb{p}\right]} + 4\,\bar{K}_{\left[\bb{n}\bb{s}\bb{p}\right]\left[\bb{p}\right]}+\bar{K}_{\left[\bb{n}\bb{s}\right]\left[\bb{n}\bb{s}\right]}\,,  \\
		&A_n \, \bar{K}_{\left[\bb{n}\bb{s}\right]\left[\bb{p}\bb{p}\right]} =\Dim\,\bar{K}_{\left[\bb{n}\bb{s}\right]\left[\bb{p}\bb{p}\right]}+ 4\,\bar{K}_{\left[\bb{n}\bb{s}\bb{p}\bb{p}\right]}+\bar{K}_{\left[\bb{n}\bb{s}\right]\left[\bb{n}\bb{s}\right]}\,,\\
		&A_n \, \bar{K}_{\left[\bb{n}\bb{s}\bb{p}\right]\left[\bb{p}\right]} =\left( \Dim+1 \right)\, \bar{K}_{\left[\bb{n}\bb{s}\bb{p}\right]\left[\bb{p}\right]} + \bar{K}_{\left[\bb{n}\bb{s}\right]\left[\bb{p}\right]\left[\bb{p}\right]}+\bar{K}_{\left[\bb{n}\bb{s}\bb{p}\bb{p}\right]}+\bar{K}_{\left[\bb{n}\bb{p}\bb{s}\bb{p}\right]}+\bar{K}_{\left[\bb{n}\bb{s}\bb{n}\bb{s}\right]}\,,\\
		&A_n \, \bar{K}_{\left[\bb{n}\bb{s}\bb{p}\bb{p}\right]} =\left(\Dim+1\right)\, \bar{K}_{\left[\bb{n}\bb{s}\bb{p}\bb{p}\right]} +\bar{K}_{\left[\bb{n}\bb{s}\right]\left[\bb{p}\bb{p}\right]} +\bar{K}_{\left[\bb{n}\bb{s}\bb{p}\right]\left[\bb{p}\right]} +\bar{K}_{\left[\bb{n}\bb{p}\bb{s}\bb{p}\right]}+\bar{K}_{\left[\bb{n}\bb{s}\bb{n}\bb{s}\right]}\,,\\
		&A_n \, \bar{K}_{\left[\bb{n}\bb{p}\bb{s}\bb{p}\right]} =\Dim\, \bar{K}_{\left[\bb{n}\bb{p}\bb{s}\bb{p}\right]} +2\, \bar{K}_{\left[\bb{n}\bb{s}\bb{p}\right]\left[\bb{p}\right]} + 2\,\bar{K}_{\left[\bb{n}\bb{s}\bb{p}\bb{p}\right]}+\bar{K}_{\left[\bb{n}\bb{s}\right]\left[\bb{n}\bb{s}\right]}\,,\\
		&A_n \, \bar{K}_{\left[\bb{n}\bb{s}\right]\left[\bb{n}\bb{s}\right]} = \, 2\Dim\, \bar{K}_{\left[\bb{n}\bb{s}\right]\left[\bb{n}\bb{s}\right]} + 4\,\bar{K}_{\left[\bb{n}\bb{s}\bb{n}\bb{s}\right]}\,,\\
		&A_n \, \bar{K}_{\left[\bb{n}\bb{s}\bb{n}\bb{s}\right]} = \,2(\Dim+1)\,\bar{K}_{\left[\bb{n}\bb{s}\bb{n}\bb{s}\right]}+ 2\,\bar{K}_{\left[\bb{n}\bb{s}\right]\left[\bb{n}\bb{s}\right]}\,.
	\end{aligned}
\end{equation}

Before expanding the factorized form of the projector we need to express the above expression in the normalized basis $K_{\zeta}=\dfrac{1}{\mathrm{st}(\zeta)}\,\bar{K}_{\zeta}$.
The stability index of each basis element of $\mathbb{C}[\mathfrak{b}(\mathcal{A})]_4$ reads:
\begin{equation*}
	\def\arraystretch{1.4}
	\begin{array}{llllc}
		\mathrm{st}(\left[\bb{n}\bb{s}\right]\left[\bb{p}\right]\left[\bb{p}\right])=4\,, & \mathrm{st}(\left[\bb{n}\bb{s}\right]\left[\bb{p}\bb{p}\right])=4\,,  & 
		\mathrm{st}(\left[\bb{n}\bb{s}\bb{p}\right]\left[\bb{p}\right])=1\,, &
		\mathrm{st}(\left[\bb{n}\bb{s}\bb{p}\bb{p}\right])=1\,, \\ \mathrm{st}(\left[\bb{n}\bb{p}\bb{s}\bb{p}\right])=2\,, &
		\mathrm{st}(\left[\bb{n}\bb{s}\right]\left[\bb{n}\bb{s}\right])=8\,,&
		\mathrm{st}(\left[\bb{n}\bb{s}\bb{n}\bb{s}\right])=4\,. &\\
	\end{array}
\end{equation*}
Therefore, in the normalized basis the relations \eqref{eq:Laplace_Br_4} become: 
\begin{equation}\label{eq:A_action_normalized}
	\begin{aligned}
		& A_n \, K_{\left[\bb{n}\bb{s}\right]\left[\bb{p}\right]\left[\bb{p}\right]} = \Dim\,K_{\left[\bb{n}\bb{s}\right]\left[\bb{p}\right]\left[\bb{p}\right]} + K_{\left[\bb{n}\bb{s}\bb{p}\right]\left[\bb{p}\right]}+2 \, K_{\left[\bb{n}\bb{s}\right]\left[\bb{n}\bb{s}\right]}\,,  \\
		&A_n \, K_{\left[\bb{n}\bb{s}\right]\left[\bb{p}\bb{p}\right]} =\Dim\,K_{\left[\bb{n}\bb{s}\right]\left[\bb{p}\bb{p}\right]}+ K_{\left[\bb{n}\bb{s}\bb{p}\bb{p}\right]}+2\, K_{\left[\bb{n}\bb{s}\right]\left[\bb{n}\bb{s}\right]}\,,\\
		&A_n \, K_{\left[\bb{n}\bb{s}\bb{p}\right]\left[\bb{p}\right]} =\left( \Dim+1 \right)\, K_{\left[\bb{n}\bb{s}\bb{p}\right]\left[\bb{p}\right]} + 4\, K_{\left[\bb{n}\bb{s}\right]\left[\bb{p}\right]\left[\bb{p}\right]}+K_{\left[\bb{n}\bb{s}\bb{p}\bb{p}\right]}+2\, K_{\left[\bb{n}\bb{p}\bb{s}\bb{p}\right]}+4\, K_{\left[\bb{n}\bb{s}\bb{n}\bb{s}\right]}\,,\\
		&A_n \, K_{\left[\bb{n}\bb{s}\bb{p}\bb{p}\right]} =\left(\Dim+1\right)\, K_{\left[\bb{n}\bb{s}\bb{p}\bb{p}\right]} +4\,K_{\left[\bb{n}\bb{s}\right]\left[\bb{p}\bb{p}\right]} +K_{\left[\bb{n}\bb{s}\bb{p}\right]\left[\bb{p}\right]} +2\,K_{\left[\bb{n}\bb{p}\bb{s}\bb{p}\right]}+4\,K_{\left[\bb{n}\bb{s}\bb{n}\bb{s}\right]}\,,\\
		&A_n \, K_{\left[\bb{n}\bb{p}\bb{s}\bb{p}\right]} =\Dim\, K_{\left[\bb{n}\bb{p}\bb{s}\bb{p}\right]} + K_{\left[\bb{n}\bb{s}\bb{p}\right]\left[\bb{p}\right]} +K_{\left[\bb{n}\bb{s}\bb{p}\bb{p}\right]}+4\,K_{\left[\bb{n}\bb{s}\right]\left[\bb{n}\bb{s}\right]}\,,\\
		&A_n \, K_{\left[\bb{n}\bb{s}\right]\left[\bb{n}\bb{s}\right]} = \, 2\Dim\, K_{\left[\bb{n}\bb{s}\right]\left[\bb{n}\bb{s}\right]} + 2\,K_{\left[\bb{n}\bb{s}\bb{n}\bb{s}\right]}\,,\\
		&A_n \, K_{\left[\bb{n}\bb{s}\bb{n}\bb{s}\right]} = \,2(\Dim+1)\,K_{\left[\bb{n}\bb{s}\bb{n}\bb{s}\right]}+ 4\,K_{\left[\bb{n}\bb{s}\right]\left[\bb{n}\bb{s}\right]}\,.\\  
	\end{aligned}
\end{equation}
We have all the rules necessary to expand the factorized expression \eqref{eq:traceless_P4_Poly} of the traceless projector $P_4$. We omit calculations of the expansion and pass directly to the final result:
\begin{equation}\label{eq:projector4}
	P_4 = 1 + \sum_{\zeta\in \mathfrak{B}_4}\, k_\zeta \, K_{\zeta}\,,
\end{equation}
with,
\begin{equation*}
	\def\arraystretch{1.8}
	\begin{array}{rclrcl}
		k_{\left[\bb{n}\bb{s}\right]\left[\bb{p}\right]\left[\bb{p}\right]} & = & -\,\dfrac{\Dim^2(\Dim+4)-4}{(\Dim-2)\Dim(\Dim+2)(\Dim+4)}\,, &
		k_{\left[\bb{n}\bb{s}\right]\left[\bb{p}\bb{p}\right]} & = &\dfrac{4}{(\Dim-2)\Dim(\Dim+2)(\Dim+4)}\,, \\[14pt]
		k_{\left[\bb{n}\bb{s}\bb{p}\right]\left[\bb{p}\right]} & = & \dfrac{\Dim+3}{(\Dim-2)(\Dim+2)(\Dim+4)}\,, &
		k_{\left[\bb{n}\bb{s}\bb{p}\bb{p}\right]} & = & -\,\dfrac{1}{(\Dim-2)(\Dim+2)(\Dim+4)}\,, \\[14pt]
		k_{\left[\bb{n}\bb{p}\bb{s}\bb{p}\right]} & = & -\,\dfrac{2}{(\Dim-2)\Dim(\Dim+4)}\,, &
		k_{\left[\bb{n}\bb{s}\right]\left[\bb{n}\bb{s}\right]} & = & \dfrac{\Dim(\Dim+3)+6}{(\Dim-2)(\Dim-1)(\Dim+2)(\Dim+4)}\,, \\[14pt]
		\multicolumn{6}{c}{k_{\left[\bb{n}\bb{s}\bb{n}\bb{s}\right]}=-\,\dfrac{3\Dim+2}{(\Dim-2)(\Dim-1)(\Dim+2)(\Dim+4)}\,,}
	\end{array}
\end{equation*}
and all other coefficients being zero.\medskip

The expressions of the traceless projectors for $n=5$ and $n=6$ can be found in the appendix of \cite{KMP_central_idempotents}, and in the Mathematica notebook joined with the article \cite{bulgakova2022construction} where there is also the expression for $n=7$.



%
	%

	%
	\appendix 
\chapter{Irreducible decomposition of Riemann tensors}\label{app:Irreducible_Riemanns}
\numberwithin{equation}{section}
The computations for the irreducible decompositions of the Riemann tensors presented below are detailed in the Mathematica notebook \cite{xMAGRiemann}.
\section{The case of zero non-metricity}\label{app:IrredRiemannT}

We consider here the irreducible decomposition of the Riemann tensor when $\tensor{Q}{_a_b_c}=0$ and $\tensor{T}{^a_b_c}\neq0$. In this case, the Riemann tensor is antisymmetric with respect to the first and second pairs of indices:
\begin{equation*}
	\tensor{\mathcal{R}}{_{[ab][cd]}}=\tensor{\mathcal{R}}{_{abcd}}\,.
\end{equation*}
As for the general case the algebraic Bianchi identity is not a symmetry of the Riemann tensor because it involves the covariant derivative of the torsion tensor.
The space of Riemann tensor $V_{\mathcal{R}}$ with zero non-metricity and non zero torsion decomposes as the following direct sum of $\GL(\Dim,\mathbb{C})$-irreducible representation:
\begin{equation*}\label{eq:centralGL_decompositionVRT}
	V_{\mathcal{R}}=\, V^{\Yboxdim{3pt}\yng(1,1)}\otimes V^{\Yboxdim{3pt}\yng(1,1)}=V^{\Yboxdim{3pt}\yng(2,2)}\oplus V^{\Yboxdim{3pt}\yng(2,1,1)}\oplus V^{\Yboxdim{3pt}\yng(1,1,1,1)}\,,
\end{equation*}
where we have used the Littlewood-Richardson rule. This decomposition is unique and more explicitly we have: 
\begin{equation}
	\tensor{\mathcal{R}}{_{abcd}}=\tensor{{\accentset{\left(\Yboxdim{2.5pt}\yng(2,2)\right)}{\mathcal{R}}}}{_{abcd}}+\tensor{{\accentset{\left(\Yboxdim{2.5pt}\yng(2,1,1)\right)}{\mathcal{R}}}}{_{abcd}}+\tensor{{\accentset{\left(\Yboxdim{2.5pt}\yng(1,1,1,1)\right)}{\mathcal{R}}}}{_{abcd}}\, 
\end{equation}
with
\begin{equation}
	\begin{aligned}
		&\tensor{{\accentset{\left(\Yboxdim{2.5pt}\yng(2,2)\right)}{\mathcal{R}}}}{_{abcd}}:=\tensor{(\mathcal{R}\cdot Z^{\Yboxdim{2.5pt}\yng(2,2)})}{_{abcd}}=\frac{1}{2}\left(\tensor{\mathcal{R}}{_{abcd}}+\tensor{\mathcal{R}}{_{cdab}}\right)-\tensor{\mathcal{R}}{_{[abcd]}}\,, \hspace{1cm} 
		\tensor{{\accentset{\left(\Yboxdim{2.5pt}\yng(1,1,1,1)\right)}{\mathcal{R}}}}{_{abcd}}:=\tensor{(\mathcal{R}\cdot Z^{\Yboxdim{2.5pt}\yng(1,1,1,1)} )}{_{abcd}}=\tensor{\mathcal{R}}{_{[abcd]}}\,,\\
		&\tensor{{\accentset{\left(\Yboxdim{2.5pt}\yng(2,1,1)\right)}{\mathcal{R}}}}{_{abcd}}:=\tensor{(\mathcal{R}\cdot Z^{\Yboxdim{2.5pt}\yng(2,1,1)})}{_{abcd}}=\frac{1}{2}\left(\tensor{\mathcal{R}}{_{abcd}}-\tensor{\mathcal{R}}{_{cdab}}\right)\,.
	\end{aligned}
\end{equation}
The explicit expressions of the projectors $Z^{\mu}$, expanded in the conjugacy class sum basis of $\mathcal{Z}_n$, are given in \eqref{eq:centralYoung_4}.\smallskip

The branching rules from $\GL(\Dim,\mathbb{C})$ to $\Or(1,\Dim-1)$ are 
\begin{equation*}
	V^{\Yboxdim{3pt}\yng(2,2)}=D^{\Yboxdim{3pt}\yng(2,2)}\oplus D^{\Yboxdim{3pt}\yng(2)} \oplus D^{\Yboxdim{3pt}\emptyset}\,,\hspace{1cm} V^{\Yboxdim{3pt}\yng(2,1,1)}=D^{\Yboxdim{3pt}\yng(2,1,1)}\oplus D^{\Yboxdim{3pt}\yng(1,1)}\,, \hspace{1cm} V^{\Yboxdim{3pt}\yng(1,1,1,1)}=D^{\Yboxdim{3pt}\yng(1,1,1,1)}\,,
\end{equation*}
and one has $V_\mathcal{R}=D^{\Yboxdim{3pt}\yng(2,2)}\oplus D^{\Yboxdim{3pt}\yng(2,1,1)} \oplus D^{\Yboxdim{3pt}\yng(1,1,1,1)} \oplus D^{\Yboxdim{3pt}\yng(2)}\oplus D^{\Yboxdim{3pt}\yng(1,1)}\oplus D^{\Yboxdim{3pt}\emptyset}$. This decomposition is multiplicity free and we therefore have a unique $\Or(1,\Dim-1)$ irreducible decomposition of the Riemann tensor. The decomposition of each $\GL(\Dim,\mathbb{C})$ irreducible sector is given by 
\begin{equation*}
	\tensor{{\accentset{\left(\Yboxdim{3pt}\yng(2,2)\right)}{\mathcal{R}}}}{_{abcd}}=\tensor{{\accentset{\left(\Yboxdim{3pt}\yng(2,2)\right)}{\underline{\mathcal{R}}}}}{_{\, abcd}}+\tensor{{\accentset{\left(\Yboxdim{3pt}\yng(2,2),\,\Yboxdim{3pt}\yng(2)\right)}{\mathcal{R}}}}{_{\, abcd}}+\tensor{{\accentset{\left(\Yboxdim{3pt}\yng(2,2),\,\emptyset\right)}{\mathcal{R}}}}{_{\, abcd}}\,,\hspace{1cm}
	\tensor{{\accentset{\left(\Yboxdim{3pt}\yng(2,1,1)\right)}{{\mathcal{R}}}}}{_{abcd}}=\tensor{{\accentset{\left(\Yboxdim{3pt}\yng(2,1,1)\right)}{\underline{\mathcal{R}}}}}{_{\,abcd}}+\,\,\,\tensor{{\accentset{\left(\Yboxdim{3pt}\yng(2,1,1),\,\Yboxdim{3pt}\yng(1,1)\right)}{{\mathcal{R}}}}}{_{abcd}}\,,\hspace{1cm}
	\tensor{{\accentset{\left(\Yboxdim{3pt}\yng(1,1,1,1)\right)}{{\mathcal{R}}}}}{_{abcd}}=\tensor{{\accentset{\left(\Yboxdim{3pt}\yng(1,1,1,1)\right)}{\underline{\mathcal{R}}}}}{_{\, abcd}}\,,
\end{equation*}
where the traceless tensors of the decompositions are given by $\tensor{{\accentset{\left(\lambda\right)}{\underline{\mathcal{R}}}}}{_{\, abcd}}=\big(\accentset{\left(\lambda\,\right)}{\mathcal{R}}\cdot P^{\lambda}_4\big)_{abcd}$ with $|\lambda|=4$, while the $2$-traceless tensors are given by   are given by $\tensor{{\accentset{\left(\mu,\,\lambda\right)}{\mathcal{R}}}}{_{\, abcd}}=\big(\accentset{\left(\mu\,\right)}{\mathcal{R}}\cdot P^{\lambda}_4\big)_{abcd}$ with $|\lambda|=2$, and $\tensor{{\accentset{\left(\Yboxdim{3pt}\yng(2,2),\,\emptyset\right)}{\mathcal{R}}}}{_{\, abcd}}=\big(\accentset{\left(\Yboxdim{3pt}\yng(2,2)\,\right)}{\mathcal{R}}\cdot P^{\emptyset}_4\big)_{abcd}$. The explicit expressions of the projectors $P^{\lambda}_4$, expanded in the conjugacy class sum basis of $\mathcal{C}_n$, are given in \eqref{eq:proj_4TL} and \eqref{eq:proj_Riemann}.\smallskip

The $\Or(1,\Dim-1)$-irreducible decomposition of the Riemann tensor is then given by
\begin{equation}\label{eq:Irreducible_RiemannT}
	\Scale[0.98]{\tensor{\mathcal{R}}{_{abcd}}=\underbrace{\tensor{{\accentset{\left(\Yboxdim{3pt}\yng(2,2)\right)}{\underline{\mathcal{R}}}}}{_{\, abcd}}+\tensor{{\accentset{\left(\Yboxdim{3pt}\yng(2,1,1)\right)}{\underline{\mathcal{R}}}}}{_{\,abcd}}+\tensor{{\accentset{\left(\Yboxdim{3pt}\yng(1,1,1,1)\right)}{\underline{\mathcal{R}}}}}{_{\, abcd}}}_{\text{traceless}}+\,\,\underbrace{\tensor{{\accentset{\left(\Yboxdim{3pt}\yng(2,2),\,\Yboxdim{3pt}\yng(2)\right)}{\mathcal{R}}}}{_{\, abcd}}+\,\,\,\tensor{{\accentset{\left(\Yboxdim{3pt}\yng(2,1,1),\,\Yboxdim{3pt}\yng(1,1)\right)}{{\mathcal{R}}}}}{_{abcd}}}_{\text{2-traceless}}+\underbrace{\tensor{{\accentset{\left(\Yboxdim{3pt}\yng(2,2),\,\emptyset\right)}{\mathcal{R}}}}{_{\, abcd}}}_{\text{full trace}}\,.}
\end{equation} 
Applying the projection operators \eqref{eq:main_res_f_traceless_central_idempotent} to the Riemann tensor we obtain
\begin{equation}
	\begin{aligned}
		&\tensor{{\accentset{\left(\Yboxdim{3pt}\yng(2,2),\,\Yboxdim{3pt}\yng(2)\right)}{\mathcal{R}}}}{_{\, abcd}}=\frac{2}{\Dim-2}\left(\, \tensor{\accentset{(1)}{\mathcal B}}{_{[a|c|}}\tensor{g}{_{b]d}} - \tensor{\accentset{(1)}{\mathcal B}}{_{[a|d|}}\tensor{g}{_{b]c}} \,\right)\,,  \hspace{1.cm}
		\tensor{{\accentset{\left(\Yboxdim{3pt}\yng(2,2),\,\emptyset\right)}{\mathcal{R}}}}{_{\, abcd}}=\frac{\mathcal{R}}{\Dim \left( \Dim-1  \right)} \Bigl(\tensor{g}{_{ac}}\tensor{g}{_{bd}}-\tensor{g}{_{ad}}\tensor{g}{_{bc}}\Bigr)\,, \\
		&\tensor{{\accentset{\left(\Yboxdim{3pt}\yng(2,1,1),\,\Yboxdim{3pt}\yng(1,1)\right)}{{\mathcal{R}}}}}{_{abcd}}=\frac{2}{\Dim-2}\left(\, \tensor{\accentset{(2)}{\mathcal B}}{_{[a|c|}}\tensor{g}{_{b]d}} - \tensor{\accentset{(2)}{\mathcal B}}{_{[a|d|}}\tensor{g}{_{b]c}} \,\right)\,,
	\end{aligned}
\end{equation}
where we have introduced the traceless building blocks
\begin{equation*}
	\tensor{\accentset{(1)}{\mathcal B}}{_{ab}}=\tensor{\underline{\accentset{(1)}{\mathcal R}}}{_{(ab)}}\,, \hspace{1cm} \tensor{\accentset{(2)}{\mathcal B}}{_{ab}}=\tensor{\accentset{(1)}{\mathcal R}}{_{[ab]}}\,.
\end{equation*}
This decomposition is identical to the one given in \cite{Hayashi_1980}.

\section{The case of zero torsion}\label{app:IrredRiemannQ}

We consider here the irreducible decomposition of the Riemann tensor when $\tensor{T}{^a_b_c}=0$ and $\tensor{Q}{_a_b_c}\neq0$. We will see that in this case the decomposition in not unique. 


When $\tensor{T}{^a_b_c}=0$, the symmetries of indices of the Riemann tensor are the same as in the general setting: antisymmetric in the first pair of indices. However we have to take into account the algebraic Bianchi identity which is a Young symmetry of the Riemann tensor: 
\begin{equation}\label{eq:bianchi_alg}
	\tensor{\mathcal{R}}{_{[abc]d}}=0\,.
\end{equation}
This identity implies directly that the totally antisymmetric part of the Riemann tensor is zero. Besides, from the structure of the Young seminormal idempotents \eqref{eq:seminormalYoung4} it is clear that
\begin{equation}
	\tensor{(\mathcal{R}\cdot  Y^{\,\Scale[0.4]{\young(12,3,4)}})}{_{abcd}}= \tensor{(\mathcal{R}\cdot Y^{\,\Scale[0.4]{\young(14,2,3)}})}{_{abcd}}=0\,.
\end{equation}
Hence, the space of Riemann tensor $V_{\mathcal{R}}$ with zero torsion and non zero non-metricity decomposes as the following direct sum of $GL(\Dim,\mathbb{C})$-irreducible representation:
\begin{equation*}
	V_{\mathcal{R}}=\, V^{\Yboxdim{3pt}\yng(3,1)}\oplus \, V^{\Yboxdim{3pt}\yng(2,2)}\oplus \,\, V^{\Yboxdim{3pt}\yng(2,1,1)} \, .
\end{equation*}
Each irreducible representation appears without multiplicity. Therefore we can infer that the $GL(\Dim,\mathbb{C})$ irreducible decomposition of the Riemann tensor is unique, more precisely one has:
\begin{equation}
	\tensor{\mathcal{R}}{_{abcd}}=\tensor{{\accentset{\left(\Yboxdim{2.5pt}\yng(3,1)\right)}{\mathcal{R}}}}{_{abcd}}+\tensor{{\accentset{\left(\Yboxdim{2.5pt}\yng(2,2)\right)}{\mathcal{R}}}}{_{abcd}}+\tensor{{\accentset{\left(\Yboxdim{2.5pt}\yng(2,1,1)\right)}{\mathcal{R}}}}{_{abcd}}\,,
\end{equation} 
with 
\begin{equation}
	\begin{aligned}
		&\tensor{\accentset{\left(\Yboxdim{2.5pt}\yng(2,1,1)\right)}{\mathcal{R}}}{_{abcd}}:=\tensor{(\mathcal{R}\cdot Z^{\Yboxdim{2.5pt}\yng(2,1,1)})}{_{abcd}}=\frac{3}{4}\left(\tensor{{\mathcal{R}}}{_{[ab|c|d]}}+\tensor{{\mathcal{R}}}{_{[a|c|bd]}}-\tensor{{\mathcal{R}}}{_{[a|d|bc]}}\right)\,,\\
		&\tensor{{\accentset{\left(\Yboxdim{2.5pt}\yng(3,1)\right)}{\mathcal{R}}}}{_{abcd}}:=\tensor{(\mathcal{R}\cdot Z^{\Yboxdim{2.5pt}\yng(3,1)})}{_{abcd}}=\frac{3}{4}\left(\tensor{{\mathcal{R}}}{_{a(bcd)}}+\tensor{{\mathcal{R}}}{_{(a|b|cd)}}\right) \,,\\
		&\tensor{{\accentset{\left(\Yboxdim{2.5pt}\yng(2,2)\right)}{\mathcal{R}}}}{_{abcd}}:=\tensor{(\mathcal{R}\cdot Z^{\Yboxdim{2.5pt}\yng(2,2)})}{_{abcd}}=\tensor{{\mathcal{R}}}{_{[ab]cd}}-\tensor{{\accentset{\left(\Yboxdim{2.5pt}\yng(2,1,1)\right)}{\mathcal{R}}}}{_{abcd}}-\tensor{{\accentset{\left(\Yboxdim{2.5pt}\yng(3,1)\right)}{\mathcal{R}}}}{_{abcd}}\,.
	\end{aligned}
\end{equation}
The explicit expressions of the projectors $Z^{\mu}$, expanded in the conjugacy class sum basis of $\mathcal{Z}_n$, are given in \eqref{eq:centralYoung_4}. Before proceeding to the $\Or(1,\Dim-1)$ irreducible decomposition, let us mention that the metric contraction of the Bianchi identity \eqref{eq:bianchi_alg} yield the following relation between traces of the Riemann tensor: 
\begin{equation}
	\tensor{\accentset{(3)}{\mathcal{R}}}{_{ab}}=2 \,\tensor{\accentset{(1)}{\mathcal{R}}}{_{[ab]}}\,.
\end{equation}

The branching rules of interest from $\GL(\Dim,\mathbb{C})$ to $\Or(1,\Dim-1)$ are 
\begin{equation*}
	V^{\Yboxdim{3pt}\yng(3,1)}=D^{\Yboxdim{3pt}\yng(3,1)}\oplus D^{\Yboxdim{3pt}\yng(2)} \oplus D^{\Yboxdim{3pt}\yng(1,1)}\,\hspace{1cm}
	V^{\Yboxdim{3pt}\yng(2,2)}=D^{\Yboxdim{3pt}\yng(2,2)}\oplus D^{\Yboxdim{3pt}\yng(2)} \oplus D^{\Yboxdim{3pt}\emptyset}\,\hspace{1cm} V^{\Yboxdim{3pt}\yng(2,1,1)}=D^{\Yboxdim{3pt}\yng(2,1,1)}\oplus D^{\Yboxdim{3pt}\yng(1,1)}\,,
\end{equation*}
and one has $V_\mathcal{R}=D^{\Yboxdim{3pt}\yng(3,1)}\oplus D^{\Yboxdim{3pt}\yng(2,2)}\oplus D^{\Yboxdim{3pt}\yng(2,1,1)} \oplus 2\, D^{\Yboxdim{3pt}\yng(2)}\oplus 2\, D^{\Yboxdim{3pt}\yng(1,1)}\oplus D^{\Yboxdim{3pt}\emptyset}$. The $O(1,\Dim-1)$-irreducible decomposition of the Riemann tensor is not unique because some irreducible representation appear with multiplicities. Indeed, $m(\Yboxdim{5pt}\yng(2))=m(\Yboxdim{5pt}\yng(1,1))=2$. Nevertheless, the $\Or(1,\Dim-1)$-irreducible representations appearing in each $\GL(\Dim,\mathbb{C})$ irreducible sector are multiplicity free. Therefore the two step irreducible decomposition is unique and one has:
\begin{equation}
	\begin{aligned}
		&\tensor{{\accentset{\left(\Yboxdim{3pt}\yng(3,1)\right)}{\mathcal{R}}}}{_{abcd}}=\tensor{{\accentset{\left(\Yboxdim{3pt}\yng(3,1)\right)}{\underline{\mathcal{R}}}}}{_{\, abcd}}+\,\,\tensor{{\accentset{\left(\Yboxdim{3pt}\yng(3,1),\,\Yboxdim{3pt}\yng(2)\right)}{\mathcal{R}}}}{_{\, abcd}}+\,\,\tensor{{\accentset{\left(\Yboxdim{3pt}\yng(3,1),\,\Yboxdim{3pt}\yng(1,1)\right)}{\mathcal{R}}}}{_{\, abcd}}\,,\hspace{2.5cm}
		\tensor{{\accentset{\left(\Yboxdim{3pt}\yng(2,2)\right)}{\mathcal{R}}}}{_{abcd}}=\tensor{{\accentset{\left(\Yboxdim{3pt}\yng(2,2)\right)}{\underline{\mathcal{R}}}}}{_{\, abcd}}+\,\,\tensor{{\accentset{\left(\Yboxdim{3pt}\yng(2,2),\,\Yboxdim{3pt}\yng(2)\right)}{\mathcal{R}}}}{_{\, abcd}}+\,\,\tensor{{\accentset{\left(\Yboxdim{3pt}\yng(2,2),\,\emptyset\right)}{\mathcal{R}}}}{_{\, abcd}}\,,
		\\
		&
		\tensor{{\accentset{\left(\Yboxdim{3pt}\yng(2,1,1)\right)}{{\mathcal{R}}}}}{_{abcd}}=\tensor{{\accentset{\left(\Yboxdim{3pt}\yng(2,1,1)\right)}{\underline{\mathcal{R}}}}}{_{\,abcd}}+\,\,\,\tensor{{\accentset{\left(\Yboxdim{3pt}\yng(2,1,1),\,\Yboxdim{3pt}\yng(1,1)\right)}{{\mathcal{R}}}}}{_{abcd}}\,,
	\end{aligned}
\end{equation}
where the traceless tensors of the decompositions are given by $\tensor{{\accentset{\left(\lambda\right)}{\underline{\mathcal{R}}}}}{_{\, abcd}}=\big(\accentset{\left(\lambda\,\right)}{\mathcal{R}}\cdot P^{\lambda}_4\big)_{abcd}$ with $|\lambda|=4$, while the $2$-traceless tensors are given by   are given by $\tensor{{\accentset{\left(\mu,\,\lambda\right)}{\mathcal{R}}}}{_{\, abcd}}=\big(\accentset{\left(\mu\,\right)}{\mathcal{R}}\cdot P^{\lambda}_4\big)_{abcd}$ with $|\lambda|=2$, and $\tensor{{\accentset{\left(\Yboxdim{3pt}\yng(2,2),\,\emptyset\right)}{\mathcal{R}}}}{_{\, abcd}}=\big(\accentset{\left(\Yboxdim{3pt}\yng(2,2)\,\right)}{\mathcal{R}}\cdot P^{\emptyset}_4\big)_{abcd}$. The explicit expressions of the projectors $P^{\lambda}_4$, expanded in the conjugacy class sum basis of $\mathcal{C}_n$, are given in \eqref{eq:proj_4TL} and \eqref{eq:proj_Riemann}.\smallskip

The  two-step $\Or(1,\Dim-1)$-irreducible decomposition of the Riemann tensor is then given by
\begin{equation}\label{eq:Irreducible_RiemannQ}
	\Scale[0.98]{\tensor{\mathcal{R}}{_{abcd}}=\underbrace{\tensor{{\accentset{\left(\Yboxdim{3pt}\yng(3,1)\right)}{\underline{\mathcal{R}}}}}{_{\, abcd}}+\tensor{{\accentset{\left(\Yboxdim{3pt}\yng(2,2)\right)}{\underline{\mathcal{R}}}}}{_{\, abcd}}+\tensor{{\accentset{\left(\Yboxdim{3pt}\yng(2,1,1)\right)}{\underline{\mathcal{R}}}}}{_{\,abcd}}}_{\text{traceless}}+\,\,\underbrace{\,\,\tensor{{\accentset{\left(\Yboxdim{3pt}\yng(3,1),\,\Yboxdim{3pt}\yng(2)\right)}{\mathcal{R}}}}{_{\, abcd}}+\,\,\tensor{{\accentset{\left(\Yboxdim{3pt}\yng(2,2),\,\Yboxdim{3pt}\yng(2)\right)}{\mathcal{R}}}}{_{\, abcd}}+\,\,\tensor{{\accentset{\left(\Yboxdim{3pt}\yng(3,1),\,\Yboxdim{3pt}\yng(1,1)\right)}{\mathcal{R}}}}{_{\, abcd}}+\,\,\,\tensor{{\accentset{\left(\Yboxdim{3pt}\yng(2,1,1),\,\Yboxdim{3pt}\yng(1,1)\right)}{{\mathcal{R}}}}}{_{abcd}}}_{\text{2-traceless}}+\underbrace{\tensor{{\accentset{\left(\Yboxdim{3pt}\yng(2,2),\,\emptyset\right)}{\mathcal{R}}}}{_{\, abcd}}}_{\text{full trace}}\,.}
\end{equation}
Applying the projection operators \eqref{eq:main_res_f_traceless_central_idempotent} to each $\GL(\Dim,\mathbb{C})$ irreducible tensor we obtain
\begin{equation}
	\begin{aligned}
		&\tensor{{\accentset{\left(\Yboxdim{3pt}\yng(3,1),\,\Yboxdim{3pt}\yng(2)\right)}{\mathcal{R}}}}{_{\, abcd}}=\frac{1}{\Dim} \left(\, \tensor{\accentset{(1)}{\mathcal B}}{_{[a|d|}}\tensor{g}{_{b]c}} +\tensor{\accentset{(1)}{\mathcal B}}{_{[a|c|}}\tensor{g}{_{b]d}}  \, \right)\,, \hspace{1.55cm} \tensor{{\accentset{\left(\Yboxdim{3pt}\yng(3,1),\,\Yboxdim{3pt}\yng(1,1)\right)}{\mathcal{R}}}}{_{\, abcd}}=\frac{1}{2\left(\, \Dim+2 \, \right)} \left( \tensor{\accentset{(4)}{\mathcal B}}{_{ab}}\tensor{g}{_{cd}}+\tensor{\accentset{(4)}{\mathcal B}}{_{[a|d|}}\tensor{g}{_{b]c}}+\tensor{\accentset{(4)}{\mathcal B}}{_{[a|c|}}\tensor{g}{_{b]d}} \right)\,, \\
		&\tensor{{\accentset{\left(\Yboxdim{3pt}\yng(2,2),\,\Yboxdim{3pt}\yng(2)\right)}{\mathcal{R}}}}{_{\, abcd}}=\frac{1}{\Dim-2}\left(\, \tensor{\accentset{(2)}{\mathcal B}}{_{[a|c|}}\tensor{g}{_{b]d}} - \tensor{\accentset{(2)}{\mathcal B}}{_{[a|d|}}\tensor{g}{_{b]c}} \,\right)\,,  \hspace{1.cm}
		\tensor{{\accentset{\left(\Yboxdim{3pt}\yng(2,2),\,\emptyset\right)}{\mathcal{R}}}}{_{\, abcd}}=\frac{\mathcal{R}}{\Dim \left( \Dim-1  \right)} \Bigl(\tensor{g}{_{ac}}\tensor{g}{_{bd}}-\tensor{g}{_{ad}}\tensor{g}{_{bc}}\Bigr)\,, \\
		&\tensor{{\accentset{\left(\Yboxdim{3pt}\yng(2,1,1),\,\Yboxdim{3pt}\yng(1,1)\right)}{{\mathcal{R}}}}}{_{abcd}}=\frac{1}{2(\Dim-2)}\left(\tensor{\accentset{(3)}{\mathcal{B}}}{_{ab}}\tensor{g}{_{cd}}+\tensor{\accentset{(3)}{\mathcal B}}{_{[a|c|}}\tensor{g}{_{b]d}}-3\,\tensor{\accentset{(3)}{\mathcal B}}{_{[a|d|}}\tensor{g}{_{b]c}}\right)\,,
	\end{aligned}
\end{equation}
where we have introduced the traceless building blocks
\begin{equation}
	\begin{aligned}
		&\tensor{\accentset{(1)}{\mathcal B}}{_{ab}}=\tensor{\accentset{(1)}{\mathcal{R}}}{_{(ab)}}-\tensor{\accentset{(2)}{\mathcal{R}}}{_{(ab)}}\,, \hspace{3cm} \tensor{\accentset{(2)}{\mathcal B}}{_{ab}}=\tensor{\underline{\accentset{(1)}{\mathcal{R}}}}{_{(ab)}}+\tensor{\underline{\accentset{(2)}{\mathcal{R}}}}{_{(ab)}}\,, \\
		&\tensor{\accentset{(3)}{\mathcal B}}{_{ab}}=\tensor{\accentset{(1)}{\mathcal{R}}}{_{[ab]}}+\tensor{\accentset{(2)}{\mathcal{R}}}{_{[ab]}}\,, \hspace{3.1cm} \tensor{\accentset{(4)}{\mathcal B}}{_{ab}}=3\,\tensor{\accentset{(1)}{\mathcal{R}}}{_{[ab]}}-\tensor{\accentset{(2)}{\mathcal{R}}}{_{[ab]}}\,.
	\end{aligned}
\end{equation} 
	\chapter{Mathematical tools : differential geometry and representation theory}\label{app:maths}
\section{General notions}\label{sec:general_definitions}

\subsection{Differential geometry}\label{subsec:definitions_geo_diff}

\paragraph{Coordinate frame.} Let $\lbrace x^\alpha \rbrace$ with $\alpha=1,\ldots,\, \Dim$, be the set of coordinate functions $x^\alpha\,:\, U \mapsto \mathbb{R}$ on an open subset $U\subset \mathcal{M}$ related to the coordinate map $x \,: U\to V\subset \mathbb{R}^\Dim\, $.\smallskip

A coordinate frame (resp. co-frame) is the data of a coordinate basis of the tangent space $T_p\mathcal{M}$ (resp. co-tangent space $T^{*}_p\mathcal{M}$) for each $p\in U$. The local coordinate system (chart) $(U\,,\, x)$ gives rise to the  coordinate (holonomic) frame $\lbrace \dfrac{\partial}{\partial x^\alpha}\,, \alpha=1,\ldots,\, \Dim\rbrace$.\footnote{We also use the notation  $\partial_{\alpha}$.} The coordinate co-frame $\lbrace \mathrm{d}x^\alpha\rbrace$ on $U$ is defined by the relation $\mathrm{d}x^\alpha(\partial_{\beta})=\tensor{\delta}{^\alpha_\beta}$. \medskip 

For any smooth function $f\,:\,\mathcal{M}\to \mathbb{R}^\Dim$ the action of the operator $\dfrac{\partial}{\partial x^\alpha}$ is defined by  

\begin{equation}
	\dfrac{\partial f}{\partial x^\alpha}\Big{|}_p:=\dfrac{\partial \left(f\circ x^{-1}\right)}{\partial \textit{z}^{\,\alpha}}\Big{|}_{\textit{z}}
\end{equation}
where $\textit{z}=x(p)$, and $\dfrac{\partial}{\partial \textit{z}^{\,\alpha}}$ is the usual partial derivative of functions on $\mathbb{R}^\Dim$.


\paragraph{Linear connection on $\mathcal{M}$.}
	A linear connection on $\mathcal{M}$ is a mapping (see for example \cite[Chapter I.7]{choquet2015introduction} and also \cite[Chapter V. Section B]{bruhat1982analysis})
	\begin{equation}
		\nabla : \ X\mapsto \nabla(X)
	\end{equation}
	from smooth vector fields $X$ on $\mathcal{M}$ to differentiable tensor fields of type (1,1) on $\mathcal{M}$ such that 
	\begin{equation}\label{eq:defining_prop_connection}
		\begin{array}{ll}
			\nabla (X + Z)= \nabla X +\nabla  Z& \hspace{1.5cm} \textit{Linearity}\\[10pt]
			\nabla (fX)= \nabla f \otimes X+ f\nabla X& \hspace{1.5cm} \textit{Leibniz rule}
		\end{array}
	\end{equation}
	where $f$ is a differentiable function on $\mathcal{M}$ and $\nabla (f):=\diff f$ with $\mathrm{d} f= \partial_\alpha (f) \mathrm{d}x^\alpha$. The tensor  $\nabla X$ is called the \textit{absolute covariant derivative} of $X$.\medskip
	
\paragraph{Covariant derivative.} The covariant derivative $\nabla_Y X$ of $X$ along $Y$ is a vector field defined by 
	\begin{equation*}
		\nabla_Y X:=\nabla X (Y) \quad \text{or} \quad \nabla_Y X(\cdot):=\nabla X(Y,\cdot).
	\end{equation*}
	From the defining properties \eqref{eq:defining_prop_connection} of $\nabla$ one has
	\begin{equation}\label{eq:defining_prop_connection2}
		\begin{array}{ll}
			\nabla_Y (X + Z)= \nabla_Y X +\nabla_Y Z& \hspace{1.5cm} \textit{Linearity}\\[10pt]
			\nabla_Y (fX)= Y(f) X+ f \, \nabla_Y X& \hspace{1.5cm} \textit{Leibniz rule}
		\end{array}
	\end{equation}
where
	\begin{equation}
	Y(f)=\nabla_Y(f)\,.
	\end{equation}
For two arbitrary tensors $T$ and $S$ we have 
	\begin{equation}
		\begin{array}{ll}
			\nabla_Y(\, T+S \,)=\nabla_Y(\, T \,)+\nabla_Y(\, S \,)\,& \hspace{0.5cm} \textit{Linearity}\\
			\nabla_Y(\, T\otimes S \,)=\nabla_Y(\, T\,)\otimes S+ T \otimes \nabla_Y(\, S \,)\, &\hspace{0.5cm} \textit{Leibniz rule}
		\end{array}
	\end{equation}

\paragraph{Connection coefficients (Christoffel symbols).}	The connection coefficients $\tensor{\Gamma}{^{\sigma}_{\alpha\beta}}$ of $\nabla$ in a coordinate frame are defined by the relation
	\begin{equation}\label{eq:def_connection_coefficients_coordinate}
		\nabla (\partial_\beta):=\tensor{\Gamma}{^{\sigma}_{\alpha\beta}}\, \diff x^\alpha\otimes\partial_\sigma\,, \hspace{0.5cm}
		\text{or equivalently}\hspace{0.5cm} \nabla_\alpha (\partial_\beta):=\tensor{\Gamma}{^{\sigma}_{\alpha\beta}}\, \partial_\sigma\,.
	\end{equation} 
	Hence for the covariant derivative $\nabla\,X$ of a vector field $X=X^\sigma\partial_\sigma$ we have:
	\begin{equation}
		\nabla\, X =\left(\nabla_\alpha X^{\sigma} \right) \diff x^\alpha\otimes \partial_\sigma \hspace{0.5cm}\text{with}\hspace{0.5cm} \nabla_\alpha X^{\sigma} =\partial_\alpha X^{\sigma} +  \tensor{\Gamma}{^{\sigma}_{\alpha\beta}}X^{\beta}\,,
	\end{equation}
	and for the covariant derivative $\nabla_Y\,X$ of $X$ along $Y$ we have
	\begin{equation}
		\nabla_Y X =\Big(Y^\alpha\nabla_\alpha X^\sigma \Big)\partial_\sigma.
	\end{equation}
With the above definitions the covariant derivative of a one form $A=A_\mu \, \mathrm{d}x^\mu$ along $Y$ is given by 
\begin{equation}
		\nabla_Y A =\Big(Y^\beta\nabla_\beta \, A_\sigma \Big)\mathrm{d}x^\sigma\, \hspace{0.5cm}\text{with}\hspace{0.5cm} \nabla_\beta A_{\sigma}=\partial_\beta A_{\sigma} - \tensor{\Gamma}{^{\alpha}_{\beta\sigma}}A_{\alpha}\,.
\end{equation}
For a tensor $T=\tensor{T}{^{\alpha_1\ldots \alpha_p}_{\beta_1\ldots\beta_q}}\partial_{\alpha_1}\otimes \ldots\otimes\partial_{\alpha_p}\otimes \mathrm{d}x^{\beta_1}\otimes\ldots\otimes\mathrm{d}x^{\beta_q}$ of type $(p,q)$ one has 
\begin{equation}
	\nabla_Y T =\Big(Y^\rho\nabla_\rho \, \tensor{T}{^{\alpha_1\ldots \alpha_p}_{\beta_1\ldots\beta_q}} \Big)\partial_{\alpha_1}\otimes \ldots\otimes\partial_{\alpha_p}\otimes \mathrm{d}x^{\beta_1}\otimes\ldots\otimes\mathrm{d}x^{\beta_q}\,,
\end{equation}
where
\begin{equation*}
\begin{aligned}
\nabla_\rho \, \tensor{T}{^{\alpha_1\ldots \alpha_p}_{\beta_1\ldots\beta_q}} = \partial_\rho \, \tensor{T}{^{\alpha_1\ldots \alpha_p}_{\beta_1\ldots\beta_q}} &+\tensor{\Gamma}{^{\alpha_1}_{\rho\mu}}\tensor{T}{^{\mu \ldots \alpha_p}_{\beta_1\ldots\beta_q}}+\ldots+\tensor{\Gamma}{^{\alpha_p}_{\rho\mu}}\tensor{T}{^{\alpha_1 \ldots \mu}_{\beta_1\ldots\beta_q}}\\
&-\tensor{\Gamma}{^{\mu}_{\rho\beta_1}}\tensor{T}{^{\alpha_1 \ldots \alpha_p}_{\mu\ldots\beta_q}}-\ldots- \tensor{\Gamma}{^{\mu}_{\rho\beta_q}}\tensor{T}{^{\alpha_1 \ldots \alpha_p}_{\beta_q\ldots\mu}}\,.
\end{aligned}
\end{equation*}
\paragraph{Covariant derivative in a moving frame.} In a moving frame  $\lbrace \bm{e}_a \st a=1\,,\ldots \Dim \rbrace$ (resp. coframe $\lbrace \bm{\theta}^a\st a=1\,,\ldots \Dim \rbrace$) on an open subset $U\subset\mathcal{M}$ we denote by $\tensor{\omega}{^{a}_{bc}}$ the connection coefficients of $\nabla$:
\begin{equation}
\nabla(\bm{e}_c):=\tensor{\omega}{^a_b_c}\, \bm{e}_a \otimes\bm{\theta}^b\,.
\end{equation}
	For the covariant derivative $\nabla\,X$ of a vector field $X=X^b\, \bm{e}_b$ expressed in a moving frame we have:
	\begin{equation}
		\nabla\, X =\left(\nabla_a X^{b} \right) \bm{\theta}^a\otimes \bm{e}_b,\hspace{0.5cm}\text{with}\hspace{0.5cm} \nabla_a X^{b}=\bm{e}_a(X^{b}) + \,\tensor{\omega}{^{b}_{ac}} X^{c}
	\end{equation}
	For the covariant derivative $\nabla_Y\,X$ of $X$ along $Y$ expressed in a linear frame we have:
	\begin{equation}
		\nabla_Y X =Y^a\left(\nabla_a X^{b} \right) \bm{e}_b\,.
	\end{equation}
Similarly for a one form $A=A_a \,\bm{\theta}^a$ one has 
\begin{equation}
		\nabla_Y A =Y^a\left(\nabla_a A_{b} \right) \bm{\theta}^b\, \hspace{0.5cm}\text{with}\hspace{0.5cm} \nabla_a A_{b}=\bm{e}_a(A_{b}) - \,\tensor{\omega}{^{c}_{ab}}\, A_{c}\,.
\end{equation}
These expressions generalize naturally to tensors of arbitrary types.


\subsection{Miscellaneous}\label{subsec:definitions}
\paragraph{Direct sum.} Let $V$ be a vector space and let $V_k$ for $k=1\,,\ldots\,,N$ be subspaces of $V$. The \textit{sum} of $V_1\,,\ldots\,,V_N$ denoted $V_1+\ldots+V_N$ is defined to be set of all possible sums of elements of $V_1\,,\ldots\,,V_N$:
\begin{equation}
V_1+\ldots+V_N=\lbrace \, v_1+\ldots+v_N\,\st\, v_k\in V_k \hspace{0.2cm}\text{for} \hspace{0.2cm} k=1\,,\ldots\,,N \,\rbrace\,.
\end{equation}

We say that $V$ is the direct sum of $V_1\,,\ldots,V_N$ written
\begin{equation}
	V=\bigoplus_{k=1}^{N}V_k
\end{equation}
if every $v\in V$ can be decomposed \textit{uniquely} as 
\begin{equation}
	v=v_1+\ldots+v_N\,, \hspace{0.5cm} \text{with} \hspace{0.5cm} v_k\in V_k \,.
\end{equation}

\begin{lemma}[Lem 4.7 \cite{dym2013linear}]
 The sum $V=V_1+\ldots+V_N$ is direct if and only if each set of nonzero vectors $\lbrace \, v_1\,,\,\ldots\,,\, v_N \rbrace$ with $v_k\in V_k$ for $k=1\,,\ldots\,,N$ is a linearly independent set of vectors. 
\end{lemma}

\begin{proposition}[Prop 1.9 \cite{axler2023linear}] Suppose that $V_1$ and $V_2$ are subspaces of $V$. Then $V=V_1 \oplus V_2$ if and only if $V=V_1+V_2$ and $V_1\cap V_2 =\lbrace 0\rbrace$.
\end{proposition}

\begin{proposition} Suppose $V$ admits the following direct sum decomposition: 
	\begin{equation}\label{eq:decomp0}
		V=\bigoplus_{k=1}^{N} V_k\,.
	\end{equation}
Then, for any $i\neq j$ holds $V_i\cap V_j=\lbrace 0\rbrace$.
\end{proposition}
\paragraph{Projections.} A linear transformation $P\in\text{End}(V)$ is a \textit{projection} if it is idempotent: $P^2=P$. Alternatively, $P$ is a \textit{projection} if and only if for all $w\in\text{Im}(P)$ we have $P(w)=w$. The vector space $V$ decomposes as:
\begin{equation*}
	V=\text{Ker}(P)\oplus\text{Im}(P)\,,
\end{equation*}
where $P$ is the identity on $\text{Im}(P)$.

Let now $V$ be an complex vector space with a symmetric non-degenerate inner product $\langle \cdot \,,\, \cdot\rangle\,:\, V\times V\, \to \C$. For a linear transformation $T\in\End(V)$ the adjoint operator $T^{\,\dagger}$ is defined by the relation 
\begin{equation}
	\langle T^{\,\dagger}(v_1)\,,\, v_2\rangle=\langle v_1\,,\, T(v_2)\rangle \hspace{0.5cm} \text{for all } \, v_1\,,\,v_2\,\in \, V \,.
\end{equation}

An \textit{orthogonal projection} of $V$ onto a subspace $W$ is a projection $P \,:\, V\to W$ such that $\text{Im}(P)=W$ and $\text{Ker}(P)$ is its orthogonal complement. That is, for any $w\in W=\text{Im}(P)$
\begin{equation}
\langle w \,,\, u\rangle=0\, \hspace{0.5cm} \text{for all } u\in \text{Ker}(P)\,.
\end{equation}

\begin{lemma}\label{lem:orthogonal_selfadjoint}
A projection $P$ is orthogonal if and only if it is self-adjoint: $P=P^{\dagger}$.
\end{lemma}
\begin{proof}
For the implication from left to right, let $P$ be an orthogonal projection. The latter is equivalent to the following: for any $v_1\,,\,v_2\,\in\, V$, $\langle P(v_1)\,,\, v_2-P(v_2)\rangle=\langle v_1-P(v_1)\,,\, P(v_2)\rangle=0$. This may be written as
\begin{equation}
\langle P(v_1)\,,\, v_2 \rangle = \langle v_1\,,\, P(v_2) \rangle =\langle P(v_1)\,,\, P(v_2) \rangle\,.
\end{equation}
Hence we conclude that $P=P^{\dagger}$.\smallskip

Let $P$ be a self-adjoint projector. Then for any $v_1\,,\,v_2\,\in\, V$ one has $\langle P(v_1)\,,\, v_2-P(v_2)\rangle=\langle v_1\,,\, P^{\dagger}(v_2)-P^{\dagger}P(v_2)\rangle=\langle v_1\,,\, (P-P^2)(v_2)\rangle=0$.
\end{proof}
\begin{lemma}\label{lem:orthogonal_direct_sum_decomposition}
Suppose $V$ admits the following direct sum decomposition: 
\begin{equation}\label{eq:decomp}
	V=\bigoplus_{k=1}^{N} V_k\,.
\end{equation}
Let $P_k$ be the projectors onto the subspaces $V_k$. Then, the following assertions are equivalent, 
\begin{itemize}
	\item[i)] The subspaces $V_k$ are orthogonal.
	\item[ii)] The projectors $P_k$ are self-adjoint.
\end{itemize}
Besides, if \eqref{eq:decomp} is an orthogonal decomposition then the projectors $P_k$ are pairwise orthogonal ($P_iP_j=0$ for $i\neq j$) and form a partition of unity in $\End(V)$: 
\begin{equation}
\id_{V}=\sum_{k=1}^N P_k \,.
\end{equation}
\end{lemma}
\begin{proof}
The first part is a direct consequence of Lemma \ref{lem:orthogonal_selfadjoint}. Let \eqref{eq:decomp} be an orthogonal direct sum decomposition: for any $v\in V$ their exists a unique decomposition 
\begin{equation}\label{eq:decomp_v}
v=v_1+\ldots+v_N
\end{equation}
with $v_k\in V_k$ such that $\langle v_i\,,\, v_j \rangle=0$ for $i\neq j$. Then for any $i\neq j$ one has 

\begin{equation}
	\begin{aligned}
		0 &=\langle v_i\,,\, v_j \rangle  \\						
		&=\langle P_i(v_i)\,,\, v_j \rangle\\	
		&=\langle v_i\,,\, P_i(v_j) \rangle\,.	
	\end{aligned}
\end{equation}
Hence, $P_i(v_j)=0$ and then $P_iP_j=0$. Applying $P_i$ to \eqref{eq:decomp_v} one has $P_i(v)=v_i$ and hence the projectors $P_k$ form a partition of unity.
\end{proof}

\paragraph{Tensor product.} The tensor product $V\otimes W$ of two vector spaces $V$ and $W$ over a field $\mathbb{F}$ is a vector space to which is associated a bilinear map $\otimes\,:\, V\times W\to V\otimes W$ that maps the pair $(v,w)$, to an element of $V\otimes W$ denoted $v\otimes w$. That is for any $\lambda\in \mathbb{F}$, $v\in V$ and $w\in W$ one has 
\begin{equation}
(\lambda v)\otimes w = v\otimes (\lambda w)=\lambda( v\otimes w )\,.
\end{equation}
Besides, for every bilinear map $h :\, V\times W \to Z$ where $Z$ is a vector space, there is a unique linear map $\tilde{h}\,:\, V\otimes W \to Z$, such that $h(v,w)=\tilde{h}(v\otimes w)$. This is called the universal property of the tensor product. 
\paragraph{Lagrange interpolation.} The Lagrange polynomial $P(x)$ associated with a data set $(x_i,y_i)$, $0<\, i \leqslant k$ is the unique polynomial of lowest degree $\leqslant k$ that interpolates the data set $(x_i,y_i)$. It is given by 
\begin{equation}
	P(x)=\sum_{j=1}^{k} P_j(x)\, \hspace{0.5cm} \text{with} \hspace{0.5cm} P_j(x)= y_j \, \prod_{\begin{array}{c}
			{\scriptstyle i=1}\\
			{\scriptstyle i\neq j}
	\end{array}}^{k}\frac{x_j-x}{x_j-x_i}\,.
\end{equation}
One can easily verify that $P(x)$ indeed interpolates $(x_i,y_i)$:
\begin{equation}
	P_j(x_i)=\delta_{ij} y_j \,\, \Rightarrow \,\,P(x_i)=y_i\,.
\end{equation}


\section{Representation theory}\label{sec:rep_theory}

The content of the following presentation can be found with more details for example in \cite{travis_lectures_2021,ceccherini2010representation} (see also the Ph.D manuscript \cite[Chap. $2$]{yegor_thesis}).

\subsection{Representation theory of a finite group $G$.}

\begin{theorem}[Maschke's theoreom]\label{theo:Maschke}
	Let $(\rho,V)$ be a representation of a finite group $G$. Then it can be decomposed into a direct sum of irreducible representations: 
	\begin{equation}
		V=\bigoplus_{i=1}^{N} V_i\,.
	\end{equation}
\end{theorem}
\noindent This property is called \textit{complete reducibility}, or \textit{semisimplicity}.\medskip

Let $(\rho,V)$ and $(\lambda,W)$ be two representations of $G$, and let $\text{Hom}(V,W)$ denote the set of all linear transformation $T\,:\, V\to W$. If the representation $\rho$ and $\lambda$ are such that 
\begin{equation}
	T\rho(g)=\lambda(g) T, \quad \text{for all } g\in G
\end{equation}
we say that $T$ \textit{intertwines} $\rho$ and $\lambda$. We denote by $\text{Hom}_{G}(V,W)$ the set of all such transformations which are sometimes refered to as \textit{$G$-linear map} or \textit{$G$-module homomorphisms}. Two representations $(\rho,V)$ and $(\lambda,W)$ are said to be $\textit{equivalent}$ if there exist $T\in\text{Hom}_{G}(V,W)$ which is bijective.

\begin{lemma}[Schur's lemma]
	Let $(\lambda,V)$ and $(\sigma,W)$ be two irreducible representations of $G$ over the complex field $\C$.
	\begin{itemize}
		\item[i)] If $T\in\textup{Hom}_{G}(V,W)$ then either $T=0$ or $T$ is an isomorphism (the representation $\rho$ and $\lambda$ are equivalent):
		$T$ is an invertible square matrix
		\item[ii)] $\textup{Hom}_{G}(V,V)=\mathbb{C}\id_{V}.$\\
		In other words any operator $T$ which commutes with $\rho(g)$ is proportional to the identity matrix.
	\end{itemize}        
\end{lemma}
\begin{proof}
	\begin{itemize}
		\item[\it{i)}] Let $T\in \text{Hom}_G(V,W)$, then $\text{Ker}(T)\subseteq V$ and $\text{Im}(T)\subseteq W$ are $G$-invariant. Besides, by irreducibility of $(\lambda,V)$ and $(\sigma,W)$ they have no non trivial $G$- invariant subspaces. So either $\text{Ker}(T)=V$ so that $T=0$ or $\text{Ker}(T)=\lbrace 0 \rbrace$ and necessarily $\text{Im}(T)=W$ so that $T$ is an isomorphism.
		\item[\it{ii)}] Let $T\in \text{Hom}_G(V,V)$. Since $\mathbb{C}$ is algebraically closed, $T$ has at least one eigenvalue. Let $v$ be an eigenvector of $T$ with eigenvalue $\alpha$. Then $(T-\alpha I_V)v=0$ and $\text{Ker}(T-\alpha I_V)$ is non trivial. Besides $T-\alpha I_V\in \text{Hom}_G(V,V)$ and by \textit{i)}
		$T=\alpha \id_V$. 
	\end{itemize}
\end{proof}
The algebra $\text{Hom}_{G}(V,V)\subseteq \text{End}(V)$ is called the \textit{centralizer algebra} of the representation $(\rho,V)$.
\begin{corollary}\label{cor:schur}
	If $(\rho_V,V)$ and $(\rho_W,W)$ are irreducible representations of G then 
	\begin{equation}
		\textup{dim}(\textup{Hom}_G(V,W))=\left\{
		\begin{array}{ll}
			1\quad \text{if } \rho_V\sim\rho_W,\\
			0\quad \text{if } \,\rho_V\nsim\rho_W.
		\end{array}
		\right.
	\end{equation}
\end{corollary}
\begin{proof}
	Let $V$ and $W$ be two equivalent representation. Choose any two intertwining operators $T_1\,:\, V\to W$ and $T_2\,:\, V\to W$,
	\begin{equation}
		T_1\,\rho_V(g)=\rho_W(g) T_1\,,\hspace{0.5cm}T_2\,\rho_V(g)=\rho_W(g)\,T_2\,.
	\end{equation}
	Then,
	\begin{equation}
		T_2\,T_1^{\shortminus 1}\rho_W(g)=\rho_W(g) T_2\,T_1^{\shortminus 1}\,.
	\end{equation}
	That is  $T_2\,T_1^{\shortminus 1}=\alpha \id_W$, so $T_2=\alpha\, T_1$.
\end{proof}

\subsection{Representation theory of a finite dimensional algebra $\mathcal{A}$.}\label{subsec:rep_theo_Algebra}

In this section we assume that the finite dimensional algebra $\mathcal{A}$ is associative and unital. 

\paragraph{Modules.} In the context of algebras, a $\mathcal{A}$-module is a vector space acted upon by the algebra $\mathcal{A}$, in other words it is a representation of $\mathcal{A}$.

\paragraph{Semisimplicity.} A module isomorphic to a direct sum of irreducible modules is called \textit{semisimple}. An algebra $\mathcal{A}$ is semisimple if its left (or right) regular module is semisimple. 

\begin{theorem}[Maschke's theoreom]\label{theo:Maschke2}
	Let $G$ be a finite group and $\mathbb{F}$ a field whose characteristic does not divide the order of $G$. Then, the group algebra $\mathbb{F}G$ is semisimple. 
\end{theorem}

\paragraph{Ideals.} A left ideal of an algebra $\mathcal{A}$ is a subalgebra $I\subseteq \mathcal{A}$ such that $a\,I\subseteq I$ for all $a\in A$. Similarly, a right ideal of an algebra $\mathcal{A}$ is a subalgebra $I\subseteq \mathcal{A}$ such that $I\, a\subseteq I$ for all $a\in \mathcal{A}$. A two-sided ideal is a subspace that is both a left and right ideal. 

\begin{theorem} Any finite dimensional semisimple algebra $\mathcal{A}$ is a direct sum of its minimal left (or right) ideals $I_{i}$ ($i=1\,,\ldots\,,N$)
\begin{equation}\label{eq:left_ideals_decomp}
\mathcal{A}=\bigoplus_{i=1}^{N} \, I_i\,.
\end{equation}
\end{theorem}

We denote by $\Irr(\mathcal{A})$ the set of all classes of isomorphic irreducible $\mathcal{A}$-modules.
\begin{corollary} Let $V^{\lambda}$, with $\lambda\in\Irr(\mathcal{A})$, denote the pairwise inequivalent irreducible modules over a finite dimensional semisimple algebra $\mathcal{A}$. Then $\mathcal{A}$ admits the following decomposition with $m_\lambda\geqslant 1$: 
\begin{equation}
\mathcal{A}\cong \bigoplus_{\lambda\in \Irr(\mathcal{A})} \left(V^{\lambda}\right)^{\oplus m_\lambda}\,,
\end{equation}
where $m_\lambda$ is called the multiplicity of $V^{\lambda}$ in $\mathcal{A}$.
\end{corollary}

\paragraph{Complete set of pairwise orthogonal primitive idempotents.} Let $\mathcal{A}$ be a unital finite dimensional algebra over a field $\mathbb{F}$. If $e_1$ and $e_2$ are idempotents we say that they are \textit{orthogonal} if and only if $e_1e_2=e_2e_1=0$. An idempotent $e\in \mathcal{A}$ is called \textit{primitive} if and only if $e$ cannot be written as a sum of non zero idempotents $e_1$ and $e_2$ in $\mathcal{A}$.\medskip

A \textit{complete set of pairwise orthogonal primitive idempotents} is a set $\lbrace e_i\rbrace_{i\in I}$ of elements satisfying $\sum_{i\in I}e_i=1_\mathcal{A}$ with $e_ie_j=\delta_{ij}e_i$ for $i,j\in I$, with $|I|$ maximal. The left ideals in \eqref{eq:left_ideals_decomp} can be realized as $I_i=\mathcal{A}\, e_i$.

\paragraph{Isotypic components and isotypic decomposition (\cite[Def. $1.2.9$]{ceccherini2010representation}, \cite[Def. $7.2.7$]{ceccherini2010representation}).} Let $\mathcal{A}$ be a finite dimensional algebra and let $V$ be an $\mathcal{A}$-module. We denote by $\Irr(\mathcal{A})$ the set of all classes of isomorphic irreducible $\mathcal{A}$-modules, and by $W^\lambda$ a representative of the class $\lambda\in \Irr(\mathcal{A})$ appearing in $V$ with multiplicity $m_\lambda$.\smallskip

For each $\lambda\in \Irr(\mathcal{A})$ the $\lambda$-isotypic component of $V$ is defined as  
\begin{equation}
	\mathcal{W}^\lambda=m_\lambda W^\lambda=\left(W^\lambda\right)^{\oplus m_\lambda}\,,
\end{equation}
while the $\mathcal{A}$-isotypic decomposition of $V$ is defined as 
\begin{equation}
V=\bigoplus_{\lambda\in \Irr(\mathcal{A})}\mathcal{W}^\lambda\,.
\end{equation}

\paragraph{Matrix rings, matrix algebras.} A matrix ring is a set of matrices with entries in a ring $R$ that form a ring under matrix addition and matrix multiplication. The set of $\Dim \times \Dim$ matrices with entries in $R$ is a matrix ring denoted $M_\Dim(R)$. When $R$ is a commutative ring, the matrix ring $M_{\Dim}(R)$ is an associative algebra over $R$ and $M_{\Dim}(R)$ is called a matrix algebra.\medskip

Consider the algebra of complex matrices $M_\Dim(\C)$ of size $\Dim\times\Dim$.  Any two vectors in $\C^\Dim$ are related by matrix multiplication of an element in  $M_\Dim(\C)$. Hence, $\C^\Dim$ is an irreducible $M_\Dim(\C)$-module. 

\begin{theorem}[Wedderburn-Artin Theorem]\label{theo:Wedderburn}
	Let  $\mathcal{A}$ be a finite dimensional algebra over $\C$, and let $V^\lambda$ with $\lambda\in \Irr(\mathcal{A})$ denote all pairwise inequivalent irreducible $\mathcal{A}$-modules. Then the following properties are equivalent:
	\begin{itemize}
	\item[i)] $\mathcal{A}$ is semisimple.
	\item[ii)] $\mathcal{A}$ is isomorphic to a direct sum of full matrix algebras
	\begin{equation}
	\mathcal{A}\cong \bigoplus_{\lambda\in \Irr(\mathcal{A})} M_{\Dim_\lambda}(\C)\,, \hspace{0.5cm} \text{with}\hspace{0.5cm} \Dim_\lambda=\dim(V^\lambda)\,.
	\end{equation}
	\item[iii)] \begin{equation}
	\mathcal{A}\cong \bigoplus_{\lambda\in \Irr(\mathcal{A})} \End(V^\lambda)\,.
	\end{equation}
	\end{itemize}
\end{theorem}
See \cite{travis_lectures_2021} for the proof.

\subsection{Double Centralizer Theorem and Schur-Weyl dualities}\label{subsec:Double Centralizer Theorem}

\begin{lemma}[Schur's lemma]
	Let $(\lambda,V)$ and $(\sigma,W)$ be two irreducible left modules over a finite dimensional $\C$-algebra $\mathcal{A}$.
	\begin{itemize}
		\item[i)] If $T\in\textup{Hom}_{\mathcal{A}}(V,W)$ then either $T=0$ or $T$ is an isomorphism.
		\item[ii)] Suppose $V$ is finite dimensional. If  $T\in \textup{Hom}_{\mathcal{A}}(V,V)$ then $T=\alpha \id_{V}$ for some $\alpha\in \C$.
	\end{itemize}        
\end{lemma}

\begin{example}
Let $T \in M_\Dim(\C)$ be such that
\begin{equation}
	T\,M=M\, T \hspace{1cm} \text{for all $M\in M_\Dim(\C) $}
\end{equation} 
Then by Schur's lemma $T$ is proportional to the identity matrix. 
\end{example}

\begin{theorem}[Double Centralizer Theorem]\label{theo:double_centralizer_Theorem}
	Given a finite dimensional vector space $V$, let $A$ be a completely reducible (semisimple) subalgebra of $\,\textup{End}(V)$, and $\,B=\textup{End}_A(V)$ the centralizer algebra of $A$ in $\textup{End}(V)$. Then :
	\begin{enumerate}
		\item  $B$ is semisimple 
		\item  $A$ is the centralizer algebra of $B$ in $\textup{End}(V)$: $A=\textup{End}_B(V)$
		\item As a  $B\otimes A$-module (that is under the joint action of $B$ on the left and $A$ on the right) we have the decomposition 
		\begin{equation}
			V\cong \bigoplus_{i} W_{i} \otimes U_{i},
		\end{equation}
		where $U_{i}$ are irreducible modules of $A$ and each $W_{i}\cong \textup{Hom}_{A}(U_i,V)$ is either an irreducible representation of $B$ or zero. Furthermore, the non zero $W_i$'s are all the simple modules of B.\\
	\end{enumerate}
\end{theorem}
Recall that $\textup{Hom}_{A}(U_i,V)$ is the set of operators $T: U_i\to V$ such that 
\begin{equation}
T \rho_i(g) = \rho(g) T\,, \hspace{1cm} \text{for all $g\in A$}
\end{equation}
where $\rho_i$ and $\rho$ are the homomorphism associated with representations space $U_{i}$ and $V$ of $A$. See for example \cite{stevens2016schur} and also \cite[Theorem 7.3.2]{ceccherini2010representation} for the proof.\medskip


Originally, the Schur-Weyl duality relates the irreducible tensor representations of $\GL(\Dim,\mathbb{C})$ and $\sn$. One identifies $A$ as the algebra of operators in $\End(V^{\otimes n})$ generated by the action of $\sn$ on $V^{\otimes n}$, while $B$ is the algebra of operators in $\End(V^{\otimes n})$ generated by the action of $\GL(\Dim,\mathbb{C})$ on $V^{\otimes n}$.\medskip

%

\begin{theorem}[Schur-Weyl duality 1]
	Let $V$ be a complex vector space with $\dim(V)=\Dim$. As a representation of $\,\GL(\Dim,\C) \times \sn$, $V^{\otimes n}$ decomposes as 
	\begin{equation}
		V^{\otimes n }\cong \bigoplus_{\mu \in \mathcal{P}_n(\Dim)} V^{\mu}\otimes L^{\mu}
	\end{equation}
	where $V^\mu$ \lp respectively $L^{\mu}$\rp\, are pairwise inequivalent irreducible representation of $\text{GL}(\Dim)$ \lp respectively $\sn$\rp.
\end{theorem}

\begin{theorem}[Schur-Weyl duality 2]
	Let $V$ be a complex vector space equipped with a symmetric bilinear form, with $\dim(V)=\Dim$. As a representation of $\, \Or(\Dim,\C)\times\bn $, $V^{\otimes n}$ decomposes as 
	\begin{equation}
		V^{\otimes n}\cong \bigoplus_{\lambda\in \Lambda_n(\Dim)} D^{\lambda}\otimes M^{\lambda}_n
	\end{equation}
	where $D^{\lambda}$ \lp respectively $M^{\lambda}_n$\rp\, are pairwise inequivalent irreducible representations of $\text{O}(\Dim,\C)$ \lp respectively $\bn(\Dim)$\rp.\\
\end{theorem}

\subsection{Littlewood-Richardson rules}\label{subsec:Littlewood_Richardson_rules}
The set of all Young diagrams is weakly ordered by inclusion: $\lambda\subset \mu$ implies that any box of $\lambda$ is also present in $\mu$. For a pair $\lambda\subset\mu $ define the {\it skew-shape Young diagram} $\mu\backslash \lambda$ as a set-theoretical difference of the corresponding Young diagrams. We set by definition $|\mu\backslash \lambda| = |\mu| - |\lambda|$. For example, 
\begin{equation}\label{eq:skew-shape_example}
	\text{given}\quad\mu=\Yboxdim{9pt}\yng(4,2,2,1) \quad\text{and}\quad \lambda=\Yboxdim{9pt}\yng(2,1)\,,\quad\text{one has}\quad \mu\backslash \lambda =\Yboxdim{9pt}\young(\times\times ~~,\times ~,~~,~)
\end{equation}
\vskip 4 pt
\noindent
To each box of a (skew-shape) Young diagram at a position $(i,j)$ we associate its {\it content} $c(i,j) = j - i$. In this respect we define the content of any Young diagram as a sum:
\begin{equation}
	c_\lambda = \sum_{(i,j)\in \lambda} c(i,j)\,.
\end{equation}
One defines the content of a skew-shape Young diagram to be $c_{\mu\backslash\lambda} = c_\mu-c_\lambda$. For example, for the contents of the skew shape in \eqref{eq:skew-shape_example} one has
\begin{equation*}
	\mu\backslash \lambda =\Yboxdim{9pt}\young(\times\times \two\three,\times \zero,\mitwo \mione,\mithree) \quad \Rightarrow \quad c_{\mu\backslash \lambda} = -1\,.
\end{equation*}
\vskip 4 pt
The $\mathbb{N}$-span of Young diagrams is endowed with the structure of an associative commutative monoid with the unit element given by $\emptyset$. For the product of two diagrams $\lambda,\nu$ we shall write 
\begin{equation}\label{eq:LRRule}
	\lambda \LRp \nu = \sum_{\mu} \tensor{C}{^{\,\mu}_\lambda_\nu}\,\mu\,.
\end{equation}
The structure constants $\tensor{C}{^{\,\mu}_\lambda_\nu}$ are referred to as Littlewood-Richardson coefficients. They are calculated via the {\it Littlewood-Richardson rule}, which admits a number of equivalent ways to formulate it in terms of {\it semi-standard tableaux} \cite{Fulton}. 

\paragraph{First formulation of the rule.} Let us recall that a tableau of shape $\mu\backslash\lambda$ is any map which associates a positive integer to each box of $\mu\backslash\lambda$. A tableau is called semi-standard if the numbers in each row (respectively, column) form a weakly (respectively, strongly) increasing sequence. We will say that a partition $\nu = (\nu_1\geqslant \dots\geqslant \nu_r)$ is a {\it weight} of a semi-standard tableau $\mathsf{t}(\mu\backslash\lambda)$ if the latter contains exactly $\nu_i$ occurrences of the entry $i$. To any tableau $\mathsf{t}$ one associates a {\it row word} $w(\mathsf{t})$ by reading the entries of boxes along each line from left to right proceeding from the bottom line to the top one. A word is called {\it Yamanouchi word} (equivalently, Littlewood-Richardson word or reverse lattice word) if any its suffix contains at least as many $1$'s as $2$'s, at least as many $2$'s as $3$'s, {\it etc.} For example,
\begin{equation}\label{eq:example_Yamanouchi}
	\def\arraystretch{1.4}
	\begin{array}{l}
		\text{for}\;\;\mu\backslash \lambda =\Yboxdim{8pt}\young(\times\times~~,\times~,~)\;\;\text{and the weight}\;\;\nu = \Yboxdim{8pt}\young(~~~,~)\\
		\text{one has exactly two semistandard tableaux}\;\; \mathsf{t}_1 = \Yboxdim{8pt}\young(\times\times\one\one,\times\one,\two)\,,\;\;\mathsf{t}_2 = \Yboxdim{8pt}\young(\times\times\one\one,\times\two,\one)\,,\\
		\text{such that the row words are Yamanouchi words}\;\;w(\mathsf{t}_1) = 2111\;\;\text{and}\;\;w(\mathsf{t}_1) = 1211\,.
	\end{array}
\end{equation}
\vskip 4 pt
\noindent With this at hand, one arrives at the following definition of Littlewood-Richardson coefficients:
\begin{equation}\label{eq:definition_LR_1}
	\def\arraystretch{1.4}
	\begin{array}{ll}
		\tensor{C}{^{\,\mu}_\lambda_\nu} & \text{is the number of semi-standard tableaux of the shape}\;\;\mu\backslash\lambda\;\;\text{and weight}\;\;\nu\\
		\hfill & \text{whose row word is a Yamanouchi word.}
	\end{array}
\end{equation}
If either $\lambda\not \subset \mu$ or $|\nu| \neq |\mu| - |\lambda|$, one puts $\tensor{C}{^{\,\mu}_\lambda_\nu} = 0$. It appears that $\tensor{C}{^{\,\mu}_\lambda_\nu} = \tensor{C}{^{\,\mu}_{\nu\lambda}}$, so $\tensor{C}{^{\,\mu}_\lambda_\nu} \neq 0$ also implies $\nu \subset \mu$. From the above example one obtains $C^{(4,2,1)}_{(2,1),(3,1)} = 2$.
\vskip 4 pt

\paragraph{Second formulation of the rule.} 
We will also need an equivalent definition of Littlewood-Richardson coefficients based on the {\it jeu de taquin}. Recall that a corner of a Young diagram is any box with no other boxes on the right and below. An inside corner of a skew-shape $\mu\backslash\lambda$ (with $\lambda \subset \mu$) is a corner of $\lambda$ which is not a corner of $\mu$. Any chosen inner corner of a semi-standard tableau $\mathsf{t}$ (thought as an empty box) can be removed by the following {\it sliding process}: at any step consider the neighbour(s) on the right and below the empty box, then slide the smallest one into the empty box, while if the two are equal, slide the one below. The process continues until the empty box becomes a corner of $\mu$, and is removed afterwards. The resulting tableau is again semi-standard, so one can repeatedly perform the sliding process until the shape of the semi-standard tableau becomes a Young diagram. The whole process is called jeu de taquin, and the resulting semi-standard tableau $\mathsf{t}^{\prime}$ is the same for any order of processing the inner corners. It is called the {\it rectification of} $\mathsf{t}$, $\mathsf{t}^{\prime} = \mathrm{Rect}(\mathsf{t})$. For example, for the tableau $\mathsf{t}_1$ in \eqref{eq:example_Yamanouchi} one has
\begin{equation}\label{eq:example_Rect}
	\def\arraystretch{2.4}
	\begin{array}{ll}
		\mathsf{t}^{\prime}_1 = \mathrm{Rect}\left(\Yboxdim{8pt}\young(\times\times\one\one,\times \one,\two)\right) =  \Yboxdim{8pt}\young(\one\one\one,\two)\,,\;\;\text{namely,} \;\;& \Yboxdim{8pt}\young(\times\times\one\one,\times \one,\two) \;\;\rightarrow \;\;\Yboxdim{8pt}\young(\times \one \one \one,\times \times,\two) \;\; \rightarrow \;\;\Yboxdim{8pt}\young(\times\one \one \one,\times,\two)\,,\\
		\hfill & \Yboxdim{8pt}\young(\times \one \one\one,\times ,\two) \;\;\rightarrow\;\; \young(\times \one \one\one,\two ,\times) \;\; \rightarrow \;\; \young(\times \one \one\one,\two)\;\;,\\
		\hfill & \Yboxdim{8pt}\young(\times \one \one\one,\two)\;\;\rightarrow\;\;\young(\one \times\one \one,\two)\;\;\rightarrow\;\;\young(\one \one \times\one,\two)\;\; \rightarrow\;\;\young(\one \one \one\times,\two)\;\;\rightarrow\;\;\young(\one \one \one,\two)\,.
	\end{array}
\end{equation}
One can check that for the other tableau in \eqref{eq:example_Yamanouchi} one has the same $\mathrm{Rect}(\mathsf{t}_2) = \mathsf{t}^{\prime}_1$.
\vskip 4 pt

For any Young diagram $\nu$ denote $\mathsf{E}(\nu)$ to be the semi-standard tableau of weight $\nu$, {\it i.e.} such that each $i$th row is filled with $i$. Then we arrive at the following equivalent definition of Littlewood-Richardson coefficients:
\begin{equation}\label{eq:definition_LR_2}
	\def\arraystretch{1.4}
	\begin{array}{ll}
		\tensor{C}{^{\,\mu}_\lambda_\nu} & \text{is the number of semi-standard tableaux of the shape}\;\;\mu\backslash\lambda\;\;\text{and weight}\;\;\nu\\
		\hfill & \text{whose rectification is $\mathsf{E}(\nu)$.}
	\end{array}
\end{equation}
By comparing the two examples \eqref{eq:example_Yamanouchi} and \eqref{eq:example_Rect}, one can verify that both definitions \eqref{eq:definition_LR_1} and \eqref{eq:definition_LR_2} lead to the same result $C^{(4,2,1)}_{(2,1)\, (3,1)} = 2$.
\vskip 4 pt

Let us consider a configuration obtained after a number of sliding processes during the jeu de taquin applied to a semi-standard tableau of a shape $\mu\backslash\lambda$, and let us keep the empty boxes at the end of each sliding process. Then one has a chain of three diagrams $\tau \subset \sigma \subset \mu$, where $\mu\backslash\sigma$ is the set of empty boxed resulting from the sliding processes, $\sigma\backslash\tau$ is a semi-standard tableau, and $\tau$ is the set of empty boxes not involved in the performed sliding processes. Let us say that a box of $\mu\backslash\sigma$ is an {\it addable corner} if adding it to $\sigma$ leads to a Young diagram. Note that each addable corner results from a sliding process applied to an inner corner, and that each particular sliding process is invertible. Let us define the {\it reverse sliding process} for any addable corner: at any step consider the neighbour(s) on the left and above the empty box, then slide the greater one into the empty box, while if the two are equal, slide the one above. The process continues until the empty box becomes an inner corner. In this respect, any chain $\tau \subset \sigma \subset \mu$, with a  semi-standard skew-shape diagram $\sigma\backslash\tau$, can be considered as an intermediate configuration of the jeu de taquin, with both types of slidings possible. Upon exhausting direct sliding processes, the unique terminal configuration (the rectification) with $\tau = \emptyset$ was considered above. On the other hand, starting from the terminal configuration of the shape $\mu\backslash\nu$, with the semi-standard tableau $\mathsf{E}(\nu)$, and going backwards by different sequences of reverse slidings until $\sigma = \mu$ leads to different semi-standard tableaux $\mathsf{t}^{(\mathrm{init})}$ of different shapes $\mu\backslash\lambda^{(\mathrm{init})}$.\\

Define $\mu\slashdiv\nu$ to be the set of so obtained diagrams $\lambda^{(\mathrm{init})}$. By construction, the terminal configuration $\mathsf{E}(\nu) = \mathrm{Rect}\big(\mathsf{t}^{(\mathrm{init})}\big)$ is the same for all $\mathsf{t}^{(\mathrm{init})}$, so according to the definition \eqref{eq:definition_LR_2}, the set $\mu\slashdiv\nu$ contains such diagrams $\lambda$ that $\tensor{C}{^{\,\mu}_\lambda_\nu} \neq 0$, and only them. For example, keeping empty boxes upon constructing $\mathrm{Rect}(\mathsf{t}_1)$ in \eqref{eq:example_Rect} and applying different sequences of reverse sliding processes gives three skew-shape diagrams:
\begin{equation}
	\Yboxdim{8pt}\young(\one \one \one \times,\two\times,\times)\;\;\xrightarrow{\text{rev. slides}}\;\; \Yboxdim{8pt}\young(\times\times\times\one ,\one \one,\two)\,,\;\; \Yboxdim{8pt}\young(\times\times\one\one ,\times \one,\two)\,,\;\; \Yboxdim{8pt}\young(\times\times\one\one ,\times \two,\one)\,,\;\; \Yboxdim{8pt}\young(\times\one\one\one ,\times \two,\times)\,,\quad \text{thus}\quad \Yboxdim{8pt}\young(~~~~,~~,~)\slashdiv\young(~~~,~) = \left\{\young(~,~,~)\,,\; \young(~~,~)\,,\;\young(~~~)\right\}\,.
\end{equation}

\chapter{Proofs}

\section{Proofs for chapter \ref{chap:Irreducible_MAG}}\label{app:proof_chap_2}

\subsection{Proof of Proposition \ref{prop:irreducible_distortion_GL}}\label{subsec:proof_of_prop_irreducible_distortion_GL}

\begin{proof}
	Let $\overline{Y}_S,\,\overline{Y}_A\in \mathbb{C}\mathfrak{S}_3 $ be such that 
	\begin{equation}\label{eq:ZAZS}
		\overline{Y}_S= z\, Y_{S} \, z^{\shortminus 1}\,, \hspace{1cm} \overline{Y}_A= z\, Y_{A} \, z^{\shortminus 1}\,, \hspace{0.5cm} \text{for some $z\in \mathbb{C}\mathfrak{S}_3$ such that $z^{\shortminus 1}=z^*$}.
	\end{equation}
	One has that projective invariance of $\, \accentset{\left(\Yboxdim{3pt}\yng(2,1)\right)}{C}\cdot \overline{Y}_A$ is equivalent to $ Y_{S} \, \overline{Y}_A =0$, 
	that is $\, \accentset{\left(\Yboxdim{3pt}\yng(2,1)\right)}{C}\cdot \overline{Y}_A$ must be antisymmetric with respect to permutations of the first and third indices. 
	
	\medskip On the other hand, because $\id_{\C\Sn{3}}=\overline{Y}_S+\overline{Y}_A$, one has that $\, \accentset{\left(\Yboxdim{3pt}\yng(2,1)\right)}{C}\cdot \overline{Y}_S$ must be symmetric with respect to permutations of the first and third indices. The result follows from the uniqueness of the decomposition of an order two tensor into symmetric and antisymmetric parts.
	%
	%
\end{proof}

\subsection{Proof of Proposition \ref{prop:lagrangians}}\label{subsec:proof_of_prop_lagrangians}

\begin{proof}
	Let us introduce some notations. For any $b\in \bn$ we denote by  $\Upsilon_b\,:\, V^{\otimes n} \times V^{\otimes n}\to \mathbb{C}$ the scalar-valued function defined for any $T_1,\, T_2 \in V^{\otimes n}$ by $\Upsilon_b(T_1,T_2)=\langle T_1\,\cdot b \, , \, T_2 \rangle$ and we set $\Upsilon(T_1,T_2)=\langle T_1 \, , \, T_2 \rangle$. 	For any $T_1,\, T_2 \in V^{\otimes n}$ let $b_1\,, b_2 \in \bn$ be such that $\Upsilon_{b_1}(T_1,T_2)=\alpha \Upsilon_{b_2}(T_1,T_2)$ for some non-zero constant $\alpha\in \C$. Then, due to the non-degeneracy of the scalar product, one has $\Upsilon_{b_1}-\alpha\Upsilon_{b_2}=0$ if and only if $\mathfrak{r}(b_1)=\alpha \, \mathfrak{r}(b_2)$.\medskip

	For any diagrams $b_1,\, b_2 \in B_n$, the two functions $\Upsilon_{b_1}$ and $\Upsilon_{b_2}$ are said to be equivalent, which is denoted $\Upsilon_{b_1}\propto \Upsilon_{b_2}$, if for any $T_1, T_2\, \in V^{\otimes n}$ there exist a non-zero constant $\alpha\in \C$ such that $\mathfrak{r}(b_1)=\alpha \, \mathfrak{r}(b_2)$. For any $b \in \bn (\Dim)$ and any two irreducible representations $D^{i}=V^{\otimes n}\cdot \accentset{(i)}{P}$, $D^{j}=V^{\otimes n}\cdot \accentset{(j)}{P}$ of $\Or(\Dim,\C)$ in $V^{\otimes n}$ we denote by $\Upsilon_b\big|_{(i,j)}\,$
	the restriction of $\Upsilon_b$ to $D^{i}\times D^{j}$.	For convenience we recall Schur's lemma applied to our context.\medskip


	\begin{lemma}[Schur's lemma applied to $\Or(\Dim,\mathbb{C})$]
		Let $(\,\rho^i,D^i)$ and $(\,\rho^j,D^j)$ be two irreducible representations of $\Or(\Dim,\mathbb{C})$, and let $\mathfrak{b}_{ij}\in \Hom_{\Or(\Dim,\C)}(D^i,D^j)$, that is:
		\begin{equation}
			\,\mathfrak{b}_{ij}\,:\, D^j\to D^i\, \hspace{0.5cm} \text{is such that}  \hspace{0.5cm} \mathfrak{b}_{ij}\,\rho^i(R)=\rho^j(R)\,\mathfrak{b}_{ij}\,,\hspace{0.5cm}\forall \, R\in \Or(\Dim,\C).
		\end{equation}
		\begin{itemize}
			\item[i)] 	If $i\neq j$, then either $\mathfrak{b}_{ij}=0$ or $\mathfrak{b}_{ij}$ is an isomorphism. In the latter case one says that $D^i$ is equivalent to $D^j$ and $\mathfrak{b}_{ij}$ is called an intertwiner.
			\item[ii)] 	If $i=j$, then $\mathfrak{b}_{ii}$ is proportional to the identity. 
			\item[iii)] If $D^i\cong D^j$ then  any two intertwiners $\mathfrak{b}^{(1)}_{ij}$ and $\mathfrak{b}^{(2)}_{ij}$ are proportional, $\mathfrak{b}^{(1)}_{ij}=\alpha \,\mathfrak{b}^{(2)}_{ij}$ for some $\alpha\in \C$. 
		\end{itemize}
	\end{lemma}
	The result follows from the following facts. For any $\, b_1, b_2 \in B_n$, one has:\medskip

	\noindent \textbf{Fact $1$}: If $D^i\ncong D^j$, $\Upsilon_{b_1}\big|_{(i,j)}=0\,$.\smallskip
	
	\noindent \textbf{Fact $2$}: $\Upsilon_{b_1}\big|_{(i,i)}\, \propto \, \Upsilon\big|_{(i,i)}$ if and only if $\,\accentset{(i)}{P} \, b_1\, \accentset{(i)}{P}\neq 0$.\smallskip
	
	\noindent \textbf{Fact $3$}: If $D^i\cong D^j$, then $\Upsilon_{b_1}\big|_{(i,j)}\, \propto \, \Upsilon_{b_2}\big|_{(i,j)}$ if and only if $\accentset{(i)}{P}\, b_{1}\accentset{(j)}{P}\neq 0$ and $\accentset{(i)}{P} \, b_{2}\, \accentset{(j)}{P}\neq 0$. 

%
\noindent \textbf{Proof of fact $1$}:  From point $\textit{i)}$ of Schur's lemma we have that $\,\accentset{(i)}{P} \, b_1 \,\accentset{(j)}{P}\,$ acts by zero on $D^i$.\smallskip

\noindent \textbf{Proof of fact $2$}: 
If $\,\accentset{(i)}{P} \, b_1 \accentset{(i)}{P} =0$ then $\Upsilon_{b_1}\big|_{(i,i)}=0$ is not equivalent to $\Upsilon\big|_{(i,i)}\neq 0$. If $\,\accentset{(i)}{P} \, b_1 \accentset{(i)}{P} \neq 0$ one has from point $\textit{ii)}$ of Schur's lemma  $\accentset{(i)}{P}\, b_1\, \accentset{(i)}{P}$ acts as a scalar multiple of the identity on the irreducible representation $D^{i}$.\smallskip
	
\noindent\textbf{Proof of fact $3$}: Let us first stress that in this case $\Upsilon\big|_{(i,j)} = 0$ because $\accentset{(i)}{P}\,\accentset{(j)}{P}=0$. If $\accentset{(i)}{P}\, b_1\, \accentset{(j)}{P}\neq 0$ and $\accentset{(i)}{P}\, b_2\, \accentset{(j)}{P}\neq 0$, we have $\Upsilon_{b_1}\big|_{(i,j)}\, \propto \, \Upsilon_{b_2}\big|_{(i,j)}$ as a direct consequence of point $\textit{iii)}$ of Schur's lemma. The image of $\accentset{(i)}{P}\, b_1\, \accentset{(j)}{P}$ in $\Hom_{\Or(\Dim,\C)}(D^i,D^j)$ is an isomorphism. Hence, from the non-degeneracy of the scalar product, if $\accentset{(i)}{P}\,b_{1}\,\accentset{(j)}{P}\neq 0$ then $\Upsilon_{b_1}\big|_{(i,j)}\neq 0$. The implication from left to right is straightforward.\medskip
	
\end{proof}

\section{Proofs for chapter \ref{chap:projectors_GL}}\label{app:proof_chap_3}

\subsection{Proof of Lemma \ref{lem:simple_ev_Tn}}\label{subsec:proof_simple_ev_Tn}

\begin{proof}
	By direct computation of the contents of Young diagrams one can infer that if $n<6$ or $n=7$ the spectrum of $T_n$ is simple. For $n=6$ one has $c_{(4,1^2)}=c_{(3,3)}=3.$
	For $n>7$ there is always at least two Young diagrams which are symmetric with respect to the diagonal line, that is with zero content. For $n$ even a general pair of such diagram are described by the integer partitions $\mu=(n/2,2,1^{n/2-2})$ and $\nu=(n/2-1,3,2,1^{n/2-4})$. For $n$ odd one has $\mu=((n+1)/2,1^{(n-1)/2})$ and $\nu=((n-3)/2,3,3,1^{(n-9)/2})$.
\end{proof}

\subsection{Proof of Lemma \ref{lem:induction_center_Sn}}\label{subsec:proof_induction_center_Sn}

\begin{proof}
	Let us show that the action of $\mathcal{L}$ on the averaged conjugacy class sums is given by $\mathcal{L}(\bar{K}_\nu)=\bar{K}_{\nu_{[1]}}$. By definition \eqref{eq:def_L_li} one has
			\begin{equation}
				\mathcal{L}(\bar{K}_\nu)=\sum_{i=1}^n\sum_{\gamma\in \Sn{n\shortminus1}} c_i \,\gamma \, s \, \gamma^{\shortminus 1} \, c_i^{\shortminus 1}\,,
			\end{equation}
			where $s$ is an element of the conjugacy class $C_\nu$ viewed as a element of $\sn$, hence $s(n)=n$ and $s\in C_{\nu_{[1]}}$. One the other hand one has 
			\begin{equation}
				\bar{K}_{\nu_{[1]}}=\sum_{\sigma\in \sn} \sigma s \sigma^{\shortminus 1}= \sum_{i=1}^{n}\sum_{\begin{array}{c}
						{\scriptstyle \sigma\in \sn\,}\\
						{\scriptstyle\sigma(n)=i}
				\end{array}} \sigma s \sigma^{\shortminus 1}\,.
			\end{equation}
			For any $\sigma\in \sn$, such that $\sigma(n)=i$ one has $(c_i^{\shortminus 1}\sigma)(n)=n$, that is $c_i^{\shortminus 1}\sigma\in \Sn{n\shortminus1}$. Then for any $i\in \lbrace 1\,, 2, \ldots, n\rbrace$
			\begin{equation}
				\sum_{\begin{array}{c}
						{\scriptstyle \sigma\in \sn\,}\\
						{\scriptstyle\sigma(n)=i}
				\end{array}} c_i^{\shortminus 1}\sigma s \sigma^{\shortminus 1}c_i=\sum_{\gamma\in \Sn{n\shortminus1}} \,\gamma \, s \, \gamma^{\shortminus 1} \,,
			\end{equation}
		hence
			\begin{equation}
			\sum_{\begin{array}{c}
					{\scriptstyle \sigma\in \sn\,}\\
					{\scriptstyle\sigma(n)=i}
			\end{array}}\sigma s \sigma^{\shortminus 1}=\sum_{\gamma\in \Sn{n\shortminus1}} \, c_i  \gamma \, s \, \gamma^{\shortminus 1}c_i^{\shortminus 1}\,,
		\end{equation}
and $\mathcal{L}(\bar{K}_\nu)=\bar{K}_{\nu_{[1]}}$.
		\end{proof}

\subsection{Proof of Proposition \ref{prop:line_induction_sn}}\label{subsec:proof_line_induction_sn}
\begin{proof}
	We first expand $\mathcal{L}(Z^{\nu})$ in terms of the conjugacy class sums. From \eqref{eq:centralYoung_Characters} and \eqref{lem:induction_center_Sn} we have 
	\begin{equation}
		\mathcal{L}(Z^{\nu})=\frac{\mathrm{d}_\nu}{(n-1)!}\sum_{\rho\vdash n-1} \frac{\stab(\rho_{[1]})}{\stab(\rho)}\, \chi^{\nu}_\rho K_{\rho_{[1]}}.
	\end{equation}
	Recall remark $iii)$ below \eqref{eq:character_central_idempotent_Sn}
	\begin{equation*}
		K_{\rho_{[1]}} Z^{\mu}=\dfrac{n!}{\mathrm{d}_\mu \stab(\rho_{[1]})}\chi^{\mu}_{\rho_{[1]}}Z^{\mu},
	\end{equation*}
	and then by linearity 
	\begin{equation}
		\mathcal{L}(Z^{\nu})Z^\mu=n\dfrac{\mathrm{d}_\nu}{\mathrm{d}_\mu} \left(\sum_{\rho\vdash n-1}\,\frac{1}{\stab(\rho)}\chi^{\nu}_\rho\chi^{\mu}_{\rho_{[1]}}\right)Z^{\mu}.
	\end{equation}
	Finally, from the branching rule \eqref{eq:branching_restriction_Sn} one has
	\begin{equation*}
		\chi^{\mu}_{\rho_{[1]}}=\sum_{\beta \in \mathcal{R}_\mu}\chi^{\beta}_\rho\,,
	\end{equation*}
	and the result follows from the orthogonality relations of the irreducible characters of $\C\mathfrak{S}_{n-1}$ \eqref{eq:first_ortho_Characters}.
\end{proof}
\section{Proofs for chapter \ref{chap:projectors_O}}\label{app:proof_chap_4}

\subsection{Proof of Lemma \ref{lem:arcinduction_tens}}\label{subsec:proof_Lemma_arc_induc_tens}

\begin{proof}
	
	Let us first demonstrate:\medskip
	\begin{flushleft}
		\begin{tabular}{ll}
			\textbf{Fact 2:} & \hspace{2cm} If $\,|\beta|=|\lambda|\,$ and $\,\beta\neq \lambda\,$ then $\,T^{(\beta,2)}\cdot \mathcal{A}(P^\lambda_{n-2})=0$.\medskip
		\end{tabular}
	\end{flushleft}
	\begin{mdframed}[style=mystyle,frametitle=Proof of Fact 2.]\label{ex:preLemma_Arc_induction_2}
		Consider the following product: 
		\begin{equation}
			T^{(\beta,2)}\cdot \mathcal{A}(P^\lambda_{n-2})=\sum_{1\leqslant i <j \leqslant n} \,\,\, \raisebox{-.42\height}{\includegraphics[scale=0.7]{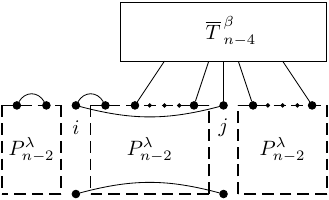}}
		\end{equation}
		There are three possible configurations for the position of the arc $(ij)$:
		\begin{itemize}
			\item[$1)$] The arc $(ij)$ is  directly connected to $\overline{T}^\beta_{n-4}$:\medskip
			
			Since the tensor $\overline{T}^\beta_{n-4}$ is traceless, $T^{(\beta,2)}\cdot \mathfrak{a}_{ij}(P^\lambda_{n-2})=0$. 
			\item[$2)$] The arc $(ij)$ is connected to the arcs of $T^{(\beta,2)}$:\medskip
			
			\forceindent By analogy with case $2)$ of \textbf{Fact $1$}, configurations where the arc $(ij)$ is connected to only one arc (either $(ij)=(12)$ or $(ij)=(34)$) of $T^{(\beta,2)}$ yields zero. Indeed, $g\otimes\overline{T}^\beta_{n-4}\in D^\beta$ and $D^\beta\cdot P^{\lambda}_{n-2}=0$.\medskip 
			
			Hence we are left with cases where two arcs of $T^{(\beta,2)}$ are involved, for example $(ij)=(14)$:
			\begin{equation}
				T^{(\beta,2)}\cdot \mathfrak{a}_{ij}(P^\lambda_{n-2})=\raisebox{-.43\height}{\includegraphics[scale=0.8]{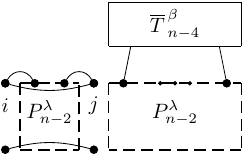}}\,
			\end{equation}
			Let us analyze the structure of $P^{\lambda}_{n-2}$ in more detail. 
			Here, $\lambda\vdash n-4$ and $P^{\lambda}_{n-2}\in J_{1}$ gives a $2$-traceless projection of tensors. All diagrams entering $P^{\lambda}_{n-2}$ which are in $J_2$ annihilate $\overline{T}^\beta_{n-4}$. Hence it is enough to focus on the element $P^{\lambda}_{n-2}\Big{|}_{1}$, the part of $P^{\lambda}_{n-2}$ where diagrams have only $1$ arcs.\medskip
			
			\forceindent All diagrams with two passing lines connected to the resting available arc nodes of $T^{(\beta,2)}$ must have an upper arc connected to $\overline{T}^\beta_{n-4}$, thus resulting in zero contribution.\medskip
			
			\forceindent Consider the part $Z^{\lambda}_{(a_0)}$ within $P^\lambda_{n-2}\Big{|}_{1}$ where diagrams are such that the upper arc $a_0$ connects the arc of $T^{(\beta,2)}$ as depicted below
			\begin{equation}\label{eq:all_arc_connected}
				T^{(\beta,2)}\cdot \mathfrak{a}_{ij}(Z^{\lambda}_{(a_1a_2)})=\raisebox{-.43\height}{\includegraphics[scale=0.8]{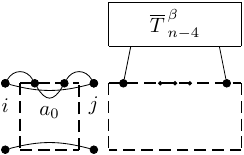}}\,
			\end{equation}  
			Each part of $Z^{\lambda}_{(a_0)}$ with diagrams having a fixed position of the lower arc, acts on $\overline{T}^{\,\beta}_{n-4}$ by permutations resulting in a projection to the space of traceless tensors $D^{\lambda}$.	Hence we conclude that this configuration yields zero. 
			
		
		\forceindent Consider the part $Z^{\lambda}_{(a_0)}$ within $P^\lambda_{n-2}\Big{|}_{1}$ where diagrams are such that their upper arc $a_0$ connects an arc of $T^{(\beta,2)}$ with a node of $\overline{T}^\beta_{n-4}$ as depicted below
		
		\begin{equation}\label{eq:folded_diagram_3tl0}
			T^{(\beta,2)}\cdot \mathfrak{a}_{ij}(Z^{\lambda}_{(a_0)})=\raisebox{-.43\height}{\includegraphics[scale=0.8]{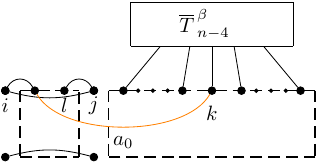}}\,
		\end{equation}  
		The configuration \eqref{eq:folded_diagram_3tl0} can be cast in the following forms
		\begin{equation}\label{eq:unfolded_diagram_3tl0}
			\raisebox{-.43\height}{\includegraphics[scale=0.75]{fig/case_1_Arc_indcution_proof_3_traceless_arc_b.pdf}}=\frac{1}{\Dim}\hspace{0.3cm}	\raisebox{-.43\height}{\includegraphics[scale=0.75]{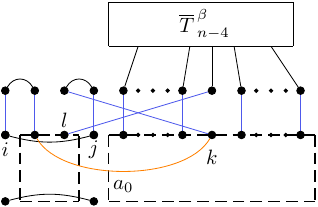}}=\frac{1}{\Dim}\hspace{0.3cm}\raisebox{-.43\height}{\includegraphics[scale=0.75]{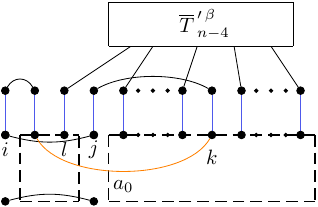}}\,
		\end{equation}
		The above sequence of diagrams translates to
		\begin{equation}
			T^{(\beta,2)}\cdot \mathfrak{a}_{ij}(Z^{\lambda}_{(a_0)})=\dfrac{1}{\Dim}\,T^{(\beta,2)}\cdot \left(s_{lk}\, \mathfrak{a}_{ij}(Z^{\lambda}_{(a_0)}) \right)=\dfrac{1}{\Dim}\,T^{\,\prime\,(\beta,2)}\cdot \mathfrak{a}_{ij}(Z^{\lambda}_{(a_0)})\,,
		\end{equation}
		where 
		\begin{equation}
			T^{\,\prime\,(\beta,2)}=T^{(\beta,2)}\cdot s_{lk}=\raisebox{-.5\height}{\includegraphics[scale=0.8]{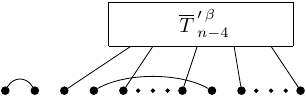}}\,\,\,\in D^\beta\,.
		\end{equation}
		From the right most diagram of \eqref{eq:unfolded_diagram_3tl0} it is clear that all passing lines within $Z^{\lambda}_{(a_0)}$ connects to $\overline{T}^{\,\prime\,\beta}_{n-4}$ resulting in its projection to $D^\lambda$. Hence we have 
		\begin{equation}
			T^{(\beta,2)}\cdot \mathfrak{a}_{ij}(Z^{\lambda}_{(a_0)})=0,
		\end{equation}
		and 
		\begin{equation}
			T^{(\beta,2)}\cdot \mathfrak{a}_{ij}(P^{\lambda}_{n-2})=0.
		\end{equation}
		This leaves us with the last possible configuration for the position of the arc $(ij)$.
		\item[$3)$] One node of the arc $(ij)$ is connected to a node within the arc of $T^{(\beta,2)}$ and the other one is connected to $\overline{T}^\beta_{n-4}$:
		\begin{equation}
			T^{(\beta,2)}\cdot \mathfrak{a}_{ij}(P^\lambda_{n-2})=\raisebox{-.43\height}{\includegraphics[scale=0.7]{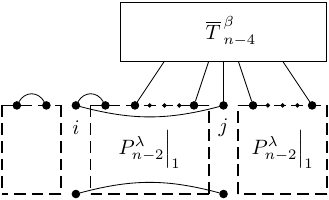}}\,.
		\end{equation}
		This case is similar to case $3)$ of the previous example, exempt we now have to consider diagrams entering $P^{\lambda}_{n-2}$ which have one arc.\medskip
		
		\forceindent Consider the part $Z^{\lambda}_{(a_0)}$ within $P^\lambda_{n-2}\Big{|}_{1}$ where diagrams are such that their upper arc connects with two nodes $a_1$, $a_2$, where $a_2$ is a node from $\overline{T}^{\,\beta}_{n-4}$. For example,
		
		\begin{equation}\label{eq:folded_diagram_3tl}
			T^{(\beta,2)}\cdot \mathfrak{a}_{ij}(Z^{\lambda\,,(1,4)}_{})=\raisebox{-.43\height}{\includegraphics[scale=0.6]{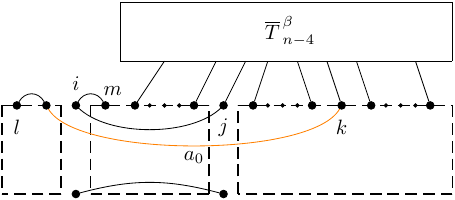}}\,
		\end{equation}
		We can now repeat a similar analysis to the previous example. The configuration \eqref{eq:folded_diagram_3tl} can be cast in following forms
		\begin{equation}\label{eq:2chains_arcs}
			\hspace{-1cm}\raisebox{-.43\height}{\includegraphics[scale=0.58]{fig/case_2_Arc_indcution_proof_3_traceless_arc_2lines.pdf}}=\frac{1}{\Dim^2}\hspace{0.3cm}	\raisebox{-.43\height}{\includegraphics[scale=0.58]{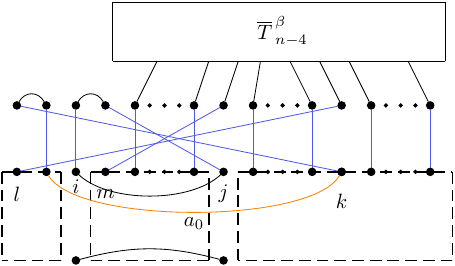}}=\frac{1}{\Dim^2}\hspace{0.3cm}\raisebox{-.43\height}{\includegraphics[scale=0.58]{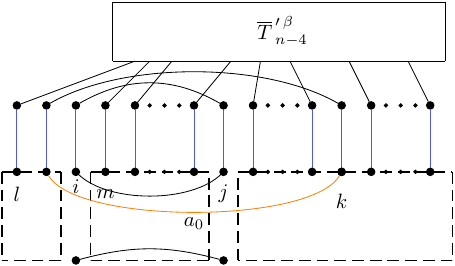}}\,
		\end{equation}
		The above sequence of diagrams translates to
		\begin{equation}
			T^{(\beta,2)}\cdot \mathfrak{a}_{ij}(Z^{\lambda}_{(a_0)})=\dfrac{1}{\Dim}\,T^{(\beta,2)}\cdot \left(s_{lk}s_{mj}\, \mathfrak{a}_{ij}(Z^{\lambda}_{(a_0)}) \right)=\dfrac{1}{\Dim}\,T^{\,\prime\,(\beta,2)}\cdot \mathfrak{a}_{ij}(Z^{\lambda}_{(a_0)})\,,
		\end{equation}
		where 
		\begin{equation}
			T^{\,\prime\,(\beta,2)}=T^{(\beta,2)}\cdot \left(s_{lk}s_{mj}\right)=\raisebox{-.5\height}{\includegraphics[scale=0.8]{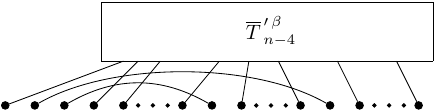}}\,\,\,\in D^\beta\,.
		\end{equation}
		We have that for any such position of the arc $a_0$, permutation lines of $Z^{\lambda}_{(a_0)}$ directly connect to $\overline{T}^{\,\prime\,\beta}_{n-4}$ which gives zero.\\
		
		\forceindent Consider the last part $Z^{\lambda}_{(a_0)}$ within $P^\lambda_{n-2}\Big{|}_{1}$ where diagrams are such that the upper arc $a_0$ connects the arcs of $T^{(\beta,2)}$. For example,
		
		\begin{equation}\label{eq:folded_diagram_3tlb}
			T^{(\beta,2)}\cdot \mathfrak{a}_{ij}(Z^{\lambda}_{(a_0)})=\raisebox{-.43\height}{\includegraphics[scale=0.8]{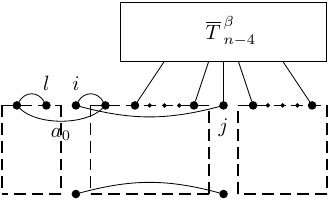}}\,
		\end{equation}
		This  configuration is very similar to \eqref{eq:folded_diagram_3tl0} and by analogy we can directly that \eqref{eq:folded_diagram_3tlb}. Nevertheless, for completeness, we reproduce the analysis.
		The configuration \eqref{eq:folded_diagram_3tlb} can be cast in following forms
		\begin{equation}\label{eq:1chain_arcs}
			\raisebox{-.43\height}{\includegraphics[scale=0.7]{fig/case_2_Arc_indcution_proof_3_traceless_arc_14.pdf}}=\frac{1}{\Dim}\hspace{0.3cm}	\raisebox{-.43\height}{\includegraphics[scale=0.7]{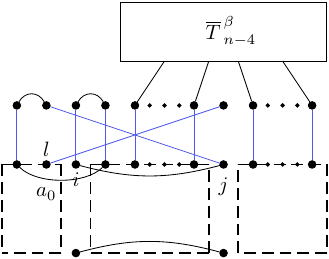}}=\frac{1}{\Dim}\hspace{0.3cm}\raisebox{-.43\height}{\includegraphics[scale=0.7]{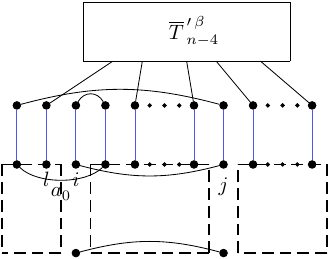}}\,
		\end{equation}
		The above sequence of diagrams translates to
		\begin{equation}
			T^{(\beta,2)}\cdot \mathfrak{a}_{ij}(Z^{\lambda}_{(a_0)})=\dfrac{1}{\Dim}\,T^{(\beta,2)}\cdot \left(s_{lj}\, \mathfrak{a}_{ij}(Z^{\lambda}_{(a_0)}) \right)=\dfrac{1}{\Dim}\,T^{\,\prime\,(\beta,2)}\cdot \mathfrak{a}_{ij}(Z^{\lambda}_{(a_0)})\,,
		\end{equation}
		where 
		\begin{equation}
			T^{\,\prime\,(\beta,2)}=T^{(\beta,2)}\cdot s_{lj}=\raisebox{-.5\height}{\includegraphics[scale=0.7]{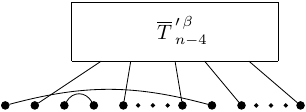}}\,\,\,\in D^\beta\,.
		\end{equation}
		We have that for any such position of the arc $a_0$, permutation lines of $Z^{\lambda}_{(a_0)}$ directly connect to $\overline{T}^{\,\prime\,\beta}_{n-4}$ which gives zero.\medskip
		
		Finally we conclude that 
		\begin{equation}
			T^{(\beta,2)}\cdot\mathcal{A}(P^{\lambda}_{n-2})=0\,.
		\end{equation}
	\end{itemize}
\end{mdframed}
	We now consider the general case. It is enough to restrict our attention to the element $T^{(\beta\,,\,f)}$ \eqref{eq:std_module_farcs}. As already mentioned above, if $f\neq f_\lambda$ one has $T^{(\beta\,,\,f)}\cdot\mathcal{A}(P^{\lambda}_{n-2})=0$, so it only remains to analyze the case $|\lambda|=|\beta|$, so that $f_\lambda=f$.\medskip
	
	Consider the action of $\mathcal{A}(P^{\lambda}_{n-2})$ on $T^{(\beta\,,\,f)}$ 
	
	\begin{equation}\label{eq:T_AP_gen}
		T^{(\beta,f)}\cdot \mathcal{A}(P^{\lambda}_{n-2})=\sum_{1\leqslant i<j \leqslant n}\,\,\raisebox{-.42\height}{\includegraphics[scale=0.7]{fig/Action_AP_tensor_standard_module.pdf}}
	\end{equation}
	\vskip 4pt
	One has $P^{\lambda}_{n-2}\in J_{f-1}$ and therefore $\mathcal{A}(P^{\lambda}_{n-2})\in J_f$. For any diagram $d$ entering $P^{\lambda}_{n-2}$  with at least $f$ arcs holds $T^{(\beta,f)}\cdot \mathfrak{a}_{ij}(d)=0$. Indeed, $\mathfrak{a}_{ij}(d)\in J_{f+1}$ and by construction $T^{(\beta,f)}$ is annihilated by $J_{f+1}$. Hence it will be enough to focus attention on $P^{\lambda}_{n-2}\Big{|}_{f-1}$, the part of $P^{\lambda}_{n-2}$ where diagrams have only $f-1$ arcs.\medskip
	
	There are three possible configurations for the position of the arc $(ij)$:
	\begin{itemize}
		\item[$1)$] The arc $(ij)$ is  directly connected to $\overline{T}^\beta_{n-2f}$:\medskip
		
		Since the tensor $\overline{T}^\beta_{n-2f}$ is traceless, $T^{(\beta,f)}\cdot \mathfrak{a}_{ij}(P^\lambda_{n-2})=0$. 
		
		\item[$2)$] The arc $(ij)$ is connected to the arcs of $T^{(\beta,f)}$:\medskip
		
		\textit{i)} First consider the case where the arc $(ij)$ is connected to an arc of $T^{(\beta,f)}$ as depicted below:
		\begin{equation}
			T^{(\beta,f)}\cdot \mathfrak{a}_{ij}(P^{\lambda}_{n-2})=\,\raisebox{-.44\height}{\includegraphics[scale=0.7]{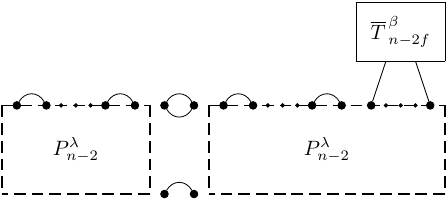}}
		\end{equation}
		Because, $g^{\otimes f-2}\otimes\overline{T}^\beta_{n-2f}\in D^\beta$ and $D^\beta\cdot P^{\lambda}_{n-2}=0$, one has 
		\begin{equation}
			T^{(\beta,f)}\cdot \mathfrak{a}_{ij}(P^{\lambda}_{n-2})=0\,.
		\end{equation}
		\textit{ii)} Next, consider the case where the arc $(ij)$ is connected to two arcs of $T^{(\beta,f)}$ as depicted below:
		\begin{equation}
			T^{(\beta,f)}\cdot \mathfrak{a}_{ij}(P^{\lambda}_{n-2})=\,\raisebox{-.44\height}{\includegraphics[scale=0.7]{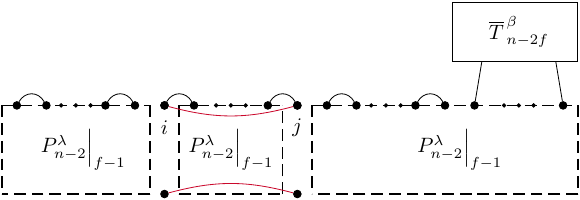}}
		\end{equation}
		As explained below \eqref{eq:T_AP_gen}, we focus on the element $P^{\lambda}_{n-2}\Big{|}_{f-1}$ as all other diagrams annihilate $T^{(\beta,f)}$.\medskip 
		
		Select a configuration of the $f-1$ upper arcs within $P^{\lambda}_{n-2}\Big{|}_{f-1}$ and denote their set by $\mathfrak{a}$. Then define $Z^{\lambda}_{\mathfrak{a}}$ as the part within $P^{\lambda}_{n-2}\Big{|}_{f-1}$ with upper arcs in $\mathfrak{a}$, and consider the following products:  
		
		\begin{equation}
			T^{(\beta,f)}\cdot \mathfrak{a}_{ij}(Z^{\lambda}_{\mathfrak{a}})=\,\raisebox{-.44\height}{\includegraphics[scale=0.7]{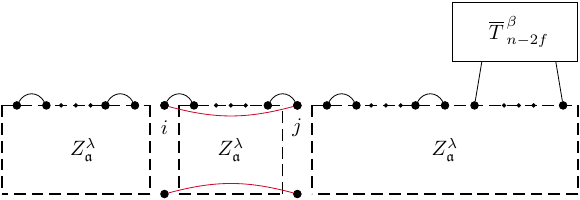}}
		\end{equation}
		Let us show that $T^{(\beta,f)}\cdot \mathfrak{a}_{ij}(Z^{\lambda}_{\mathfrak{a}})=0$ for any set $\mathfrak{a}$.
		Consider the following two types of upper arcs configurations: 
		\begin{itemize}
			\item[\textbf{(a)}] All $f-1$ upper arcs of $Z^{\lambda}_{\mathfrak{a}}$ are connected to the arcs of $T^{(\beta,f)}$.
			
			\item[\textbf{(b)}] At least one arc of $Z^{\lambda}_{\mathfrak{a}}$ connects a node of $\overline{T}^\beta_{n-2f}$ to an arc of $T^{(\beta,f)}$.
		\end{itemize}
		
		Configurations of type \textbf{(a)} are reminiscent of the configuration \eqref{eq:all_arc_connected}. Each part of $Z^{\lambda}_{\mathfrak{a}}$ with diagrams having a fixed configuration of the lower arcs, acts on $\overline{T}^{\,\beta}_{n-2f}$ by permutations resulting in a projection to the space of traceless tensors $D^{\lambda}$.	Hence we conclude that this configuration yields zero.\medskip
		
		Consider the configurations of type \textbf{(b)}. Denote by $a_0$ an arc in $\mathfrak{a}$ which connects a node of $\overline{T}^\beta_{n-2f}$ to an arc $a_1$ of $T^{(\beta,f)}$. 
		One has the following three possibilities:
		\begin{itemize}
			\item[$1.$] $a_1$ is connected to $\overline{T}^\beta_{n-2f}$ via two arcs.  This case results in an arc contraction of $\overline{T}^\beta_{n-2f}$ which yield zero. 
			\begin{equation}
				\raisebox{-.44\height}{\includegraphics[scale=0.7]{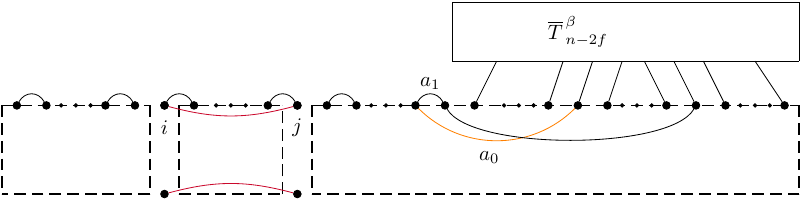}}
			\end{equation}
			\item[$2.$] $a_1$ is connected to a passing line of a diagram within $Z^{\lambda}_{\mathfrak{a}}$.
			\begin{equation}
				\raisebox{-.44\height}{\includegraphics[scale=0.7]{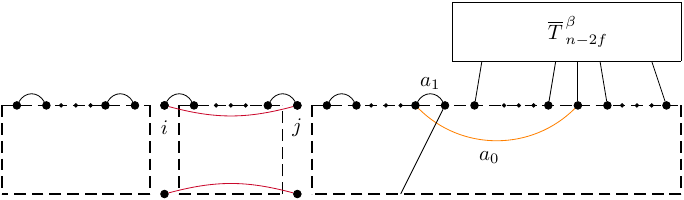}}
			\end{equation}
			\item[$3.$] $a_1$ is connected to both $\overline{T}^\beta_{n-2f}$ and to another arc of $T^{(\beta,f)}$.
			\begin{equation}
				\raisebox{-.44\height}{\includegraphics[scale=0.7]{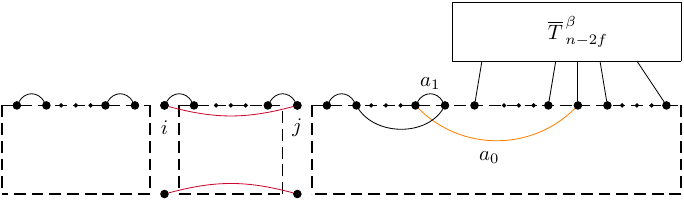}}
			\end{equation}
			This case initiates a \textit{chain of arcs} which either ends up with case $1$, or case $2.$ The former case yields zero while in the latter case one has for example
			\begin{equation}\label{eq:conf_gen_1}
				T^{(\beta,f)}\cdot \mathfrak{a}_{ij}(d)=\,\raisebox{-.44\height}{\includegraphics[scale=0.7]{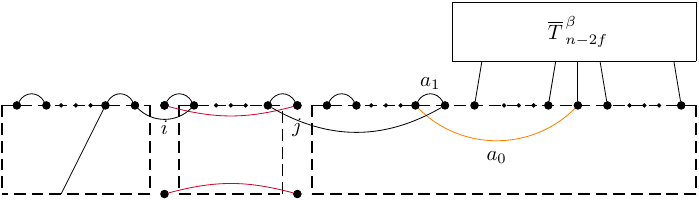}}
			\end{equation}
		\end{itemize}
		
		Any non zero configuration of type \textbf{(b)} will result in $n_c$ chains of arcs. Hence, there are $n-n_c-2f$ nodes of $\overline{T}^\beta_{n-2f}$ which are connected to $n-n_c-2f$ passing lines of the diagrams within $Z^\lambda_{\mathfrak{a}}$. These are exactly all passing lines of the diagrams within $Z^\lambda_{\mathfrak{a}}$ which are not connected to the endpoints of the chains. Indeed, there are $n-2$ lines (lower and upper arcs and passing lines) in each diagram of $Z^\lambda_{\mathfrak{a}}$. Within these $n-2$ lines there are $2(f-1)$ arcs, $n_c$ lines attached to chains of arcs, and hence $n-n_c-2f$ passing lines connected to $\overline{T}^\beta_{n-2f}$.
		As a result, passing lines of diagrams within $Z^{\lambda}_{\mathfrak{a}}$ are directly connected to a traceless tensor $\overline{T}^{\,\prime\,\beta}_{n-2f}\in D^\beta$ where
		\begin{equation}
			\overline{T}^{\,\prime\,\beta}_{n-2f}= \overline{T}^{\,\beta}_{n-2f}\cdot s \,, \hspace{1cm}\text{for some $\,s\in \sn$}\,,
		\end{equation}
		by analogy with \eqref{eq:unfolded_diagram_3tl0} and \eqref{eq:1chain_arcs}.
		Each part of $Z^{\lambda}_{\mathfrak{a}}$ with diagrams having a fixed configuration of the lower arcs, acts on $\overline{T}^{\,\prime\,\beta}_{n-2f}$ by permutations resulting in a projection to the space of traceless tensors $D^{\lambda}$. Hence for upper arc configuration $\mathfrak{a}$,
		\begin{equation}
			T^{(\beta,f)}\cdot \mathfrak{a}_{ij}(Z^{\lambda}_{\mathfrak{a}})=0,
		\end{equation}
		and therefore we conclude that
		\begin{equation}
			T^{(\beta,f)}\cdot \mathfrak{a}_{ij}(P^{\lambda}_{n-2})=0.
		\end{equation}

		\item[$3)$] One node of the arc $(ij)$ is connected to a node within the arc of $T^{(\beta,f)}$ and the other one is connected to $\overline{T}^\beta_{n-2f}$:
		
		\begin{equation}
			T^{(\beta,f)}\cdot \mathfrak{a}_{ij}(P^{\lambda}_{n-2})=\,\raisebox{-.44\height}{\includegraphics[scale=0.7]{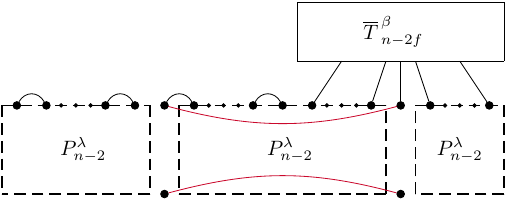}}
		\end{equation}
		
		These configurations can be analyzed similarly to the case $\textbf{(b)}$. The only difference being that the maximum number of chains of arcs is $f$ here, compared to $f-1$ for type $\textbf{(b)}$ configurations. See for example \eqref{eq:2chains_arcs} where the number of chain of arc is equal to $f=2$. Hence we can directly assert that 
		\begin{equation}
			T^{(\beta,f)}\cdot \mathfrak{a}_{ij}(P^{\lambda}_{n-2})=0\,,
		\end{equation}
		which concludes the proof.\medskip
	\end{itemize}
	%
	%
\end{proof}

\section{Proofs for chapter \ref{chap:cn}}\label{app:proof_chap4}

\subsection{Proof of Lemma \ref{lem:trace_rules}}\label{app:proof_lem_rules}

\begin{proof}
	The diversity of expressions on the left-hand-sides of the rules \eqref{eq:trace_rulea_1}-\eqref{eq:trace_rulea_4} which are necessary to fix $\tau_{a}$ (and $\tau_{t}$), is essentially limited by $\mathfrak{b}(\mathcal{A})$-linearity. Representatives of a particular form on the left-hand-sides of the rules \eqref{eq:trace_rulea_1}, \eqref{eq:trace_rulea_2} (similarly \eqref{eq:trace_rulet_1}, \eqref{eq:trace_rulet_2}) are always accessible via cyclic permutations. Taking into account the inversion, the same conclusion is valid for the rules in \eqref{eq:trace_rulea_3}, \eqref{eq:trace_rulea_4} (similarly \eqref{eq:trace_rulet_3}, \eqref{eq:trace_rulet_4})  due to the following simple facts.\medskip
	
	\textbf{Fact 1.} For $[\db{p}u] = [\db{p}I(u)]$ either $u$ or $I(u)$ is fit because $[\bb{p}u]\in \mathfrak{b}(\mathcal{A})$.\medskip
	
	\textbf{Fact 2.} For $[\db{p}u\db{s}v] = [\db{p}I(v)\db{s}I(u)]$ either $u$ or $I(v)$ is fit. Indeed, since $[\bb{p}u\bb{s}v]\in \mathfrak{b}(\mathcal{A})$, either one of the bracelets $[u]$, $[v]$ is in $\mathfrak{b}(\mathcal{A})$ or one of them is empty.\medskip
	
	\textbf{Fact 3.} For $[\db{p}u\db{s}v] = [\db{p}I(v)\db{s}I(u)]$ either $[u],[v]\notin \mathfrak{b}(\mathcal{A})$ or $[u],[v]\in \mathfrak{b}(\mathcal{A})$. In the former case either $|u|_{\bb{n}} > |u|_{\bb{s}}$ or $|v|_{\bb{n}} > |v|_{\bb{s}}$, while in the latter $u,v$ or $I(u), I(v)$ are both fit.\medskip
	
	\noindent Hence we conclude that representatives on the left-hand-sides of \eqref{eq:trace_rulea_1}-\eqref{eq:trace_rulea_4} (resp. \eqref{eq:trace_rulet_1}-\eqref{eq:trace_rulet_4}) are sufficient to define $\tau_{a}$ (resp. $\tau_{t}$).\medskip
	
	\begin{itemize}
		\item[{\it 1)}] For the rules \eqref{eq:trace_rulea_1}, the only alternative representative is $[\dd{s} I(u)]$. Application of \eqref{eq:trace_rulea_1} leads to $\tau_{a}\big([\dd{s} I(u)]\big) = 2\Dim\,[\bb{s} I(u)]$. But the latter equals to $\tau_{a}\big([\dd{s} u]\big)$ by inversion of the representative combined with a cyclic permutation. Hence, \eqref{eq:trace_rulea_1} and \eqref{eq:trace_rulet_1} are defined unambiguously.
		
		\item[{\it 2)}] For the first rule in \eqref{eq:trace_rulea_2}, the equivalent representatives are $[\db{s}v\db{s}u] = [\db{s}I(u)\db{s}I(v)] = [\db{s}I(v)\db{s}I(u)]$. By direct application of \eqref{eq:trace_rulea_2} one calculates $\tau_{a}\big([\db{s}v\db{s}u]\big) = 2\big([\db{s}v\db{s}I(u)] + [\db{s}v][\db{s}u]\big)$. But $[\db{s}v\db{s}I(u)] = [\db{s}u\db{s}I(v)]$ (the two representatives are related by inversion followed by a cyclic permutation), which reproduces $\tau_{a}\big([\db{s}u\db{s}v]\big)$. Other alternatives are analysed along the same lines.
		
		For the second rule in \eqref{eq:trace_rulea_2}, one has to analyse $[\db{s}v][\db{s}u]$, together with the alternatives $[\db{s}I(u)]$ and $[\db{s}I(v)]$ for each factor. By direct application of \eqref{eq:trace_rulea_2} one gets $\tau_{a}\big([\db{s}v][\db{s}u]\big) = 2\big([\db{s}v\db{s}u] + [\db{s}v\db{s}I(u)]\big)$. But $[\db{s}v\db{s}u] = [\db{s}u\db{s}v]$ (by a cyclic permutation) and $[\db{s}v\db{s}I(u)] = [\db{s}u\db{s}I(v)]$ (by composition of inversion and a cyclic permutation), which leads to $\tau_{a}\big([\db{s}u][\db{s}v]\big)$. Consider also $\tau_{a}\big([\db{s}u][\db{s}I(v)]\big) = 2\big([\db{s}u\db{s}I(v)] + [\db{s}u\db{s}v]\big)$ which equals $\tau_{a}\big([\db{s}u][\db{s}v]\big)$ directly. Other alternatives are analysed along the same lines. Hence, \eqref{eq:trace_rulea_2} and \eqref{eq:trace_rulet_2} are defined unambiguously.
		\item[{\it 3)}]Note that the admissible representatives on the left-hand-sides of the rules \eqref{eq:trace_rulea_3} are fixed unambiguously.\medskip
		\item[{\it 4)}] For the first rule in \eqref{eq:trace_rulea_4}, let $u,v$ are fit. Then for the alternative representative one has $\tau_{a}\big([\db{p}v\db{p}u]\big) = [\bb{n}v\bb{s}I(u)] = [\bb{n}u\bb{s}I(v)]$. Else, let $|u|_{\bb{s}} > |u|_{\bb{n}}$. Then one checks $\tau_{a}\big([\db{p}I(u)\db{p}I(v)]\big) = [\bb{n}I(u)][\bb{s}I(v)] = [\bb{n}u][\bb{s}v]$.
		
		For the second rule in \eqref{eq:trace_rulea_4} one considers $\tau_{a}\big([\db{p}v][\db{p}u]\big) = [\bb{n}v\bb{s}I(u)] = [\bb{n}u\bb{s}I(v)]$.
	\end{itemize}
\end{proof}

\subsection{Proof of Theorem \ref{thm:Laplace}}\label{app:proof_Laplace}
\begin{proof}
	We restrict our attention on the operator $\Delta_a$, as everything we present here can be directly transposed to prove that for $\Delta_t = \frac{1}{2}\,\tau_t\circ \partial^{2}$. Recall that from the definition of the average conjugacy class sum \eqref{eq:averaged_class_sum}, one has for any $b\in \dbn$ such that $\GCT(b)=\zeta$
	\begin{equation}\label{eq:An_averaged_class_sum}
		A_n\,\bar{K}_\zeta=\sum_{s \in \sn } s \, A_n b \, s^{\shortminus 1}\,.
	\end{equation}
	The product $A_n\, b = \sum_{1\leqslant i < j\leqslant n} d_{ij}\, b$ consists in summing over all possibilities of placing an arc at a pair of lower nodes of $b$, which is mimicked by imposing Leibniz rule in the definition of $\partial$. Without loss of generality, we assume a convenient representative $b$ in the conjugacy class, such that all bracelets in $\GCT(b)$ come from mutually non-intersecting cycles, and such that each independent cycle is either a cycle permutation, or, if $b$ has at least one arc, there are no intersections among the vertical lines, while each arc in the upper or lower row is incident to a pair of nodes $(i,i+1)$ (for $i \leqslant n-1$) or $(1,n)$ (for example, see \eqref{eq:convenient_representative}). The independent cycles will be schematically illustrated by rectangular blocks, with only particular significant lines specified explicitly. For example, 
	\begin{equation}
		\raisebox{-.4\height}{\includegraphics[scale=0.4]{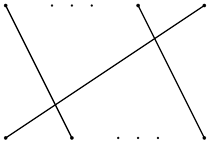}}\rightarrow \raisebox{-.4\height}{\includegraphics[scale=0.4]{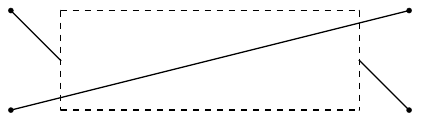}} \;, \hspace{2cm} \raisebox{-.4\height}{\includegraphics[scale=0.4]{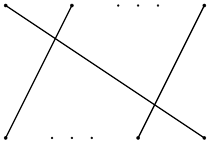}}\rightarrow \raisebox{-.4\height}{\includegraphics[scale=0.4]{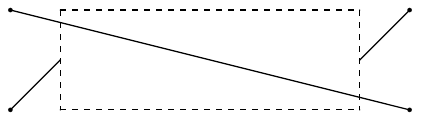}}\;,
	\end{equation}
	\begin{equation}
		\raisebox{-.4\height}{\includegraphics[scale=0.4]{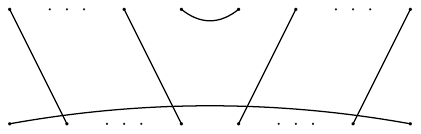}}\rightarrow \raisebox{-.4\height}{\includegraphics[scale=0.4]{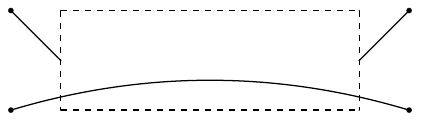}}\;, \hspace{1cm}
		\raisebox{-.4\height}{\includegraphics[scale=0.4]{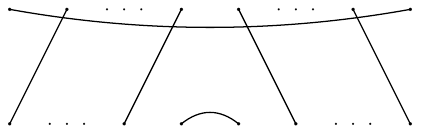}}\rightarrow \raisebox{-.4\height}{\includegraphics[scale=0.4]{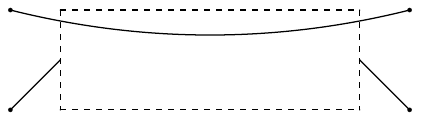}}
	\end{equation}
	There are two major cases how one can place the arc upon multiplication by $d_{ij}$: either it  joins two nodes within a single cycle, or the two nodes belong to two different cycles. Only the bracelets coming from these blocks will be affected by multiplication by $d_{ij}$, so all other blocks can be ignored while analysing each particular $i<j$. This property is manifested in the definition of $\tau$ as a $\mathbb{C}[\mathfrak{b}(\mathcal{A})]$-linear operation. We will work in terms of particular representatives in the bracelets, and read off the resulting representatives coming from the initial ones. To fix a representative, we will highlight the starting point by a circle around a node, such that one starts reading along the adjacent line (if there is a diagram attached below, one ignores it). We always imply the word $u$ read off by following the lines hidden behind the first block, while for the second block (if any) the word is $v$. Additional arrows entering each block fix the direction of reading which leads to the corresponding word. Passing the block in the opposite direction leads to $I(u)$ instead of $u$ and $I(v)$ instead of $v$. For example, from the following schema one reads off the representative $[\bb{n}u\bb{s}I(v)]$:
	\begin{equation}
		\raisebox{-.4\height}{\includegraphics[scale=0.7]{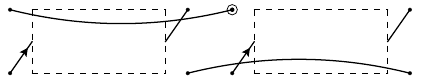}}
	\end{equation}
	
	Each letter in a word parametrising a bracelet corresponds to a line with two endpoints. Upon multiplication $d_{ij}\, b$, put $0$, $1$ or $2$ dots above a letter in each representative depending on how many endpoints are occupied by the $(i,j)$-arc attached to the lower row of $b$. It is clear that $\bb{n}$ can never acquire a dot, $\bb{p}$ can acquire at most $1$ dot, while $\bb{s}$ can acquire up to $2$ dots. This is exactly encoded in the definition of $\partial$ by putting $\partial(\bb{n}) = \partial(\db{p}) = \partial(\dd{s}) = 0$. To read off the rules for $\tau$ such that the formula for $\Delta_a$ in \eqref{eq:Laplace} holds, it is sufficient to consider the following cases for the possibilities of attaching
	the arc.
	\begin{itemize}
		\item[{\it 1)}] The arc is attached to another arc, which leads to a cycle:
		\begin{equation}
			\begin{aligned}
				\dfrac{1}{2}\tau\big(\left[\dd{s}u\right]\big) \;\;=\;\; \hspace{0.1cm}\raisebox{-.4\height}{\includegraphics[scale=0.4]{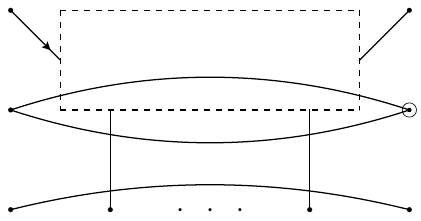}}\;\;=\;\;\Dim\;  \raisebox{-.4\height}{\includegraphics[scale=0.4]{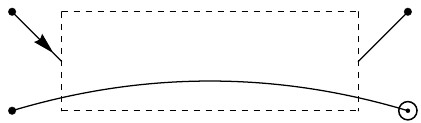}} \;\;=\;\; \Dim\, [\bb{s}u]\,,
			\end{aligned}    
		\end{equation}
		so one arrives at the rule \eqref{eq:trace_rulea_1}.
		
		\item[{\it 2)}] The arc is attached to one of the endpoints of two particular arcs of $b$. One sums over the four possibilities in this case:
		\begin{equation}
			\def\arraystretch{2}
			\begin{array}{ll}
				\tau\big(\left[\db{s}u\db{s}v\right]\big) & =\;\;\hspace{0.1cm}\raisebox{-.4\height}{\includegraphics[scale=0.6]{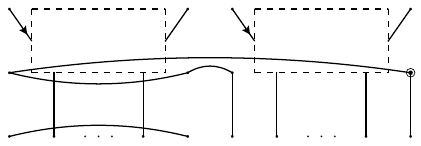}}\;+\;\raisebox{-.4\height}{\includegraphics[scale=0.6]{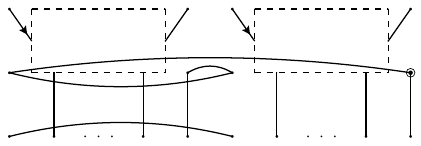}}\;+ \\
				\hfill &\,\hspace{0.6cm}\raisebox{-.4\height}{\includegraphics[scale=0.6]{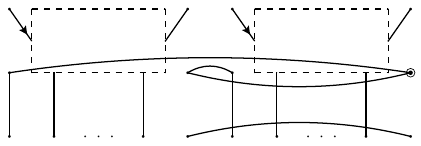}}\;+\;\raisebox{-.4\height}{\includegraphics[scale=0.6]{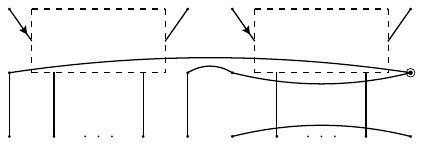}}\\
				\multicolumn{2}{l}{ = \;2\,\raisebox{-.4\height}{\includegraphics[scale=0.6]{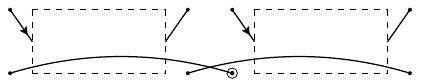}}\; + \;2\,\raisebox{-.4\height}{\includegraphics[scale=0.6]{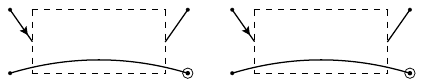}}\;=\;2\,[\bb{s}u\bb{s}I(v)] + 2\,[\bb{s}u][\bb{s}v]\,,}
			\end{array}    
		\end{equation}
		so one reproduces the first rule in \eqref{eq:trace_rulea_2}.
		
		\item[{\it 3)}] Next, let us consider the case when one endpoint of the arc is attached to a passing line, while the other one occupies one of the endpoints of a particular arc in $b$. Then one sums over the two possibilities, where the structure of the representatives assumes that $u$ is fit
		\begin{equation}
			\begin{array}{ll}
				\tau\big(\left[\db{p}u\db{s}v\right]\big) & =\hspace{0.1cm} \raisebox{-.4\height}{\includegraphics[scale=0.6]{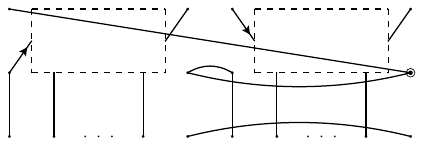}}  \;+\; \raisebox{-.4\height}{\includegraphics[scale=0.6]{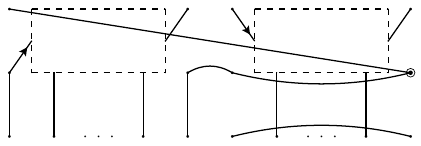}}\\
				\multicolumn{2}{l}{= \raisebox{-.4\height}{\includegraphics[scale=0.6]{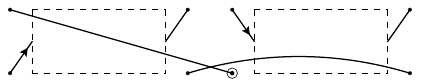}}\;+\;\raisebox{-.4\height}{\includegraphics[scale=0.6]{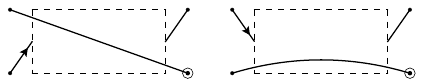}} = [\bb{p}u\bb{s}I(v)]\;+\;[\bb{p}u][\bb{s}v]\,,}
			\end{array}    
		\end{equation}
		so one recovers the first rule in \eqref{eq:trace_rulea_3}.
		\item[{\it 4)}] Finally, the arc can occupy the lower endpoints of two passing lines. First, suppose that $u$ is fit (so is $v$), which leads to
		\begin{equation}
			\begin{aligned}
				\tau\big(\left[\db{p}u\db{p}v\right]\big) & =\hspace{0.1cm}\raisebox{-.4\height}{\includegraphics[scale=0.6]{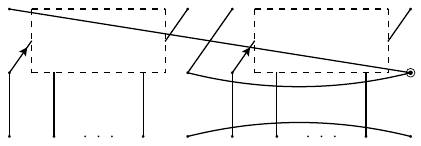}}\;=\;\hspace{0.1cm}\,\raisebox{-.4\height}{\includegraphics[scale=0.6]{fig/theo3_f_r1fit.pdf}}\;=\; [\bb{n}u\bb{s}I(v)]\,,
			\end{aligned}    
		\end{equation}
		in agreement with the first rule in \eqref{eq:trace_rulea_4} for the case of a fit representative. In the other case, when neither $u$ nor $v$ is not fit, with $|u|_{\bb{s}} > |u|_{\bb{n}}$, one has
		\begin{equation}
			\begin{aligned}
				\tau\big(\left[\db{p}u\db{p}v\right]\big) & =\;\hspace{0.1cm}\raisebox{-.4\height}{\includegraphics[scale=0.6]{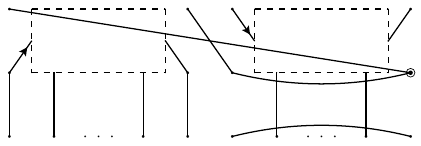}}\;=\;\hspace{0.1cm}\,\raisebox{-.4\height}{\includegraphics[scale=0.6]{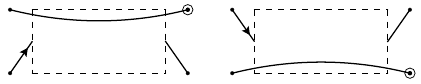}}\;=\; [\bb{n}u][\bb{s}v]\,,
			\end{aligned}    
		\end{equation}
		which is again in agreement with the first rule in \eqref{eq:trace_rulea_4}. 
	\end{itemize}
	We are left with the cases when the arc connects two independent cycles.
	\begin{itemize}
		\item[{\it 5)}] If the arc occupies the endpoints of two particular arcs in $b$, one sums over four possibilities
		\begin{equation}
			\def\arraystretch{2}
			\begin{array}{ll}
				\tau\big(\left[\db{s}u\right]\left[\db{s}v\right]\big) = & \;\raisebox{-.4\height}{\includegraphics[scale=0.6]{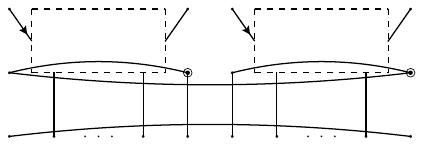}}\;+\;\raisebox{-.4\height}{\includegraphics[scale=0.6]{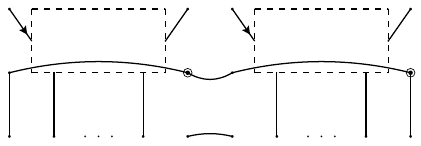}}\\
				\hfill & + \;\raisebox{-.4\height}{\includegraphics[scale=0.6]{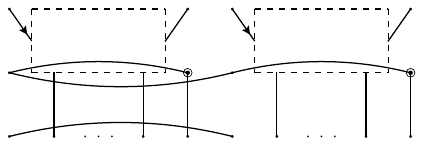}}\;+\;\raisebox{-.4\height}{\includegraphics[scale=0.6]{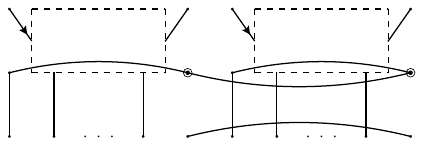}} \\
				\multicolumn{2}{l}{=\hspace{0.1cm}\,2\,\raisebox{-.4\height}{\includegraphics[scale=0.6]{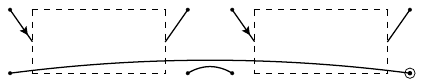}}\,+\,2\,\raisebox{-.4\height}{\includegraphics[scale=0.6]{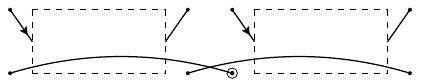}}\;=\; 2\, [\bb{s}u\bb{s}v]\;+\;2\, [\bb{s}u\bb{s}I(v)]\,,}
			\end{array}
		\end{equation}
		which reproduces the second rule in \eqref{eq:trace_rulea_2}.
		\item[{\it 6)}] If the arc occupies the passing line in one cycle and an endpoint of an arc in the other cycle of $b$, one sums over two possibilities ($u$ is fit)
		\begin{equation}
			\def\arraystretch{2}
			\begin{array}{ll}
				\tau\big(\left[\db{p}u\right]\left[\db{s}v\right]\big) = & \raisebox{-.4\height}{\includegraphics[scale=0.6]{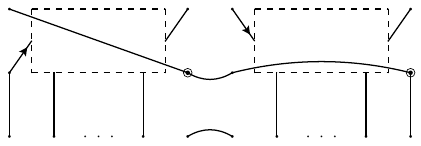}}\;+\;\raisebox{-.4\height}{\includegraphics[scale=0.6]{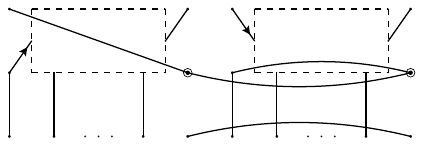}} \\
				\multicolumn{2}{l}{=\hspace{0.1cm}\,\raisebox{-.4\height}{\includegraphics[scale=0.6]{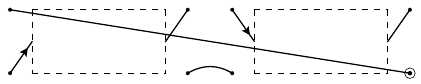}}\,+\,\raisebox{-.4\height}{\includegraphics[scale=0.6]{fig/theo3_e_r2.pdf}}\; = \; [\bb{p}u\bb{s}v] \;+\; [\bb{p}u\bb{s}I(v)]\,,}
			\end{array}    
		\end{equation}
		in agreement with the second rule in \eqref{eq:trace_rulea_3}.
		\item[{\it 7)}] Finally, if the arc occupies two passing lines in two independent cycles of $b$ (with both $u,v$ fit) one has
		\begin{equation}
			\begin{aligned}
				\tau\big(\left[\db{p}u\right]\left[\db{p}v\right]\big) & = \raisebox{-.4\height}{\includegraphics[scale=0.6]{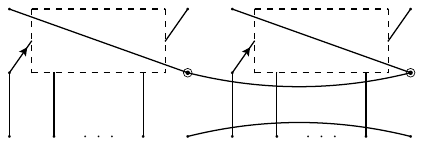}}\;=\;\hspace{0.1cm}\,\raisebox{-.4\height}{\includegraphics[scale=0.6]{fig/theo3_f_r1fit.pdf}}\;=\; [\bb{n}u\bb{s}I(v)]\,,
			\end{aligned}    
		\end{equation}
		which coincides with the second rule in \eqref{eq:trace_rulea_4}.
	\end{itemize}
	Due to Lemma \ref{lem:trace_rules}, the considered cases are sufficient to prove the assertion.
\end{proof}

\subsection{Proof of Lemma \ref{lem:sym}}\label{app:Lemma_sym}

\begin{proof}
	Let us first prove the assertion for $b\in \dbn$ such that $\zeta = \GCT(b)$ is a single bracelet. Let us show that upon a convenient choice of $b$ among the conjugate diagrams, any $t \in \Stab_{\Sn{n}}(b)$ can be written as $t = c^{m} r^{p}$ for some $m = 0,\dots, n-1$ and $p = 0,1$, where
	\begin{equation}
		c\; =\; \raisebox{-.4\height}{\includegraphics[scale=0.6]{fig/csym.pdf}}\;, \hspace{3cm} r\;= \;\raisebox{-.4\height}{\includegraphics[scale=0.6]{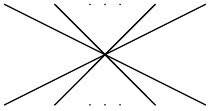}}
	\end{equation}
	(the fact that $c = r$ for $n = 2$ does not lead to any problem). Let us say that two lines in a diagram are {\it adjacent} if there exists a pair of vertically aligned nodes which are endpoints of these lines, {\it i.e.} these endpoints are identified by the step 2 in the definition of the map $\GCT$ in Section \ref{sec:classes-bracelets_map}. For any $s\in \Sn{n}$, conjugation $b\to sb s^{-1}$ preserves the adjacency relation, which is exactly the property encoded by the bracelet $\GCT(b)$ (by placing letters along an unoriented circle, one exactly defines the nearest neighbours). If a conjugation transformation preserves the diagram (with $t\in \Stab_{\Sn{n}}(b)$), then the  result of such transformation can be described as follows: each passing line takes the place of another passing line and each upper (respectively, lower) arc takes the place of another upper (respectively, lower) arc. Without loss of generality, take the convenient diagram in the conjugacy class: $b = c$ if $b$ is a permutation, else take
	\begin{equation}\label{eq:convenient_representative}
		b=\hspace{0.1cm}\raisebox{-.4\height}{\includegraphics[scale=0.7]{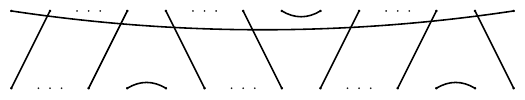}}\;.
	\end{equation}
	In view of preservation of adjacency of lines, the particularly chosen diagram $b$ can by preserved by adjoint transformations composed only by cyclic permutations of nodes and inversion. 
	\vskip 5 pt
	With this at hand, let us establish the bijection between the centraliser $\Stab_{\Sn{n}}(b)$ and the turnover stabiliser $S(w_b)$ for a particular representative $w_b\in \zeta$. Namely, to read off the word $w_b$ (as described in steps 1, 2 in the definition of the map $\GCT$), start at the upper right node and follow the adjacent edge. With a slight abuse of notation, define the action of permutations $c$ and $r$ on words of the length $n$: $c(\bb{a}_1\dots\bb{a}_n) = \bb{a}_{n} \bb{a}_{1}\dots \bb{a}_{n-1}$ and $r(\bb{a}_1\dots\bb{a}_n) = cI(\bb{a}_1\dots\bb{a}_n) = \bb{a}_{1}\bb{a}_{n}\dots\bb{a}_{2}$. Then it is straightforward that $c^{m}r^{p}\in \Stab_{\Sn{n}}(b)$ iff $c^{m}r^{p}\in S(w_b)$, which proves the assertion.
\end{proof} 
	\chapter{Examples of full trace projectors}\label{app:Traceprojector}

\section{The full trace projector $P^{(\emptyset)}_{12}$}\label{sec:P12}
	\begin{equation*}
	P^{(\emptyset)}_{12}=\raisebox{-0.955\height}{\includegraphics[scale=1.2]{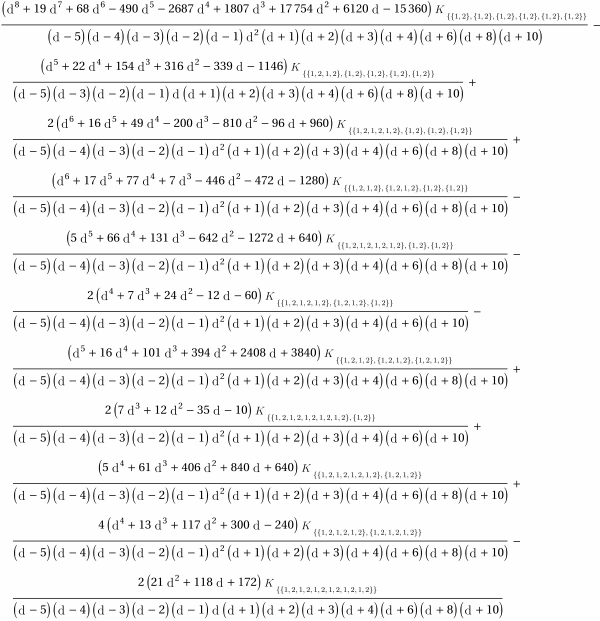}}\,
\end{equation*}

\section{The full trace projector $P^{(\emptyset)}_{14}$}\label{sec:P14}

\begin{equation*}
	P^{\emptyset}_{14}=\raisebox{-0.965\height}{\includegraphics[scale=1.5]{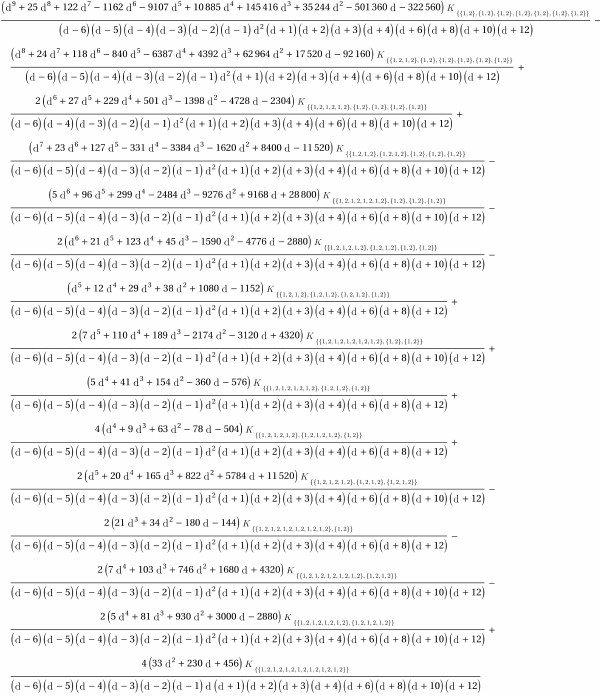}}\,
\end{equation*}

\section{The full trace projector $P^{(\emptyset)}_{16}$}\label{sec:P16}

\begin{equation*}
	P^{\emptyset}_{16}=\raisebox{-0.975\height}{\includegraphics[scale=1.6]{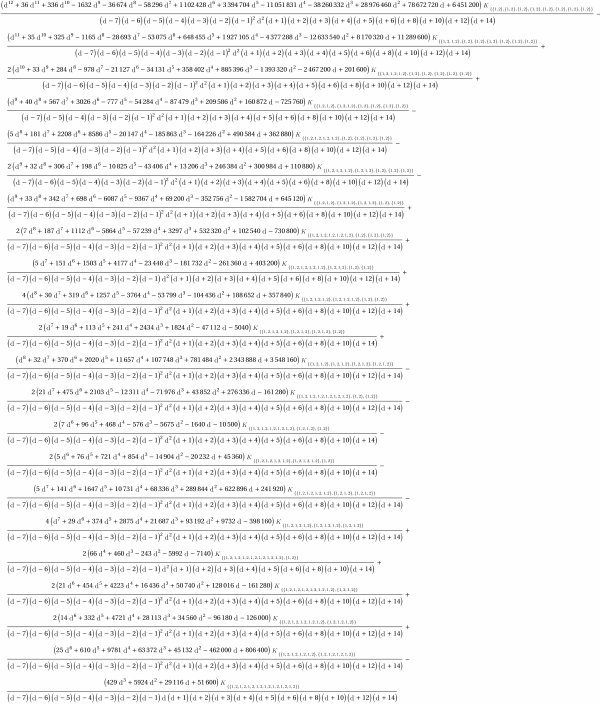}}\,
\end{equation*}

%
	\bibliography{biblio_thesis.bib}
	
\newpage
\thispagestyle{empty}
\newgeometry{hmargin=0.3cm,vmargin=0.3cm}
\hspace{-0.7cm}

\includegraphics[width=0.22\textwidth]{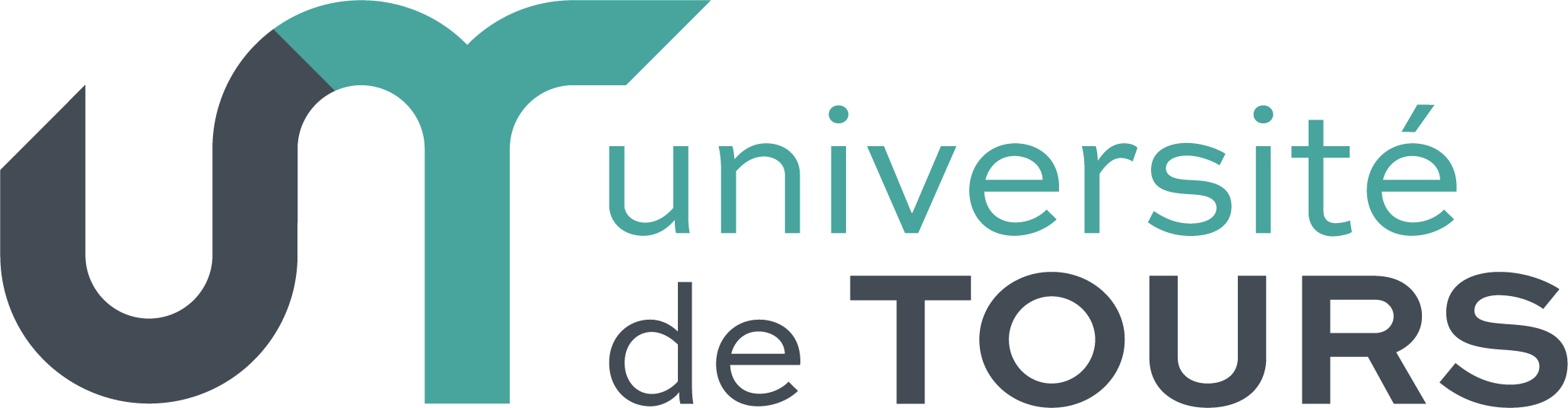}
\hfill
\includegraphics[width=0.13\textwidth]{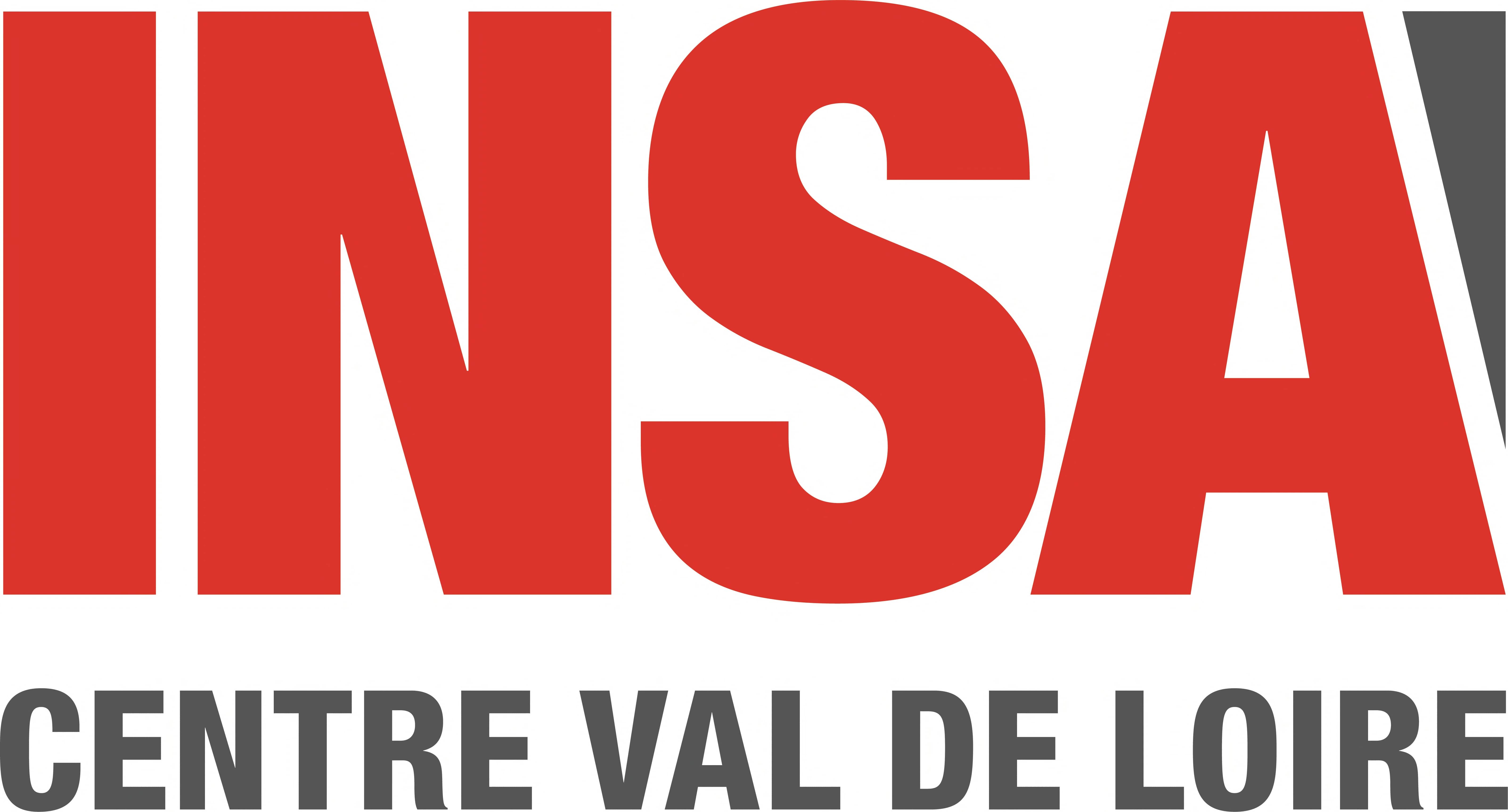}
\hfill
\includegraphics[width=0.15\textwidth]{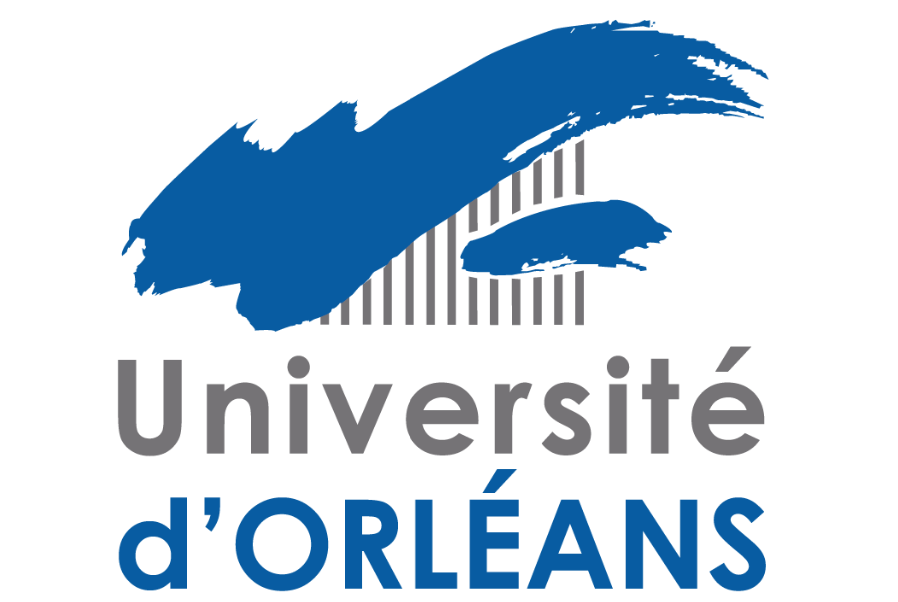}\\
\vspace{-0.4cm}

\begin{sffamily}
	
\begin{center}
	\LARGE{Thomas HELPIN}\\ 
	\textbf{Décompositions irréductibles des tenseurs via l'algèbre de Brauer et applications à la gravitation métrique-affine.}
\end{center}

\vspace{0.1cm}
{
	\fontsize{10pt}{10pt}\selectfont
	
	\fbox{
		\begin{minipage}{19.5cm}
			{\setlength{\parindent}{6pt}
				\noindent\textbf{Résumé} :\smallskip
				
		\noindent Dans la première partie de cette thèse, on utilise la théorie des représentations des groupes et des algèbres afin d'obtenir une décomposition irréductible des tenseurs dans le contexte de la gravitation métrique-affine. En particulier, on considère l'action du groupe orthogonal $\Or(1, \Dim-1)$ sur le tenseur de Riemann associé à une connexion affine, avec torsion et non-métricité, définit sur une variété pseudo-Riemannienne. Cette connexion est l'ingrédient caractéristique de la gravitation métrique-affine.
		
		La décomposition irréductible du tenseur de Riemann effectuée dans cette thèse est conçue pour l'étude des théories invariantes projectives de la gravitation métrique-affine. Les propriétés suivantes impliquent l'unicité de la décomposition. Premièrement, elle s'effectue en deux étapes: on considère l'action de $\GL(\Dim,\mathbb{R})$ puis l'action de $\Or(1,\Dim-1)$. Deuxièmement, le nombre de tenseurs invariants projectifs dans la décomposition est maximal. Troisièmement, la décomposition est orthogonale par rapport au produit scalaire canonique induit par la métrique. La même procédure est appliquée à la distorsion de la connexion affine. Enfin, à partir de ces décompositions, nous obtenons les Lagrangiens quadratiques généraux en la distortion et en la courbure de Riemann.
		
		\medskip
		\noindent Dans la deuxième partie de cette thèse, nous construisons les opérateurs de projection utilisés pour obtenir les décompositions mentionnées précédemment. Ces opérateurs sont réalisés en termes de l'algèbre du groupe symétrique $\C\sn$ et de l'algèbre de Brauer $\bn(\Dim)$, qui sont respectivement liées à l'action de $\GL(\Dim,\C)$ (et sa forme réelle $\GL(\Dim,\mathbb{R}$)) et à l'action de $\Or(\Dim,\C)$ (et sa forme réelle $\Or(1,\Dim-1)$) sur les tenseurs via la dualité de Schur-Weyl.
		
		Tout d'abord, nous proposons une approche alternative aux formules connues pour les idempotents centraux de $\C\sn$. Ces éléments réalisent une décomposition réductible unique, connue sous le nom de décomposition isotypique. Cette décomposition s'avère remarquablement pratique pour aboutir à la décomposition irréductible par rapport à $\GL(\Dim,\mathbb{R})$ recherchée.
		
		Ensuite, nous construisons les éléments de $\bn(\Dim)$ qui réalisent la décomposition isotypique d'un tenseur par rapport à l'action de $\Or(\Dim,\C)$. Cette décomposition est irréductible sous $\Or(\Dim,\C)$ lorsqu'elle est appliquée à un tenseur $\GL(\Dim,\C)$ irréductible d'ordre 5 ou moins. En conséquence directe de la construction, nous proposons une solution au problème de décomposition d'un tenseur arbitraire en sa partie sans trace, doublement sans trace, et ainsi de suite.
		
		Enfin, dans le dernier chapitre, nous présentons une technique qui optimise l'utilisation des opérateurs de projection par des systèmes de calcul formel.  Ces résultats ont conduit au développement de plusieurs packages Mathematica liés au bundle \textit{xAct} pour le calcul tensoriel en théorie des champs. Des fonctions particulières sont présentées tout au long du manuscrit. \smallskip
		
		\noindent\textbf{Mots-clés} : Gravitation métrique-affine, tenseur de Riemann, invariance projective, décomposition irreductible, dualités de Schur-Weyl, algèbre de Brauer.
			}			
		\end{minipage}
	}
	\vspace{0.5cm}
	
	\fbox{
	\begin{minipage}{19.5cm}
		{\setlength{\parindent}{6pt}	
			\noindent\textbf{Abstract}:\smallskip
			
			\noindent In the first part of this thesis, we make use of representation theory of groups and algebras to perform an irreducible decomposition of tensors in the context of metric-affine gravity. In particular, we consider the action of the orthogonal group $\Or(1,\Dim-1)$ on the Riemann tensor associated with an affine connection defined on a $\Dim$-dimensional pseudo-Riemannian manifold. This connection, with torsion and non-metricity, is the characteristic ingredient of metric-affine theories of gravity.  \smallskip
			
			The irreducible decomposition of the Riemann tensor carried out in this thesis is devised for the study of projective invariant theories of metric-affine gravity. The following properties imply the uniqueness of the decomposition. Firstly, it is performed in two steps: we consider the action of $\GL(\Dim,\R)$ and then the action of $\Or(1,\Dim-1)$. Secondly, the number of projective invariant irreducible tensors in the decomposition is maximal. Thirdly, the decomposition is orthogonal with respect to the canonical scalar product induced by the metric. The same procedure is applied to the distortion of the affine connection. Finally, from these decompositions, we derive the general quadratic Lagrangians in the distortion and in the Riemann tensor.
			
			\medskip
			\noindent In the second part of this thesis, we construct the projection operators used for the aforementioned decomposition. They are realized in terms of the symmetric group algebra $\C\sn$ and of the Brauer algebra $\bn(\Dim)$ which are related respectively to the action of $\GL(\Dim,\C)$ (and its real form $\GL(\Dim,\mathbb{R})$) and to the action of $\Or(\Dim,\C)$ (and its real form $\Or(1,\Dim-1)$) on tensors via the Schur-Weyl duality.
			
			\smallskip
			First of all, we give an alternative approach to the known formulas for the central idempotents of $\C\sn$. These elements provide a unique reducible decomposition, known as the isotypic decomposition. For our purposes, this decomposition is remarkably handy to arrive at the sought after irreducible decomposition with respect to $\GL(\Dim,\mathbb{R})$.
			
			\smallskip		
			Then, we construct the elements in $\bn(\Dim)$ which realize the isotypic decomposition of a tensor under the action of $\Or(\Dim,\C)$. This decomposition is irreducible under $\Or(\Dim,\C)$ when applied to an irreducible $\GL(\Dim,\C)$ tensor of order $5$ or less. As a by product of the construction, we give a solution to the problem of decomposing an arbitrary tensor into its traceless part, doubly traceless part and so on.
			
			\smallskip
			Finally, in the last chapter  we present a technique which optimizes the use of the projection operators by computer algebra systems. These results led to the development of several Mathematica packages linked to the \textit{xAct} bundle for tensor calculus in field theory. Particular functions are presented along the manuscript. \smallskip 
			
			\noindent\textbf{Keywords:} Metric-affine gravity, Riemann tensor, projective invariance, irreducible decomposition, Schur-Weyl dualities, Brauer algebra.
		}
	\end{minipage}
	}
	\vfill
}
	
\end{sffamily}

\end{document}